\documentclass[12pt]{article}

\usepackage[margin=1in]{geometry}
\usepackage[longnamesfirst,round]{natbib}
\usepackage{amsmath}
\usepackage{bbm}
\usepackage{mathtools}
\usepackage{microtype}
\usepackage{amssymb}
\usepackage{float}
\usepackage{amsthm}
\usepackage{subfig}
\usepackage[hang,flushmargin]{footmisc}
\usepackage{enumitem}
\usepackage[dvipsnames]{xcolor}
\usepackage[colorlinks,
  linkcolor=blue,
  citecolor=blue,
  urlcolor=blue,
  breaklinks
]{hyperref}
\usepackage[norefs,nocites]{refcheck}

\newlist{inlineroman}{enumerate*}{1}
\setlist[inlineroman]{afterlabel=~,label=(\roman*)}

\newcommand{\R}{{\mathbb{R}}}
\newcommand{\N}{{\mathbb{N}}}
\newcommand{\E}{{\mathbb{E}}}
\renewcommand{\P}{{\mathbb{P}}}
\newcommand{\Q}{{\mathbb{Q}}}
\newcommand{\cF}{{\mathcal{F}}}
\newcommand{\cM}{{\mathcal{M}}}
\newcommand{\cG}{{\mathcal{G}}}
\newcommand{\cQ}{{\mathcal{Q}}}
\newcommand{\cH}{{\mathcal{H}}}
\newcommand{\cC}{{\mathcal{C}}}
\newcommand{\cK}{{\mathcal{K}}}
\newcommand{\pack}{{\mathrm{pack}}}
\newcommand{\op}{{\mathrm{op}}}
\newcommand{\bX}{{\mathbf{X}}}
\newcommand{\bA}{{\boldsymbol{A}}}
\newcommand{\bI}{{\boldsymbol{I}}}
\newcommand{\bB}{{\boldsymbol{B}}}
\newcommand{\bU}{{\boldsymbol{U}}}
\newcommand{\bD}{{\boldsymbol{D}}}
\newcommand{\bK}{{\boldsymbol{K}}}
\newcommand{\bR}{{\boldsymbol{R}}}
\newcommand{\cX}{{\mathcal{X}}}
\newcommand{\cT}{{\mathcal{T}}}
\newcommand{\cB}{{\mathcal{B}}}
\newcommand{\cS}{{\mathcal{S}}}
\newcommand{\cA}{{\mathcal{A}}}
\newcommand{\cU}{{\mathcal{U}}}
\newcommand{\cR}{{\mathcal{R}}}
\newcommand{\T}{{\mathsf{T}}}
\newcommand{\Gauss}{{\mathrm{Gauss}}}
\newcommand{\poly}{{\mathrm{poly}}}
\newcommand{\Sob}{{\mathrm{Sob}}}
\newcommand{\BS}{{\mathrm{Brier}}}

\newcommand{\nrep}{{n_\mathrm{rep}}}
\DeclareMathOperator*{\argmin}{argmin}
\DeclareMathOperator*{\sargmin}{sargmin}
\DeclareMathOperator{\Vol}{Vol}
\DeclareMathOperator{\Exp}{Exp}
\DeclareMathOperator{\KL}{KL}
\DeclareMathOperator{\TV}{TV}
\DeclareMathOperator*{\essinf}{ess\,inf}

\newcommand{\diffi}{\,\mathrm{d}}
\newcommand{\diffd}{\mathrm{d}}
\newcommand{\dbar}{\;\|\;}

\newcommand{\myfigscale}{0.72}

\newtoggle{journal}
\togglefalse{journal}

\theoremstyle{plain}
\newtheorem{theorem}{Theorem}
\newtheorem{lemma}[theorem]{Lemma}
\newtheorem{corollary}[theorem]{Corollary}
\newtheorem{proposition}[theorem]{Proposition}
\newtheorem{definition}[theorem]{Definition}
\newtheorem{example}[theorem]{Example}

\title{Upgrading survival models with CARE}
\author{
  William G.\ Underwood\textsuperscript{1,*},
  Henry W.\ J.\ Reeve\textsuperscript{2},
  Oliver Y.\ Feng\textsuperscript{3}, \\
  Samuel A.\ Lambert\textsuperscript{4},
  Bhramar Mukherjee\textsuperscript{5} and
  Richard J.\ Samworth\textsuperscript{6}
}
\date{\today}

\makeatletter%
\begin{document}

\maketitle

\footnotetext[1]{
  Department of Statistics,
  University of Warwick.
}
\footnotetext[2]{
  School of Artificial Intelligence,
  Nanjing University.
}
\footnotetext[3]{
  Department of Statistics,
  London School of Economics and Political Science.
}
\footnotetext[4]{
  British Heart Foundation Cardiovascular Epidemiology Unit, \\
  Department of Public Health and Primary Care,
  University of Cambridge.
}
\footnotetext[5]{
  Department of Biostatistics,
  Yale School of Public Health.
}
\footnotetext[6]{
  Statistical Laboratory,
  University of Cambridge.
}

\let\thefootnote\relax
\footnotetext[1]{
  \textsuperscript{*}Corresponding author:
  \href{mailto:william.underwood@warwick.ac.uk}{%
  \textcolor{blue}{\texttt{william.underwood@warwick.ac.uk}}}
}
\newcommand{\thefootnote}{\arabic{footnote}}

\setcounter{page}{0}\thispagestyle{empty}

\newcommand{\anon}{1}
 
\begin{abstract}
Clinical risk prediction models are regularly updated
as new data, often with additional covariates,
become available.
We propose CARE
(Convex Aggregation of relative Risk Estimators) as
a general approach for combining existing `external' estimates
with a new data set in a time-to-event survival analysis setting.
Our method initially employs the new data to fit a flexible family
of reproducing kernel estimators via penalised partial likelihood
maximisation.
The final relative risk estimator is then constructed as a convex
combination of
the kernel and external estimators, with
the regularisation parameters
and convex combination coefficients
selected using cross-validation.
We establish high-probability bounds for the $L_2$-error
of our proposed aggregated estimator, showing that
it achieves a rate of convergence that is at least as good as both
the optimal kernel estimator and the best existing external model.
Empirical results from simulation studies align with
the theoretical results, and we illustrate the improvements
our methods provide for cardiovascular disease risk modelling.
Our methodology is implemented in the Python package
\texttt{care-survival}.
 \end{abstract}

\vspace*{10mm}

\vspace*{4mm}

\clearpage
\pagebreak

\section{Introduction}

In biomedicine, clinical prediction models for health outcomes based
on covariates
provide an important resource to aid medical decision-making.
These models are typically updated
regularly as additional potentially relevant variables are measured,
usually on a separate set of individuals. A motivating example for
this work is the SCORE2 model
\citep{esc2021score2} for estimating 10-year fatal and
non-fatal cardiovascular disease risk in Europe.
The model was initially trained using data on $677{,}684$ individuals
under a Fine--Gray competing risks framework
\citep{fine1999proportional},
including age, smoking status, blood pressure and cholesterol as covariates.
A similar time-to-event model named SCORE2-OP was developed concurrently
\citep{esc2021score2op},
specifically targeted towards individuals over $70$ years old.
Since the initial release of these models, several updates have
been proposed, adding new predictor variables to improve risk estimation
\citep{matsushita2023including,cosmi2024score,%
dienhart2024including,wong2025integration,xie2025metabolomics}.
Although the new predictors can sometimes be calculated
directly from existing covariates by applying appropriate transformations,
it is often the case that
the model must be updated using an alternative data source
\citep{matsushita2023including,xie2025metabolomics}.
Other existing survival models include
PREDICT for breast cancer and prostate cancer prognosis
\citep{wishart2010predict,wishart2012predict,wishart2014inclusion,
thurtle2019individual,thurtle2019understanding,pandiaraja2024utilisation},
QRISK for cardiovascular disease risk prediction
\citep{hippisley2007derivation,hippisley2008predicting,
hippisley2017development} and
PLCO\textsubscript{M2012} for lung cancer diagnosis
\citep{tammemagi2011lung,tammemagi2013selection}.
A common approach for upgrading such models is to
incorporate directly fitted coefficients from an existing study,
here in the form of estimated hazard ratios.
While validation studies can provide some reassurance, the underlying
strong assumption is that the
imported coefficients,
which may be based on a reduced set of variables, remain
valid in the presence of the expanded set of covariates in the model.
For example, the PREDICT tool for breast cancer prognosis was updated
twice using this approach, adding
the biomarker proteins HER2 and KI67 as new covariates
\citep[see][respectively]{wishart2012predict,wishart2014inclusion}.

In this paper, we propose general methodology for fitting and
upgrading survival analysis prediction models when new data
(potentially with additional covariates) become available.
Our framework encompasses the common situations where the original
individual-level training data may no
longer be available, for instance due to patient confidentiality,
regulatory or proprietary reasons.
Instead, we may only be able to
query the model in the sense of obtaining predictions at an arbitrary
test point of our choosing, typically through an online user
interface, or in the case of a parametric model, through published
fitted parameter estimates.
In this setting, we have access to one or more existing pre-trained models,
as well as a new independent data set
consisting of the full expanded set of covariates,
censored event times and censoring indicators.

Following some background material in Section~\ref{sec:set_up}, our
first contribution, in Section~\ref{sec:methodology}, is to introduce a
principled reproducing
kernel Hilbert space (RKHS) framework that provides considerable
flexibility for modelling the relative risk function. In the special case
of a linear kernel, this reduces to
(a regularised version of) the celebrated Cox proportional
hazards model \citep{cox1972regression}, but fully nonparametric
choices are also possible. We demonstrate that such models can be
fitted in a computationally efficient manner via a restricted
penalised partial likelihood approach.
In order to
borrow strength across all available information sources, we propose
a data-driven procedure named CARE
(Convex Aggregation of relative Risk Estimators),
which combines the existing models and our new RKHS estimator.
The final prediction function is constructed as a convex combination
of the old and new prediction functions,
with the convex combination weights and the RKHS regularisation parameter
chosen via a joint partial likelihood cross-validation strategy.

In Section~\ref{sec:analysis}, we turn to the theoretical properties
of our proposed CARE procedure. Our first main results consider the
RKHS estimator alone. Theorem~\ref{thm:rate} provides a
high-probability bound on the error of this estimator in terms of the
sample size, the regularisation parameter and a quantity reflecting
the spectral properties of the natural integral operator associated
with the kernel, while Theorem~\ref{thm:minimax} establishes
minimax rate-optimality of our procedure, under some further mild
conditions on the kernel spectrum.
Theorem~\ref{thm:parameter_tuning} then gives
an oracle inequality to justify the use of cross-validation for
parameter tuning.
Theorem~\ref{thm:model_selection} in
Section~\ref{sec:model_selection} presents our main performance
guarantee for the CARE procedure that combines the RKHS
estimator with the prediction functions from the existing models.
It demonstrates that our joint cross-validation strategy performs nearly
as well as the oracle convex combination of the
pre-trained models and the RKHS estimators. In
particular, CARE is able to adapt both to situations where
the best pre-trained model is a superior predictor compared to the
newly constructed RKHS estimator (for example, because the sample
size of the new study is too small),
and to settings where the optimal RKHS estimator is a better
predictor than all of the existing estimators.
Section~\ref{sec:implementation} provides practical implementation
details for our CARE procedure, including fitting recommendations via
first-\ and second-order optimisation methods. In
Section~\ref{sec:empirical}, we present our empirical results,
first verifying our theoretical results through a series of
simulation studies.
We then consider upgrading an existing cardiovascular disease risk
model (SCORE2)
using cohort data from the UK Biobank \citep{littlejohns2019uk}.
By incorporating a new data set with additional covariates
(relating to area-specific social deprivation levels and polygenic scores)
we show that CARE can improve the concordance index
\citep{harrell1982evaluating}
of the SCORE2 model by
$1.52\%$ for females and by $2.26\%$ for males.
Proofs and further technical and computational details are gathered in
\iftoggle{journal}{the supplementary material%
  \if1\anon\ \citep{underwood2026caresupplement}\fi,
  where results are prefixed with
the letter `S'}{the appendix},
along with additional empirical results.

Recent years have seen a growing interest among health researchers in
developing methodologies for aggregating estimators and data sets.
In the parametric setting, approaches include
joint likelihood maximisation \citep{zhang2020generalized},
penalised constrained maximum likelihood for
heterogeneous populations \citep{zhai2022data},
empirical Bayes composite estimation \citep{gu2023meta},
constrained empirical likelihood \citep{ni2023empirical}
and generalised method of moments \citep{fang2025integrated}.
Semiparametric and nonparametric methods include
universal aggregation \citep{goldenshluger2009universal},
affine aggregation \citep{dalalyan2012sharp},
efficient data fusion
\citep{li2023efficient,graham2026towards,hu2026semiparametric}, model
aggregation via cross-validation \citep{yu2025unified}
and pseudo-labelling \citep{wang2026pseudo}.
In the context of time-to-event models,
\citet{liu2014estimating} use external data sources to improve
estimation of the baseline hazard function for within-cohort
Cox proportional hazards analysis of rare diseases.
\citet{huang2016efficient} describe a method for
fitting a Cox model by an empirical likelihood approach,
incorporating external information in the form of
published group-level survival probabilities;
see also \citet{chen2021combining} and \citet{sheng2021synthesizing}
for extensions.  \citet{ding2023fitting,ding2024efficient}
consider models in which auxiliary information
is available in the form of fixed-time survival
probabilities, stratified by subgroups;
estimation then proceeds by penalised
likelihood estimation or a penalised
generalised method of moments.
\citet{wang2025kullback} propose a procedure for
fitting a generalised parametric Cox model by
maximising a penalised version of the log-likelihood
that imposes Kullback--Leibler-based shrinkage towards
external pre-trained survival models.
In our methodology, we opt for
cross-validation of convex combinations
rather than explicit penalisation in order
to obtain an automatic data-driven procedure
for incorporating the external information.

Related to our model upgrading task, \emph{transfer learning}
\citep{cai2021transfer,reeve2021adaptive} also considers settings of
multiple data sources. Specifically, one typically studies source
and target data sets, both of which are available and measure the
same features. The primary emphasis is on adapting to the
distributional shift between source and target, with the aim of
improving inference regarding some aspect of the target distribution.

A separate line of work has sought to develop general theory and
methodology for estimation in survival analysis models where the
relative risk is modelled more flexibly than via a linear predictor.
Early works include the introduction of local polynomials for partial
likelihood estimation \citep{fan1997local,chen2007local}; see also
more recent work on high-dimensional regularised Cox models
\citep{bradic2011regularization,huang2013oracle,yu2021confidence}.
\citet{liu2020local} propose RKHS methods with Sobolev kernels for
censored survival data,
building on work of
\citet{osullivan1988nonparametric,osullivan1993nonparametric}.
\citet{chung2018partial} incorporate nonparametric shape constraints
\citep{samworth2018special} into hazard ratio modelling, while
\citet{qu2016optimal} consider kernel estimation in the functional
Cox model. In a more applied vein, \citet{montaseri2025survival}
study heart failure via an RKHS framework (without regularisation)
for an accelerated failure time model.

An alternative approach to survival modelling is landmark analysis
\citep{van2007dynamic,van2008dynamic,parast2013landmark,fournier2018r2},
in which survival is modelled conditional
on being event-free up to a fixed landmark time.
Covariate values are typically fixed at the landmark, and survival
beyond this time is analysed within the landmark cohort,
often using a Cox model.
This yields interpretable, dynamic predictions, but does not model
the full time-varying hazard function.
Since existing studies for disease risk modelling
only rarely use landmark analysis,
we do not address this framework in our methodology.

\section{Set-up and preliminary results}
\label{sec:set_up}

We begin by providing some notation used
throughout the paper, and
give a description of the model for the data generating process
based on a censored time-to-event formulation.
We also present some further introductory material on reproducing
kernel Hilbert spaces and associated norms.

\subsection{Notation}

For $n \in \N$, let $[n] \vcentcolon= \{1, \ldots, n\}$.
We write $a \land b$ for the minimum of two real numbers
$a$ and~$b$, and $a \lor b$ for their maximum.
For positive real-valued sequences $(a_n)$ and $(b_n)$, we write
$a_n \lesssim b_n$ if there exists $C > 0$
and $N \in \N$ with
$a_n \leq C b_n$ for all $n \geq N$.
Similarly, we write
$a_n \asymp b_n$ if there exists $C \geq 1$
and $N \in \N$ with
$1/C \leq a_n / b_n \leq C$ for all $n \geq N$.
For $a \in \R$ and $d \in \N$,
define $a^{\otimes d} \vcentcolon= (a, \ldots, a)^\T \in \R^d$.
For a real square symmetric matrix $\bA$,
we write $\bA \succeq 0$ if $\bA$
is positive semi-definite and
$\bA \succ 0$ if $\bA$ is positive definite, with $\preceq$ and
$\prec$ defined analogously.
We write $\lambda_{\min}(\bA)$ and $\lambda_{\max}(\bA)$
for the minimum and maximum eigenvalues of $\bA$, respectively.
The operator norm is denoted by $\|\bA\|_{\mathrm{op}}$,
while the entrywise $\ell_\infty$-norm is $\|\bA\|_{\infty}$.
The $d \times d$ identity matrix is denoted $\bI_d$,
and we write $0_d \vcentcolon= (0, \ldots, 0)^\T \in \R^d$
and $1_d \vcentcolon= (1, \ldots, 1)^\T \in \R^d$.
For a measurable space $(\cX,\cA)$,
define $\cB(\cX)$ to be the set
of bounded, measurable functions from $\cX$ to $\R$.
The supremum norm of $f \in \cB(\cX)$ is
$\|f\|_\infty \vcentcolon= \sup_{x \in \cX} |f(x)|$.
For a totally ordered finite set $\cS$ and a function $f: \cS \to \R$,
we write $\sargmin_{s \in \cS} f(s) \vcentcolon=
\min \bigl\{s \in \cS: f(s) = \min\{f(t): t \in \cS\}\bigr\}$.
The uniform distribution on a compact interval $[a, b]$
is denoted by $\cU[a, b]$.

\subsection{Data and partial likelihood functions}
\label{sec:data}

Let $(\cX,\cA)$ denote a measurable space, and let
$(X, T_{\mathrm S}, T_{\mathrm C})$ be a random vector taking values in
$\cX \times (0, \infty) \times (0, \infty)$.
Here, $X$ represents a covariate, $T_{\mathrm S}$ a survival
time and $T_{\mathrm C}$ a censoring time.
We emphasise that $X$ may take values in a wide range of
possible covariate spaces, including $\R^d$ for any $d \in \N$
as well as non-Euclidean domains such as sequence or graph spaces.
As such, we refrain from explicitly referring to the
`dimension' of the covariate.
Suppose that $T_{\mathrm S}$
and $T_{\mathrm C}$ are conditionally independent given $X$, and
define the observed time $T \vcentcolon= T_{\mathrm S} \land T_{\mathrm C}$
and the censoring indicator
$I \vcentcolon= \mathbbm{1}_{\{T_{\mathrm C} < T_{\mathrm S}\}}$.
The \emph{at-risk process} is given by $R(t) \vcentcolon= \mathbbm{1}_{\{ T
\geq t \}}$ for $t \in [0, 1]$.
Define $q: \cX \times [0, 1] \to \R$ by
$q(x, t)
\vcentcolon= \P(T \geq t \mid X=x)$
and assume that
$q_1 \vcentcolon= \inf_{x \in \cX} q(x, 1) > 0$.
Let $P_X$ denote the distribution of~$X$.
We consider a proportional hazards model with conditional hazard
function
\begin{equation}
  \label{eq:hazard}
  h(t \mid x)
  \vcentcolon= - \frac{\diffd}{\diffd t}
  \log \P (T_{\mathrm S} > t \mid X=x)
  = \lambda_0(t) e^{f_0(x)}
\end{equation}
for $t \in [0, 1]$ and $P_X$-almost all $x$,
where $\lambda_0: [0, 1] \to [0, \infty)$
is a measurable baseline hazard function and
$f_0 \in \cB(\cX)$ is a function representing the relative risk.
For any
$P_X$-integrable $f: \cX \to \R$, define
$P_X(f) \vcentcolon= \int_\cX f(x) \diffi P_X(x)$.
To ensure the identifiability of $\lambda_0$ and $f_0$,
we impose the constraint that $P_X(f_0) = 0$
\citep[see][]{osullivan1993nonparametric}.
Observe that
\begin{align*}
  \Lambda &\vcentcolon=
  \int_0^1 \lambda_0(t) \diffi t
  = e^{-f_0(x)}
  \bigl\{
    \log \P(T_{\mathrm S} > 0 \mid X = x)
    - \log \P(T_{\mathrm S} \geq 1 \mid X = x)
  \bigr\} \\
  &\phantom{\vcentcolon}\leq e^{\|f_0\|_\infty} \log (1/q_1) < \infty,
\end{align*}
where the second equality holds for $P_X$-almost all
$x \in \cX$.
The observed data $(X_i, T_i, I_i)_{i \in [2n]}$ are assumed to be
independent and identically distributed copies of $(X,T,I)$.
We refer to the first $n$ triples as the \emph{training set},
while the rest form the \emph{validation set}.
For $i \in [2n]$, write
$N_i(t) \vcentcolon= \mathbbm{1}_{\{ T_i \leq t,\, I_i = 0 \}}$
for the event process \citep{fleming2013counting} and
$R_i(t) \vcentcolon= \mathbbm{1}_{\{ T_i \geq t \}}$.
Defining $N(t) \vcentcolon= \sum_{i=1}^n N_i(t) / n$ and
\begin{equation*}
  S_n(f, t) \vcentcolon= \frac{1}{n} \sum_{i=1}^n
  R_i(t) e^{f(X_i)}
\end{equation*}
for $f: \cX \to \R$, the negative log-partial likelihood
of this model is
\begin{equation}
  \label{eq:likelihood}
  \ell_n(f)
  \vcentcolon=
  \int_{0}^{1} \log S_n(f, t) \diffi {N(t)}
  - \frac{1}{n} \sum_{i=1}^n f(X_i) N_i(1).
\end{equation}
This is well-defined because
$S_n(f, T_i) \geq R_i(T_i) e^{f(X_i)} / n = e^{f(X_i)} / n > 0$
for $i \in [n]$.
For $f \in \cB(\cX)$ and $t \in [0, 1]$,
let $S(f, t) \vcentcolon= \int_\cX q(x, t) e^{f(x)} \diffi P_X(x)$
and define the limiting negative log-partial likelihood by
\begin{equation*}
  \ell_{\star}(f)
  \vcentcolon=
  \int_{0}^{1}
  \log\bigl(S(f, t)\bigr)
  S(f_0, t) \lambda_0(t) \diffi t
  -
  \int_{0}^{1}
  D S(f_0, t)(f)
  \lambda_0(t)
  \diffi t,
\end{equation*}
where $D S(f_0, t)$ is a Gateaux derivative
\citep[Chapter~7]{aliprantis2006infinite}.
This is finite, since
$S(f_0, t) \leq e^{\|f - f_0\|_\infty} S(f, t)$, and we
interpret $0 \log 0$ as zero.
Both $\ell_n(f)$ and $\ell_\star(f)$
are invariant under addition of constants to $f$.
The following lemma may be regarded as a Fisher
consistency property of $f_0$ \citep{fisher1922mathematical}.
\begin{lemma}[Characterisation of $f_0$]%
  \label{lem:characterisation_f0}
  We have $f_0 \in \argmin_{f \in \cB(\cX)} \ell_\star(f)$.
  Further, $f_0$ is unique in the sense that if
  $\tilde f \in \argmin_{f \in \cB(\cX)} \ell_\star(f)$
  and $P_X(\tilde f) = 0$,
  then $\tilde{f}(x) = f_0(x)$
  for $P_X$-almost every $x \in \cX$.
\end{lemma}

\subsection{Reproducing kernel Hilbert spaces}
\label{sec:rkhs}

A \emph{kernel} is a function $k: \cX \times \cX \to \R$
that is symmetric and is positive definite in the sense that for any
$n \in \N$ and $x_1,\ldots,x_n \in \mathcal{X}$, the matrix
$\bK \in \R^{n \times n}$ with entries
$\bK_{i j} \vcentcolon= k(x_i,x_j)$
is symmetric and positive semi-definite. We will assume that
the measurable space $(\mathcal{X},\mathcal{A})$ and
the kernel $k$ are sufficiently regular;
see Appendix~\ref{sec:kernel_assumptions}%
\iftoggle{journal}{ in the supplementary material%
\if1\anon\ \citep{underwood2026caresupplement}\fi}{}
for precise conditions and further technical details.
By Mercer's theorem
\citep[Theorem~A.8.8]{samworth2026modern},
there exists a decreasing sequence $(\nu_r)_{r \in \mathcal{R}}$
of positive numbers
with $\sum_{r \in \mathcal{R}} \nu_r < \infty$ and associated continuous
functions $(e_r)_{r \in \mathcal{R}}$ in
$L_2(P_X)$, where either $\mathcal{R} = [R]$ for some $R \in \N$ or
$\mathcal{R} = \N$, endowed with the usual order,
such that
$k(x, y) = \sum_{r \in \mathcal{R}} \nu_r e_r(x) e_r(y)$
for $x,y \in \mathcal{X}$.
Write $\langle \cdot, \cdot \rangle_{L_2}
\vcentcolon= \langle \cdot, \cdot \rangle_{L_2(P_X)}$.
Let $(\mathcal{H},\langle \cdot,\cdot \rangle_{\mathcal{H}})$
be the reproducing kernel Hilbert space (RKHS)
associated with $k$, and
suppose further that $(x, y) \mapsto k(x, y) - \delta$
is a positive definite function on $\cX \times \cX$
for some $\delta > 0$. This ensures that
$\cH$ contains the unit constant function on $\cX$,
which we denote by $1_\cX$.
We assume throughout that the model is correctly
specified in the sense that the
true relative risk function satisfies $f_0 \in \cH$.

In preparation for the upcoming regularisation framework
for RKHS estimation,
we define another inner product, for $\gamma \geq 0$, by
$\langle f, g \rangle_{\cH_\gamma} \vcentcolon=
\langle f, g \rangle_{L_2} + \gamma \langle f, g \rangle_\cH$.
For $r \in \mathcal R$,
write $a_r \vcentcolon= 1/(1 + \gamma / \nu_r)$ and
define $K \vcentcolon= \sup_{x \in \cX} k(x, x)$,
$H_{2,\gamma} \vcentcolon= \sum_{r \in \mathcal R} a_r$
and $H_{\infty,\gamma} \vcentcolon= \sup_{x \in \cX}
\sum_{r \in \mathcal R} a_r e_r(x)^2$,
so that
$H_{2,\gamma} \leq H_{\infty,\gamma}
\leq K / \gamma$.

We give some examples of kernel functions that satisfy
our assumptions in Section~\ref{sec:kernel_examples},
while practical guidance
on the choice of kernel is presented
in Section~\ref{sec:choice_kernel}.
In Lemma~\ref{lem:infty_bounds},
we bound $\|\cdot\|_\infty$ on $\cH$ in terms of
$\|\cdot\|_\cH$ and $\|\cdot\|_{\cH_\gamma}$ respectively.

\begin{lemma}[Supremum norm bounds]%
  \label{lem:infty_bounds}
  Let $f \in \mathcal{H}$ and $\gamma > 0$.
  Then $\|f\|_\infty^2 \leq K \|f\|_\cH^2$
  and $\|f\|_\infty^2 \leq H_{\infty,\gamma} \|f\|_{\cH_\gamma}^2$.
\end{lemma}

\section{Methodology}
\label{sec:methodology}

\subsection{Kernel estimator}
\label{sec:estimator}

We begin by introducing our basic estimator, which is obtained by minimising
a penalised version of the negative log-partial
likelihood over the RKHS $\cH$.  For $\gamma > 0$ and $f \in \cH$,
define the penalised negative log-partial likelihood
\begin{equation}
  \label{eq:likelihood_penalised}
  \ell_{n,\gamma}(f)
  \coloneqq \ell_n(f) + \gamma \|f\|_\cH^2,
\end{equation}
where $\ell_n(f)$ is defined in \eqref{eq:likelihood}.
We show in Proposition~\ref{prop:representation} below that
there is a unique estimator
\begin{equation}
  \label{eq:argmin}
  \hat f_{n,\gamma} \coloneqq
  \argmin_{f \in \cH:\, P_n(f) = 0} \ell_{n, \gamma}(f),
\end{equation}
where $P_n(f) \vcentcolon= \sum_{i=1}^n f(X_i) / n$.  The centring
constraint $P_n(f) = 0$ in our optimisation is natural in view of the
identifiability condition $P_X(f_0) = 0$ and the fact that the
log-likelihood is unaffected by the addition of constants.
Before stating this result, we introduce some notation.
Let $\cA_n$ be a maximal subset of $[n]$ such that
$\bigl\{
  (\|1_\cX\|_\cH^2, 1, \ldots, 1)^\T,
  \bigl(1, k(X_i, X_1), \ldots, k(X_i, X_n)\bigr)^\T :
  i \in \cA_n
\bigr\} \subseteq \R^{n+1}$
is linearly independent. For $x, y \in \cX$,
let $\tilde k(x, y) \vcentcolon= k(x, y) - \bar k(y)$, where $\bar
k(y) \vcentcolon= \sum_{i=1}^n k(X_i, y) / n$. Further, for $\beta
\in \R^{\cA_n}$, define
$f_\beta(x) \vcentcolon= \sum_{i \in \cA_n} \tilde k(x, X_i) \beta_i$,
so that $P_n(f_\beta) = 0$.

\begin{proposition}[Representation of $\hat f_{n,\gamma}$]%
  \label{prop:representation}
  The function $\beta \mapsto \ell_{n,\gamma}(f_\beta)$ is
  strongly convex
  on $\R^{\cA_n}$. Writing $\hat{\beta}$ for its unique minimiser,
  we have that $\hat f_{n,\gamma} \vcentcolon= f_{\hat\beta}$ is the
  unique minimiser of $f \mapsto \ell_{n, \gamma}(f)$
  over $f \in \mathcal{H}$ with $P_n(f) = 0$.
\end{proposition}

Proposition~\ref{prop:representation} shows that in order to find
the minimiser in~\eqref{eq:argmin},
it suffices to consider functions of the form
$f_\beta$, where $\beta \in \R^{\cA_n}$.
This representation allows the feasible set of
the optimisation problem to be simplified from $\cH$
(which may be infinite-dimensional)
to $\R^{\cA_n}$ (which has at most $n$ dimensions).
In Section~\ref{sec:implementation}, we give further details on the
calculation of $\|1_\cX\|_\cH$ and
the construction of the set $\cA_n$, and discuss
practical procedures for computing $\hat \beta$ and $\hat f_{n,\gamma}$.

\subsection{Parameter tuning}

Partial likelihood cross-validation can be used to select the
regularisation parameter $\gamma$.
Let $\Gamma \subseteq (0, \infty)$ be a finite set of candidate
parameters.
With $\tilde N(t) \vcentcolon= \sum_{i=n+1}^{2n} N_i(t) / n$ and
$\tilde S_n(f, t) \vcentcolon= \sum_{i={n+1}}^{2n}
R_i(t) e^{f(X_i)} / n$,
define the negative partial log-likelihood on the validation data by
\begin{equation*}
  \tilde \ell_n(f)
  \vcentcolon=
  \int_{0}^{1} \log \tilde S_n(f, t) \diffi {\tilde N(t)}
  - \frac{1}{n} \sum_{i=n+1}^{2n} f(X_i) N_i(1),
\end{equation*}
where $f: \cX \to \R$.
For each $\gamma \in \Gamma$, fit $\hat f_{n,\gamma}$
on the training data as in \eqref{eq:argmin}.
Finally, set
\begin{equation}
  \label{eq:gamma_hat}
  \hat\gamma \vcentcolon= \sargmin_{\gamma \in \Gamma}
  \tilde \ell_n\bigl(\hat f_{n,\gamma}\bigr).
\end{equation}

\subsection{Model aggregation with CARE}
\label{sec:care}

We now provide a framework for combining the
kernel estimators $(\hat f_{n,\gamma} : \gamma \in \Gamma)$
with a collection of `external' estimators $(\hat f_m: m \in \cM)$,
where $\cM = [|\cM|]$ is a non-empty finite set
and $\hat f_m$ takes values in $\cB(\cX)$ for each $m \in \cM$.  In
applications, we do not need
the external estimators $\hat f_m$
to use precisely the same set of covariates
as the kernel estimators $\hat f_{n,\gamma}$.
Indeed, the only requirement is that the validation likelihood
$\tilde \ell_n$ can be evaluated at each $\hat f_m$.
Specifically, if the validation data set contains a covariate
not used by an external estimator
(that is, a new covariate is being added),
it can simply be ignored when evaluating the validation likelihood.
Conversely, if the external estimator uses an extra covariate
not present in the validation data
(that is, an existing covariate is being removed),
then one should marginalise over (an estimate of)
the distribution of the missing covariate given
the observed covariates.
This may be challenging if the original data
set is no longer available; however, in most practical settings,
one seeks to add new covariates to existing models
rather than remove them.

To ensure the validity of our cross-validation methodology,
we assume that $(\hat f_m: m \in \cM)$
are computed on a data set
that is independent of the training and validation data
$(X_i, T_i, I_i)_{i \in [2n]}$.  However, we do not need each $\hat
f_m$ to itself be a
function of independent and identically distributed data points;
nor do we require that $\hat f_m$ be statistically independent
of $\hat f_{m'}$, where $m \neq m'$.
We further suppose that $\max_{m \in \cM} \|\hat f_m\|_\infty \leq M$
almost surely, where $M$ is a random variable taking values in $(0, \infty)$.
For $m \in \cM$, define
$\tilde f_m \vcentcolon= \hat f_m - P_n(\hat{f}_m) 1_\cX$.
Let $\Theta \subseteq \Delta_\cM \coloneqq
\bigl\{\theta \in [0, 1]^{\cM}:
\sum_{m \in \cM} \theta_m \leq 1\bigr\}$
be a finite set
and consider the convex combination estimators
$(\check f_{n, \gamma, \theta}:\gamma \in \Gamma,\theta \in \Theta)$,
where we define
\begin{equation*}
  \check f_{n, \gamma, \theta}(x)
  \vcentcolon=
  \biggl(1 - \sum_{m \in \cM} \theta_m\biggr)
  \hat f_{n,\gamma}(x) + \sum_{m \in \cM} \theta_m \tilde f_m(x)
\end{equation*}
for $x \in \cX$. We construct $\check\gamma$ and $\check\theta$
via partial likelihood cross-validation, setting
\begin{equation}
  \label{eq:gamma_theta_check}
  \bigl(\check\gamma,\check\theta\bigr)
  \vcentcolon= \sargmin_{(\gamma, \theta) \in \Gamma \times \Theta}
  \tilde \ell_n\bigl(\check f_{n,\gamma,\theta}\bigr),
\end{equation}
using the lexicographic order on $\Theta$.
In general, $\check\gamma$ defined in \eqref{eq:gamma_theta_check}
may not take the same value as $\hat\gamma$ given in
\eqref{eq:gamma_hat}.  In Section~\ref{sec:cv},
we give some practical recommendations for solving the optimisation
problem \eqref{eq:gamma_theta_check}.
We define the CARE estimator as $\check f_{n, \check\gamma, \check\theta}$.

\section{Theoretical analysis}%
\label{sec:analysis}

In this section, we establish high-probability bounds for the deviations of
our proposed estimators from the true relative risk function $f_0$.
We are primarily interested
in approximation under the $L_2(P_X)$-norm; that is,
in controlling $\|f- f_0\|_{L_2}$.
This choice of metric can be motivated by
treating $f$ as an estimator
of the log-relative hazard, since for $f \in L_2(P_X)$
and $t \in [0, 1]$, and
with $h(t \mid x)$ as in \eqref{eq:hazard},
\begin{equation*}
  \|f - f_0\|_{L_2}^2
  = \E \Biggl\{
    \biggl(
      \log \frac{h(t \mid X)}{\lambda_0(t)}
      - f(X)
    \biggr)^2
  \Biggr\}.
\end{equation*}

\subsection{Kernel estimator}%
\label{sec:kernel}

Our first main theoretical result concerns the proximity of
$\hat f_{n,\gamma}$ to $f_0$,
where $\gamma > 0$ is a fixed regularisation parameter.

\begin{theorem}[Rate of convergence]%
  \label{thm:rate}
  There exists $C_0 > 0$, depending only on
  $\|f_0\|_\cH$, $\Lambda$, $q_1$, $\|1_\cX\|_\cH$ and $K$,
  such that if $\gamma H_{\infty,\gamma} \leq 1 / C_0$
  and $1 \leq s \leq \sqrt{n} / (C_0 H_{\infty,\gamma})$,
  then with probability at least $1 - e^{-s^2}$,
  \begin{equation}
    \label{eq:rate}
    \bigl\| \hat f_{n,\gamma} - f_0 \bigr\|_{L_2}
    \leq
    \bigl\| \hat f_{n,\gamma} - f_0 \bigr\|_{\cH_\gamma}
    \leq
    C_0
    \biggl(
      \frac{\sqrt{H_{2,\gamma}} + s}{\sqrt n}
      + \sqrt{\gamma}
    \biggr).
  \end{equation}
\end{theorem}

The right-hand side of \eqref{eq:rate} consists of two terms,
corresponding to a bias--variance trade-off.
Specifically, the `variance' term $H_{2,\gamma} / n$ is a decreasing
function of both the sample size $n$ and the regularisation
parameter $\gamma$,
while the `bias' term is $\sqrt\gamma$.
As such, choosing a small value of $\gamma > 0$ allows for less bias
at the expense of potentially increasing the variance.
The rate of convergence is driven by the
kernel complexity quantity $H_{2,\gamma}$,
corresponding to the effective dimension of the problem
\citep{mendelson2003performance,
caponnetto2007optimal,koltchinskii2011oracle}.
The conditions under which \eqref{eq:rate}
holds depend on the quantity $H_{\infty,\gamma}$,
corresponding to a uniform upper bound on the leverage scores
\citep[Definition~2]{rudi2015less}
associated with the $\gamma$-regularised kernel.
If $\cX$ is compact, the assumption $\gamma H_{\infty,\gamma} \leq 1 / C_0$ is
satisfied for sufficiently small $\gamma > 0$ since $\gamma H_{\infty,\gamma}
\to 0$ as $\gamma \searrow 0$.

The best bias--variance trade-off in~\eqref{eq:rate}, in terms of the
order in $n$, is obtained by taking $\gamma = \gamma^\star$, where
$\gamma^\star \in \argmin_{\gamma > 0}
(H_{2,\gamma} / n + \gamma) \to 0$ as $n \to \infty$.
Moreover, we show that this rate
is minimax optimal in Theorem~\ref{thm:minimax} below.
For this choice of $\gamma^\star$,
and when $H_{2, \gamma^\star} \asymp H_{\infty, \gamma^\star}$,
we have $\sqrt{n} / H_{\infty,\gamma^\star}
\asymp 1/\sqrt{\gamma^\star H_{2,\gamma^\star}}$,
so for every fixed $s \geq 1$,
the condition $s \leq \sqrt{n} / (C_0 H_{\infty,\gamma^\star})$
is satisfied for sufficiently large~$n$.
If $\gamma^{a} H_{2,\gamma} \asymp 1$ as $\gamma \searrow 0$
for some $a \in [0, 1)$, as is common in applications
(see Section~\ref{sec:kernel_examples} below), then we obtain
$\gamma^\star \asymp n^{-1/(1 + a)}$, yielding an overall
convergence rate for $\hat{f}_{n,\gamma^\star}$ of $n^{-1/(2 + 2a)}$.

When covariates take values in $\R^d$ for some $d \in \N$,
the dimension $d$ affects the rate of convergence only through
the quantities $H_{2,\gamma}$ and $H_{\infty,\gamma}$;
our framework therefore handles
covariates with diverging dimension in a fully non-asymptotic fashion.

\subsubsection{Examples of kernel functions}
\label{sec:kernel_examples}

While $\gamma H_{\infty,\gamma} \to 0$ as $\gamma \searrow 0$
when $\cX$ is compact, sharper upper bounds on~$H_{\infty,\gamma}$
(and, therefore, on $H_{2,\gamma}$) are sometimes available
when a particular kernel is specified.
In this section, we present some
examples of kernels that satisfy our assumptions and,
where possible, derive bounds on the corresponding
quantities $H_{\infty,\gamma}$ and $\|1_\cX\|_\cH$.
See Appendix~\ref{sec:proofs}%
\iftoggle{journal}{ in the supplementary material%
\if1\anon\ \citep{underwood2026caresupplement}\fi}{}
for details.

First, in Example~\ref{ex:polynomial}, we demonstrate that for
a multivariate polynomial kernel on a bounded domain,
$H_{\infty,\gamma}$ is bounded as $\gamma \to 0$,
provided that the polynomial degree $p$ and the distribution of $X$
are fixed. In Example~\ref{ex:sobolev}, we consider
shifted first- and second-order Sobolev kernels
under a uniform measure.
For the first-order Sobolev kernel, we
show that $H_{\infty,\gamma}$ increases
at a rate no faster than
$1/\sqrt{\gamma}$ as $\gamma \searrow 0$.
Lastly, in Example~\ref{ex:gaussian},
we present a shifted Gaussian kernel.

\begin{example}[Polynomial kernel]
  \label{ex:polynomial}
  Let $d \in \N$ and $\cX \subseteq \R^d$,
  and take $p \in \N$ and $a \geq 0$.
  For $x, y \in \cX$,
  define the polynomial kernel of degree $p$ by
  \begin{equation*}
    k^\poly_{p,a}(x, y)
    \vcentcolon=
    \bigl(x^\T y + a\bigr)^p.
  \end{equation*}
  Assume that $\cX$ is a bounded,
  non-empty, Borel measurable subset of $\R^d$ containing $0_d$
  and let~$P_X$ be a Borel distribution on $\cX$.
  Then $\cR$ is a finite set with $|\cR| \leq \binom{d+p}{p}$.
  Moreover, there exists $C(p, d, P_X) > 0$, depending only on~$p$,~$d$
  and $P_X$, such that for all $a \geq 0$ and $\gamma > 0$, we have
  $\max_{r \in \cR} \nu_r \leq C(p, d, P_X)$,
  $\max_{r \in \cR} \|e_r\|_\infty \leq C(p, d, P_X)$
  and $H_{2,\gamma} \leq H_{\infty,\gamma} \leq C(p, d, P_X)$.
  Further, $1_\cX \in \cH$ if and only if $a > 0$;
  in that case, $\|1_\cX\|_\cH^2 = 1/a^p$.
\end{example}

\begin{example}[Shifted Sobolev kernels]
  \label{ex:sobolev}
  Let $\cX = [0, 1]$ and take $a \geq 0$.
  For $x, y \in \cX$,
  the first- and second-order shifted Sobolev kernels
  are defined respectively by
  \begin{align*}
    k^\Sob_{1,a}(x, y)
    &\vcentcolon=
    a + (x \land y),
    &k^\Sob_{2,a}(x, y)
    &\vcentcolon=
    a + \int_0^{x \land y} (x - z) (y - z) \diffi z.
  \end{align*}
  Let $P_X$ be Lebesgue measure on $\cX \vcentcolon= [0, 1]$
  and take $\gamma > 0$.
  Then, for $k_{1,a}^\Sob$, we have $\cR = \N$,
  $\nu_r \leq 1 / \{(r-1)^2 \pi^2\}$ for each $r \geq 2$,
  $\sup_{r \in \cR} \|e_r\|_\infty \leq \sqrt 2$
  and $H_{2,\gamma} \leq H_{\infty,\gamma} \leq 2 + 1 / \sqrt{\gamma}$.
  Let $\cH_1$, $\cH_2$ denote the RKHSs associated respectively with
  $k_{1,a}^\Sob$ and $k_{2,a}^\Sob$.
  Then $1_\cX \in \cH_1$ if and only if $a > 0$;
  in that case, $\|1_\cX\|_{\cH_1}^2 = 1/a$.
  Likewise, $1_\cX \in \cH_2$ if and only if $a > 0$,
  and again in that case, $\|1_\cX\|_{\cH_2}^2 = 1/a$.
\end{example}

\begin{example}[Shifted Gaussian kernel]
  \label{ex:gaussian}
  Let $d \in \N$ and $\cX \subseteq \R^d$,
  and take $a \geq 0$ and
  $\Sigma \in \R^{d \times d}$ with $\Sigma \succ 0$.
  For $x, y \in \cX$, the $d$-dimensional shifted Gaussian kernel
  is given by
  \begin{equation*}
    k^\Gauss_{\Sigma,a}(x, y)
    \vcentcolon=
    a + \exp\bigl(-(x-y)^\T \Sigma^{-1} (x-y)\bigr).
  \end{equation*}
  If $\cX$ is unbounded, then
  $1_\cX \in \cH$ if and only if $a > 0$;
  in that case, $\|1_\cX\|_\cH^2 = 1/a$.
\end{example}

By Theorem~\ref{thm:rate}, we deduce that
for the polynomial kernel, $\gamma^\star \asymp 1/n$, giving
a parametric rate of convergence of $1/\sqrt{n}$.
For the shifted first-order Sobolev kernel,
$\gamma^\star \asymp n^{-2/3}$ and the corresponding high-probability
bound from Theorem~\ref{thm:rate} is
$\| \hat f_{n,\gamma} - f_0 \|_{L_2} \lesssim n^{-1/3}$.
More generally, if the eigenfunctions $(e_r)_{r \in \mathcal{R}}$ are
uniformly bounded, in the sense that $\sup_{x \in \cX} \sup_{r \in
\cR} |e_r(x)| \leq C$,
then both $H_{2,\gamma}$ and
$H_{\infty,\gamma}$ are controlled by the decay rate of the eigenvalues
$(\nu_r)_{r \in \cR}$. For example, if
$\cR = \N$ and $\nu_r \leq 1 / r^b$ for some $b > 1$, then
$H_{\infty,\gamma} \leq C \int_0^\infty 1 / (1 + \gamma t^b) \diffi t
= C \pi / \bigl\{b \gamma^{1/b} \sin(\pi / b)\bigr\}$.
Thus, $\gamma^\star \asymp n^{-b / (1 + b)}$
and the high-probability bound from Theorem~\ref{thm:rate} is
$\| \hat f_{n,\gamma} - f_0 \|_{L_2} \lesssim n^{-b / (2 + 2 b)}$.

\subsection{Minimax lower bound}
\label{sec:minimax}
\refused{sec:minimax}

In this section we present a minimax lower bound
for relative risk function estimation in an RKHS,
establishing rate-optimality of our Theorem~\ref{thm:rate}.
In order to demonstrate that the fundamental
difficulty of this problem is driven by
estimation in the RKHS,
we consider a subset of the model
described in Section~\ref{sec:data} in which
no censoring is present, the baseline hazard function
is constant, and the distribution of the covariates
is fixed.
More precisely, suppose that
$T_\mathrm{C} = 1$
almost surely, so that $I = 0$ and $T = T_\mathrm{S}$,
and assume that \eqref{eq:hazard} holds
with $\lambda_0(t) = 1$ for all $t \in [0, 1]$.
Since the baseline hazard function $\lambda_0$ is
fixed, we may drop the identifiability condition that
$P_X(f_0) = 0$.
Assume without further loss of generality that
$K \coloneqq \sup_{x \in \cX} k(x, x) = 1$.  Recalling that either
$\mathcal{R}$ is finite or $\mathcal{R} = \mathbb{N}$, we suppose in
the latter case that
for all $r \in \cR$, we have
\begin{inlineroman}
  \item
    $\sum_{r' \in \cR, r' > r} \nu_{r'} \leq C_K r \nu_r$; and
  \item
    $\nu_{r+1} \geq \nu_r / C_K$,
\end{inlineroman}
for some $C_K > 0$.  These conditions
are satisfied by all commonly used kernels,
including first- and second-order Sobolev
kernels, as well as kernels with exponentially decaying eigenvalues;
see \citet[Section~2]{yang2017randomized} and
\citet[Section~3]{ma2023optimally}.
Let $\cH_1 \coloneqq \{f \in \cH: \|f\|_\cH \leq 1\}$ and write
$\cF_n$ for the set of measurable functions from
$\bigl(\cX \times [0, 1]\bigr)^n$ to $\cH_1$.

\begin{theorem}[Minimax lower bound]%
  \label{thm:minimax}
  There exists $C_0 > 0$
  depending only on $\nu_1$,
  $R$ (if $\cR = [R]$) and
  $C_K$ (if $\cR = \N$) such that
  \begin{align*}
    \inf_{\hat f_n \in \cF_n}
    \sup_{f_0 \in \cH_1}
    \E\bigl(
      \|\hat f_n - f_0\|_{L_2}^2
    \bigr)^{1/2}
    &\geq
    \frac{1}{C_0}
    \inf_{\gamma > 0}
    \biggl(
      \sqrt{\frac{H_{2,\gamma}}{n}}
      + \sqrt{\gamma}
    \biggr).
  \end{align*}
\end{theorem}

The lower bound in Theorem~\ref{thm:minimax} matches
the leading bias and variance terms in the upper bound obtained
in Theorem~\ref{thm:rate}, thereby establishing
rate-optimality of our procedure.

\subsection{Parameter tuning}%
\label{sec:parameter_tuning}

To state our main result on the RKHS estimator with the
cross-validated choice of $\gamma$,
define $\gamma^+ \vcentcolon= \max \Gamma$ and
$\gamma^- \vcentcolon= \min \Gamma$, and
suppose that $\Gamma \subseteq [n^{-\xi}, n^\xi]$ for some $\xi \geq 1$.

\begin{theorem}[Parameter tuning]%
  \label{thm:parameter_tuning}
  Let $\hat\gamma$ be as in \eqref{eq:gamma_hat}.
  There exists $C_0 > 0$ depending only on
  $\|f_0\|_\cH$, $\Lambda$, $q_1$, $\|1_\cX\|_\cH$, $K$ and $\xi$
  such that if $\gamma^+ H_{\infty,\gamma^+} \leq 1 / C_0$,
  $H_{\infty,\gamma^-} \sqrt{(\log n) / n} \leq 1/C_0$
  and $1 \leq s \leq \sqrt{n} / (C_0 H_{\infty,\gamma^-})$,
  then with probability at least $1 - e^{-s^2}$,
  \begin{equation*}
    \bigl\|\hat f_{n,\hat\gamma} - f_0\bigr\|_{L_2}
    \leq
    C_0
    \biggl(
      \min_{\gamma \in \Gamma}
      \bigl\|\hat f_{n,\gamma} - f_0\bigr\|_{L_2}
      + \frac{s + \sqrt {\log n}}{\sqrt n}
    \biggr).
  \end{equation*}
\end{theorem}

The cross-validated choice
$\hat\gamma$ therefore achieves the same convergence rate
as the oracle choice of $\gamma$, up to a constant factor
and an additive term that is at most $\sqrt{(\log n) / n}$.
Applying Theorem~\ref{thm:rate}, and
under the conditions specified in
Theorem~\ref{thm:parameter_tuning}, we deduce that
with probability at least $1 - e^{-s^2}$,
\begin{equation*}
  \bigl\|\hat f_{n,\hat\gamma} - f_0\bigr\|_{L_2}
  \leq
  C_0
  \Biggl(
    \sqrt{\frac{H_{2,\gamma^{\star}}}{n}}
    + \sqrt{\gamma^{\star}}
    + \frac{s + \sqrt {\log n}}{\sqrt n}
  \Biggr).
\end{equation*}
Whenever $\gamma^a H_{2,\gamma} \asymp 1$ as $\gamma \searrow 0$
for some $a \in [0, 1)$, the optimal rate according to
Theorem~\ref{thm:rate} and Theorem~\ref{thm:minimax}
is $n^{-1/(2 + 2a)}$.
If $a > 0$, such as with the first-order Sobolev
kernel in Example~\ref{ex:sobolev} where $a = 1/2$,
then $n^{1/(2 + 2a)} \sqrt{(\log n) / n} \to 0$
as $n \to \infty$, so the logarithmic term appearing in
Theorem~\ref{thm:parameter_tuning}
is asymptotically negligible.
In such a setting, we obtain the same rate of convergence
when using $\hat\gamma$ as when using the oracle choice $\gamma^\star$.
When $a = 0$, such as for the
polynomial kernel of Example~\ref{ex:polynomial},
the upper bound on the convergence rate with cross-validation
is worse by a factor of at most $\sqrt{\log n}$
than with the optimal tuning parameter~$\gamma^\star$.

Our cross-validation scheme for $\gamma$ can be extended
to include other parameters,
such as the choice of kernel, as well as hyperparameters
including bandwidths, polynomial degrees and shape parameters.
While these may each induce a different RKHS,
Theorem~\ref{thm:parameter_tuning} is stated with respect to the
$L_2(P_X)$-norm rather than any kernel-dependent metrics and can
therefore be applied to arbitrary collections
of estimators in $L_2(P_X)$.
See Appendix~\ref{sec:kernel_selection}%
\iftoggle{journal}{ in the supplementary material%
\if1\anon\ \citep{underwood2026caresupplement}\fi}{}
for more details.

\subsection{Model aggregation with CARE}
\label{sec:model_selection}

Our final main theoretical result
gives an upper bound on the rate of convergence
of the CARE estimator $\check f_{n,\check\gamma,\check\theta}$.

\begin{theorem}[CARE procedure]%
  \label{thm:model_selection}
  Let $(\check\gamma,\check\theta)$ be the cross-validated choices
  as defined in~\eqref{eq:gamma_theta_check}.
  There exists $C_0 > 0$ depending only on
  $\|f_0\|_\cH$, $\Lambda$, $q_1$,
  $\|1_\cX\|_\cH$, $M$, $K$ and $\xi$
  such that if
  $\gamma^+ H_{\infty,\gamma^+} \leq 1 / C_0$,
  $H_{\infty,\gamma^-} \sqrt{(\log n) / n} \leq 1/C_0$
  and $1 \leq s \leq \sqrt{n} / (C_0 H_{\infty,\gamma^-})$,
  then with probability at least $1 - e^{-s^2}$,
  \begin{align*}
    \bigl\|
    \check f_{n,\check\gamma,\check\theta} - f_0
    \bigr\|_{L_2}
    &\leq
    C_0
    \biggl(
      \min_{(\gamma,\theta) \in \Gamma \times \Theta}
      \bigl\|\check f_{n,\gamma,\theta} - f_0\bigr\|_{L_2}
      + \frac{s + \sqrt{|\cM| + \log n}}{\sqrt n}
    \biggr).
  \end{align*}
\end{theorem}

The cross-validated CARE estimator
$\check f_{n,\check\gamma,\check\theta}$ therefore attains a
convergence rate under $L_2(P_X)$ comparable
to that of the oracle choices of
$\gamma \in \Gamma$ and $\theta \in \Theta$.
In particular, $\check\gamma$ may depend on properties of
the external estimators $\hat f_m$.

Denote the index of the best external estimator by
$m^\star \vcentcolon= \sargmin_{m \in \cM}
\|\tilde f_m - f_0\|_{L_2}$.
Then under the conditions of Theorem~\ref{thm:model_selection}, by
combining this result with Theorem~\ref{thm:rate}, we have with
probability at least $1 - e^{-s^2}$ that
\begin{align}
  \label{eq:sequential}
  \bigl\|\check f_{n,\check\gamma,\check\theta}
  &- f_0\bigr\|_{L_2}
  \leq
  C_0
  \Biggl(
    \Biggl\{
      \sqrt{\frac{H_{2,\gamma^\star}}{n}}
      + \sqrt{\gamma^\star}
    \Biggr\}
    \land
    \bigl\|\tilde f_{m^\star} - f_0\bigr\|_{L_2}
    + \frac{s + \sqrt {|\cM| + \log n}}
    {\sqrt n}
  \Biggr).
\end{align}

Therefore, $\check f_{n,\check\gamma,\check\theta}$ attains a rate of
convergence that is at least as good as both
\begin{inlineroman}
  \item the optimal kernel estimator $\hat f_{n,\gamma^\star}$, and
  \item the best empirically centred external
    estimator $\tilde f_{m^\star}$.
\end{inlineroman}
In this sense, CARE is able to automatically adapt to the
predictive values of the external models
$\tilde f_m$ relative to those of the kernel estimators
$\hat f_{n,\gamma}$.
In particular, the performance of the CARE procedure is
not significantly degraded when some of the external estimators
may be poor predictors, for
example, if they are misspecified,
or if they were originally fitted to a small data set.
When the kernel estimator is well-specified
(a relatively mild assumption in practice when
using a nonparametric infinite-dimensional RKHS),
the CARE procedure may be thought of
as a double regularisation, applying shrinkage towards the origin
(via $\gamma$)
and also towards the external estimators (via $\theta$).

Our method can accommodate certain forms of covariate shift
with no modification; this is important when the
external estimators may have been trained on samples from
populations that differ from that of the data
$(X_i, T_i, I_i)_{i \in [2n]}$.
To this end, suppose that
$\|\tilde f_{\tilde m} - f_0\|_{\tilde L_2} \leq \delta$
for some $\tilde m \in \cM$, where
$\tilde L_2 \vcentcolon= L_2(\tilde P_X)$
and $\tilde P_X$ is a distribution on $\cX$
with Radon--Nikodym derivative
$\diffd P_X / \diffd \tilde P_X \leq C^2$ almost surely for some
non-random $C \in (0,\infty)$.
Then $\|\tilde f_{\tilde m} - f_0\|_{L_2} \leq C \delta$
and the high-probability bound given in \eqref{eq:sequential}
holds with $\|\tilde f_{m^\star} - f_0\|_{L_2}$ replaced by $C \delta$.

In practice, an estimate of the survival probability
$\P(T_\mathrm{S} > t \mid X = x)$ is often required,
where $t$ is a fixed time and $x$ is a specified covariate value.
Given an estimator $f$ of the relative risk function
(which may be obtained using the methods
in Section~\ref{sec:model_selection}),
we propose the plug-in estimator
\begin{equation}
  \label{eq:survival_prob}
  \hat P(t \mid x, f)
  \vcentcolon=
  \hat P(t \mid f)^{\exp\{f(x)\}},
\end{equation}
where $\hat P(t \mid f) \vcentcolon=
\exp \bigl( - \int_0^t \{1 / S_n(f, s)\} \diffi N(s) \bigr)$.

\section{Implementation details}
\label{sec:implementation}

In this section, we present some further details on the
practical implementation of our CARE methodology.
Specifically, we describe procedures for calculating
$\|1_\cX\|_\cH$, computing the set $\cA_n$,
solving the optimisation problem \eqref{eq:argmin}
and conducting cross-validation as
detailed in \eqref{eq:gamma_hat} and \eqref{eq:gamma_theta_check}.
We also provide guidance on the choice of kernel.
In Appendix~\ref{sec:feature_map}%
\iftoggle{journal}{ in the supplementary material%
\if1\anon\ \citep{underwood2026caresupplement}\fi}{},
we provide an alternative `feature map' formulation
for superior computational efficiency when using certain kernels,
while in Appendix~\ref{sec:nystrom}%
\iftoggle{journal}{ in the supplementary material%
\if1\anon\ \citep{underwood2026caresupplement}\fi}{},
we demonstrate how Nystr{\"o}m approximation can
be used to attain better computational
performance with general kernels.

\subsection{Calculating \texorpdfstring{$\|1_\cX\|_\cH$}{||1X||}}

In Proposition~\ref{prop:constant_norm} below, we give the exact general form
of $\|1_\cX\|_\cH$ based on properties of the kernel.
For $n \in \N$ and $\bA \in \R^{n \times n}$
with $\bA \succeq 0$, define
\begin{align*}
  \kappa(\bA)
  &\vcentcolon=
  \lim_{\delta \searrow 0}
  \frac{1}{1_n^\T (\bA + \delta \bI_n)^{-1} 1_n}
  =
  \begin{cases}
    0 & \text{if there is } v \in \R^n \text{ with }
    \bA v = 0,
    1_n^\T v = 1, \\
    \frac{1}{1_n^\T \bA^+ 1_n} & \text{otherwise},
  \end{cases}
\end{align*}
where $\bA^+$ is the Moore--Penrose pseudo-inverse of $\bA$
\citep[e.g.,][Section~6.7]{friedberg2003linear}.

\begin{proposition}[Characterisation of $\|1_\cX\|_\cH$]%
  \label{prop:constant_norm}
  For $n \in \N$ and
  $x_1, \ldots, x_n \in \cX$, define the matrix
  $\bK(x_1, \ldots, x_n) \in \R^{n \times n}$
  by $\bK(x_1, \ldots, x_n)_{i j} \vcentcolon= k(x_i, x_j)$
  for $i, j \in [n]$.
  Then $1_\cX \in \cH$ if and only if
  $\kappa_k \vcentcolon=
  \inf_{n \in \N} \inf_{x_1, \ldots, x_n \in \cX}
  \kappa \bigl(\bK(x_1, \ldots, x_n)\bigr) > 0$,
  in which case $\|1_\cX\|_\cH^2 = 1 / \kappa_k$.
\end{proposition}

See Examples~\ref{ex:polynomial}, \ref{ex:sobolev} and
\ref{ex:gaussian} for explicit calculations of
$\|1_\cX\|_\cH$ using Proposition~\ref{prop:constant_norm}.

\subsection{Constructing the set \texorpdfstring{$\cA_n$}{An}}

Let $\bK \in \R^{n\times n}$ be given by
$\bK_{i j} \vcentcolon= k(X_i,X_j)$,
and define the block matrix
\begin{equation}
  \label{eq:K_tilde}
  \tilde{\bK}
  \vcentcolon=
  \begin{pmatrix}
    \|1_\cX\|_\cH^2 & 1_n^\T \\
    1_n & \bK
  \end{pmatrix}
  \in \R^{(n+1) \times (n+1)}.
\end{equation}
Let $\bR \in \R^{(n+1) \times (n+1)}$
be the reduced row echelon form of $\tilde{\bK}$,
obtained by applying Gaussian elimination
\citep[e.g.][Theorem~3.14]{friedberg2003linear}
to $\tilde{\bK}$,
and define
$\cA_n \vcentcolon= \bigl\{j \in [n]: \exists i \in [n+1]:
  \bigl(\bR_{i,1}, \ldots, \bR_{i,j+1}\bigr)
= (0, \ldots, 0, 1)\bigr\}$,
that is, the one-decremented indices of the columns in $\bR$
that contain a leading one.
Then by \citet[Theorem~3.16]{friedberg2003linear},
$\bigl\{ \bigl(
    \tilde{\bK}_{i+1, 1}, \ldots,
    \tilde{\bK}_{i+1, n+1}
\bigr) : i \in \{0\} \cup \cA_n \bigr\}$
is a basis for
$\mathrm{span} \bigl\{ \bigl(
    \tilde{\bK}_{i, 1}, \ldots,
    \tilde{\bK}_{i, n+1}
\bigr) : i \in [n+1] \bigr\}$,
and hence $\cA_n$ satisfies the desired properties.

\subsection{Computing the estimator \texorpdfstring{$\hat f_{n,\gamma}$}{f}}

Proposition~\ref{prop:representation} shows that for $\gamma > 0$,
the kernel estimator $\hat f_{n,\gamma}$ is the solution to
a strongly convex optimisation problem.  Here we give the
explicit form of this problem and discuss iterative methods for its solution.
We use notation from Section~\ref{sec:estimator} throughout
and, for $x, y \in \cX$, we define
$\hat k(x, y) \vcentcolon= k(x, y) - \bar k(x) - \bar k(y)
+ \bar k(x) \bar k(y) \|1_\cX\|_\cH^2$.
First, for $\beta \in \R^{\cA_n}$,
we have $P_n(f_\beta) = 0$, so
the objective function in \eqref{eq:argmin} can be written as
\begin{equation*}
  \ell_{n,\gamma}(f_\beta)
  = \frac{1}{n}
  \sum_{i=1}^n
  \log
  S_n(f_\beta, T_i)
  N_i(1)
  - \frac{1}{n} \sum_{i=1}^n
  f_\beta(X_i) N_i(1)
  + \gamma \sum_{i \in \cA_n} \sum_{j \in \cA_n}
  \hat k(X_i, X_j) \beta_i \beta_j.
\end{equation*}
For $j \in \cA_n$, define
$\partial_j \ell_{n,\gamma}(f_\beta)
\vcentcolon= \partial \ell_{n,\gamma}(f_\beta) / \partial \beta_j$.
The gradient of the objective function is
\begin{align*}
  \partial_j \ell_{n,\gamma}(f_\beta)
  &=
  \frac{1}{n}
  \sum_{i=1}^n
  \frac{
    D S_n(f_\beta, T_i)\bigl(\tilde k(\cdot, X_j)\bigr)
  }{S_n(f_\beta, T_i)}
  N_i(1)
  -
  \frac{1}{n} \sum_{i=1}^n
  \tilde k(X_i, X_j) N_i(1)
  + 2 \gamma \sum_{i \in \cA_n} \hat k(X_i, X_j) \beta_i.
\end{align*}

Since this optimisation problem is strongly convex by
Proposition~\ref{prop:representation}, a variety of methods can be
used for approximate computation of
the optimal solution $\hat\beta$.
One first-order approach is to run steepest gradient descent
with the Armijo backtracking step size rule;
this is guaranteed to converge as the number of iterations
diverges to infinity \citep[Section~12.6]{lange2013optimization}.
However, computational speed-ups can be achieved by incorporating
second-order quasi-Newton methods
such as the Broyden--Fletcher--Goldfarb--Shanno (BFGS) algorithm
\citep{broyden1970convergence,fletcher1970new,%
goldfarb1970family,shanno1970conditioning};
this is how we implemented our procedure.

\subsection{Cross-validation of
\texorpdfstring{$\gamma$}{gamma} and \texorpdfstring{$\theta$}{theta}}
\label{sec:cv}

To compute $\hat\gamma$ using \eqref{eq:gamma_hat}, we suggest
taking $\Gamma$ to be a finite set of geometrically spaced values
and applying a grid search procedure,
computing $\ell_n(\hat f_{n,\gamma})$ for each $\gamma \in \Gamma$.
Our implementation starts with $\gamma = \max \Gamma$, initialising
the optimisation algorithm at zero. The remaining values of $\gamma$
may be evaluated in decreasing order, initialising the algorithm
each time with a warm start at the previous optimal solution
$\hat\beta_{n,\gamma}$.

The computational burden of solving \eqref{eq:gamma_theta_check}
for the CARE estimator
is higher than that of~\eqref{eq:gamma_hat},
requiring a joint optimisation over $(\gamma, \theta) \in \Gamma \times \Theta$.
However, for each $\gamma \in \Gamma$,
the expensive step of computing $\hat f_{n,\gamma}$
only needs to be performed once.
In practice we therefore recommend taking
$\Theta$ to be an approximately evenly spaced set of points in $\Delta_\cM$
and using the following procedure.
For each $\gamma \in \Gamma$, compute $\hat f_{n,\gamma}$ as before,
then search over $\theta \in \Theta$ to find
$\check \theta(\gamma) \vcentcolon=
\sargmin_{\theta \in \Theta} \tilde \ell_n(\check f_{n,\gamma,\theta})$.
It remains to search over $\gamma \in \Gamma$
to find $\check \gamma \vcentcolon= \sargmin_{\gamma \in \Gamma}
\tilde \ell_n(\check f_{n,\gamma,\check\theta(\gamma)})$;
finally set $\check \theta \vcentcolon= \check \theta(\check \gamma)$.

\subsection{Choice of kernel}
\label{sec:choice_kernel}

The choice of kernel in applications depends on the
nature of the covariates under consideration.
For continuous features, the most flexible kernel
functions are those that induce large, fully nonparametric
Hilbert spaces. For example, the RKHS generated by
the first-order Sobolev kernel
(Example~\ref{ex:sobolev}) contains
all absolutely continuous functions
$f: [0, 1] \to \R$ with weak derivative $f' \in L_2([0, 1])$
\citep[Example~6.16]{samworth2026modern}.
Somewhat smaller is the RKHS arising from the
second-order Sobolev kernel, which only includes functions
possessing absolutely continuous first derivative
$f'$ and weak second derivative $f'' \in L_2([0, 1])$
\citep[Example~6.17]{samworth2026modern}.
The RKHS generated by a
Gaussian kernel on $\cX \subseteq \R^d$
(Example~\ref{ex:gaussian}) is dense in the space of
continuous functions on $\cX$ equipped with the topology
induced by uniform convergence on compact sets
\cite[Theorem~17]{micchelli2006universal}.
For such nonparametric kernels, we typically have
$H_{2,\gamma} \to \infty$ as $\gamma \searrow 0$
(Example~\ref{ex:sobolev}); appropriate regularisation
is therefore crucial to prevent overfitting (see Theorem~\ref{thm:rate}).

At the other end of the spectrum, polynomial kernels
with fixed degree induce finite-dimensional Hilbert spaces and
$H_{\infty,\gamma}$ is bounded as $\gamma \searrow 0$
(Example~\ref{ex:polynomial}).
While this property limits the approximation power of the resulting
kernel estimators, they are also less prone to overfitting.
Further, the existence of a finite-dimensional feature map
permits significant computational improvements in low-dimensional
settings; see Appendix~\ref{sec:feature_map}%
\iftoggle{journal}{ in the supplementary material%
\if1\anon\ \citep{underwood2026caresupplement}\fi}{}.

When working with a diverse collection of covariates,
we construct multivariate kernels
by combining lower-dimensional kernel functions as follows.
If $k_j$ is a kernel on $\cX_j$ for
$j \in [d]$, then with $\cX \vcentcolon= \prod_{j=1}^d \cX_j$,
the function $k_{\mathrm{sum}}: \cX \times \cX \to \R$ defined by
\begin{equation*}
  k_{\mathrm{sum}}(x, y)
  \vcentcolon=
  \sum_{j=1}^d k_j(x_j, y_j)
\end{equation*}
is a kernel \citep[Proposition~6.28]{samworth2026modern}.
This formulation is especially appropriate
for an additive model structure of the form
$f_0(x) = \sum_{j=1}^d f_j(x_j)$.

\section{Empirical studies}
\label{sec:empirical}

We investigate the empirical performance of our procedures
through experiments with simulated and real data sets.
Our methodology is implemented in the
Python package \texttt{care-survival}%
\if1\anon\ \citep{underwood2025caresurvival}\fi.

\subsection{Simulated data study}
\label{sec:simulated}

To simulate data from the model specified in \eqref{eq:hazard},
we construct
$T_{\mathrm S}\vcentcolon= \Lambda^{-1}\bigl(- e^{-f_0(X)} \log U\bigr)$
where $U \sim \cU[0, 1]$ and
$\Lambda(t) \vcentcolon= \int_0^t \lambda_0(s) \diffi s$.
Since estimation of the baseline hazard is not the focus of this paper,
we restrict to relatively simple forms for~$\lambda_0$, for which
the function $\Lambda^{-1}$ can be computed in closed form.
We define the Breslow estimator
(\citealp{breslow1972contribution}; see also \citealp{huang2006note})
of the survival function $t \mapsto \P(T_\mathrm{S} > t)$ by
\begin{equation}
  \label{eq:breslow}
  \hat P(t)
  \vcentcolon=
  \exp \biggl(
    - \int_0^t \frac{n}{\sum_{i=1}^n R_i(s)}
    \diffi N(s)
  \biggr).
\end{equation}

\subsubsection{Univariate regressors}
\label{sec:simulation_univariate}

We begin by considering a setting where $X$ is univariate, not
because this is realistic for modern prediction models in
biomedicine, but because it allows us visualise our methodology by
plotting the survival times and the
fitted relative risk function.  Let $P_X = \cU[0, 1]$ and
$\lambda_0(t) = 6$ for $t \in [0, 1]$,
so that $\Lambda(t) = 6 t$ and $\Lambda^{-1}(t) = t / 6$.
We take
$f_0(x) \vcentcolon= 2 \sin(2 x) - 2 \sin^2(1)$
for $x \in [0,1]$, so that $P_X(f_0) = 0$, and
sample $T_{\mathrm C} \sim \cU[1/5, 2] \land 1$ independent of $X$.
Figure~\ref{fig:data} illustrates a sample of size
$n = 200$ from this distribution, along with
the Breslow estimator \eqref{eq:breslow}.
We use the shifted first-order Sobolev
kernel $k_{1,a}^\Sob$ defined in Example~\ref{ex:sobolev}
with $a = 1$.  By the remark in Section~\ref{sec:choice_kernel}, we
have $f_0 \in \cH$.
Since $f_0$ is not a polynomial, using
a shifted polynomial kernel as in Example~\ref{ex:polynomial}
would result in misspecification bias.
We take $\Gamma$ to consist of $50$ geometrically increasing values from
$10^{-5}$ to $10$.

\begin{figure}[H]
  \centering
  \subfloat[Survival time against covariate value.]{
    \includegraphics[scale=\myfigscale]{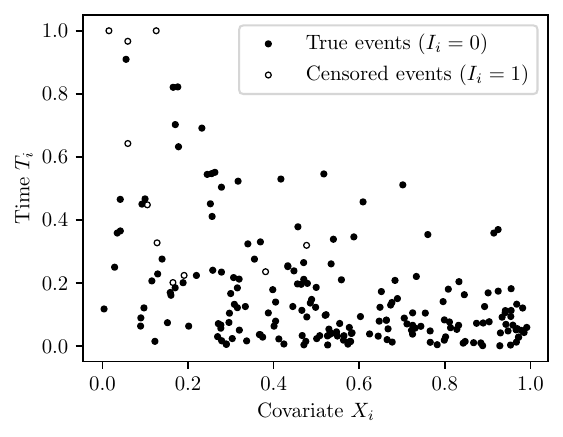}
  }
  \hspace*{5mm}
  \subfloat[Breslow survival estimator.]{
    \includegraphics[scale=\myfigscale]{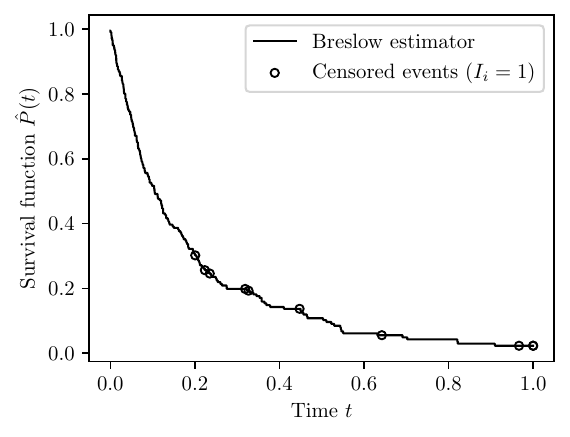}
  }
  \caption{In panel~(a), we plot a sample of size $n = 200$
    from the specified data distribution.
    In panel~(b), we display the associated Breslow estimator
  of the survival function.}
  \label{fig:data}
\end{figure}

Figure~\ref{fig:validation}(a) illustrates
our proposed cross-validation scheme for choosing the
regularisation parameter $\gamma$ based on a sample
of size $n = 200$. The negative partial
log-likelihood $\ell_n(\hat{f}_{n,\gamma})$ on the training data is
an increasing function of
$\gamma$, reflecting overfitting for small values of~$\gamma$. However,
the corresponding negative partial log-likelihood
$\tilde{\ell}_n(\hat{f}_{n,\gamma})$ on the validation
data attains a global minimum at $\gamma = \hat\gamma$.
In Figure~\ref{fig:validation}(b),
the resulting cross-validated estimator
$\hat f_{n,\hat\gamma}$ is seen to approximate
the true relative risk function $f_0$ reasonably well
when $n = 200$.

\begin{figure}[H]
  \centering
  \subfloat[Partial likelihood cross-validation curves.]{
    \includegraphics[scale=\myfigscale]{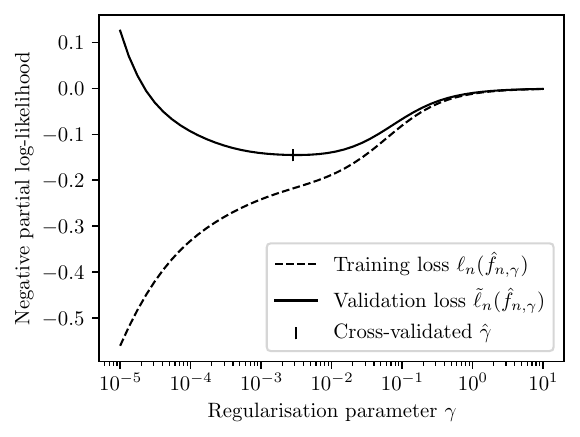}
  }
  \hspace*{5mm}
  \subfloat[Fitted estimator and true risk function.]{
    \includegraphics[scale=\myfigscale]{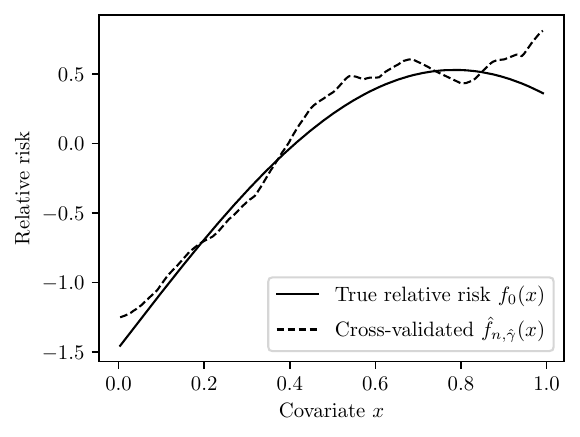}
  }
  \caption{In panel~(a), we show how
    cross-validation is conducted by minimising the negative partial
    log-likelihood on an independent sample of size $n=200$.
    In panel~(b), we illustrate how
  the resulting estimator approximates the true relative risk function.}
  \label{fig:validation}
\end{figure}

Next, we investigate the behaviour of the cross-validated RKHS estimator
as we vary the sample size $n$.
For computational efficiency reasons, we use Nystr{\"o}m
approximation with subsampling number $m \coloneqq n \land 100$
(see Appendix~\ref{sec:nystrom}%
  \iftoggle{journal}{ in the supplementary material%
  \if1\anon\ \citep{underwood2026caresupplement}\fi}{}
for details).
In all of the subsequent figures, results are averaged over
$\nrep = 200$ repetitions.
More precisely, for each value presented,
the mean $\hat\mu$ and standard deviation~$\hat\sigma$ are
calculated across $\nrep$ repetitions;
we then plot the point estimates~$\hat\mu$ in black,
and display $[\hat\mu \pm 2\hat\sigma / \sqrt{\nrep}]$ as grey shaded regions.

In Figure~\ref{fig:rkhs_dgp_1}(a),
we illustrate the dependence of the cross-validated and oracle
regularisation parameters on the sample size.
The oracle choice $\gamma^\star$
is calculated by directly minimising
$\gamma \mapsto \|\hat f_{n,\gamma} - f_0\|_{L_2}$,
which we approximate by Monte Carlo integration using an independent
sample of $500$ data points.
As the sample size increases,~$\gamma^\star$ decreases, reflecting the
bias--variance trade-off suggested by
Theorems~\ref{thm:rate} and~\ref{thm:minimax}.
Further, the cross-validated choice $\hat\gamma$ provides a good approximation
to $\gamma^\star$, even for small sample sizes.
In Figure~\ref{fig:rkhs_dgp_1}(b), we see that the
$L_2$-error $\|\hat f_{n,\hat{\gamma}} - f_0\|_{L_2}$ of the
cross-validated estimator
decreases as the sample size increases, and
it is competitive with the oracle estimator
(which requires knowledge of~$f_0$).

\begin{figure}[H]
  \centering
  \subfloat[Cross-validated and oracle choices of $\gamma$.]{
    \includegraphics[scale=\myfigscale]{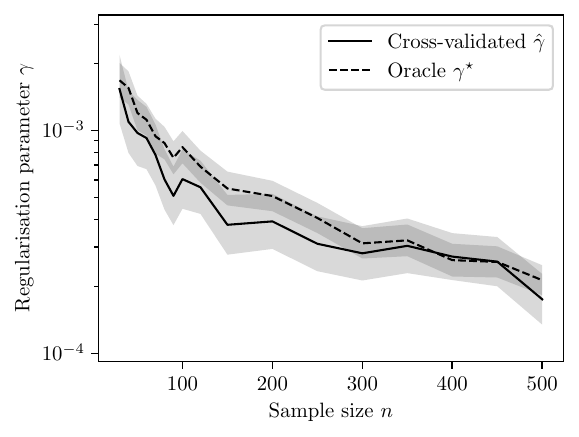}
  }
  \hspace*{5mm}
  \subfloat[Cross-validated and oracle $L_2$-error.]{
    \includegraphics[scale=\myfigscale]{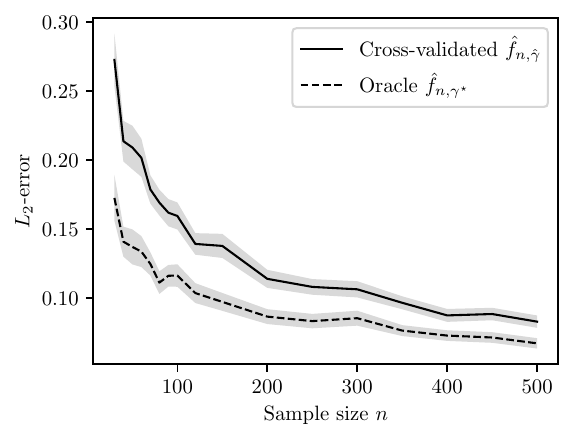}
  }
  \caption{In panel~(a), we plot the cross-validated and oracle
    choices of regularisation parameter, as a function of the sample
    size. In panel~(b), we plot the corresponding $L_2$-errors of
  the fitted prediction functions.}
  \label{fig:rkhs_dgp_1}
\end{figure}

We now assess the performance of the aggregated CARE estimator
$\check f_{n,\check\gamma,\check\theta}$.
We define a fixed `external' estimator $\tilde f$ by
perturbing the function $f_0$, setting
$\tilde f(x) \vcentcolon=
2 \sin(3x/2) - (8/3) \sin^2 (3/4)$.
The convex combination parameter $\theta$ is selected from the set
$\Theta \vcentcolon= \{(k-1)/20: k \in [21]\}$.
In Figure~\ref{fig:aggregation_dgp_1}(a), we see that
the oracle choice $\theta^\dagger$,
where $(\gamma^\dagger,\theta^\dagger)$ minimises
$(\gamma, \theta) \mapsto \|\check f_{n,\gamma,\theta} - f_0\|_{L_2}$,
decreases as the sample size increases. For small sample
sizes,~$\tilde f$ provides valuable
information due to the large variance of the kernel
estimator,
and hence~$\theta^\dagger$ is relatively large.
However, with more new data, the kernel
estimators $\hat f_{n,\gamma}$ approximate $f_0$ well
and the external model $\tilde f$ becomes less informative,
yielding a value of $\theta^\dagger$ that is closer to zero.
The cross-validation scheme produces estimates $\check\theta$
that are reasonably close to the oracle choices $\theta^\dagger$.
From Figure~\ref{fig:aggregation_dgp_1}(b), we also observe that
CARE consistently outperforms
the cross-validated kernel estimator that ignores the external
estimator, although in this smooth, one-dimensional problem the effect is
smaller at larger sample sizes.
When the sample size is sufficiently large,
the CARE estimator is also seen to be superior
to the external estimator $\tilde f$.

\begin{figure}[H]
  \centering
  \subfloat[Cross-validation and oracle choices of $\theta$.]{
    \includegraphics[scale=\myfigscale]{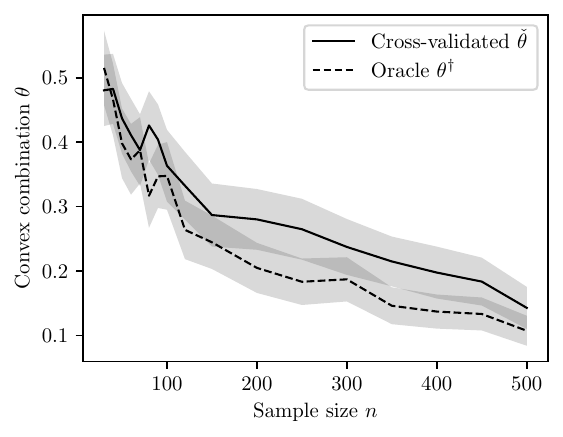}
  }
  \hspace*{5mm}
  \subfloat[Cross-validated and oracle performance.]{
    \includegraphics[scale=\myfigscale]{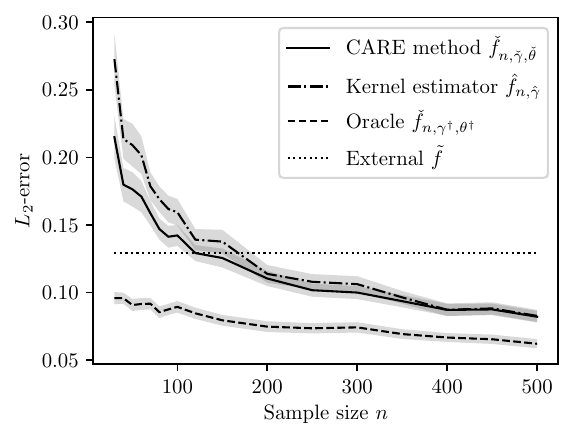}
  }
  \caption{In panel~(a), we plot the cross-validated and oracle
    choices of the convex combination parameter as a function of the
    sample size. In panel~(b), we compare the corresponding
  CARE, kernel, oracle and external estimators.}
  \label{fig:aggregation_dgp_1}
\end{figure}

\subsubsection{Multivariate regressors}
\label{sec:simulation_multivariate}

We now turn to a setting with multiple covariates. With $d = 10$,
let $P_X = \cU[0, 1]^d$ and
$\lambda_0(t) = 6$ as before.
We consider the sparsely supported
relative risk function $f_0$ given by
$f_0(x) = \sum_{j=1}^5 \{2\sin(2 x_j) - 2 \sin^2(1)\}$
for $x = (x_1,\ldots,x_d)^\T \in [0,1]^d$, which depends only on
the first five components of $x$.
We draw $T_{\mathrm C} \sim \cU[1/5, 2] \land 1$ independently of
$X$, and employ
an additive shifted first-order Sobolev kernel $k$ defined by
$k(x, y) \vcentcolon= a + \sum_{j=1}^d (x_j \land y_j)$ with $a = 1$.
We take $\Gamma$ to consist of $50$ geometrically increasing values from
$10^{-7}$ to $0.1$.

In Figure~\ref{fig:rkhs_dgp_2}(a),
we illustrate the dependence of the cross-validated and oracle
regularisation parameters on the sample size.
As before, $\gamma^\star$ decreases as the sample size increases.
The cross-validated choice $\hat\gamma$ provides
a reasonable approximation to $\gamma^\star$, although it
appears to exhibit some upward bias.
In Figure~\ref{fig:rkhs_dgp_2}(b), the
$L_2$-error of the cross-validated estimator
decreases as the sample size increases, and is competitive with the
oracle estimator.

\begin{figure}[H]
  \centering
  \subfloat[Cross-validated and oracle choices of $\gamma$.]{
    \includegraphics[scale=\myfigscale]{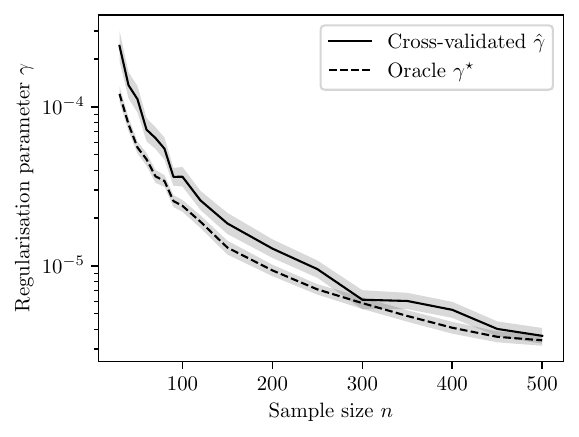}
  }
  \hspace*{5mm}
  \subfloat[Cross-validated and oracle performance.]{
    \includegraphics[scale=\myfigscale]{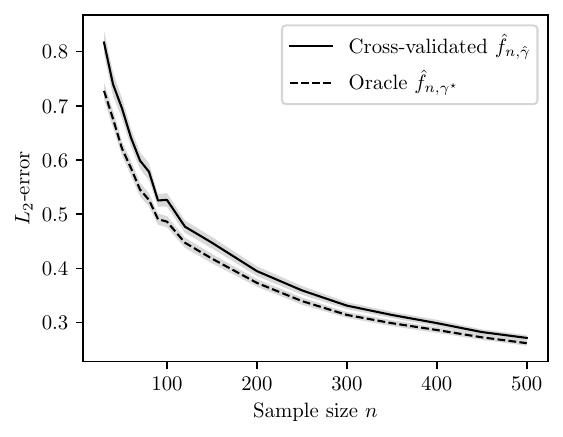}
  }
  \caption{In panel~(a), we plot the cross-validated and oracle
    choices of regularisation parameter, as a function of the sample
    size. In panel~(b), we plot the corresponding $L_2$-errors of
  the fitted prediction functions.}
  \label{fig:rkhs_dgp_2}
\end{figure}

Finally, we investigate the CARE estimator.
For the `external' estimator $\tilde f$, we retain the additive
structure, defining
$\tilde f(x) \vcentcolon=
\sum_{j=1}^4 (\sin 2 - \cos 2 - 1) (6x_j - 3)$.
That is, for the first four entries of $x$, we
take the optimal linear approximation to $f_0$ under
the $L_2(P_X)$-norm.
In Figure~\ref{fig:aggregation_dgp_2}(a), we see that
again the oracle choice $\theta^\dagger$
decreases as the sample size increases.
For small sample sizes, the linear approximation
$\tilde f$ provides useful
information and hence~$\theta^\dagger$ is relatively large.
With more data, the kernel
estimators $\hat f_{n,\gamma}$ approximate
the non-linear function~$f_0$ well
and reliance on the external model $\tilde f$ diminishes,
yielding $\theta^\dagger$ closer to zero.
The cross-validated estimates $\check\theta$
are reasonably close to the oracle choices $\theta^\dagger$.
In Figure~\ref{fig:aggregation_dgp_2}(b), we also observe that
CARE consistently outperforms the cross-validated kernel estimator that
ignores the external estimator.
For most sample sizes, the CARE estimator is also superior
to the external estimator $\tilde f$.
Some additional simulation results are given in
Appendix~\ref{sec:additional_simulation}%
\iftoggle{journal}{ in the supplementary material%
\if1\anon\ \citep{underwood2026caresupplement}\fi}{}.

\begin{figure}[H]
  \centering
  \subfloat[Cross-validation and oracle choices of $\theta$.]{
    \includegraphics[scale=\myfigscale]{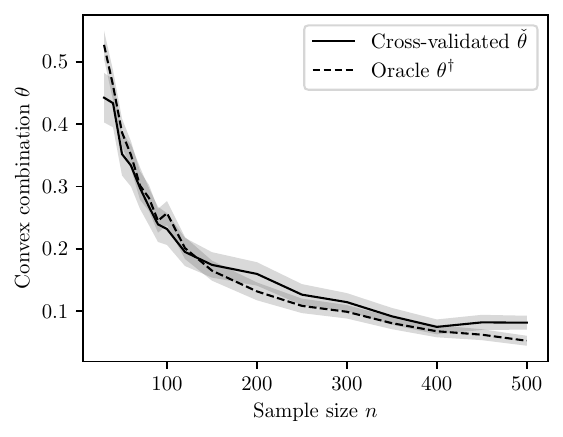}
  }
  \hspace*{5mm}
  \subfloat[Cross-validated and oracle performance.]{
    \includegraphics[scale=\myfigscale]{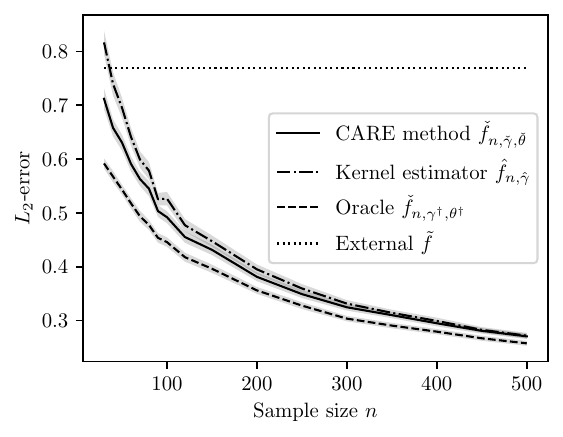}
  }
  \caption{In panel~(a), we plot the cross-validated and oracle
    choices of the convex combination parameter as a function of the
    sample size. In panel~(b), we compare the corresponding
  CARE, kernel, oracle and external estimators.}
  \label{fig:aggregation_dgp_2}
\end{figure}

\subsection{Cardiovascular disease risk prediction}
\label{sec:score2}

We consider the problem of estimating 10-year
fatal and non-fatal cardiovascular disease risk in England,
using cohort data taken from the UK Biobank (UKB, \citealp{littlejohns2019uk}).
As an external estimator, we employ a version of the SCORE2 model
\citep{esc2021score2}, trained on a data set of size%
\footnote{Large sample sizes are necessary for effective
  models since diagnoses of cardiovascular disease are relatively rare
  (less than 5\% of the population in any 10-year period, as indicated
by Figure~\ref{fig:data_score2}).}
$173{,}438$ from the Emerging Risk Factor Collaboration
\citep{emerging2007emerging}
that is independent of the UK Biobank data.  This estimator was
fitted using a Fine--Gray model
\citep{fine1999proportional} and,
as such, targets the subdistribution hazard function
rather than the standard hazard function given in \eqref{eq:hazard}.
However, our framework makes no assumptions on the
consistency properties (or lack thereof) of the
external estimators; it suffices that SCORE2 contains
predictive information about the cardiovascular risk.
Under appropriate assumptions, our main methodology and
theory could be adapted for the Fine--Gray setting,
but we prefer the relative risk formulation for two main reasons.
First, it is simpler, avoiding complex modelling assumptions
and unusual risk sets.
Second, the primary motivation of our methodology is to
estimate the impact of a new covariate, treatment or test on
the incidence of an event of interest.
For such explanatory tasks,
it is generally preferable to target the standard hazard function
\eqref{eq:hazard} rather than the subdistribution hazard function
\citep{andersen2012competing}.

SCORE2 uses the following covariates for prediction:
\texttt{age} (at baseline survey),
\texttt{smoking} (binary-valued smoking status),
\texttt{sbp} (systolic blood pressure),
\texttt{tchol} (total cholesterol) and
\texttt{hdl} (high density lipoprotein cholesterol),
along with the two-way interactions of \texttt{age}
with \texttt{smoking}, \texttt{sbp}, \texttt{tchol} and \texttt{hdl}.
We propose to upgrade the SCORE2 model using
the CARE methodology
defined in Section~\ref{sec:care}
by including three new covariates,
also available in the UK Biobank data,
which are known to be associated
with cardiovascular disease risk:
\texttt{imd} (multiple social deprivation index),
\texttt{pgs000018}
(polygenic score associated with coronary artery disease;
see \citealp{inouye2018genomic})
and \texttt{pgs000039}
(polygenic score associated with ischaemic stroke;
see \citealp{abraham2019genomic}).
To be precise, we set $|\cM| = 1$
and the external estimator is taken to be the
fitted SCORE2 relative risk function $\hat f_1$.
Kernel estimators $\hat f_{n,\gamma}$
are fitted to the UK Biobank data, using
an expanded set of covariates,
and the CARE methodology combines these
two sources of information.
Polygenic scores were calculated using data and software
(\texttt{pgsc\_calc}) from the Polygenic Score Catalog
\citep{lambert2021polygenic,lambert2024enhancing}.
Approximately $15\%$ of the individuals in the study had missing
covariates, and these people were removed from the data set prior to analysis.
Estimators for female and male participants were constructed separately,
as for the original SCORE2 model. We used a quadratic polynomial
kernel, allowing for non-linear effects while taking advantage
of the computationally efficient feature map representation;
see Appendix~\ref{sec:feature_map}%
\iftoggle{journal}{ in the supplementary material%
\if1\anon\ \citep{underwood2026caresupplement}\fi}{}.
We also tried using a first-order additive Sobolev kernel
and Nystr{\"o}m approximation but found this did not
improve the performance;
see Appendix~\ref{sec:additional_score2}%
\iftoggle{journal}{ in the supplementary material%
\if1\anon\ \citep{underwood2026caresupplement}\fi}{}.

In Figure~\ref{fig:data_score2}, we plot Breslow estimators \eqref{eq:breslow}
of the survival functions, stratified by sex and \texttt{imd}, with
\texttt{imd} scores at most the median labelled `low'
and scores above the median labelled `high'.
Higher levels of area-specific deprivation
are associated with an increased incidence of cardiovascular disease
and a lower survival probability \citep{conrad2024trends}.

\begin{figure}[H]
  \centering
  \subfloat[Breslow survival estimators for females.]{
    \includegraphics[scale=\myfigscale]{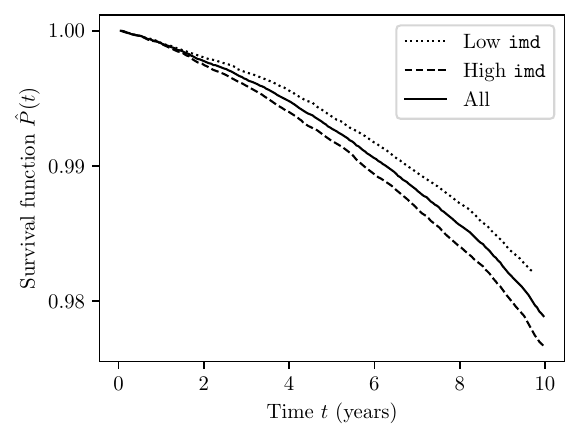}
  }
  \hspace*{5mm}
  \subfloat[Breslow survival estimators for males.]{
    \includegraphics[scale=\myfigscale]{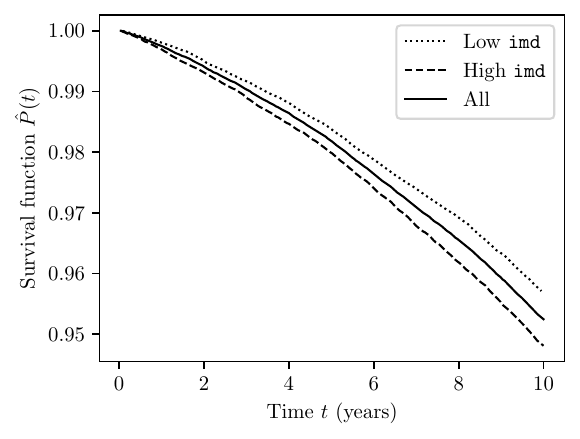}
  }
  \caption{
    We display the Breslow estimators of the survival function.
    Panel~(a) uses data from the female participants
  and panel~(b) is based on the male participants.}
  \label{fig:data_score2}
\end{figure}

\subsubsection{Model 1: upgrading SCORE2 for UKB risk stratification}

We update the SCORE2 model by incorporating data from
UK Biobank.
To fit the cross-validated CARE estimator, we first set aside
$n_{\mathrm{test}} \vcentcolon= \lfloor N / 3 \rfloor$ data points as
a test set, where~$N$ is the total number of UK Biobank data points available
($N = 162{,}682$ for females and $N = 121{,}333$ for males).
For a range of values $n$ satisfying $1 \leq n \leq n_{\mathrm{test}}$,
we then draw~$n$ training points and $n$ validation points
at random from the remaining data, without replacement.
The sampling of training and validation data is repeated
$20$ times with different randomisation seeds
and the results are averaged; uncertainty estimates
are calculated as for the simulated data in
Section~\ref{sec:simulated}. In
Figure~\ref{fig:selection_score2_model_1} we see that,
for both the female and male participants,
the cross-validated choice of the convex combination parameter
tends to decrease as $n$ increases.

\begin{figure}[H]
  \centering
  \subfloat[Cross-validated choices of $\theta$ for females.]{
    \includegraphics[scale=\myfigscale]{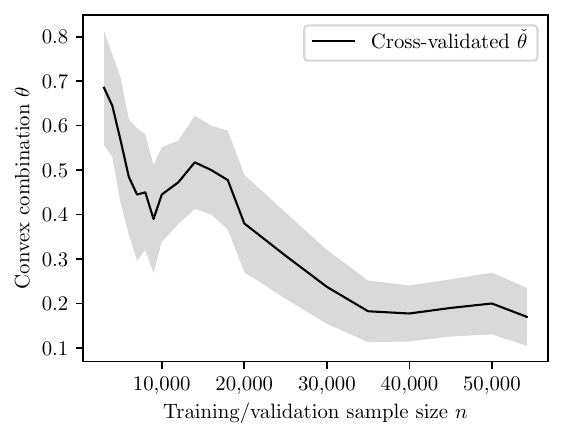}
  }
  \hspace*{5mm}
  \subfloat[Cross-validated choices of $\theta$ for males.]{
    \includegraphics[scale=\myfigscale]{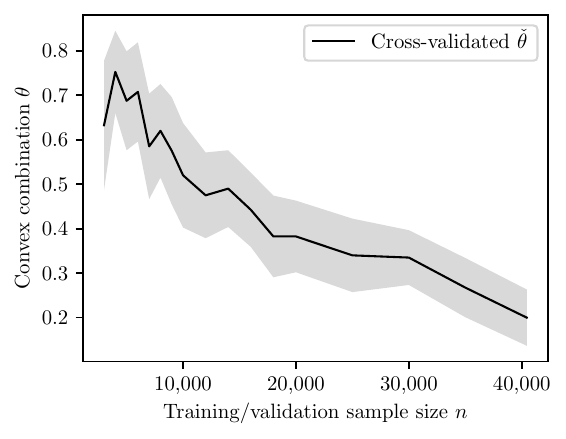}
  }
  \caption{We plot the cross-validated choices of the convex
    combination parameter $\theta$ in Model~1 as a function of the
  training and validation sample sizes.}
  \label{fig:selection_score2_model_1}
\end{figure}

Since the true relative risk function $f_0$ is unknown,
we cannot compute the $L_2$-error of the estimators; instead, we
measure performance of a relative risk estimate $f$
using the concordance index \citep{harrell1982evaluating}
and the integrated Brier score
\citep{graf1999assessment,li2021semiparametric}.
The concordance index is defined on
the test data $(X_i,T_i,I_i)_{i \in [n_{\mathrm{test}}]}$ as
\begin{align*}
  \mathrm{Concordance}(f)
  \vcentcolon=
  \frac{
    \sum_{i=1}^{n_{\mathrm{test}}}
    \sum_{j=1}^{n_{\mathrm{test}}}
    \mathbbm 1_{\{f(X_i) < f(X_j)\}}
    \mathbbm 1_{\{T_i > T_j\}}
    (1 - I_j)
  }{
    \sum_{i=1}^{n_{\mathrm{test}}}
    \sum_{j=1}^{n_{\mathrm{test}}}
    \mathbbm 1_{\{T_i > T_j\}}
    (1 - I_j)
  }.
\end{align*}
The integrated Brier score is
$\int_0^1 \BS(f, t) \diffi w(t)$,
where $w$ is a measure on $[0, 1]$ and
\begin{equation*}
  \BS(f, t)
  \coloneqq
  \frac{1}{n_\mathrm{test}}
  \sum_{i=1}^{n_\mathrm{test}}
  \biggl(
    \frac{\hat P(t \mid X_i, f)^2 \mathbbm 1\{T_i \leq t\} (1 - I_i)}
    {\hat P_{\mathrm C}(T_i)}
    +
    \frac{\{1 - \hat P(t \mid X_i, f)\}^2 \mathbbm 1\{T_i > t\}}
    {\hat P_{\mathrm C}(t)}
  \biggr),
\end{equation*}
with $\hat P(t \mid x, f)$ and $\hat P(t \mid f)$
as in \eqref{eq:survival_prob}, and
$\hat P_{\mathrm C}(t) \vcentcolon=
\exp \bigl( - \sum_{i=1}^{n_\mathrm{test}} \mathbbm 1\{T_i \leq t\} I_i
\big/ \sum_{j=1}^{n_\mathrm{test}} \mathbbm 1\{T_j \geq T_i\} \bigr)$.
For our experiments, we take $w$ to be the discrete uniform measure on
$\bigl\{i/24: i \in \{0, \ldots, 24\}\bigr\}$,
after scaling the time variable to $[0, 1]$.

Figure~\ref{fig:performance_score2_model_1} illustrates the
dependence of the concordance index and integrated Brier score on~$n$.
The CARE estimator $\check f_{n,\check\gamma,\check\theta}$ consistently
improves upon the external estimator from the SCORE2 model as soon as
$n \geq 10{,}000$, indicating that only a relatively small additional
sample size is needed for the benefits of aggregation to become
visible. Likewise, CARE tends to perform
better than the kernel estimator $\hat{f}_{n,\hat{\gamma}}$ (which
ignores the SCORE2 data), particularly for smaller values of~$n$.
When using all of the available data ($n = n_{\mathrm{test}}$),
the concordance index improved by
$0.51\%$ for females and by $0.44\%$ for males,
while the integrated Brier score improved by
$0.06\%$ for females and by $0.03\%$ for males.
We emphasise that these integrated Brier scores are
computed using estimated survival probabilities that are
recalibrated for the UKB data \eqref{eq:survival_prob}.
As such, the observed improvements from CARE are probably due to the
non-linear features providing a superior predictive signal,
and not because of any potential miscalibration in SCORE2.

\begin{figure}[H]
  \centering
  \iftoggle{journal}{\hspace*{2mm}}{\hspace*{3mm}}
  \subfloat[Concordance index for females.]{
    \includegraphics[scale=\myfigscale]{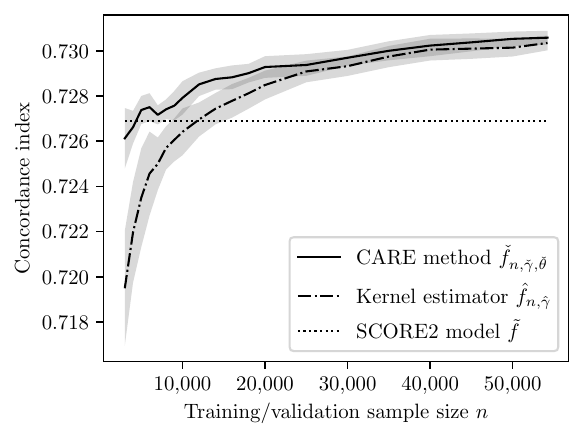}
  }
  \hspace*{8mm}
  \subfloat[Concordance index for males.]{
    \includegraphics[scale=\myfigscale]{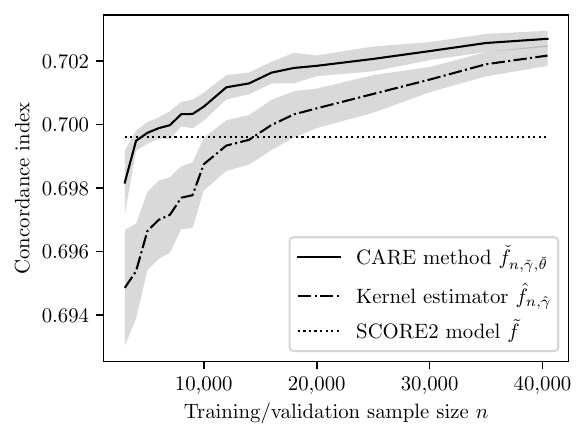}
  } \\
  \subfloat[Integrated Brier score for females.]{
    \includegraphics[scale=\myfigscale]{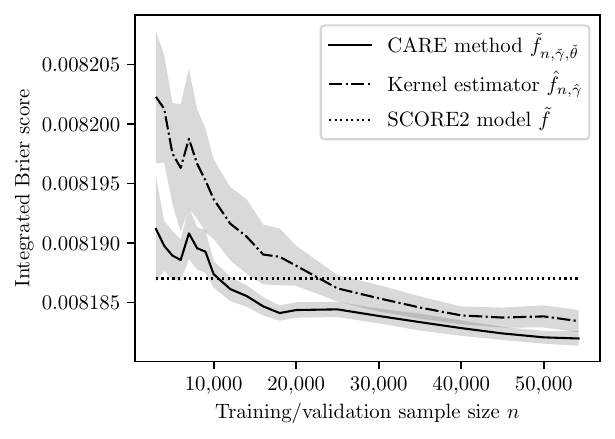}
  }
  \hspace*{6mm}
  \subfloat[Integrated Brier score for males.]{
    \includegraphics[scale=\myfigscale]{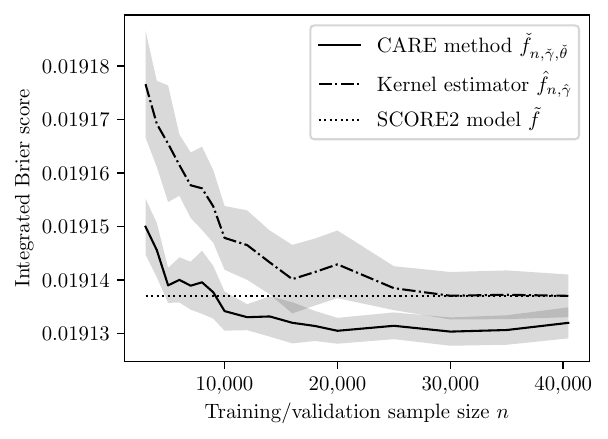}
  }
  \caption{We plot the concordance index and integrated Brier score
    in Model 1 against the training and validation sample size.
    CARE, kernel and external SCORE2 estimators are shown.
    For concordance index, higher is better;
  for integrated Brier score, lower is better.}
  \label{fig:performance_score2_model_1}
\end{figure}

\subsubsection{Model 2: adding features to SCORE2
for improved risk stratification}

For our second model, we use the \texttt{imd}
area-level deprivation feature and the
polygenic scores \texttt{pgs000018} and \texttt{pgs000039}
from UKB as additional covariates, in combination with SCORE2.
From Figure~\ref{fig:selection_score2_model_2}, we see a similar pattern
as in Figure~\ref{fig:selection_score2_model_1}, with the weight
assigned to the external estimator decreasing with $n$.
Figure~\ref{fig:performance_score2_model_2} indicates the
benefits of incorporating the deprivation index and
polygenic scores into the model, using our cross-validated estimator.
When using all of the available data,
the concordance index improved by
$1.52\%$ for females and by $2.26\%$ for males,
while the integrated Brier score improved by
$0.19\%$ for females and by $0.44\%$ for males.
We conclude that the new covariates
\texttt{imd}, \texttt{pgs000018}
and \texttt{pgs000039}
accounted for approximately
$66.6\%$ %
of the total concordance improvement for females and
$80.4\%$ %
of the improvement for males,
with the rest due to the use
of a quadratic kernel.
Further analysis in
Appendix~\ref{sec:additional_score2}%
\iftoggle{journal}{ in the supplementary material%
\if1\anon\ \citep{underwood2026caresupplement}\fi}{}
shows that the polygenic score covariates
\texttt{pgs000018} and \texttt{pgs000039}
contributed approximately
$77.9\%$ %
of these improvements for females and
$86.4\%$ %
for males, with the remainder attributable
to the deprivation index variable \texttt{imd}.

\begin{figure}[H]
  \centering
  \subfloat[Cross-validated choices of $\theta$ for females.]{
    \includegraphics[scale=\myfigscale]{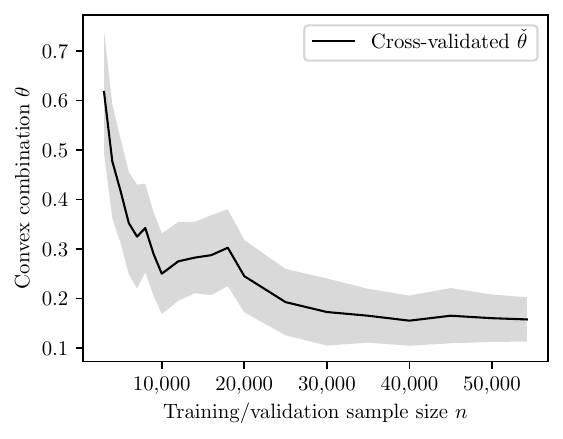}
  }
  \hspace*{5mm}
  \subfloat[Cross-validated choices of $\theta$ for males.]{
    \includegraphics[scale=\myfigscale]{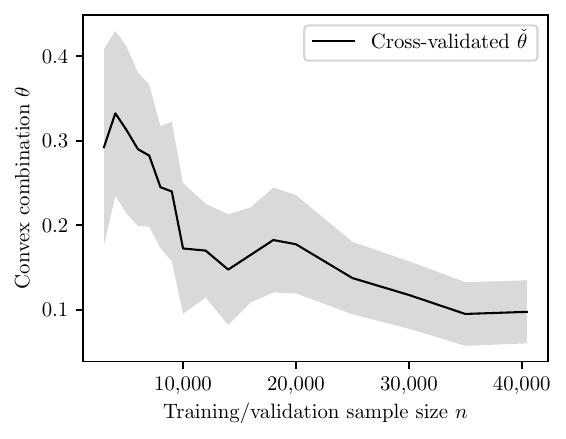}
  }
  \caption{We plot the cross-validated choices of the convex
    combination parameter $\theta$ in Model~2 as a function of the
  training and validation sample sizes.}
  \label{fig:selection_score2_model_2}
\end{figure}

\begin{figure}[H]
  \centering
  \hspace*{1mm}
  \subfloat[Concordance index for females.]{
    \includegraphics[scale=\myfigscale]{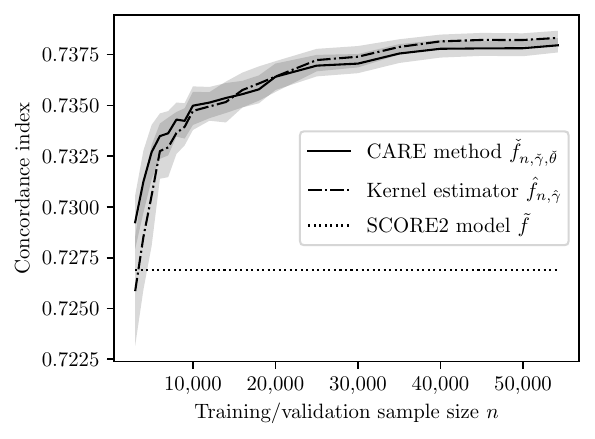}
  }
  \hspace*{5mm}
  \subfloat[Concordance index for males.]{
    \includegraphics[scale=\myfigscale]{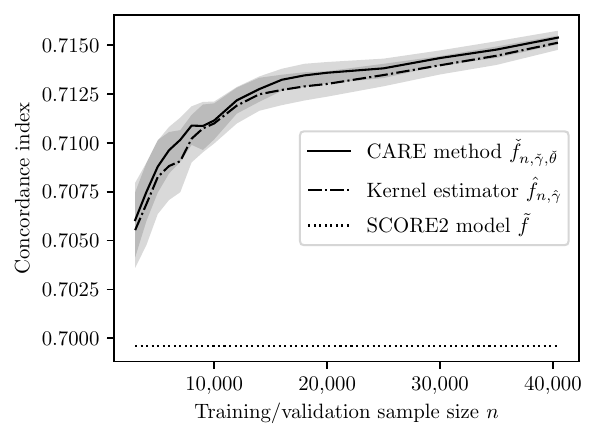}
  } \\
  \subfloat[Integrated Brier score for females.]{
    \includegraphics[scale=\myfigscale]{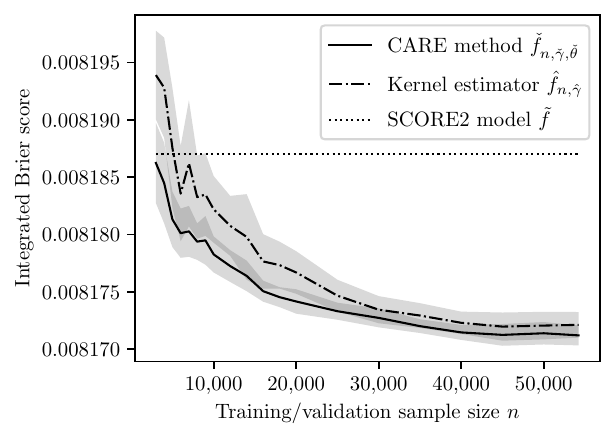}
  }
  \hspace*{4mm}
  \subfloat[Integrated Brier score for males.]{
    \includegraphics[scale=\myfigscale]{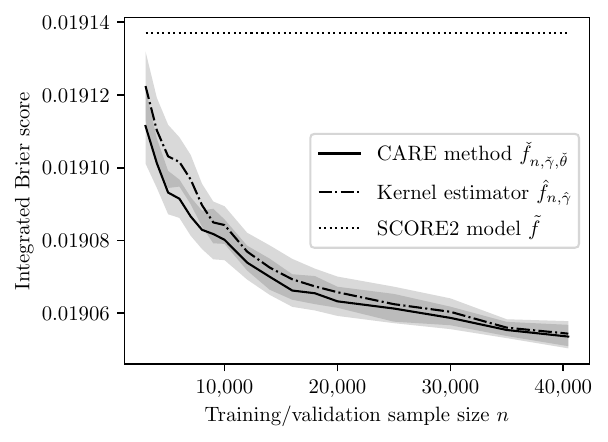}
  }
  \caption{We plot the concordance index and integrated Brier score
    in Model 2 against the training and validation sample size,
  for the CARE, kernel and external SCORE2 estimators.}
  \label{fig:performance_score2_model_2}
\end{figure}

These models highlight an important robustness property of the
CARE methodology: when the existing SCORE2 model can be improved,
it is successfully ignored by our convex aggregation procedure
in favour of the superior kernel estimator that uses the new covariates.
See Appendix~\ref{sec:additional_score2}%
\iftoggle{journal}{ in the supplementary material%
\if1\anon\ \citep{underwood2026caresupplement}\fi}{}
for further empirical results;
in summary, we show that CARE can be used to
upgrade a fitted survival model for the task of in-cohort prediction,
and can incorporate new features to further improve performance.

\section{Acknowledgements and funding}

WGU and RJS
were supported by the latter's European Research Council Advanced
Grant 101019498;
WGU was also funded
by the August 2025 G-Research Early Career Grant.
This research has been conducted using the
UK Biobank Resource under application 608471.
The work was supported by grants from the British Heart Foundation
(RG/18/13/33946 and RG/F/23/110103), British Heart Foundation Chair
Award (CH/12/2/29428) and by Health Data Research UK, which is funded
by the UK Medical Research Council, Engineering and Physical Sciences
Research Council, Economic and Social Research Council, Department of
Health and Social Care (England), Chief Scientist Office of the
Scottish Government Health and Social Care Directorates, Health and
Social Care Research and Development Division (Welsh Government),
Public Health Agency (Northern Ireland), British Heart Foundation and
the Wellcome Trust, National Institute for Health and Care Research
(NIHR) Cambridge Biomedical Research Centre (NIHR203312), BHF Centre
for Research Excellence (RE/18/1/34212 and RE/24/130011).
 
\pagebreak
\appendix

\section{Further computational considerations}

We discuss two techniques for improving the computational
properties of our procedure.

\subsection{Feature map formulation of the RKHS estimator}
\label{sec:feature_map}

Certain kernel functions, for example polynomial kernels,
admit finite-dimensional feature maps, and the aim of this subsection
is to show how this can be exploited for substantial computational
gains. In fact, even when the feature map is infinite-dimensional,
the scheme we outline below can be used to compute an approximate
solution to~\eqref{eq:argmin}. Suppose that there exist $q \in
\mathbb{N}$ and a
measurable function
$\phi = (\phi_1,\ldots,\phi_{q+1}): \cX \to \R^{q+1}$ with linearly
independent coordinate functions satisfying
$k(x, y) = \phi(x)^\T \phi(y)$ and $\phi_{q+1}(x) = c$ for all $x \in
\mathcal{X}$ and some $c \in \mathbb{R} \setminus \{0\}$.
Define $\tilde \phi(\cdot) \vcentcolon=
\bigl(\phi_j(\cdot) - P_n(\phi_j): j \in [q]\bigr)$ and,
for $\alpha \in \R^q$, let
$f_\alpha(\cdot) \vcentcolon= \alpha^\T \tilde\phi(\cdot).$
Noting that $P_n(\hat f_{n,\gamma}) = 0$,
the representer theorem (Proposition~\ref{prop:representation}) asserts
the existence of $\hat{\alpha}$ taking values in $\R^q$ with
\begin{align*}
  \hat f_{n,\gamma}(\cdot)
  &= \hat\alpha^\T \tilde\phi(\cdot)
  = f_{\hat\alpha}(\cdot).
\end{align*}
For $j \in [q+1]$ and $x \in \cX$,
by the reproducing property, we have
\[
  \phi_j(x) = \langle \phi_j, k(\cdot, x) \rangle_\cH
  = \langle \phi_j, \phi(x)^\T \phi \rangle_\cH
  = \sum_{r=1}^{q+1} \phi(x)_r \langle \phi_j, \phi_r \rangle_\cH,
\]
so $\langle \phi_j, \phi_r \rangle_\cH = \mathbbm 1_{\{j = r\}}$
for each $j, r \in [q+1]$ by linear independence.
Thus, for every $\alpha = (\alpha_1,\ldots,\alpha_q)^\T \in \mathbb{R}^q$,
\begin{align*}
  \|f_{\alpha}\|_\cH^2
  &=
  \sum_{j=1}^q
  \sum_{r=1}^q
  \alpha_j \alpha_r
  \langle \tilde \phi_j, \tilde \phi_r \rangle_\cH
  =
  \sum_{j=1}^q
  \sum_{r=1}^q
  \alpha_j \alpha_r
  \bigg\langle
  \phi_j - P_n(\phi_j) \frac{\phi_{q+1}}{c},
  \phi_r - P_n(\phi_r) \frac{\phi_{q+1}}{c}
  \bigg\rangle_\cH \\
  &= \sum_{j=1}^{q} \alpha_j^2
  + \biggl(\frac{1}{c} \sum_{j=1}^q \alpha_j P_n(\phi_j)\biggr)^2.
\end{align*}
This allows us to calculate
the penalised negative partial log-likelihood~\eqref{eq:likelihood_penalised}
for any function of the form $f_\alpha$, and hence to
solve the optimisation problem in~\eqref{eq:argmin}
over $\alpha \in \mathbb{R}^q$.
When the sample size is much
larger than the dimension of the features, this problem can be significantly
easier to solve than the more general version obtained through
Proposition~\ref{prop:representation}.

\subsection{Nystr{\"o}m approximation}
\label{sec:nystrom}

When a nonparametric kernel is used, there is no finite-dimensional
feature map available.
However, when the number of training data points is large,
the computational complexity of calculating
$\hat f_{n,\gamma}$ can be reduced using
Nystr{\"o}m approximation \citep{williams2000using}.
This method proceeds by writing
$\hat f_{n,\gamma} \approx \sum_{i \in \cS_n} \tilde k(x, X_i) \beta_i$,
where $\cS_n$ is a randomly selected subset of $[n]$ of cardinality $m < n$.
This increases the bias of the procedure by
shrinking the space of admissible functions.
However, when $m \ll n$, this approximation
can greatly reduce the computational cost of solving
\eqref{eq:argmin}, as the decision variable in the resulting
optimisation problem is now only $m$-dimensional.
See Appendix~\ref{sec:additional_empirical}%
\iftoggle{journal}{ in the supplementary material%
\if1\anon\ \citep{underwood2026caresupplement}\fi}{}
for empirical results on using Nystr{\"o}m approximation
to upgrade the SCORE2 model.
We refer the reader to \citet{bach2013sharp} and \citet{alaoui2015fast}
for some theoretical guarantees on Nystr{\"o}m approximation for RKHS methods;
typically, it suffices to take $m$ on the order of the effective
dimension $H_{2,\gamma}$ of the $\gamma$-regularised kernel.

\section{Kernel assumptions}
\label{sec:kernel_assumptions}

We provide some technical conditions
and results regarding the kernel function $k$
which are required for our main theoretical results.
Let~$\mathcal{X}$ be a locally compact metric space,
equipped with its
Borel $\sigma$-algebra $\mathcal{A}$, and assume that every open set
in $\mathcal{X}$ is
$\sigma$-compact, i.e.~a countable union of compact sets. This
ensures that~$P_X$ is a Radon probability measure
\citep[Theorem~7.8]{folland1999real}.
We will further assume that~$k$
is continuous, that $P_X(U) > 0$ for every non-empty open set $U
\subseteq \mathcal{X}$,
and that $x \mapsto k(x,x)$ is $P_X$-integrable on $\mathcal{X}$.
By Mercer's theorem
\citep[Theorem~A.8.8]{samworth2026modern}, $T_k$ admits
a decreasing sequence $(\nu_r)_{r \in \mathcal{R}}$
of positive eigenvalues
with $\sum_{r \in \mathcal{R}} \nu_r < \infty$ and associated continuous
orthonormal eigenfunctions $(e_r)_{r \in \mathcal{R}}$ in
$L_2(P_X)$, with $T_k(e_r) = \nu_r e_r$ for $r \in \mathcal{R}$, where
either $\mathcal{R} = [R]$ for some $R \in \N$ or
$\mathcal{R} = \N$, endowed with the usual order.
We then have the spectral decomposition
$k(x, y) = \sum_{r \in \mathcal{R}} \nu_r e_r(x) e_r(y)$
for $x,y \in \mathcal{X}$,
where the convergence is absolute and uniform on
compact subsets of $\mathcal{X} \times \mathcal{X}$.
Writing
$\ell_2(\mathcal{R}) \vcentcolon= \bigl\{ (\alpha_r)_{r \in \mathcal{R}}
: \sum_{r \in \mathcal{R}} \alpha_r^2 < \infty \bigr\}$, let
\begin{equation*}
  \cH \vcentcolon=
  \biggl\{ f_\alpha \vcentcolon= \sum_{r \in \mathcal{R}}
  \alpha_r \sqrt{\nu_r} e_r : \alpha \in \ell_2(\mathcal{R}) \biggr\},
\end{equation*}
and for
$\alpha \vcentcolon= (\alpha_r)_{r \in \mathcal{R}} \in \ell_2(\mathcal{R})$
and $\beta \vcentcolon= (\beta_r)_{r \in \mathcal{R}} \in \ell_2(\mathcal{R})$,
define the inner product
$\langle f_\alpha, f_{\beta} \rangle_\cH
\vcentcolon= \sum_{r \in \mathcal{R}} \alpha_r \beta_r$.
Then by \citet[Corollary~A.8.9]{samworth2026modern}, the pair
$(\mathcal{H},\langle \cdot,\cdot \rangle_{\mathcal{H}})$
forms a Hilbert space with
orthonormal basis $(\sqrt{\nu_r} e_r)_{r \in \mathcal{R}}$, and is in
fact the unique reproducing kernel Hilbert space (RKHS) associated with $k$.
Further, every $f \in \cH$ is continuous
\citep[Theorem~A.8.8 and Corollary~A.8.9]{samworth2026modern}
and satisfies $f \in L_2(P_X)$, so
$f = \sum_{r \in \mathcal{R}} \langle f, e_r \rangle_{L_2} e_r$
where $\langle \cdot, \cdot \rangle_{L_2}
\vcentcolon= \langle \cdot, \cdot \rangle_{L_2(P_X)}$.
From this it follows that
$\langle f, g \rangle_\cH = \sum_{r \in \mathcal{R}} \frac{1}{\nu_r}
\langle f, e_r \rangle_{L_2} \langle g, e_r \rangle_{L_2}$,
so $\langle f, e_r \rangle_\cH = \frac{1}{\nu_r}
\langle f, e_r \rangle_{L_2}$ for each $r \in \mathcal{R}$,
and $\|f\|_\cH^2 = \sum_{r \in \mathcal{R}} \frac{1}{\nu_r}
\langle f, e_r \rangle_{L_2}^2$.
We deduce that
$\langle f, e_r \rangle_{\cH_\gamma}
= \frac{1}{a_r} \langle f, e_r \rangle_{L_2}$
and $\|f\|_{\cH_\gamma}^2 = \sum_{r \in \mathcal{R}}
\frac{1}{a_r} \langle f, e_r \rangle_{L_2}^2$.
Further, by \citet[Theorem~3.11]{paulsen2016introduction},
$\cH$ contains the unit constant function on $\cX$.
Note that $H_{\infty,\gamma}$
does not depend on the choice of
$L_2(P_X)$-orthonormal eigenfunctions $(e_r)_{r \in \cR}$
due to the uniqueness of the eigenspace associated with
$\nu_r$, for each $r \in \cR$.
Observe that $\gamma \sum_{r \in \cR} a_r e_r(x)^2 \searrow 0$
as $\gamma \searrow 0$ for each $x \in \cX$,
by the dominated convergence theorem.
Since $\gamma a_r \leq \nu_r$ and
$\sum_{r \in \cR} \nu_r e_r(x)^2 = k(x, x)$, with the
convergence uniform on compact sets, we deduce that
$\sup_{x \in \cX'} \gamma \sum_{r \in \cR} a_r e_r(x)^2 \searrow 0$
as $\gamma \searrow 0$ whenever $\cX' \subseteq \cX$ is compact.
In particular, $\gamma H_{\infty,\gamma} \searrow 0$
as $\gamma \searrow 0$ if $\cX$ is compact.

\section{Automatic kernel selection}
\label{sec:kernel_selection}

Suppose that $\cK$ is a finite set of
kernels such that each $k \in \cK$
induces an RKHS $\cH_k$
satisfying the conditions
of Section~\ref{sec:rkhs}.
For each kernel $k \in \cK$, define
\begin{equation*}
  \hat f_{n,\gamma,k} \vcentcolon=
  \argmin_{f \in \cH_k:\, P_n(f) = 0}
  \ell_{n, \gamma}(f)
\end{equation*}
and take $K_k$, $H_{2,\gamma,k}$
and $H_{\infty,\gamma,k}$
as in Section~\ref{sec:rkhs}.
Set
$K_\cK \coloneqq \max_{k \in \cK} K_k$ and
$H_{\infty,\gamma,\cK}
\coloneqq \max_{k \in \cK} H_{\infty,\gamma,k}$,
and let
\begin{align*}
  \bigl(
    \hat\gamma_\cK,
    \hat k
  \bigr)
  &\vcentcolon=
  \sargmin_{\gamma \in \Gamma,\, k \in \cK}
  \tilde\ell_n\bigl(\hat f_{n,\gamma,k}\bigr).
\end{align*}
Recall the notation
$\gamma^+ \vcentcolon= \max \Gamma$ and
$\gamma^- \vcentcolon= \min \Gamma$
from Section~\ref{sec:parameter_tuning}.

\begin{proposition}[Kernel selection]%
  \label{prop:kernel_selection}
  There exists $C_0 > 0$ depending only on
  $\max_{k \in \cK} \|f_0\|_{\cH_k}$,
  $\Lambda$, $q_1$,
  $\max_{k \in \cK} \|1_\cX\|_{\cH_k}$,
  $K_\cK$ and $\xi$ such that if
  $\gamma^+ H_{\infty,\gamma^+,\cK} \leq 1 / C_0$,
  $H_{\infty,\gamma^-,\cK} \sqrt{\{\log (n|\cK|)\} / n}
  \leq 1/C_0$ and
  $1 \leq s \leq \sqrt{n} /
  (C_0 H_{\infty,\gamma^-,\cK})$,
  then with probability at least $1 - e^{-s^2}$,
  \begin{equation*}
    \bigl\|
    \hat f_{n,\hat\gamma_\cK,\hat k} - f_0
    \bigr\|_{L_2}
    \leq
    C_0
    \biggl(
      \min_{\gamma \in \Gamma,\, k \in \cK}
      \bigl\|\hat f_{n,\gamma,k} - f_0\bigr\|_{L_2}
      + \frac{s + \sqrt{\log |\cK| + \log n}}{\sqrt n}
    \biggr).
  \end{equation*}
\end{proposition}
This oracle inequality permits
adaptation across Sobolev scales or automatic kernel selection
\citep{blanchard2008statistical,goldenshluger2011bandwidth};
we omit details for brevity.

\section{Additional empirical results}%
\label{sec:additional_empirical}

We present additional results relating to our
simulation studies and SCORE2 model analysis.

\subsection{Simulated data study}
\label{sec:additional_simulation}

We consider some further censoring distributions to better
understand the behaviour our kernel method and the CARE procedure.
The set-up is identical to the univariate and multivariate
situations of Sections~\ref{sec:simulation_univariate}
and~\ref{sec:simulation_multivariate} respectively,
modifying only the law of the censoring variable $T_{\mathrm C}$.
For `heavy censoring,' we take $T_{\mathrm C} \sim \cU[0, 1]$
and for `no censoring,' we consider $T_{\mathrm C} = 1$.

\iftoggle{journal}{}{\pagebreak}

\subsubsection{Univariate regressors, heavy censoring}

We present results for univariate regressors with
heavy censoring in Figures~\ref{fig:rkhs_dgp_3}
and~\ref{fig:aggregation_dgp_3}.

\begin{figure}[H]
  \centering
  \subfloat[Cross-validated and oracle choices of $\gamma$.]{
    \includegraphics[scale=\myfigscale]{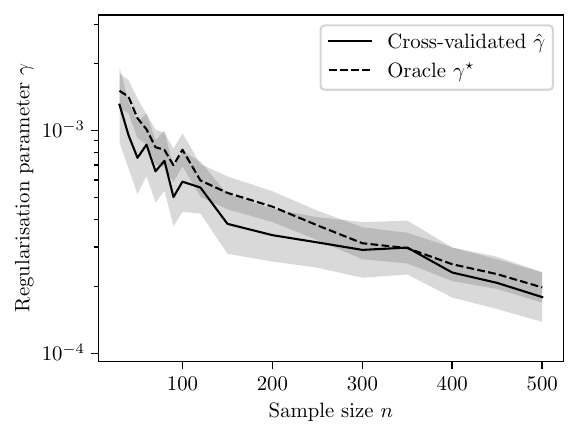}
  }
  \hspace*{5mm}
  \subfloat[Cross-validated and oracle $L_2$-error.]{
    \includegraphics[scale=\myfigscale]{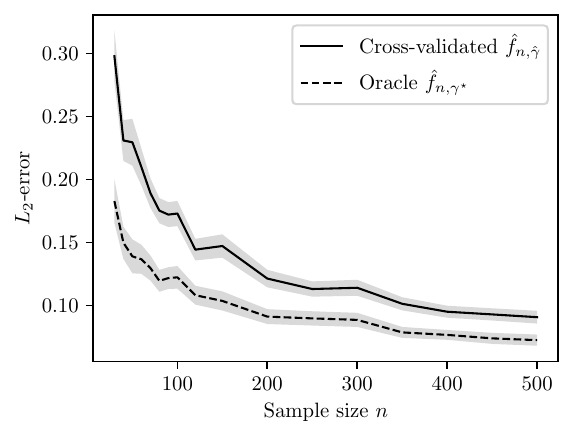}
  }
  \caption{Panel~(a): regularisation parameter. Panel~(b): $L_2$-error.}
  \label{fig:rkhs_dgp_3}
\end{figure}

\begin{figure}[H]
  \centering
  \subfloat[Cross-validation and oracle choices of $\theta$.]{
    \includegraphics[scale=\myfigscale]{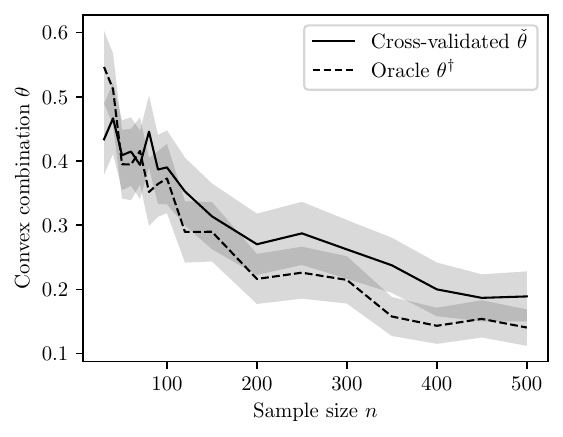}
  }
  \hspace*{5mm}
  \subfloat[Cross-validated and oracle performance.]{
    \includegraphics[scale=\myfigscale]{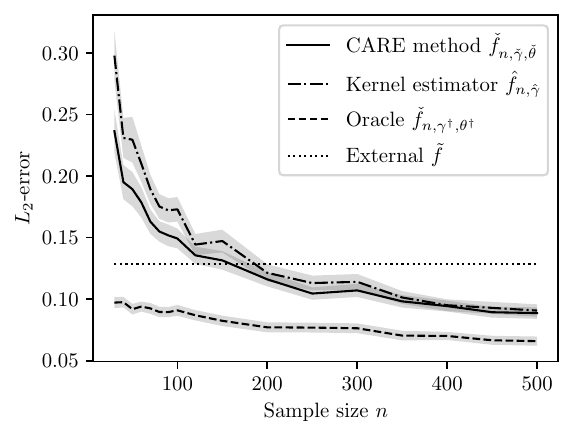}
  }
  \caption{Panel~(a): convex combination parameter.
  Panel~(b): comparison of estimators.}
  \label{fig:aggregation_dgp_3}
\end{figure}

\pagebreak

\subsubsection{Multivariate regressors, heavy censoring}

We present results for multivariate regressors with
heavy censoring in Figures~\ref{fig:rkhs_dgp_4}
and~\ref{fig:aggregation_dgp_4}.

\begin{figure}[H]
  \centering
  \subfloat[Cross-validated and oracle choices of $\gamma$.]{
    \includegraphics[scale=\myfigscale]{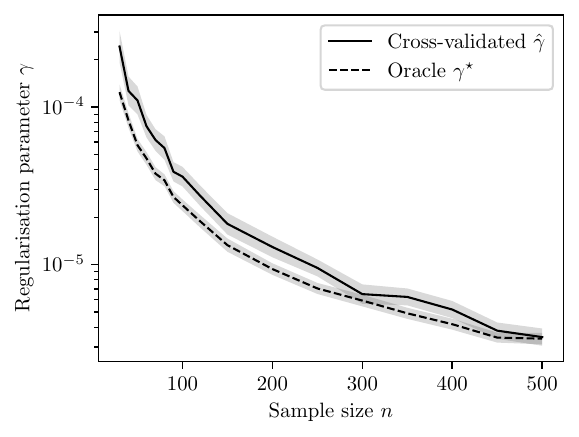}
  }
  \hspace*{5mm}
  \subfloat[Cross-validated and oracle $L_2$-error.]{
    \includegraphics[scale=\myfigscale]{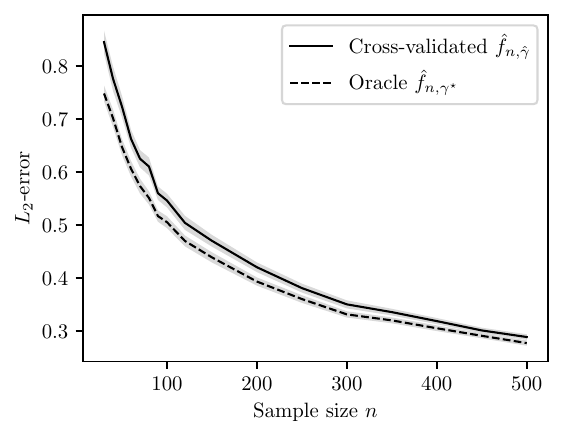}
  }
  \caption{Panel~(a): regularisation parameter. Panel~(b): $L_2$-error.}
  \label{fig:rkhs_dgp_4}
\end{figure}

\begin{figure}[H]
  \centering
  \subfloat[Cross-validation and oracle choices of $\theta$.]{
    \includegraphics[scale=\myfigscale]{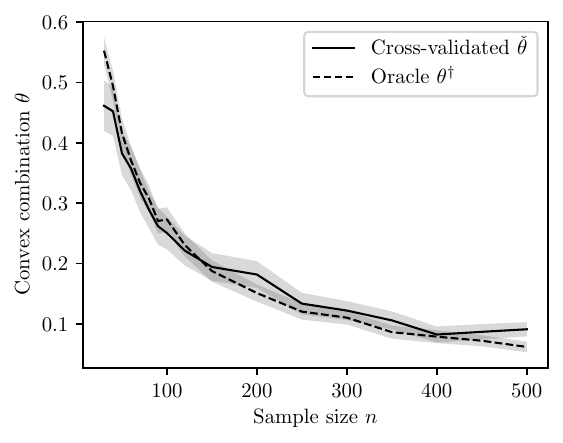}
  }
  \hspace*{5mm}
  \subfloat[Cross-validated and oracle performance.]{
    \includegraphics[scale=\myfigscale]{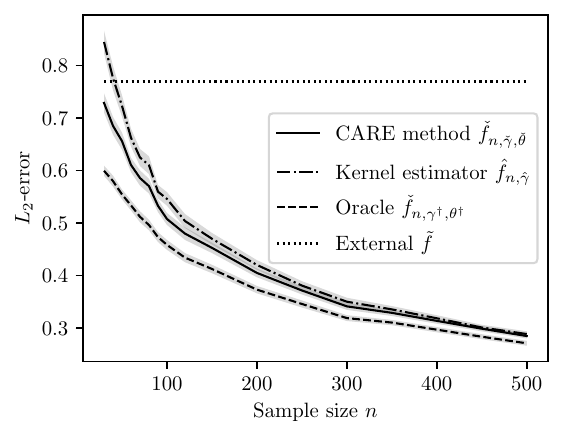}
  }
  \caption{Panel~(a): convex combination parameter.
  Panel~(b): comparison of estimators.}
  \label{fig:aggregation_dgp_4}
\end{figure}

\pagebreak

\subsubsection{Univariate regressors, no censoring}

We present results for univariate regressors with
no censoring in Figures~\ref{fig:rkhs_dgp_5}
and~\ref{fig:aggregation_dgp_5}.

\begin{figure}[H]
  \centering
  \subfloat[Cross-validated and oracle choices of $\gamma$.]{
    \includegraphics[scale=\myfigscale]{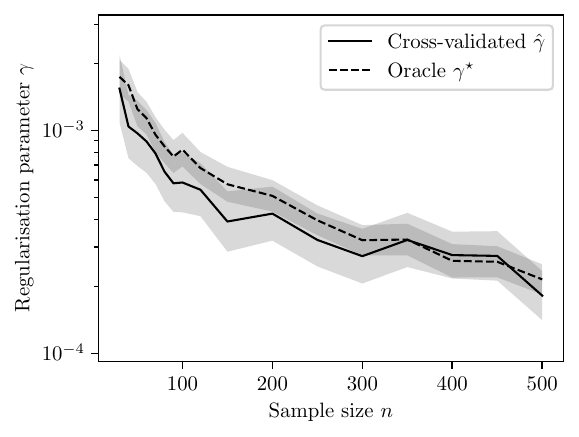}
  }
  \hspace*{5mm}
  \subfloat[Cross-validated and oracle $L_2$-error.]{
    \includegraphics[scale=\myfigscale]{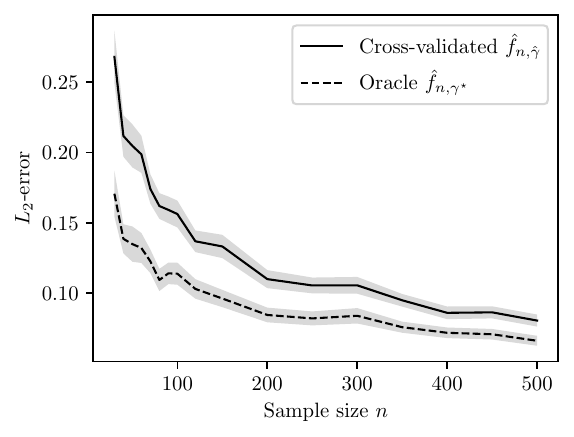}
  }
  \caption{Panel~(a): regularisation parameter. Panel~(b): $L_2$-error.}
  \label{fig:rkhs_dgp_5}
\end{figure}

\begin{figure}[H]
  \centering
  \subfloat[Cross-validation and oracle choices of $\theta$.]{
    \includegraphics[scale=\myfigscale]{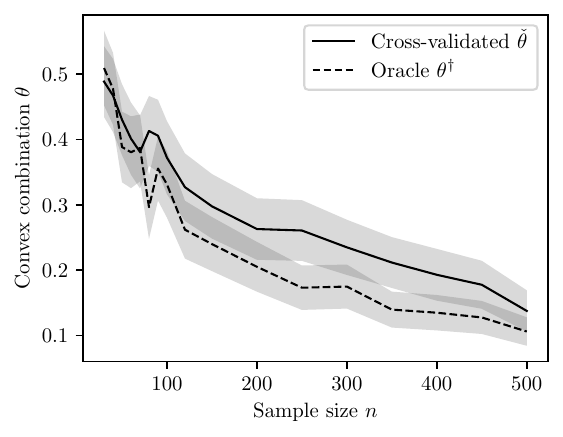}
  }
  \hspace*{5mm}
  \subfloat[Cross-validated and oracle performance.]{
    \includegraphics[scale=\myfigscale]{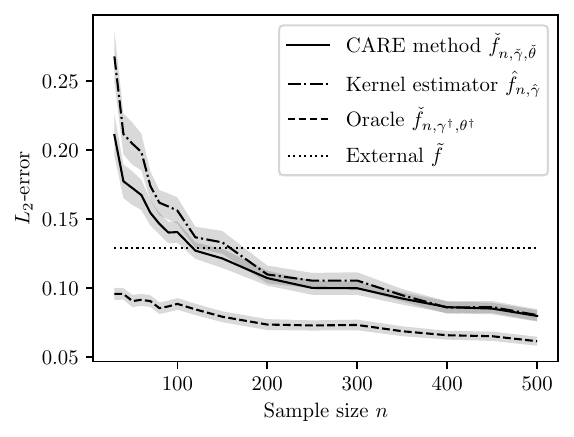}
  }
  \caption{Panel~(a): convex combination parameter.
  Panel~(b): comparison of estimators.}
  \label{fig:aggregation_dgp_5}
\end{figure}

\pagebreak

\subsubsection{Multivariate regressors, no censoring}

We present results for multivariate regressors with
no censoring in Figures~\ref{fig:rkhs_dgp_6}
and~\ref{fig:aggregation_dgp_6}.

\begin{figure}[H]
  \centering
  \subfloat[Cross-validated and oracle choices of $\gamma$.]{
    \includegraphics[scale=\myfigscale]{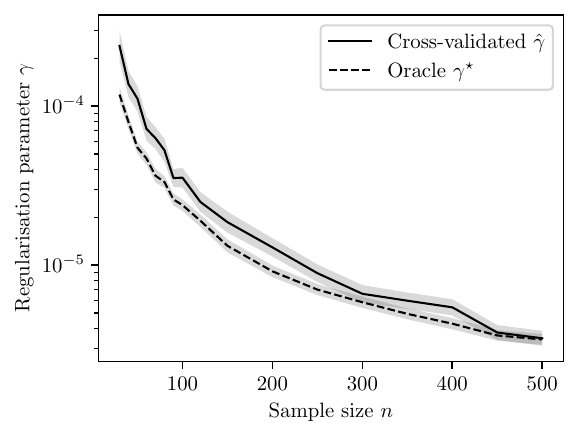}
  }
  \hspace*{5mm}
  \subfloat[Cross-validated and oracle $L_2$-error.]{
    \includegraphics[scale=\myfigscale]{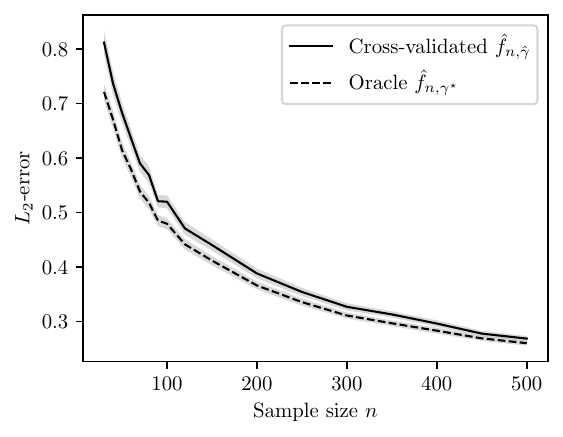}
  }
  \caption{Panel~(a): regularisation parameter. Panel~(b): $L_2$-error.}
  \label{fig:rkhs_dgp_6}
\end{figure}

\begin{figure}[H]
  \centering
  \subfloat[Cross-validation and oracle choices of $\theta$.]{
    \includegraphics[scale=\myfigscale]{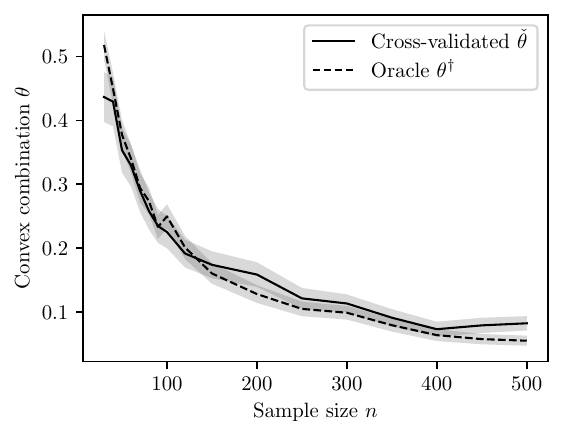}
  }
  \hspace*{5mm}
  \subfloat[Cross-validated and oracle performance.]{
    \includegraphics[scale=\myfigscale]{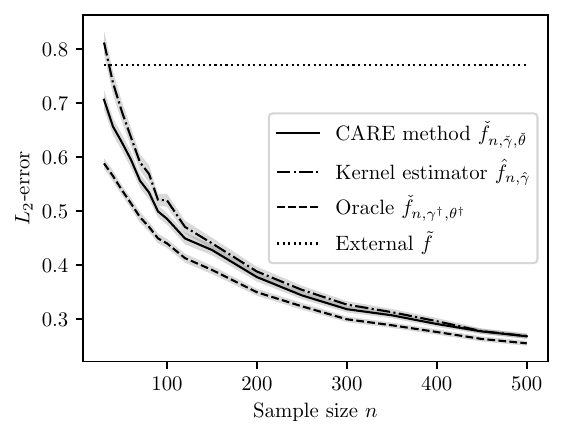}
  }
  \caption{Panel~(a): convex combination parameter.
  Panel~(b): comparison of estimators.}
  \label{fig:aggregation_dgp_6}
\end{figure}

By comparing with the results in
Section~\ref{sec:simulated},
we conclude that the relative
performance of our methodology
is broadly similar across different
censoring levels, both for
univariate and multivariate
regressors.  The `no censoring' case
generally leading to slightly better results in terms of
$L_2$-error and integrated Brier
score.

\pagebreak

\subsection{Cardiovascular disease risk prediction}
\label{sec:additional_score2}

We present additional results on
upgrading the SCORE2 cardiovascular disease risk model.

\subsubsection{Model 1}

In Figure~\ref{fig:scatter_score2_model_1}, we plot the predicted 10-year risks
in Model~1 against the corresponding predicted risks for the
recalibrated SCORE2 model.  We see that most predicted risks increase
slightly under CARE as opposed to SCORE2, perhaps particularly among
low-risk males.

\begin{figure}[H]
  \centering
  \subfloat[Predicted risk for females.]{
    \includegraphics[scale=\myfigscale]{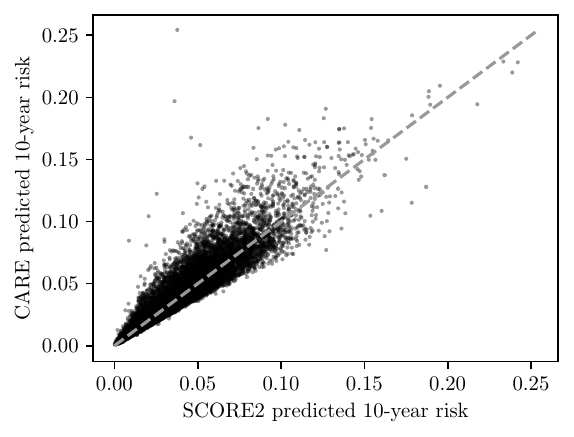}
  }
  \hspace*{5mm}
  \subfloat[Predicted risk for males.]{
    \includegraphics[scale=\myfigscale]{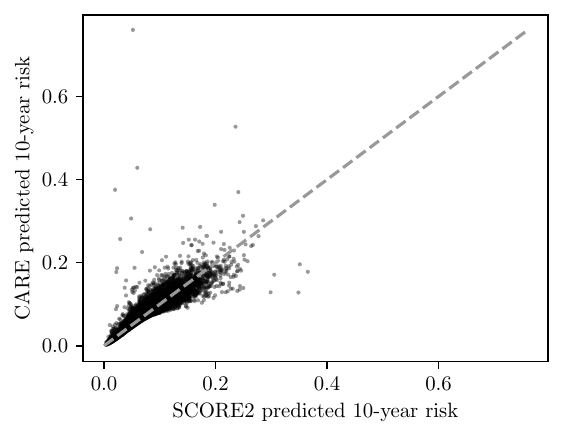}
  }
  \caption{Risk predictions for Model 1 compared to SCORE2.}
  \label{fig:scatter_score2_model_1}
\end{figure}

\subsubsection{Model 2}

In Figure~\ref{fig:scatter_score2_model_2}, we plot the predicted 10-year risks
in Model~2 against the corresponding predicted risks for the
recalibrated SCORE2 model.
We see more deviation between the SCORE2 and CARE predicted risks than
for Model~1, due to the extra predictive power offered by the
inclusion of area-level deprivation and polygenic score covariates.

\begin{figure}[H]
  \centering
  \subfloat[Predicted risk for females.]{
    \includegraphics[scale=\myfigscale]{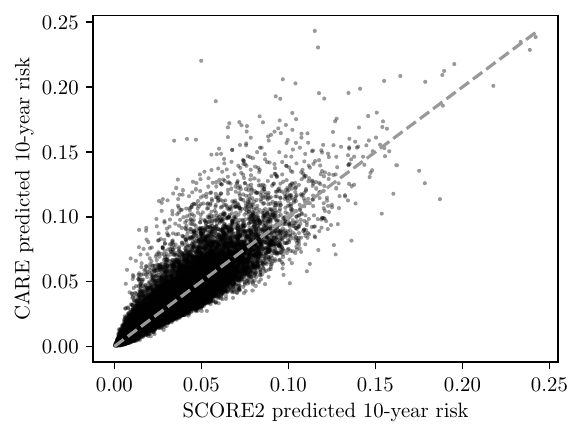}
  }
  \hspace*{5mm}
  \subfloat[Predicted risk for males.]{
    \includegraphics[scale=\myfigscale]{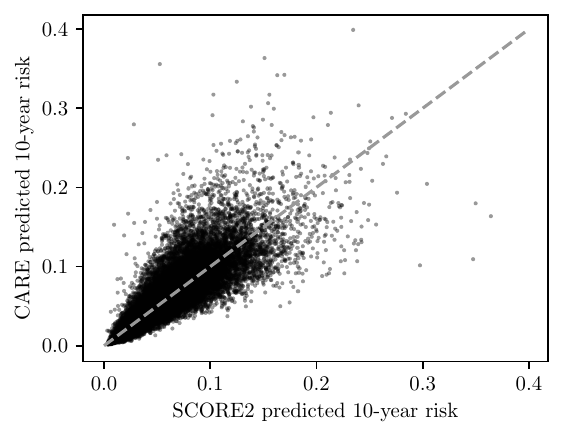}
  }
  \caption{Risk predictions for Model 2 compared to SCORE2.}
  \label{fig:scatter_score2_model_2}
\end{figure}

\subsubsection{Model 3: SCORE2 and UKB with deprivation index}

In Model 3, we consider upgrading SCORE2 using CARE,
adding just the deprivation index variable \texttt{imd}.
The cross-validated convex combination parameter is
shown in Figure~\ref{fig:selection_score2_model_3},
while Figure~\ref{fig:performance_score2_model_3}
illustrates the performance of this model as measured by
concordance index and integrated Brier score.
A scatter plot of the predicted risks against recalibrated SCORE2
is shown in Figure~\ref{fig:scatter_score2_model_3}.
When using all of the available data,
the concordance index improved by
$0.73\%$ for females and by $0.69\%$ for males,
while the integrated Brier score improved by
$0.07\%$ for females and by $0.15\%$ for males.
We conclude that when moving from Model 1 to Model 2,
the polygenic score covariates
\texttt{pgs000018} and \texttt{pgs000039}
accounted for approximately
$77.9\%$ %
of the total concordance improvement
(upon to the UKB-adjusted Model 1)
for females and
$86.4\%$ %
of the improvement for males,
with the remainder of the improvement attributable
to inclusion of the
deprivation index variable \texttt{imd}.

\begin{figure}[H]
  \centering
  \subfloat[Cross-validated choices of $\theta$ for females.]{
    \includegraphics[scale=\myfigscale]{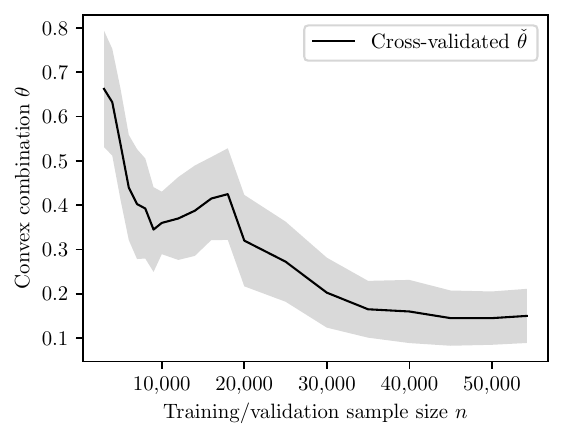}
  }
  \hspace*{5mm}
  \subfloat[Cross-validated choices of $\theta$ for males.]{
    \includegraphics[scale=\myfigscale]{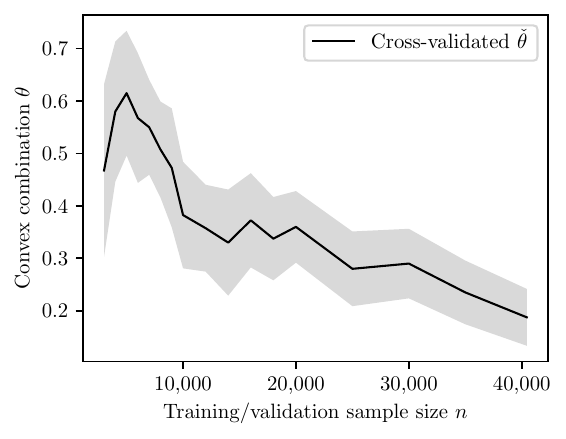}
  }
  \caption{Convex combination parameter choices in Model 3.}
  \label{fig:selection_score2_model_3}
\end{figure}

\begin{figure}[H]
  \centering
  \hspace*{2mm}
  \subfloat[Concordance index for females.]{
    \includegraphics[scale=\myfigscale]{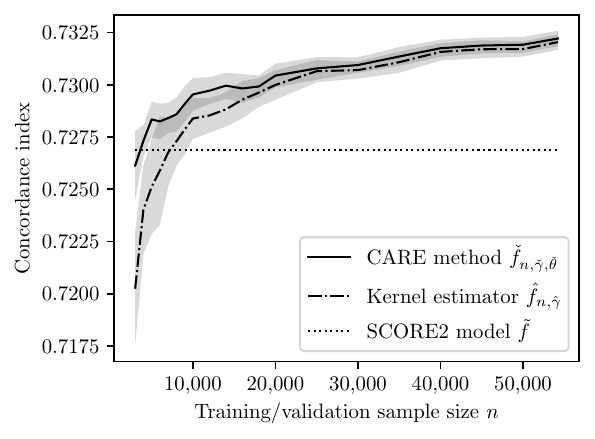}
  }
  \hspace*{8mm}
  \subfloat[Concordance index for males.]{
    \includegraphics[scale=\myfigscale]{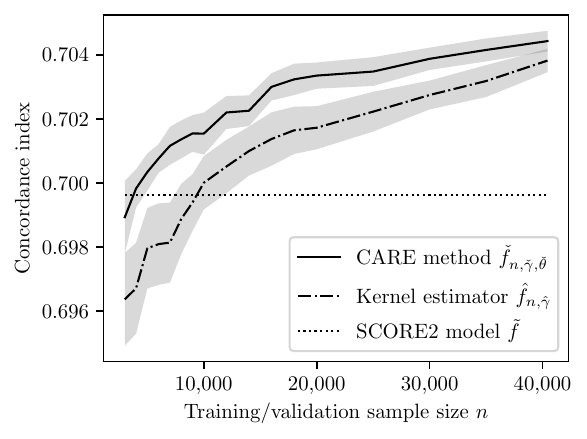}
  } \\
  \subfloat[Integrated Brier score for females.]{
    \includegraphics[scale=\myfigscale]{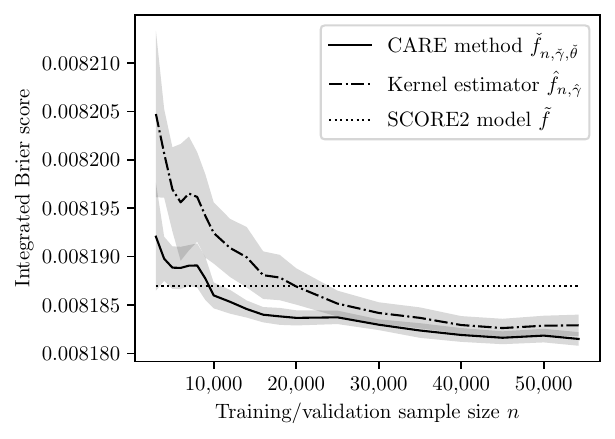}
  }
  \hspace*{6mm}
  \subfloat[Integrated Brier score for males.]{
    \includegraphics[scale=\myfigscale]{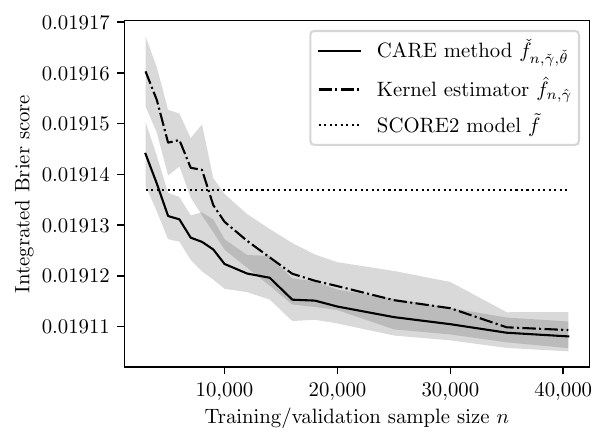}
  }
  \caption{Comparison of estimators in Model 3.}
  \label{fig:performance_score2_model_3}
\end{figure}

\begin{figure}[H]
  \centering
  \subfloat[Predicted risk for females.]{
    \includegraphics[scale=\myfigscale]{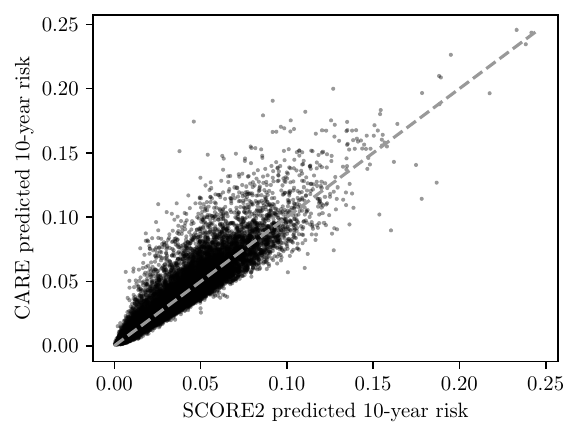}
  }
  \hspace*{5mm}
  \subfloat[Predicted risk for males.]{
    \includegraphics[scale=\myfigscale]{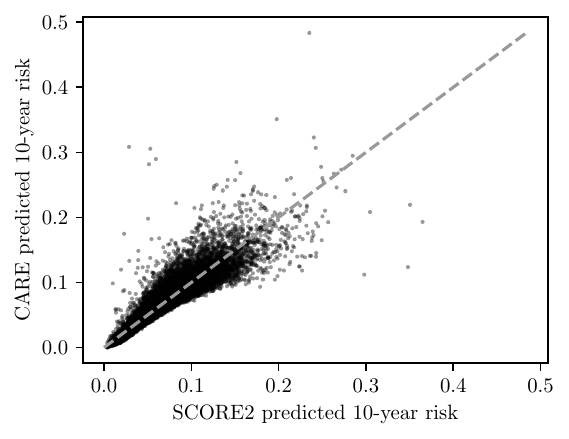}
  }
  \caption{Risk predictions for Model 3 compared to SCORE2.}
  \label{fig:scatter_score2_model_3}
\end{figure}

\subsubsection{Model 4: SCORE2 and UKB with deprivation index
and polygenic scores, linear kernel}

In Model 4, we add the deprivation index variable \texttt{imd}
and the polygenic scores \texttt{pgs000018} and \texttt{pgs000039}
from UKB to SCORE2, this time using a linear kernel with no regularisation
(corresponding to a standard Cox model).
The cross-validated convex combination parameter is
shown in Figure~\ref{fig:selection_score2_model_4},
while Figure~\ref{fig:performance_score2_model_4}
illustrates the performance of this model as measured by
concordance index and integrated Brier score.  A scatter plot of the
predicted risks against recalibrated SCORE2
is shown in Figure~\ref{fig:scatter_score2_model_4}.
When using all of the available data,
the concordance index improved by
$1.56\%$ for females and by $2.00\%$ for males,
while the integrated Brier score improved by
$0.17\%$ for females and by $0.41\%$ for males.

\begin{figure}[H]
  \centering
  \subfloat[Cross-validated choices of $\theta$ for females.]{
    \includegraphics[scale=\myfigscale]{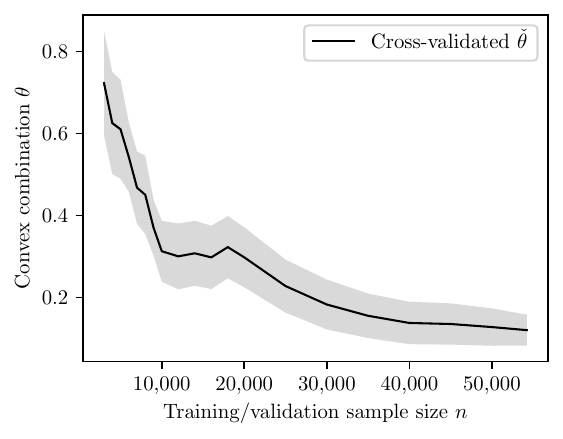}
  }
  \hspace*{5mm}
  \subfloat[Cross-validated choices of $\theta$ for males.]{
    \includegraphics[scale=\myfigscale]{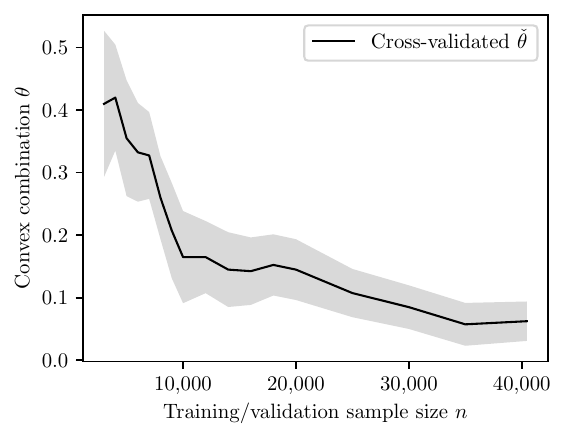}
  }
  \caption{Convex combination parameter choices in Model 4.}
  \label{fig:selection_score2_model_4}
\end{figure}

\begin{figure}[H]
  \centering
  \hspace*{2mm}
  \subfloat[Concordance index for females.]{
    \includegraphics[scale=\myfigscale]{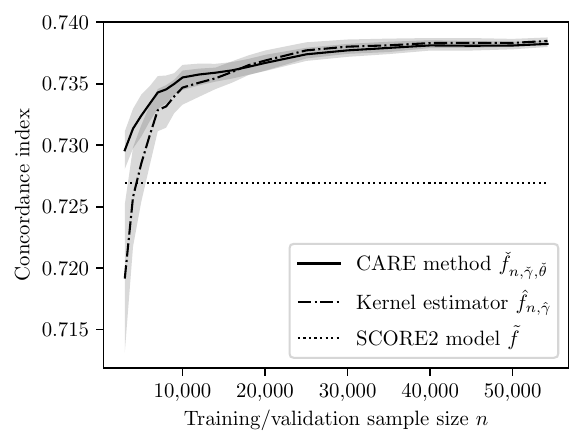}
  }
  \hspace*{8mm}
  \subfloat[Concordance index for males.]{
    \includegraphics[scale=\myfigscale]{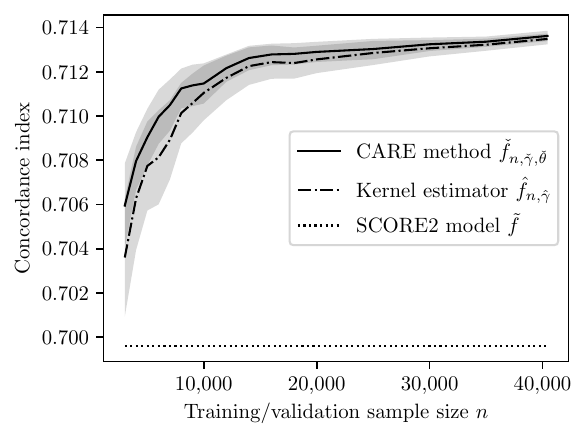}
  } \\
  \subfloat[Integrated Brier score for females.]{
    \includegraphics[scale=\myfigscale]{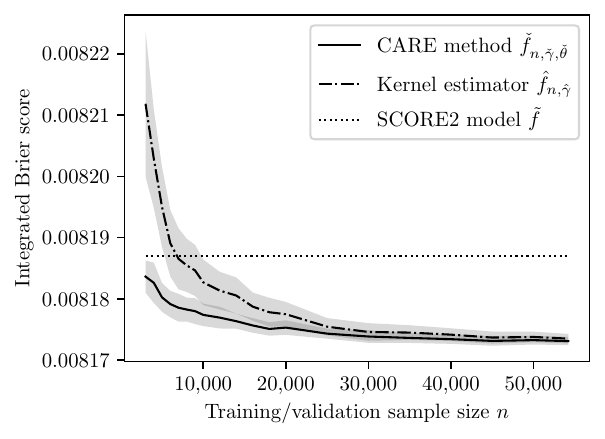}
  }
  \hspace*{6mm}
  \subfloat[Integrated Brier score for males.]{
    \includegraphics[scale=\myfigscale]{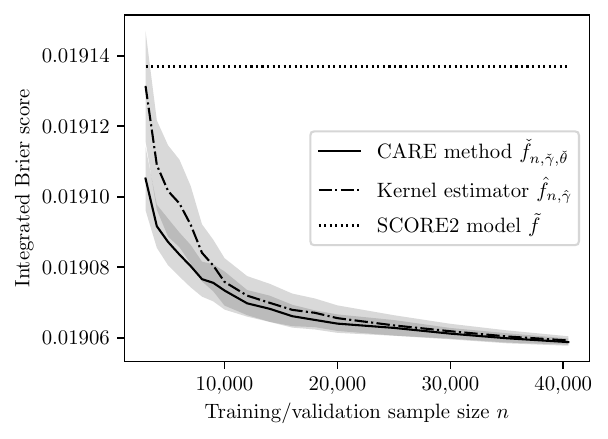}
  }
  \caption{Comparison of estimators in Model 4.}
  \label{fig:performance_score2_model_4}
\end{figure}

\begin{figure}[H]
  \centering
  \subfloat[Predicted risk for females.]{
    \includegraphics[scale=\myfigscale]{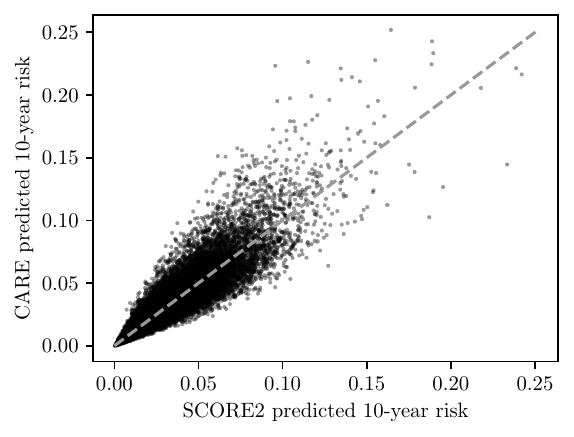}
  }
  \hspace*{5mm}
  \subfloat[Predicted risk for males.]{
    \includegraphics[scale=\myfigscale]{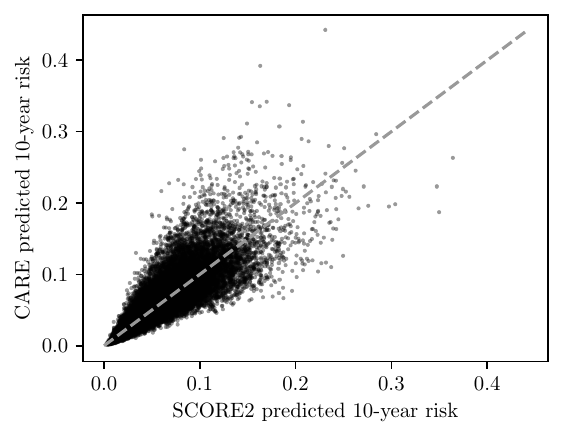}
  }
  \caption{Risk predictions for Model 4 compared to SCORE2.}
  \label{fig:scatter_score2_model_4}
\end{figure}

\subsubsection{Model 5: SCORE2 and UKB with deprivation index
and polygenic scores, Sobolev kernel}

In Model 5, we add the deprivation index variable \texttt{imd}
and the polygenic scores \texttt{pgs000018} and \texttt{pgs000039}
from UKB to SCORE2, this time using a shifted
first order Sobolev kernel, computing the RKHS estimators using
a Nystr{\"o}m subsampling size of $100$.
The cross-validated convex combination parameter is
shown in Figure~\ref{fig:selection_score2_model_5},
while Figure~\ref{fig:performance_score2_model_5}
illustrates the performance of this model as measured by
concordance index and integrated Brier score.
A scatter plot of the predicted risks against recalibrated SCORE2
is shown in Figure~\ref{fig:scatter_score2_model_5}.
When using all of the available data,
the concordance index improved by
$1.75\%$ for females and by $2.00\%$ for males,
while the integrated Brier score improved by
$0.17\%$ for females and by $0.32\%$ for males.

\begin{figure}[H]
  \centering
  \subfloat[Cross-validated choices of $\theta$ for females.]{
    \includegraphics[scale=\myfigscale]{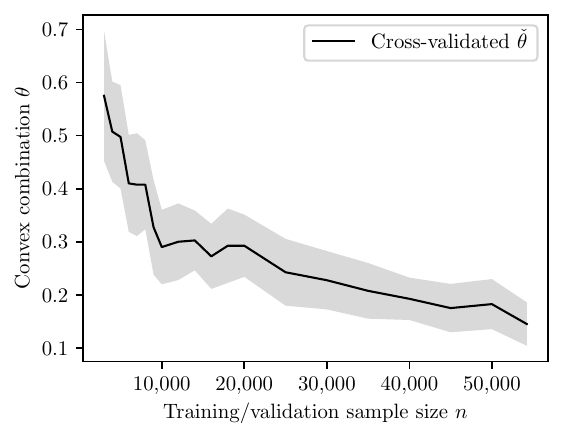}
  }
  \hspace*{5mm}
  \subfloat[Cross-validated choices of $\theta$ for males.]{
    \includegraphics[scale=\myfigscale]{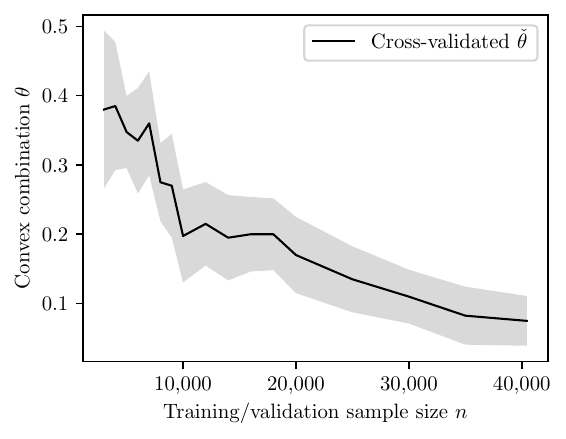}
  }
  \caption{Convex combination parameter choices in Model 5.}
  \label{fig:selection_score2_model_5}
\end{figure}

\begin{figure}[H]
  \centering
  \iftoggle{journal}{\hspace*{2mm}}{\hspace*{3mm}}
  \subfloat[Concordance index for females.]{
    \includegraphics[scale=\myfigscale]{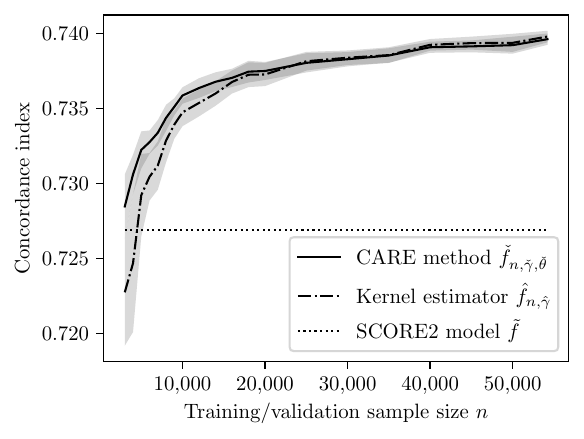}
  }
  \hspace*{8mm}
  \subfloat[Concordance index for males.]{
    \includegraphics[scale=\myfigscale]{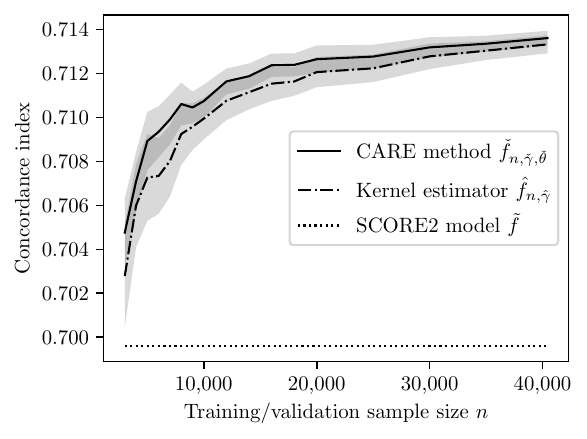}
  } \\
  \subfloat[Integrated Brier score for females.]{
    \includegraphics[scale=\myfigscale]{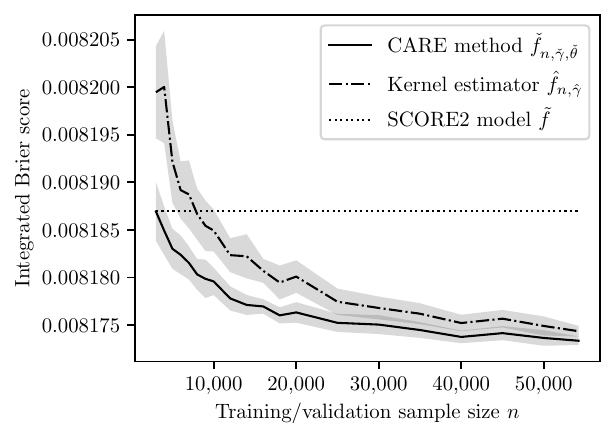}
  }
  \hspace*{6mm}
  \subfloat[Integrated Brier score for males.]{
    \includegraphics[scale=\myfigscale]{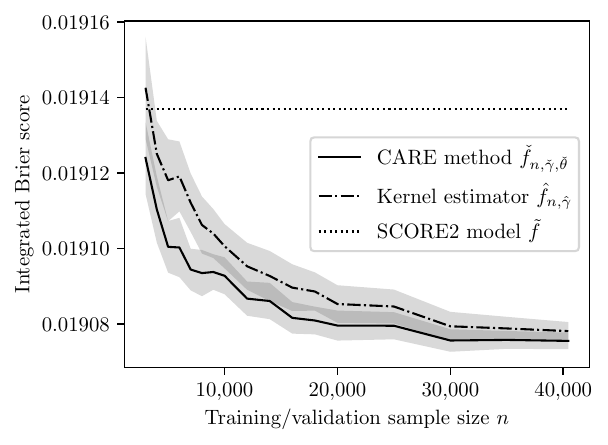}
  }
  \caption{Comparison of estimators in Model 5.}
  \label{fig:performance_score2_model_5}
\end{figure}

\begin{figure}[H]
  \centering
  \subfloat[Predicted risk for females.]{
    \includegraphics[scale=\myfigscale]{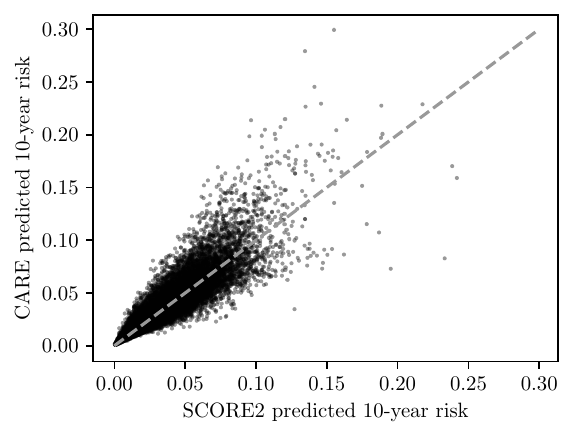}
  }
  \hspace*{5mm}
  \subfloat[Predicted risk for males.]{
    \includegraphics[scale=\myfigscale]{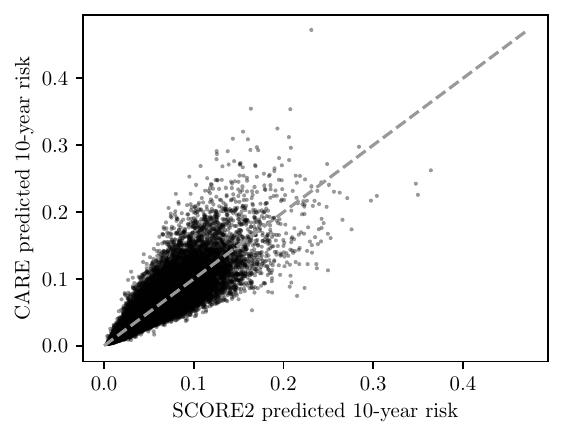}
  }
  \caption{Risk predictions for Model 5 compared to SCORE2.}
  \label{fig:scatter_score2_model_5}
\end{figure}

\section{Proofs and technical details}%
\label{sec:proofs}

In this section, we present full proofs for all of our main results, along with
auxiliary lemmas and technical details.
Several of our proofs rely on the
decoupling and symmetrisation inequalities for
$U$-processes given in Appendix~\ref{sec:u_statistics},
which may be of independent interest.

\subsection{Proofs for Section~\ref{sec:set_up}}

We begin with some properties of the
partial likelihood functions introduced in Section~\ref{sec:data}.

\subsubsection{Gateaux derivatives of \texorpdfstring{$\ell_n$}{ln}
and \texorpdfstring{$\ell_\star$}{lstar}}

Firstly, Lemmas~\ref{lem:derivatives_Sn} and \ref{lem:derivatives_ln}
characterise the Gateaux derivatives of the negative
log-partial likelihood $\ell_n$.
Next, Lemmas~\ref{lem:derivatives_S} and~\ref{lem:derivatives_lstar}
yield the Gateaux derivatives of the limiting negative log-partial
likelihood $\ell_\star$.

\begin{lemma}[Gateaux derivatives of $S_n$]%
  \label{lem:derivatives_Sn}
  Let $f, f_1, \ldots, f_m: \cX \to \R$. Then for $m \in \N_0$, the $m$-th order
  Gateaux derivative of $S_n(f, t)$
  in the direction $(f_1, \ldots, f_m)$ is given by
  \begin{equation*}
    D^m S_n(f, t)(f_1, \ldots, f_m)
    = \frac{1}{n} \sum_{i=1}^n
    R_i(t)
    e^{f(X_i)}
    \prod_{j=1}^{m} f_j(X_i).
  \end{equation*}
\end{lemma}

\begin{proof}[Proof of Lemma~\ref{lem:derivatives_Sn}]
  We proceed by induction on $m$; the result holds
  for $m = 0$ by definition of $S_n(f, t)$.
  Now for $m \in \N$ and
  $f, f_1, \ldots, f_m: \cX \to \R$,
  suppose that
  \begin{align*}
    D^{m-1} S_n(f, t)\left(f_1, \ldots, f_{m-1}\right)
    &=
    \frac{1}{n} \sum_{i=1}^n
    R_i(t)
    e^{f(X_i)}
    \prod_{j=1}^{m-1} f_j(X_i).
  \end{align*}
  By the definition of the Gateaux derivative and the induction hypothesis,
  for every $\varepsilon \in \R \setminus \{0\}$ and $t \in [0,1]$, we have
  \begin{align*}
    &\frac{1}{\varepsilon}
    \bigl(
      D^{m-1} S_n(f + \varepsilon f_m, t)(f_1, \ldots, f_{m-1})
      - D^{m-1} S_n(f, t)(f_1, \ldots, f_{m-1})
    \bigr) \\
    &\quad=
    \frac{1}{n} \sum_{i=1}^n
    R_i(t)
    \frac{1}{\varepsilon}
    \bigl(
      e^{f(X_i) + \varepsilon f_m(X_i)} - e^{f(X_i)}
    \bigr)
    \prod_{j=1}^{m-1} f_j(X_i)
    \to
    \frac{1}{n} \sum_{i=1}^n
    R_i(t)
    e^{f(X_i)}
    \prod_{j=1}^{m} f_j(X_i)
  \end{align*}
  as $\varepsilon \rightarrow 0$, as required.
\end{proof}

\begin{lemma}[Gateaux derivatives of $\ell_n$]%
  \label{lem:derivatives_ln}
  Let $f, f_1, f_2: \cX \to \R$.
  Then the first-order
  Gateaux derivative of $\ell_n(f)$
  in the direction $f_1$ exists and is given by
  \begin{align*}
    D \ell_n(f)(f_1)
    &=
    \int_{0}^{1}
    \frac{D S_n(f, t)(f_1)}{S_n(f, t)}
    \diffi {N(t)}
    -
    \frac{1}{n}
    \sum_{i=1}^n
    f_1(X_i)
    N_i(1).
  \end{align*}
  This is well-defined as $S_n(f, T_i) > 0$ for
  $i \in [n]$. The second-order
  Gateaux derivative of $\ell_n(f)$
  in the direction $(f_1, f_2)$ also exists, and is given by
  \begin{align*}
    D^2 \ell_n(f)(f_1, f_2)
    &=
    \int_{0}^{1}
    \biggl(
      \frac{D^2 S_n(f, t)(f_1, f_2)}{S_n(f, t)}
      - \frac{D S_n(f, t)(f_1)}{S_n(f, t)}
      \frac{D S_n(f, t)(f_2)}{S_n(f, t)}
    \biggr)
    \diffi {N(t)}.
  \end{align*}
  Further, $D^2 \ell_n(f)(f_1, f_1) \geq 0$,
  so $\ell_n$ is a convex function.
  Finally, for $f, f_1, f_2, f_3: \cX \to \R$,
  the third-order Gateaux derivative
  $D^3 \ell_n(f)(f_1, f_2, f_3)$ exists.
\end{lemma}

\begin{proof}[Proof of Lemma~\ref{lem:derivatives_ln}]
  By definition of the negative log-partial likelihood,
  for any $f : \cX \to \R$,
  \begin{align*}
    \ell_n(f)
    &=
    \int_{0}^{1} \log S_n(f, t) \diffi {N(t)}
    - \frac{1}{n} \sum_{i=1}^n f(X_i) N_i(1) \\
    &=
    \frac{1}{n}
    \sum_{i=1}^n
    (1 - I_i)
    \log S_n(f, T_i)
    - \frac{1}{n} \sum_{i=1}^n f(X_i) N_i(1).
  \end{align*}
  For $\varepsilon \in \R \setminus \{0\}$,
  it follows from the chain rule that
  \begin{align*}
    \frac{1}{\varepsilon}
    \bigl(
      \ell_n(f + \varepsilon f_1)
      &- \ell_n(f)
    \bigr) \\
    &=
    \frac{1}{\varepsilon n}
    \sum_{i=1}^n
    (1 - I_i)
    \bigl(
      \log S_n(f + \varepsilon f_1, T_i)
      - \log S_n(f, T_i)
    \bigr)
    - \frac{1}{n} \sum_{i=1}^n
    f_1(X_i) N_i(1) \\
    &\to
    \frac{1}{n}
    \sum_{i=1}^n
    (1 - I_i)
    \frac{D S_n(f, T_i)(f_1)}{S_n(f, T_i)}
    - \frac{1}{n} \sum_{i=1}^n
    f_1(X_i) N_i(1) \\
    &=
    \int_0^1
    \frac{D S_n(f, t)(f_1)}{S_n(f, t)}
    \diffi N(t)
    - \frac{1}{n} \sum_{i=1}^n
    f_1(X_i) N_i(1)
  \end{align*}
  as $\varepsilon \to 0$.
  Similarly, by the quotient rule, for $\varepsilon \in \R \setminus \{0\}$,
  \begin{align*}
    \frac{1}{\varepsilon}
    \bigl(
      D \ell_n(f &+ \varepsilon f_2)(f_1)
      - D \ell_n(f)(f_1)
    \bigr) \\
    &=
    \frac{1}{\varepsilon n}
    \sum_{i=1}^n
    (1 - I_i)
    \biggl(
      \frac{D S_n(f + \varepsilon f_2, T_i)(f_1)}
      {S_n(f + \varepsilon f_2, T_i)}
      - \frac{D S_n(f, T_i)(f_1)}{S_n(f, T_i)}
    \biggr) \\
    &\to
    \frac{1}{n}
    \sum_{i=1}^n
    (1 - I_i)
    \biggl(
      \frac{D^2 S_n(f, T_i)(f_1, f_2)}{S_n(f, T_i)}
      - \frac{D S_n(f, T_i)(f_1)}{S_n(f, T_i)}
      \frac{D S_n(f, T_i)(f_2)}{S_n(f, T_i)}
    \biggr) \\
    &=
    \int_{0}^{1}
    \biggl(
      \frac{D^2 S_n(f, t)(f_1, f_2)}{S_n(f, t)}
      - \frac{D S_n(f, t)(f_1)}{S_n(f, t)}
      \frac{D S_n(f, t)(f_2)}{S_n(f, t)}
    \biggr)
    \diffi {N(t)}
  \end{align*}
  as $\varepsilon \to 0$, as required.
  To show convexity of $\ell_n$, note that
  \begin{align*}
    0
    &\leq
    \int_{0}^{1}
    \frac{1}{n}
    \sum_{i=1}^n
    \frac{R_i(t) e^{f(X_i)}}{S_n(f, t)}
    \bigg(
      f_1(X_i) - \frac{D S_n(f, t)(f_1)}{S_n(f, t)}
    \bigg)^2
    \diffi N(t) \\
    &=
    \int_{0}^{1}
    \frac{1}{S_n(f, t)}
    \biggl(
      D^2 S_n(f, t)(f_1, f_1)
      - 2 \frac{D S_n(f, t)(f_1)^2}{S_n(f, t)}
      + \frac{D S_n(f, t)(f_1)^2}{S_n(f, t)}
    \biggr)
    \diffi N(t) \\
    &=
    \int_{0}^{1}
    \biggl(
      \frac{D^2 S_n(f, t)(f_1, f_1)}{S_n(f, t)}
      - \frac{D S_n(f, t)(f_1)^2}{S_n(f, t)^2}
    \biggr)
    \diffi {N(t)} =
    D^2 \ell_n(f)(f_1, f_1).
  \end{align*}
  Finally, we show that the third-order Gateaux derivative exists.
  For $\varepsilon \in \R \setminus \{0\}$,
  by the dominated convergence
  theorem, the chain rule, the product rule and the quotient rule,
  \begin{align}
    \nonumber
    &\frac{1}{\varepsilon}
    \Bigl\{
      D^2 \ell_n(f + \varepsilon f_3)(f_1, f_2)
      - D^2 \ell_n(f)(f_1, f_2)
    \Bigr\} \\
    \nonumber
    &\,=
    \int_{0}^{1}
    \frac{1}{\varepsilon}
    \biggl(
      \frac{D^2 S_n(f + \varepsilon f_3, t)(f_1, f_2)}
      {S_n(f + \varepsilon f_3, t)}
      - \frac{D^2 S_n(f, t)(f_1, f_2)}{S_n(f, t)}
    \biggr)
    \diffi {N(t)} \\
    \nonumber
    &\quad-
    \int_{0}^{1}
    \frac{1}{\varepsilon}
    \biggl(
      \frac{D S_n(f + \varepsilon f_3, t)(f_1)}
      {S_n(f + \varepsilon f_3, t)}
      \frac{D S_n(f + \varepsilon f_3, t)(f_2)}
      {S_n(f + \varepsilon f_3, t)}
      - \frac{D S_n(f, t)(f_1)}{S_n(f, t)}
      \frac{D S_n(f, t)(f_2)}{S_n(f, t)}
    \biggr)
    \diffi {N(t)} \\
    \nonumber
    &\,\stackrel{\epsilon \rightarrow 0}{\to}
    \int_{0}^{1}
    \biggl(
      \frac{D^3 S_n(f, t)(f_1, f_2, f_3)}{S_n(f, t)}
      - \frac{D^2 S_n(f, t)(f_1, f_2) D S_n(f, t)(f_3)}{S_n(f, t)^2}
    \biggr)
    \diffi {N(t)} \\
    \nonumber
    &\quad-
    \int_{0}^{1}
    \frac{D S_n(f, t)(f_1)}{S_n(f, t)}
    \biggl(
      \frac{D^2 S_n(f, t)(f_2, f_3)}{S_n(f, t)}
      - \frac{D S_n(f, t)(f_2) D S_n(f, t)(f_3)}{S_n(f, t)^2}
    \biggr)
    \diffi {N(t)} \\
    \nonumber
    &\quad-
    \int_{0}^{1}
    \frac{D S_n(f, t)(f_2)}{S_n(f, t)}
    \biggl(
      \frac{D^2 S_n(f, t)(f_1, f_3)}{S_n(f, t)}
      - \frac{D S_n(f, t)(f_1) D S_n(f, t)(f_3)}{S_n(f, t)^2}
    \biggr)
    \diffi {N(t)} \\
    \label{eq:third_derivative_ln}
    &\,=
    \int_{0}^{1}
    \biggl(
      \frac{D^3 S_n(f, t)(f_1, f_2, f_3)}{S_n(f, t)}
      - \frac{D^2 S_n(f, t)(f_1, f_2) D S_n(f, t)(f_3)}{S_n(f, t)^2} \\
      \nonumber
      &\qquad\qquad-
      \frac{D^2 S_n(f, t)(f_1, f_3) D S_n(f, t)(f_2)}{S_n(f, t)^2}
      - \frac{D^2 S_n(f, t)(f_2, f_3) D S_n(f, t)(f_1)}{S_n(f, t)^2} \\
      \nonumber
      &\qquad\qquad+
      2 \frac{D S_n(f, t)(f_1) D S_n(f, t)(f_2) D S_n(f_3)}{S_n(f, t)^3}
    \biggr)
    \diffi {N(t)}.
    \qedhere
  \end{align}
\end{proof}

\begin{lemma}[Gateaux derivatives of $S$]
  \label{lem:derivatives_S}
  Take $m \in \N_0$ and let
  $f, f_1, \ldots, f_m \in \cB(\cX)$. Then
  \begin{align*}
    \E \bigl(
      D^m S_n(f, t)(f_1, \ldots, f_m)
    \bigr)
    = \int_\cX
    q(x, t)
    e^{f(x)}
    \prod_{j=1}^{m} f_j(x)
    \diffi P_X(x)
    = D^m S(f, t)\left(f_1, \ldots, f_m\right).
  \end{align*}
\end{lemma}

\begin{proof}[Proof of Lemma~\ref{lem:derivatives_S}]
  We proceed by induction on $m$. For $m = 0$, we have
  \begin{align*}
    \E \bigl(
      S_n(f, t)
    \bigr)
    &=
    \E \bigl(
      R(t) e^{f(X)}
    \bigr)
    = \int_\cX
    q(x, t)
    e^{f(x)}
    \diffi P_X(x)
    = S(f, t).
  \end{align*}
  For $m \in \N$ and
  $f, f_1, \ldots, f_m \in \cB(\cX)$,
  suppose that
  \begin{align*}
    \E \bigl(
      D^{m-1} S_n(f, t)(f_1, \ldots, f_{m-1})
    \bigr)
    &= \int_\cX
    q(x, t)
    e^{f(x)}
    \prod_{j=1}^{m-1} f_j(x)
    \diffi P_X(x) \\
    &= D^{m-1} S(f, t)(f_1, \ldots, f_{m-1}).
  \end{align*}
  Then by Lemma~\ref{lem:derivatives_Sn},
  \begin{align*}
    \E \bigl(
      D^m S_n(f, t)(f_1, \ldots, f_m)
    \bigr)
    &=
    \E \biggl\{
      R(t)
      e^{f(X)}
      \prod_{j=1}^{m} f_j(X)
    \biggr\}
    = \int_0^1
    q(x, t)
    e^{f(x)}
    \prod_{j=1}^{m} f_j(x)
    \diffi P_X(x).
  \end{align*}
  It remains to find the Gateaux derivative
  $D^m S(f, t)(f_1, \ldots, f_m)$.
  By the induction hypothesis and the dominated convergence theorem,
  for $\varepsilon \in \R \setminus \{0\}$,
  \begin{align*}
    &\frac{1}{\varepsilon}
    \bigl(
      D^{m-1} S(f + \varepsilon f_m, t)(f_1, \ldots, f_{m-1})
      - D^{m-1} S(f, t)(f_1, \ldots, f_{m-1})
    \bigr) \\
    &\quad=
    \frac{1}{\varepsilon}
    \biggl(
      \int_\cX
      q(x, t)
      e^{f(x)}
      \bigl(
        e^{\varepsilon f_m(x)} - 1
      \bigr)
      \prod_{j=1}^{m-1} f_j(x)
      \diffi P_X(x)
    \biggr)
    \to
    \int_\cX
    q(x, t)
    e^{f(x)}
    \prod_{j=1}^{m} f_j(x)
    \diffi P_X(x)
  \end{align*}
  as $\varepsilon \to 0$, as required.
\end{proof}

\begin{lemma}[Gateaux derivatives of $\ell_\star$]%
  \label{lem:derivatives_lstar}
  Let $f, f_1, f_2 \in \cB(\cX)$.
  Then the first-order
  Gateaux derivative of $\ell_{\star}(f)$
  in the direction of $f_1$ exists and is given by
  \begin{align*}
    D \ell_{\star}(f)(f_1)
    &=
    \int_{0}^{1}
    D S(f, t)(f_1)
    \frac{S(f_0, t)}{S(f, t)}
    \lambda_0(t) \diffi t
    - \int_{0}^{1}
    D S(f_0, t)(f_1)
    \lambda_0(t)
    \diffi t.
  \end{align*}
  Similarly, the second-order Gateaux derivative of $\ell_{\star}(f)$
  in the direction of $(f_1, f_2)$ is
  \begin{align*}
    D^2 \ell_{\star}(f)(f_1, f_2)
    &=
    \int_{0}^{1}
    \biggl(
      \frac{D^2 S(f, t)(f_1, f_2)}{S(f, t)}
      - \frac{D S(f, t)(f_1)}{S(f, t)}
      \frac{D S(f, t)(f_2)}{S(f, t)}
    \biggr)
    S(f_0, t)
    \lambda_0(t) \diffi t.
  \end{align*}
  In each display we interpret $0/0$ to be zero.
\end{lemma}

\begin{proof}[Proof of Lemma~\ref{lem:derivatives_lstar}]
  By Lemma~\ref{lem:derivatives_S},
  \begin{align*}
    \ell_{\star}(f)
    &=
    \int_{0}^{1}
    \log\bigl(S(f, t)\bigr)
    S(f_0, t) \lambda_0(t) \diffi t
    - \int_{0}^{1}
    \int_\cX
    q(x, t)
    e^{f_0(x)}
    f(x)
    \diffi P_X(x)
    \lambda_0(t)
    \diffi t.
  \end{align*}
  Note that $\int_0^1 \lambda_0(t) \diffi t = \Lambda < \infty$.
  Hence by the dominated convergence theorem, since
  $f, f_0, f_1 \in \cB(\cX)$,
  and as $S(f_0, t) \leq e^{\|\tilde f - f_0\|_\infty} S(\tilde f, t)$
  for all $\tilde f \in \cB(\cX)$, we have
  \begin{align*}
    \frac{1}{\varepsilon}
    \bigl(
      \ell_\star(f + \varepsilon f_1) - \ell_{\star}(f)
    \bigr) &=
    \frac{1}{\varepsilon}
    \int_{0}^{1}
    \bigl(
      \log S(f + \varepsilon f_1, t)
      - \log S(f, t)
    \bigr)
    S(f_0, t) \lambda_0(t) \diffi t \\
    &\hspace{4cm}-
    \int_{0}^{1}
    \int_\cX
    q(x, t)
    e^{f_0(x)}
    f_1(x)
    \diffi P_X(x)
    \lambda_0(t)
    \diffi t \\
    &\to
    \int_{0}^{1}
    \frac{D S(f, t)(f_1)}{S(f, t)}
    S(f_0, t) \lambda_0(t) \diffi t
    - \int_{0}^{1}
    D S(f_0, t)(f_1)
    \lambda_0(t)
    \diffi t
  \end{align*}
  as $\varepsilon \to 0$.
  Similarly, for the second derivative,
  again by the dominated convergence theorem,
  \begin{align*}
    \frac{1}{\varepsilon}
    \bigl(
      D \ell_\star(f + \varepsilon f_2)(f_1)
      &- D \ell_{\star}(f)(f_1)
    \bigr) \\
    &=
    \frac{1}{\varepsilon}
    \int_{0}^{1}
    \biggl(
      \frac{D S(f + \varepsilon f_2, t)(f_1)}{S(f + \varepsilon f_2, t)}
      - \frac{D S(f, t)(f_1)}{S(f, t)}
    \biggr)
    S(f_0, t) \lambda_0(t) \diffi t \\
    &\to
    \int_{0}^{1}
    \biggl(
      \frac{D^2 S(f, t)(f_1, f_2)}{S(f, t)}
      - \frac{D S(f, t)(f_1)}{S(f, t)}
      \frac{D S(f, t)(f_2)}{S(f, t)}
    \biggr)
    S(f_0, t)
    \lambda_0(t) \diffi t
  \end{align*}
  as $\varepsilon \to 0$, as required.
\end{proof}

\subsubsection{Proofs of Lemmas~\ref{lem:characterisation_f0}
and~\ref{lem:infty_bounds}}

The proof of Lemma~\ref{lem:characterisation_f0} relies on the
following restricted
strong convexity result for $\ell_\star$.

\begin{lemma}[Restricted strong convexity of $\ell_\star$]%
  \label{lem:strong_convexity}
  If $f, f_1 \in \cB(\cX)$, then
  \begin{align*}
    e^{-\|f - f_0\|_\infty}
    e^{-\|f\|_\infty}
    \Lambda q_1
    \|f_1 - P_X(f_1)\|_{L_2}^2
    &\leq
    D^2 \ell_{\star}(f)(f_1, f_1)
    \leq
    e^{\|f-f_0\|_\infty} e^{\|f\|_\infty}
    \Lambda
    \|f_1\|_{L_2}^2.
  \end{align*}
\end{lemma}

\begin{proof}[Proof of Lemma~\ref{lem:strong_convexity}]
  For the upper bound, observe
  by Lemma~\ref{lem:derivatives_S} that
  \[
    D^2 S(f, t)(f_1, f_1)
    = \int_\cX q(x, t) e^{f(x)} f_1(x)^2 \diffi P_X(x)
    \leq e^{\|f\|_\infty} \|f_1\|_{L_2}^2,
  \]
  so by Lemma~\ref{lem:derivatives_lstar},
  \begin{align*}
    D^2 \ell_{\star}(f)(f_1, f_1)
    &\leq \int_{0}^{1}
    \frac{D^2 S(f, t)(f_1, f_1)}{S(f, t)}
    S(f_0, t) \lambda_0(t) \diffi t
    \leq
    e^{\|f-f_0\|_\infty} e^{\|f\|_\infty}
    \Lambda
    \|f_1\|_{L_2}^2.
  \end{align*}
  Since $\ell_\star(f) = \ell_\star(f + c 1_\cX)$
  for all $c \in \R$, we have
  \begin{align*}
    D^2 \ell_{\star}(f)(f_1, f_1)
    &= D^2 \ell_{\star}(f)(f_2, f_2),
  \end{align*}
  where $f_2 \vcentcolon= f_1 - P_X(f_1) 1_\cX$.
  For the lower bound, by Lemma~\ref{lem:derivatives_lstar} and
  Lemma~\ref{lem:derivatives_S} again,
  \begin{align*}
    D^2 \ell_{\star}(f)(f_2, f_2)
    &= \int_{0}^{1}
    \biggl(
      \frac{D^2 S(f, t)(f_2, f_2)}{S(f, t)}
      - \frac{D S(f, t)(f_2)^2}{S(f, t)^2}
    \biggr)
    S(f_0, t) \lambda_0(t) \diffi t \\
    &= \int_{0}^{1}
    \int_\cX
    \biggl(
      f_2(x) - \frac{D S(f, t)(f_2)}{S(f, t)}
    \biggr)^2
    e^{f(x)} q(x, t)
    \diffi P_X(x) \,
    \frac{S(f_0, t)}{S(f, t)}
    \lambda_0(t) \diffi t.
  \end{align*}
  Since $P_X(f_2) = 0$,
  we have $\int_\cX \bigl(f_2(x) - c\bigr)^2 \diffi P_X(x)
  \geq \|f_2\|_{L_2}^2$ for every $c \in \R$,
  so because $S(f_0, t) \geq e^{- \|f - f_0\|_\infty} S(f, t)$,
  we have
  \begin{align*}
    D^2 \ell_{\star}(f)(f_2, f_2)
    &\geq
    e^{-\|f - f_0\|_\infty}
    e^{-\|f\|_\infty}
    \Lambda q_1
    \|f_2\|_{L_2}^2.
    \qedhere
  \end{align*}
\end{proof}

\begin{proof}[Proof of Lemma~\ref{lem:characterisation_f0}]
  By Lemma~\ref{lem:derivatives_lstar},
  $D \ell_\star(f_0) = 0$. Let $f \in \cB(\cX)$.
  By Taylor's theorem and Lemma~\ref{lem:strong_convexity},
  as $P_X(f_0) = 0$, there exists $\check f$
  on the line segment between $f_0$ and $f$ with
  \begin{align}
    \label{eq:characterisation_f0_convex}
    \ell_\star(f) - \ell_\star(f_0)
    &= D^2 \ell(\check f)(f - f_0)^{\otimes 2}
    \geq
    e^{-\|\check f - f_0\|_\infty}
    e^{-\|\check f\|_\infty}
    \Lambda q_1
    \|f - P_X(f) - f_0\|_{L_2}^2 \geq 0,
  \end{align}
  so $f_0 \in \argmin_{f \in \cB(\cX)} \ell_\star(f)$.
  Moreover, if $\tilde f \in \argmin_{f \in \cB(\cX)} \ell_\star(f)$
  with $P_X(\tilde f) = 0$, then from~\eqref{eq:characterisation_f0_convex},
  we have $\|\tilde f - f_0\|_{L_2} = 0$, so
  $\tilde f(x) = f_0(x)$ holds for $P_X$-almost every $x \in \cX$.
\end{proof}

\begin{proof}[Proof of Lemma~\ref{lem:infty_bounds}]
  By Cauchy--Schwarz, for $x \in \cX$,
  \begin{align*}
    f(x)^2
    &=
    \biggl(
      \sum_{r\in \mathcal{R}} \langle f, e_r \rangle_{L_2} e_r(x)
    \biggr)^2
    =
    \biggl(
      \sum_{r\in \mathcal{R}} \frac{1}{\sqrt{\nu_r}} \langle f, e_r
      \rangle_{L_2}
      \cdot \sqrt{\nu_r} e_r(x)
    \biggr)^2 \\
    &\leq \|f\|_\cH^2
    \sum_{r\in \mathcal{R}} \nu_r e_r(x)^2
    = \|f\|_\cH^2 k(x, x)
    \leq K \|f\|_\cH^2.
  \end{align*}
  Similarly,
  \begin{equation*}
    f(x)^2
    =
    \biggl(
      \sum_{r\in \mathcal{R}} \frac{1}{\sqrt{a_r}} \langle f, e_r
      \rangle_{L_2}
      \cdot \sqrt{a_r} e_r(x)
    \biggr)^2
    \leq \|f\|_{\cH_\gamma}^2
    \sum_{r\in \mathcal{R}} a_r e_r(x)^2
    \leq H_{\infty,\gamma} \|f\|_{\cH_\gamma}^2.
    \qedhere
  \end{equation*}
\end{proof}

\subsection{Proofs for Section~\ref{sec:methodology}}

We prove our result on representation of the kernel estimator
$\hat f_{n,\gamma}$.

\begin{proof}[Proof of Proposition~\ref{prop:representation}]
  We treat $(X_i,T_i,I_i)_{i \in [n]}$ as fixed for now;
  measurability issues are addressed later.
  For $\gamma > 0$ and $f \in \cH$,
  define $\check \ell_{n,\gamma}(f)
  \vcentcolon= \ell_n(f) + \gamma
  \|f - P_n(f) 1_\cX\|_\cH^2 + \gamma P_n(f)^2$.
  Since $f \mapsto \ell_n(f)$ is invariant
  to addition of constants to $f$, the optimisation
  problem \eqref{eq:argmin} is equivalent to solving
  \begin{equation}
    \label{eq:argmin_unconstrained}
    \hat f_{n,\gamma} \vcentcolon=
    \argmin_{f \in \cH} \check\ell_{n, \gamma}(f).
  \end{equation}
  Define the subspace
  $\cH_{\bX} \vcentcolon= \mathrm{span} \bigl\{1_\cX, k(\cdot, X_1),
  \ldots, k(\cdot, X_n)\bigl\}$ of $\mathcal{H}$, which is
  finite-dimensional and hence closed.
  Further define $\pi_\bX : \cH \to \cH_\bX$ to be the projection operator
  onto $\cH_\bX$ under $\langle \cdot, \cdot\rangle_\cH$, and
  for $f \in \cH$ write $f_{\bX} \vcentcolon= \pi_\bX(f)$.
  By the reproducing property and orthogonality,
  $f(X_i) - f_{\bX}(X_i)
  = \langle f-f_{\bX}, k(\cdot, X_i)\rangle_\cH = 0$
  for each $i \in [n]$. Thus, $\ell_n(f_{\bX}) = \ell_n(f)$ and
  $P_n(f_{\bX}) = P_n(f)$. Hence, since $\pi_\bX(1_\cX) = 1_\cX$, we have
  $\pi_\bX\bigl(f - P_n(f) 1_\cX\bigr) = f_\bX - P_n(f_\bX) 1_\cX$.
  It follows that
  \[
    \|f - P_n(f) 1_\cX\|_\cH^2 = \|f_{\bX} - P_n(f_\bX) 1_\cX\|_\cH^2
    + \|f - f_{\bX}\|_\cH^2,
  \]
  so $\check\ell_{n,\gamma}(f_\bX) \leq \check\ell_{n,\gamma}(f)$,
  with equality if and only if $f \in \cH_\bX$.
  But every $f \in \cH_{\bX}$ is of the form
  $f(x) = \beta_0 + \sum_{i=1}^n k(x, X_i) \beta_i$, where $\beta_i
  \in \R$ for $i \in \{0\} \cup [n]$. Since $\ell_n(f) = \ell_n(f + a 1_\cX)$
  and $f + a 1_\cX - P_n(f + a 1_\cX) = f - P_n(f)$ for all $a \in
  \R$, we deduce that
  any minimiser~$f$ of~\eqref{eq:argmin_unconstrained}
  must satisfy $P_n(f) = 0$ and $f \in \cH_\bX$.

  Define $\tilde \cH_\bX \vcentcolon=
  \bigl\{f \in \cH : f(x) = \beta_0 + \sum_{i \in \cA_n} k(x, X_i)
  \beta_i, P_n(f) = 0, \beta \in \R^{\{0\} \cup \cA_n}\bigr\}$.
  Our final reduction is to show that it suffices to consider
  functions $f \in \tilde \cH_{\bX}$.
  Indeed, take $\delta \in \R^{\{0\} \cup [n]}$
  such that the function
  $f \in \cH_{\bX}$ defined by
  $f(x) \vcentcolon=
  \delta_0 + \sum_{i=1}^n k(x, X_i) \delta_i$ satisfies $P_n(f) = 0$.
  Then with $\tilde{\bK}$ defined as in \eqref{eq:K_tilde},
  by the reproducing property,
  \begin{align*}
    \|f - P_n(f)\|_\cH^2
    &= \|f\|_\cH^2
    = \bigg\langle
    \delta_0 1_\cX + \sum_{i=1}^n k(\cdot, X_i) \delta_i,\;
    \delta_0 1_\cX + \sum_{i=1}^n k(\cdot, X_i) \delta_i
    \bigg\rangle_\cH \\
    &= \delta_0^2 \|1_\cX\|_\cH^2
    + 2 \delta_0 \sum_{i=1}^n \delta_i
    + \sum_{i=1}^n \sum_{j=1}^n \delta_i \delta_j k(X_i, X_j)
    = \delta^\T \tilde{\bK} \delta.
  \end{align*}
  Now, by definition of $\cA_n$,
  there exist $\beta_i \in \R$
  for $i \in \{0\} \cup \cA_n$ such that
  \begin{align*}
    \beta_0 \|1_\cX\|_\cH^2 + \sum_{i \in \cA_n} \beta_i
    &=
    \delta_0 \|1_\cX\|_\cH^2 + \sum_{i=1}^n \delta_i, \\
    \beta_0 + \sum_{j \in \cA_n}
    k\bigl(X_i, X_j\bigr) \beta_j
    &=
    \delta_0 + \sum_{j=1}^n
    k(X_i, X_j) \delta_j
    \qquad\text{for all }i \in [n].
  \end{align*}
  Thus, with
  $\tilde f \in \tilde \cH_\bX$ of the form
  $\tilde f(\cdot) \vcentcolon= \beta_0 + \sum_{i \in \cA_n} k(\cdot,
  X_i) \beta_i$,
  we have $\tilde f(X_i) = f(X_i)$ for all $i \in [n]$ and so
  $\ell_n(\tilde f) = \ell_n(f)$.
  Further, $P_n(\tilde f) = P_n(f) = 0$, so
  \begin{align}
    \nonumber
    \bigl\|\tilde f - P_n(\tilde f)\bigr\|_\cH^2
    &= \|\tilde f\|_\cH^2
    = \beta_0^2 \|1_\cX\|_\cH^2 + 2 \beta_0 \sum_{i \in \cA_n} \beta_i
    + \sum_{i \in \cA_n} \sum_{j \in \cA_n} \beta_i \beta_j k(X_i, X_j) \\
    \nonumber
    &= \beta_0
    \biggl( \beta_0 \|1_\cX\|_\cH^2 + \sum_{i \in \cA_n} \beta_i \biggr)
    + \sum_{i \in \cA_n} \beta_i
    \biggl( \beta_0 + \sum_{j \in \cA_n} k(X_i, X_j) \beta_j \biggr) \\
    \nonumber
    &= \beta_0
    \biggl( \delta_0 \|1_\cX\|_\cH^2 + \sum_{i=1}^n \delta_i \biggr)
    + \sum_{i \in \cA_n} \beta_i
    \biggl( \delta_0 + \sum_{j=1}^n k(X_i, X_{j}) \delta_j \biggr) \\
    \nonumber
    &= \delta_0
    \biggl( \beta_0 \|1_\cX\|_\cH^2 + \sum_{i \in \cA_n} \beta_i \biggr)
    + \sum_{i=1}^n \delta_i
    \biggl( \beta_0 + \sum_{j \in \cA_n} k(X_i, X_j) \beta_j \biggr) \\
    \nonumber
    &= \delta_0
    \biggl( \delta_0 \|1_\cX\|_\cH^2 + \sum_{i=1}^n \delta_i \biggr)
    + \sum_{i=1}^n \delta_i
    \biggl( \delta_0 + \sum_{j=1}^n k(X_i, X_j) \delta_j \biggr) \\
    \label{eq:representation_H_norm}
    &= \delta^\T
    \tilde{\bK}
    \delta
    = \|f - P_n(f)\|_\cH^2.
  \end{align}
  Thus, $\check\ell_{n,\gamma}(\tilde f) = \check\ell_{n,\gamma}(f)$,
  so we may indeed restrict the feasible set in
  \eqref{eq:argmin_unconstrained} to $f \in \tilde \cH_\bX$.
  Since $P_n(f) = 0$ for such a function $f$, we may write it as
  $f(\cdot) = \sum_{i \in \cA_n} \tilde k(\cdot, X_i) \beta_i$.

  Next, we show that the objective
  in \eqref{eq:argmin_unconstrained}, and hence also in
  \eqref{eq:argmin},
  is strongly convex as a function of $\beta \in \R^{\cA_n}$.
  Indeed, if $\bA$ is a $d \times d$ positive
  semi-definite matrix whose first $k$ columns
  $\bB_k \in \R^{d \times k}$ are linearly independent,
  then $\bA_k \succ 0$, where $\bA_k$
  is the leading $k \times k$ principal submatrix of $\bA$. To see this,
  suppose that $\bA_k u = 0$ for some $u = (u_1,\ldots,u_k)^\T \in
  \R^k$, and define
  $v \vcentcolon= (u_1, \ldots, u_k, 0, \ldots, 0)^\T \in \R^d$.
  Then $v^\T \bA v = u^\T \bA_k u = 0$,
  so $0 = \bA v = \bB_k u$ and hence $u = 0$. Thus, with
  $\bK_{\cA_n} \in \R^{\cA_n \times \cA_n}$ defined by
  $(\bK_{\cA_n})_{i j} \vcentcolon= \bK_{i j}$
  for $i, j \in \cA_n$ and setting
  \begin{align*}
    \tilde{\bK}_{\cA_n}
    \vcentcolon=
    \begin{pmatrix}
      \|1_\cX\|_\cH^2 & 1_{\cA_n}^\T \\
      1_{\cA_n} & \bK_{\cA_n}
    \end{pmatrix}
    \in \R^{(\{0\} \cup \cA_n) \times (\{0\} \cup \cA_n)},
  \end{align*}
  we have $\tilde{\bK}_{\cA_n} \succ 0$.
  Therefore, by \eqref{eq:representation_H_norm},
  \begin{align}
    \label{eq:representation_strongly_convex}
    \beta \equiv (\beta_i)_{i \in \{0\} \cup \cA_n}
    &\mapsto
    \biggl\|
    \beta_0 + \sum_{i \in \cA_n} k(\cdot, X_i) \beta_i
    \biggr\|_\cH^2
    =
    \beta^\T \tilde{\bK}_{\cA_n} \beta
  \end{align}
  is a strongly convex quadratic function on $\R^{\{0\} \cup \cA_n}$.
  Now, the linear map from
  $\R^{\cA_n}$ to $\R^{\{0\} \cup \cA_n}$ defined by
  $(\beta_i)_{i \in \cA_n} \mapsto (\beta_i)_{i \in \{0\} \cup \cA_n}$
  with $\beta_0 \vcentcolon= -\sum_{i \in \cA_n} \bar k(X_i) \beta_i$
  has (full) rank $|\cA_n|$.
  Hence by~\eqref{eq:representation_strongly_convex}, the function
  $(\beta_i)_{i \in \cA_n} \mapsto
  \bigl\| \sum_{i \in \cA_n} \tilde k(\cdot, X_i) \beta_i
  \bigr\|_\cH^2$ is a strongly convex quadratic on~$\R^{\cA_n}$.
  Moreover, the negative log-partial likelihood $\ell_n$ is convex
  by Lemma~\ref{lem:derivatives_ln}.
  Since $\gamma > 0$,
  \begin{align*}
    (\beta_i)_{i \in \cA_n}
    &\mapsto
    \ell_{n}\biggl(\sum_{i \in \cA_n} \tilde k(\cdot, X_i) \beta_i\biggr)
    + \gamma
    \biggl\| \sum_{i \in \cA_n} \tilde k(\cdot, X_i) \beta_i \biggr\|_\cH^2
  \end{align*}
  is strongly convex, so has a unique minimiser
  $(\hat\beta_i)_{i \in \cA_n}$.
  Measurability of $\hat\beta$ follows since it is the limit of
  a sequence of iterates produced by
  gradient descent started from the origin with
  step sizes chosen by backtracking line search
  \citep[Proposition~12.6.1]{lange2013optimization}.

  For the final part, let $f_1, f_2 \in \argmin_{f \in \cH} \ell_{n,\gamma}(f)$
  with $P_n(f_1) = P_n(f_2) = 0$.
  Define $g:[0,1] \rightarrow \R$ by
  $g(a) \vcentcolon= \ell_{n,\gamma}\bigl((1-a) f_1 + a f_2\bigr)$,
  which is convex by Lemma~\ref{lem:derivatives_ln}.
  Since $g(0) = g(1)$, we have $g(a) = g(0)$ for all
  $a \in [0, 1]$. Thus, it follows that
  $0 = g''(a) = D^2 \ell_n\bigl(a f_1 + (1-a) f_2\bigr)(f_2 -
  f_1)^{\otimes 2} + \gamma \|f_1 - f_2\|_\cH^2$,
  and hence $f_1 = f_2$ by Lemma~\ref{lem:derivatives_ln} again.
\end{proof}

\subsection{Proofs for Section~\ref{sec:analysis}}

We now present the proofs of our main theoretical results. First, we
define a projected estimator $\tilde f_{n,\gamma}$
and derive the basic inequality~\eqref{eq:basic_inequality}.

\subsubsection{Projection and the basic inequality}

For a fixed $b \geq 1$ to be determined, let
\begin{align*}
  \tilde f_{n,\gamma}
  &\vcentcolon=
  f_0 - P_n(f_0) 1_\cX
  + \bigl( \hat f_{n,\gamma} - f_0 + P_n(f_0) 1_\cX \bigr)
  \Biggl\{
    1 \land
    \frac{1} {b \sqrt{H_{\infty,\gamma}}
    \bigl\|\hat f_{n,\gamma} - f_0 + P_n(f_0) 1_\cX \bigr\|_{\cH_\gamma}}
  \Biggr\}.
\end{align*}
Then $\bigl\|\tilde f_{n,\gamma} - f_0 + P_n(f_0) 1_\cX\bigr\|_{\cH_\gamma}^2
\leq 1 / (b^2 H_{\infty,\gamma})$ and hence
$\bigl\|\tilde f_{n,\gamma} - f_0 + P_n(f_0) 1_\cX\bigr\|_{\infty}
\leq 1/b$ by Lemma~\ref{lem:infty_bounds}.
Now $\tilde f_{n,\gamma}$ lies on the line segment between
$f_0 - P_n(f_0) 1_\cX$ and $\hat f_{n,\gamma}$,
so $P_n(\tilde f_{n,\gamma}) = 0$ by Proposition~\ref{prop:representation}.
Since $\ell_{n,\gamma}$ is convex by Lemma~\ref{lem:derivatives_ln},
by definition of $\hat f_{n,\gamma}$,
\begin{align*}
  \ell_{n, \gamma}\bigl(\tilde f_{n, \gamma}\bigr)
  &\leq
  \ell_{n, \gamma}\bigl(f_0 - P_n(f_0) 1_\cX\bigr)
  \lor
  \ell_{n, \gamma}\bigl(\hat f_{n, \gamma}\bigr)
  \leq
  \ell_{n, \gamma}\bigl(f_0 - P_n(f_0) 1_\cX\bigr).
\end{align*}
Since $\ell_n$ is invariant under addition of constants, we deduce that
\begin{align*}
  0
  &\geq
  \ell_{n}\bigl( \tilde f_{n, \gamma} \bigr)
  - \ell_{n}(f_0)
  + \gamma \bigl\| \tilde f_{n,\gamma} \bigr\|_\cH^2
  - \gamma \bigl\|f_0 - P_n(f_0) 1_\cX\bigr\|_\cH^2.
\end{align*}

For $f_1, f_2 \in \cH$, we have
$\|f_1\|_\cH^2 - \|f_2\|_\cH^2 = \|f_1 - f_2\|_\cH^2
+ 2 \langle f_1 - f_2, f_2 \rangle_\cH$.
By Taylor's theorem, there is a
function $\check f_{n,\gamma}$ on the segment between
$f_0 - P_n(f_0) 1_\cX$ and $\tilde f_{n, \gamma}$ such that
\begin{align}
  \nonumber
  0
  &\geq
  D \ell_n(f_0) \bigl( \tilde f_{n,\gamma} - f_0 \bigr)
  + \frac{1}{2}
  D^2 \ell_n(f_0) \bigl( \tilde f_{n,\gamma} - f_0 \bigr)^{\otimes 2}
  + \frac{1}{6}
  D^3 \ell_n\bigl(\check f_{n,\gamma}\bigr)
  \bigl( \tilde f_{n,\gamma} - f_0 \bigr)^{\otimes 3} \\
  \label{eq:basic_inequality}
  &\quad+
  \gamma \bigl\| \tilde f_{n,\gamma} - f_0 + P_n(f_0) 1_\cX\bigr\|_\cH^2
  + 2 \gamma \bigl\langle \tilde f_{n,\gamma} - f_0 + P_n(f_0) 1_\cX,
  f_0 - P_n(f_0) 1_\cX \bigr\rangle_\cH.
\end{align}

\subsubsection{Convergence of \texorpdfstring{$S_n$}{Sn},
\texorpdfstring{$D S_n$}{DSn} and \texorpdfstring{$D^2 S_n$}{D2Sn}}

Lemmas~\ref{lem:bv_cs} and~\ref{lem:operator_norm_bounds} are
auxiliary results which are applied repeatedly in subsequent proofs.
In Lemmas~\ref{lem:convergence_Sn},~\ref{lem:convergence_DSn}
and~\ref{lem:convergence_D2Sn},
we show convergence of $S_n$, $D S_n$ and $D^2 S_n$
respectively.
\begin{lemma}[A Cauchy--Schwarz inequality for integrals]%
  \label{lem:bv_cs}
  Let $h: [0, 1] \to \R$ be a function of bounded variation.
  Given a countable index set $\mathcal R$, suppose that
  $g_r : [0, 1] \to \R$ for $r \in \mathcal R$
  satisfy $\int_0^1 |g_r(t)| \, |\mathrm{d}{h(t)|} < \infty$.
  Then
  \begin{align*}
    \sum_{r \in \mathcal R}
    \biggl(
      \int_0^1 g_r(t) \diffi h(t)
    \biggr)^2
    &\leq
    \biggl(
      \int_0^1 |\mathrm d h(t)|
    \biggr)^2\sup_{t \in [0, 1]}
    \sum_{r \in \mathcal R}
    g_r^2(t).
  \end{align*}
\end{lemma}

\begin{proof}[Proof of Lemma~\ref{lem:bv_cs}]
  By the Cauchy--Schwarz inequality and Fubini's theorem,
  \begin{align*}
    \sum_{r \in \mathcal R}
    \biggl(
      \int_0^1 g_r(t) \diffi h(t)
    \biggr)^2
    &\leq
    \sum_{r \in \mathcal R}
    \int_0^1 g_r^2(t) \, |\mathrm d h(t)|
    \cdot
    \int_0^1 |\mathrm d h(t)| \notag \\
    &\leq
    \biggl(
      \int_0^1 |\mathrm d h(t)|
    \biggr)^2\sup_{t \in [0, 1]}
    \sum_{r \in \mathcal R}
    g_r^2(t).
    \qedhere
  \end{align*}
\end{proof}

\begin{lemma}[Operator norm bounds]%
  \label{lem:operator_norm_bounds}
  Let $G: \cH \to \R$ be a linear operator and $\gamma \geq 0$.
  Define $\|G\|_{\cH, \op}
  \vcentcolon= \sup_{f \in \mathcal{H}:\|f\|_\cH = 1} |G(f)|$
  and $\|G\|_{\cH_\gamma, \op}
  \vcentcolon= \sup_{f \in \mathcal{H}:\|f\|_{\cH_\gamma} = 1} |G(f)|$.
  Then
  \begin{align*}
    \|G\|_{\cH, \op}^2
    \leq
    \sum_{r \in \mathcal{R}}
    \nu_r G(e_r)^2, \qquad
    \|G\|_{\cH_\gamma,\op}^2
    \leq \sum_{r \in \mathcal{R}} a_r G(e_r)^2.
  \end{align*}
\end{lemma}

\begin{proof}[Proof of Lemma~\ref{lem:operator_norm_bounds}]
  By linearity of $G$ and Cauchy--Schwarz,
  for any $f \in \cH$,
  \begin{equation*}
    G(f)^2
    =
    \biggl(
      \sum_{r \in \mathcal{R}} \langle f, e_r \rangle_{L_2} G(e_r)
    \biggr)^2
    =
    \biggl(
      \sum_{r \in \mathcal{R}}
      \frac{1}{\sqrt{\nu_r}}
      \langle f, e_r \rangle_{L_2}
      \cdot \sqrt{\nu_r} G(e_r)
    \biggr)^2
    \leq
    \|f\|_\cH^2
    \sum_{r \in \mathcal{R}}
    \nu_r
    G(e_r)^2.
  \end{equation*}
  Similarly,
  \begin{equation*}
    G(f)^2
    =
    \biggl(
      \sum_{r \in \mathcal{R}}
      \frac{1}{\sqrt{a_r}}
      \langle f, e_r \rangle_{L_2}
      \cdot \sqrt{a_r} G(e_r)
    \biggr)^2
    \leq
    \|f\|_{\cH_\gamma}^2
    \sum_{r \in \mathcal{R}}
    a_r
    G(e_r)^2.
    \qedhere
  \end{equation*}
\end{proof}

\begin{lemma}[Convergence of $S_n$]%
  \label{lem:convergence_Sn}
  There exists a universal constant $C > 0$ such that
  for all $s \geq 1$ and $f \in \cB(\cX)$, with probability at least
  $1 - e^{-s^2}$,
  \begin{align*}
    \sup_{t \in [0, 1]}
    \bigl|S_n(f, t) - S(f, t) \bigr|
    &\leq
    \frac{C s e^{\|f\|_\infty}}{\sqrt n}.
  \end{align*}
\end{lemma}

\begin{proof}[Proof of Lemma~\ref{lem:convergence_Sn}]
  For $i \in [n]$ and $t \in [0, 1]$, define
  \begin{align*}
    \xi_i(t) \vcentcolon= R_i(t) e^{f(X_i)}
    - \int_\cX q(x, t) e^{f(x)} \diffi P_X(x),
  \end{align*}
  noting that $|\xi_i(t)| \leq 2 e^{\|f\|_\infty}$. It follows that
  with $\mathcal{I} := \{(i,j):i,j \in [n],i<j\}$,
  \begin{align}
    \nonumber
    \sup_{t \in [0, 1]}
    \bigl( S_n(f, t) - S(f, t) \bigr)^2
    &=
    \sup_{t \in [0, 1]}
    \frac{1}{n^2}
    \sum_{i=1}^n
    \sum_{j=1}^n
    \xi_i(t)
    \xi_j(t) \\
    \nonumber
    &\leq
    \sup_{t \in [0, 1]}
    \frac{1}{n^2}
    \sum_{i=1}^n
    \xi_i(t)^2
    + \sup_{t \in [0, 1]}
    \frac{2}{n^2}
    \sum_{(i,j) \in \mathcal{I}}
    \xi_i(t)
    \xi_j(t) \\
    \label{eq:diagonal_bound}
    &\leq
    \frac{4 e^{2 \|f\|_\infty}}{n}
    + \sup_{t \in [0, 1]}
    \frac{2}{n(n-1)}
    \biggl|
    \sum_{(i,j) \in \mathcal{I}}
    \xi_i(t)
    \xi_j(t)
    \biggr|.
  \end{align}
  The second term in the above display is a degenerate $U$-process
  \citep[Chapter~5]{de1999decoupling} as
  $\E \bigl\{ \xi_{i}(t) \xi_{j}(t) \mid X_i, T_i, I_i \bigr\}
  = \xi_{i}(t) \, \E \bigl\{ \xi_{j}(t) \bigr\} = 0$ for $(i,j) \in
  \mathcal{I}$ and $t \in [0, 1]$.
  We now let $\mathcal{Q}$ denote the set of finitely-supported
  probability measures on $(\mathcal{X} \times [0,1] \times \{0,1\})^2$,
  and for $Q \in \mathcal{Q}$, seek to bound the $L_2(Q)$-covering
  number (see Definition~\ref{def:covering}) of the class
  \begin{align*}
    \mathcal G
    &=
    \biggl\{
      \big((X_i, T_i, I_i), (X_j, T_j, I_j)\big)
      \mapsto \xi_{i}(t) \xi_{j}(t) :
      t \in [0, 1]
    \biggr\}.
  \end{align*}
  For $0 \leq t < t' \leq 1$ and $i, j \in [n]$,
  \begin{align*}
    |\xi_{i}(t)|
    | \xi_{j}(t) - \xi_{j}(t')|
    &\leq
    2 e^{2 \|f\|_\infty}
    \bigl\{
      \mathbbm 1_{\{t \leq T_j < t'\}}
      + \P(t \leq T < t')
    \bigr\}.
  \end{align*}
  As such, we deduce that
  \begin{align*}
    |\xi_{i}(t) \xi_{j}(t)
    - \xi_{i}(t') \xi_{j}(t')|
    &\leq
    |\xi_{i}(t)| | \xi_{j}(t) - \xi_{j}(t')|
    + |\xi_{j}(t')| | \xi_{i}(t) - \xi_{i}(t')| \\
    &\leq
    2 e^{2\|f\|_\infty}
    \bigl\{
      \mathbbm 1_{\{t \leq T_i < t'\}}
      + \mathbbm 1_{\{t \leq T_j < t'\}}
      + 2 \P(t \leq T < t')
    \bigr\}.
  \end{align*}
  It follows that
  \begin{align*}
    &\E_Q
    \Bigl[
      \bigl(
        \xi_{i}(t) \xi_{j}(t) - \xi_{i}(t') \xi_{j}(t')
      \bigr)^2
    \Bigr] \\
    &\quad\leq
    12 e^{4 \|f\|_\infty}
    \bigl\{
      \P_Q(t \leq T_i < t')
      + \P_Q(t \leq T_j < t')
      + 4 \P(t \leq T < t')^2
    \bigr\}.
  \end{align*}
  Given $\varepsilon \in (0,1)$, define $\mathcal{T}_1 :=
  \{t_1,\ldots,t_{\lceil 1/\varepsilon^2 \rceil}\}$, where
  \[
    t_m := \inf\{t \in [0,1]: \P_Q(T_i \leq t) \geq m\varepsilon^2\}
  \]
  for $m \in [\lceil 1/\varepsilon^2 \rceil - 1]$ and $t_{\lceil
  1/\varepsilon^2 \rceil} := 1$, so that given any $t \in [0,1]$,
  there exists $m \in [\lceil 1/\varepsilon^2 \rceil]$ such that
  $\P_Q(t \leq T_i < t_m) \leq \varepsilon^2$.
  Let $\cT_2$ and $\cT_3$ satisfy the same properties but for
  $T_j$ under $Q$ and $T$ under $P_T$ respectively.
  Then the set $\cT = \cT_1 \cup \cT_2 \cup \cT_3$ is a
  $\sqrt{72} e^{2 \|f\|_\infty} \varepsilon$-cover
  of $\mathcal G$ under $L_2(Q)$, of cardinality at most $6/\varepsilon^2$.
  Defining $\varepsilon' \vcentcolon = 3\varepsilon/ \sqrt{2} \in
  (0,3/\sqrt{2})$ and $M \vcentcolon = 4 e^{2 \|f\|_\infty}$,
  \begin{align*}
    N\big(\varepsilon' M,\mathcal G, \|\cdot\|_{Q,2}\big)
    \leq \frac{27}{\varepsilon^{\prime 2}}.
  \end{align*}
  Hence, since $\log(1+x) \leq 1 + \log x$ for $x \geq 1$,
  \begin{align*}
    J
    &\vcentcolon=
    \sup_{Q \in \mathcal{Q}}
    \int_0^{1}
    \log
    \bigl(
      1 +
      N\bigl(\varepsilon M, \mathcal{G}, \|\cdot\|_{Q,2}\bigr)
    \bigr)
    \diffi \varepsilon
    \leq
    1 + \int_0^{1}
    \log(27 / \varepsilon^2)
    \diffi \varepsilon = 1 + \log(27) + 2
    \leq 7.
  \end{align*}
  It follows by Lemma~\ref{lem:orliczUStat} that there exists a
  universal constant $C_1 > 0$ such that
  for all $s \geq 1$,
  \begin{align*}
    \P \biggl(
      \sup_{t \in [0, 1]}
      \frac{2}{n(n-1)}
      \biggl|
      \sum_{(i,j) \in \mathcal{I}}
      \xi_{i}(t) \xi_{j}(t)
      \biggr|
      > \frac{C_1 s e^{2 \|f\|_\infty}}{n}
    \biggr)
    \leq e^{-s}.
  \end{align*}
  By combining this with the bound \eqref{eq:diagonal_bound} on the
  diagonal elements, there exists a universal constant $C_2 > 0$ such that
  for all $s \geq 1$,
  with probability at least $1 - e^{-s^2}$,
  \begin{align*}
    \sup_{t \in [0, 1]}
    \big( S_n(f, t) - S(f, t) \big)^2
    &\leq
    \frac{C_2 s^2 e^{2 \|f\|_\infty}}{n}.
    \qedhere
  \end{align*}
\end{proof}

\begin{lemma}[Convergence of $D S_n$]%
  \label{lem:convergence_DSn}
  Let $f \in \cB(\cX)$.
  There exists a universal constant $C > 0$ such that
  for all $s \geq 1$, with probability at least $1 - e^{-s^2}$,
  \begin{align*}
    \sup_{t \in [0, 1]}
    \bigl\|D S_n(f, t) - D S(f, t)\bigr\|_{\cH_\gamma, \op}
    \leq
    C e^{\|f\|_\infty}
    \biggl(
      \frac{\sqrt{H_{2,\gamma}} + s}{\sqrt n}
      + \frac{s^2 \sqrt{H_{\infty,\gamma}}}{n}
    \biggr).
  \end{align*}
\end{lemma}

\begin{proof}[Proof of Lemma~\ref{lem:convergence_DSn}]
  Let $\varepsilon_1, \ldots, \varepsilon_n$
  be independent Rademacher variables that are independent of the
  data. Noting that
  \[
    \sup_{f_1 \in \cH:\, \|f_1\|_{\cH_\gamma} \leq 1}
    \sup_{t \in [0, 1]}
    \bigl|D S_n(f, t)(f_1) - D S(f, t)(f_1)\bigr|
  \]
  is measurable by arguments similar to those in the proof of
  Lemma~\ref{lem:uniform_convergence_Pn},
  by Lemmas~\ref{lem:derivatives_Sn} and~\ref{lem:derivatives_S}, a
  symmetrisation inequality \citep[Lemma~2.3.1]{van1996weak},
  a contraction principle \citep[Proposition~A.1.10]{van1996weak},
  Jensen's inequality and Cauchy--Schwarz, we have
  \begin{align}
    \nonumber
    &\E \biggl(
      \sup_{f_1 \in \cH:\, \|f_1\|_{\cH_\gamma} \leq 1}
      \sup_{t \in [0, 1]}
      \bigl|D S_n(f, t)(f_1) - D S(f, t)(f_1)\bigr|
    \biggr) \\
    \nonumber
    &\qquad\leq
    2 \,
    \E \biggl(
      \sup_{f_1 \in \cH:\, \|f_1\|_{\cH_\gamma} \leq 1}
      \sup_{t \in [0, 1]}
      \biggl|
      \frac{1}{n}
      \sum_{i=1}^n
      \varepsilon_i
      R_i(t)
      e^{f(X_i)}
      f_1(X_i)
      \biggr|
    \biggr) \\
    \nonumber
    &\qquad\leq
    2 e^{\|f\|_\infty}
    \E \biggl(
      \sup_{f_1 \in \cH:\, \|f_1\|_{\cH_\gamma} \leq 1}
      \sup_{t \in [0, 1]}
      \biggl|
      \frac{1}{n}
      \sum_{i=1}^n
      \varepsilon_i
      R_i(t)
      f_1(X_i)
      \biggr|
    \biggr) \\
    \nonumber
    &\qquad\leq
    2 e^{\|f\|_\infty}
    \E \biggl(
      \sup_{(\alpha_r)_{r \in \mathcal R}: \,
      \sum_{r \in \mathcal R} \alpha_r^2 / a_r \leq 1}
      \sup_{t \in [0, 1]}
      \frac{1}{n}
      \sum_{i=1}^n
      \varepsilon_i
      R_i(t)
      \sum_{r \in \mathcal R} \alpha_r e_r(X_i)
    \biggr) \\
    \nonumber
    &\qquad\leq
    2 e^{\|f\|_\infty}
    \Biggl[
      \E \biggl\{
        \sup_{(\alpha_r)_{r \in \mathcal R}: \,
        \sum_{r \in \mathcal R} \alpha_r^2 / a_r \leq 1}
        \sup_{t \in [0, 1]}
        \biggl(
          \sum_{r \in \mathcal R}
          \alpha_r
          \frac{1}{n}
          \sum_{i=1}^n
          \varepsilon_i
          R_i(t)
          e_r(X_i)
        \biggr)^2
      \biggr\}
    \Biggr]^{1/2} \\
    \nonumber
    &\qquad\leq
    2 e^{\|f\|_\infty}
    \Biggl[
      \E \biggl\{
        \sup_{(\alpha_r)_{r \in \mathcal R}: \,
        \sum_{r \in \mathcal R} \alpha_r^2 / a_r \leq 1}
        \sup_{t \in [0, 1]}
        \sum_{r \in \mathcal R}
        a_r
        \biggl(
          \frac{1}{n}
          \sum_{i=1}^n
          \varepsilon_i
          R_i(t)
          e_r(X_i)
        \biggr)^2
        \sum_{r' \in \mathcal R}
        \frac{\alpha_{r'}^2}{a_{r'}}
      \biggr\}
    \Biggr]^{1/2} \\
    \label{eq:DSn_rademacher}
    &\qquad=
    2 e^{\|f\|_\infty}
    \Biggl[
      \E \biggl\{
        \sup_{t \in [0, 1]}
        \sum_{r \in \mathcal R}
        a_r
        \biggl(
          \frac{1}{n}
          \sum_{i=1}^n
          \varepsilon_i
          R_i(t)
          e_r(X_i)
        \biggr)^2
      \biggr\}
    \Biggr]^{1/2}.
  \end{align}
  For $i \in [n]$, $r \in \cR$ and $t \in [0, 1]$,
  define $\xi_{i r}(t) \coloneqq \varepsilon_i R_i(t) e_r(X_i)$,
  so that
  \begin{align*}
    \sup_{t \in [0, 1]}
    \biggl|
    \sum_{r \in \cR} a_r \xi_{i r}(t) \xi_{j r}(t)
    \biggr|
    &\leq
    \biggl|
    \sum_{r \in \cR} a_r e_r(X_i) e_r(X_j)
    \biggr|
    \eqqcolon
    F(X_i, X_j).
  \end{align*}
  Thus, with $\mathcal{I} := \{(i,j): i,j \in [n], i < j\}$,
  we have
  \begin{align}
    \nonumber
    \sup_{t \in [0, 1]}
    \sum_{r \in \mathcal R}
    a_r
    \biggl(
      \frac{1}{n} &
      \sum_{i=1}^n
      \varepsilon_i
      R_i(t)
      e_r(X_i)
    \biggr)^2 \\
    \nonumber
    &=
    \sup_{t \in [0, 1]}
    \sum_{r \in \mathcal R}
    \frac{1}{n^2}
    \sum_{i=1}^n
    \sum_{j=1}^n
    a_r
    \xi_{i r}(t)
    \xi_{j r}(t) \\
    \nonumber
    &\leq
    \sup_{t \in [0, 1]}
    \sum_{r \in \mathcal R}
    \frac{1}{n^2}
    \sum_{i=1}^n
    a_r
    \xi_{i r}(t)^2
    + \sup_{t \in [0, 1]}
    \sum_{r \in \mathcal R}
    \frac{2}{n^2}
    \sum_{(i, j) \in \mathcal I}
    a_r
    \xi_{i r}(t)
    \xi_{j r}(t) \\
    \label{eq:DSn_Uproc}
    &\leq
    \frac{1}{n^2}
    \sum_{i=1}^n
    F(X_i, X_i)
    + \sup_{t \in [0, 1]}
    \sum_{r \in \mathcal R}
    \frac{2}{n^2}
    \sum_{(i, j) \in \mathcal I}
    a_r
    \xi_{i r}(t)
    \xi_{j r}(t).
  \end{align}
  The first term in~\eqref{eq:DSn_Uproc} satisfies
  \begin{align}
    \label{eq:DSn_diagonal}
    \E\biggl(
      \frac{1}{n^2}
      \sum_{i=1}^n
      F(X_1, X_1)
    \biggr)
    &=
    \frac{1}{n}
    \E\bigl\{
      F(X_1, X_1)
    \bigr\}
    =
    \frac{1}{n}
    \sum_{r \in \cR}
    a_r
    \E\bigl\{
      e_r(X_1)^2
    \bigr\}
    =
    \frac{1}{n}
    \sum_{r \in \cR}
    a_r
    = \frac{H_{2,\gamma}}{n}.
  \end{align}
  The second term in \eqref{eq:DSn_Uproc} is a degenerate
  second-order $U$-process
  \citep[Chapter~5]{de1999decoupling}, since by the
  dominated convergence theorem, for each $(i, j) \in \mathcal I$
  and $t \in [0, 1]$,
  \begin{align*}
    \E\biggl(
      \sum_{r \in \mathcal R}
      a_r
      \xi_{i r}(t)
      \xi_{j r}(t)
      \Bigm|
      X_i, T_i, I_i, \varepsilon_i
    \biggr)
    =
    \sum_{r \in \mathcal R}
    a_r
    \xi_{i r}(t)
    \E\bigl\{ \xi_{j r}(t) \bigr\}
    = 0.
  \end{align*}
  We now let $\mathcal{Q}$ denote the set of finitely-supported
  probability measures on $(\mathcal{X} \times [0,1] \times \{0,1\})^2$,
  and for $Q \in \mathcal{Q}$, seek to bound the $L_2(Q)$-covering
  number of the function class
  \begin{align*}
    \mathcal G
    &=
    \biggl\{
      \big((X_i, T_i, I_i, \varepsilon_i),
      (X_j, T_j, I_j, \varepsilon_j)\big)
      \mapsto
      \sum_{r \in \mathcal{R}}
      a_r \xi_{i r}(t) \xi_{j r}(t) :
      t \in [0, 1]
    \biggr\}.
  \end{align*}
  For $0 \leq t \leq t' \leq 1$ and $i, j \in [n]$,
  observe that
  \begin{align*}
    &\biggl\|
    \sum_{r \in \mathcal{R}}
    a_r
    \bigl\{
      \xi_{i r}(t') \xi_{j r}(t')
      - \xi_{i r}(t) \xi_{j r}(t)
    \bigr\}
    \biggr\|_{Q, 2}^2 \\
    &\quad=
    \biggl\|
    \sum_{r \in \mathcal{R}}
    a_r
    \xi_{i r}(t')
    \bigl\{
      \xi_{j r}(t') - \xi_{j r}(t)
    \bigr\}
    + \sum_{r \in \mathcal{R}}
    a_r
    \xi_{j r}(t)
    \bigl\{
      \xi_{i r}(t') - \xi_{i r}(t)
    \bigr\}
    \biggr\|_{Q, 2}^2 \\
    &\quad=
    \biggl\|
    \Bigl(
      R_i(t') \bigl\{ R_j(t) - R_j(t') \bigr\}
      + R_j(t) \bigl\{ R_i(t) - R_i(t') \bigr\}
    \Bigr)
    \biggl|
    \sum_{r \in \mathcal{R}}
    a_r
    e_r(X_j)
    e_r(X_i)
    \biggr|
    \biggr\|_{Q, 2}^2 \\
    &\quad\leq
    \Bigl\|
    \bigl(
      \mathbbm{1}_{\{ t \leq T_i < t' \}}
      + \mathbbm{1}_{\{ t \leq T_j < t' \}}
    \bigr)
    F(X_i, X_j)
    \Bigr\|_{Q, 2}^2 \\
    &\quad\leq
    2 \, \E_Q\Bigl\{
      \bigl(
        \mathbbm{1}_{\{ t \leq T_i < t' \}}
        + \mathbbm{1}_{\{ t \leq T_j < t' \}}
      \bigr)
      F(X_i, X_j)^2
    \Bigr\}
    \eqqcolon \mu(t, t').
  \end{align*}
  Noting that $\mu(0, 1) = 4 \|F\|_{Q,2}^2$
  and $t' \mapsto \mu(t, t')$ is increasing
  and left continuous for each $t \in [0,1]$, given $\varepsilon \in
  (0,1)$, define $\mathcal{T} \coloneqq
  \{t_1,\ldots,t_{\lceil 1/\varepsilon^2 \rceil}\}$, where
  $ t_m \coloneqq
  \inf\{t \in [0,1]: \mu(0, t) \geq 4 m \|F\|_{Q,2}^2\varepsilon^2\}$
  for $m \in [\lceil 1/\varepsilon^2 \rceil - 1]$ and $t_{\lceil
  1/\varepsilon^2 \rceil} := 1$.
  Moreover, $t \mapsto \mu(t, t')$
  is decreasing and left continuous
  for each $t' \in [0, 1]$;
  therefore, given any $t \in [0,1]$,
  there exists $m \in [\lceil 1/\varepsilon^2 \rceil]$ such that
  $t_m \leq t$ and
  $\mu(t_m, t) \leq 4 \|F\|_{Q,2}^2 \varepsilon^2$.
  Thus $\cT$ is a
  $2 \|F\|_{Q,2} \varepsilon$-cover
  of $\mathcal G$ under $L_2(Q)$, of cardinality at most
  $2/\varepsilon^2$.
  Defining $\varepsilon' \vcentcolon = 2\varepsilon \in
  (0,2)$,
  we deduce that
  \begin{equation*}
    N\big(\varepsilon' \|F\|_{Q,2},
    \mathcal G, \|\cdot\|_{Q,2}\big)
    \leq \frac{8}{\varepsilon^{\prime 2}}.
  \end{equation*}
  Hence, since $\log(1+x) \leq 1 + \log x$ for $x \geq 1$,
  \begin{align*}
    J
    \vcentcolon=
    \sup_{Q \in \mathcal{Q}}
    \int_0^{1}
    \log
    \bigl\{
      1 + N\bigl(\varepsilon \|F\|_{Q,2},
      \mathcal{G}, \|\cdot\|_{Q,2}\bigr)
    \bigr\}
    \diffi \varepsilon
    &\leq
    1 + \int_0^{1}
    \log(8 / \varepsilon^2)
    \diffi \varepsilon \\
    &= 1 + \log 8 + 2
    \leq 6.
  \end{align*}
  Further, by independence and
  orthonormality of the eigenfunctions $(e_r)_{r \in \mathcal{R}}$, we have
  \begin{align*}
    \|F\|_{P_X,2}^2
    &=
    \E\biggl\{
      \biggl(
        \sum_{r \in \cR}
        a_r e_r(X_1) e_r(X_2)
      \biggr)^2
    \biggr\} \\
    &=
    \sum_{r \in \cR}
    \sum_{r' \in \cR}
    a_r a_{r'}
    \Bigl(
      \E\bigl\{
        e_r(X_1)
        e_{r'}(X_1)
      \bigr\}
    \Bigr)^2
    \leq
    \sum_{r \in \cR}
    a_r^2
    \leq
    \biggl(
      \sum_{r \in \cR}
      a_r
    \biggr)^2
    = H_{2,\gamma}^2.
  \end{align*}
  Applying \citet[Corollary~5.6]{chen2020jackknife},
  there exists a universal constant $C > 0$ such that
  \begin{align}
    \label{eq:DSn_kato}
    \E \biggl\{
      \sup_{t \in [0, 1]}
      \sum_{r \in \mathcal R}
      \frac{2}{n^2}
      \sum_{(i, j) \in \mathcal I}
      a_r
      \xi_{i r}(t)
      \xi_{j r}(t)
    \biggr\}
    &\leq
    \frac{C J \|F\|_{P_X,2}}{n}
    \leq
    \frac{6 C H_{2,\gamma}}{n},
  \end{align}
  By~\eqref{eq:DSn_rademacher},
  \eqref{eq:DSn_Uproc}, \eqref{eq:DSn_diagonal}
  and~\eqref{eq:DSn_kato}, and by increasing the universal
  constant $C > 0$ if necessary, we deduce that
  \begin{align*}
    \E \biggl(
      \sup_{t \in [0, 1]}
      \bigl\|D S_n(f, t) - D S(f, t)\bigr\|_{\cH_\gamma,\op}
    \biggr)
    &\leq
    \frac{C e^{\|f\|_\infty} \sqrt{H_{2,\gamma}}}{\sqrt n}.
  \end{align*}
  It remains to control the stochastic fluctuations
  of this supremum above its expected value.
  To this end, by Lemma~\ref{lem:infty_bounds}, we have
  \begin{align*}
    \sigma^2
    &\coloneqq
    \sup_{f_1 \in \cH:\, \|f_1\|_{\cH_\gamma} \leq 1}
    \sup_{t \in [0, 1]}
    \sum_{i=1}^n
    \E \biggl\{
      \biggl(
        \frac{1}{n}
        R_i(t)
        e^{f(X_i)}
        f_1(X_i)
      \biggr)^2
    \biggr\} \\
    &\,\leq
    \sup_{f_1 \in \cH:\, \|f_1\|_{\cH_\gamma} \leq 1}
    \frac{e^{2 \|f\|_\infty} \|f_1\|_{L_2}^2 }{n}
    \leq
    \frac{e^{2 \|f\|_\infty}}{n}, \\
    M
    &\coloneqq
    \sup_{f_1 \in \cH:\, \|f_1\|_{\cH_\gamma} \leq 1}
    \sup_{t \in [0, 1]}
    \biggl|
    \frac{1}{n}
    R_i(t)
    e^{f(X_i)}
    f_1(X_i)
    \biggr| \\
    &\,\leq
    \sup_{f_1 \in \cH:\, \|f_1\|_{\cH_\gamma} \leq 1}
    \frac{e^{\|f\|_\infty} \|f_1\|_\infty}{n}
    \leq
    \frac{e^{\|f\|_\infty} \sqrt{H_{\infty,\gamma}}}{n}.
  \end{align*}
  By Bousquet's inequality \citep[Theorem~12.5]{boucheron2013concentration},
  we conclude that with probability at least $1 - e^{-s^2}$,
  \begin{align*}
    \sup_{t \in [0, 1]}
    \bigl\|
    D S_n(f, t) - D S(f, t)
    \bigr\|_{\cH_\gamma, \op}
    &\leq
    C e^{\|f\|_\infty}
    \biggl(
      \frac{\sqrt{H_{2,\gamma}} + s}{\sqrt n}
      + \frac{s^2 \sqrt{H_{\infty,\gamma}}}{n}
    \biggr),
  \end{align*}
  as required.
\end{proof}

\begin{lemma}[Convergence of $D^2 S_n$]%
  \label{lem:convergence_D2Sn}
  Let $f \in \cB(\cX)$.
  There is a universal constant $C > 0$ such that
  for all $s \geq 1$, with probability at least $1 - e^{-s^2}$,
  \begin{align*}
    \sup_{t \in [0, 1]}
    \sum_{r \in \mathcal{R}} \sum_{r' \in \mathcal{R}}
    a_r a_{r'}
    \big( D^2 S_n(f, t)(e_r, e_{r'}) - D^2 S(f, t)(e_r, e_{r'}) \big)^2
    &\leq
    \frac{C s^2 e^{2 \|f\|_\infty} H_{\infty,\gamma}^2}{n}.
  \end{align*}

\end{lemma}

\begin{proof}[Proof of Lemma~\ref{lem:convergence_D2Sn}]
  For $i \in [n]$ and $r, r' \in \mathcal R$, define
  \begin{align*}
    \xi_{i r r'}(t)
    &\vcentcolon=
    R_i(t) e^{f(X_i)} e_r(X_i) e_{r'}(X_i)
    - \int_\cX q(x, t) e^{f(x)} e_r(x) e_{r'}(x) \diffi P_X(x).
  \end{align*}
  For every $t \in [0,1]$ and $i \in [n]$,
  \begin{align}
    \nonumber
    \sum_{r \in \mathcal R}
    \sum_{r' \in \mathcal R}
    a_r a_{r'}
    \xi_{i r r'}(t)^2
    &\leq
    2 \sum_{r \in \mathcal R}
    \sum_{r' \in \mathcal R}
    a_r a_{r'}
    R_i(t) e^{2 f(X_i)}
    e_r(X_i)^2
    e_{r'}(X_i)^2 \\
    \nonumber
    &\hspace{1.5cm}+
    2 \sum_{r \in \mathcal R}
    \sum_{r' \in \mathcal R}
    a_r a_{r'}
    \biggl(
      \int_\cX q(x, t) e^{f(x)} e_r(x) e_{r'}(x) \diffi P_X(x)
    \biggr)^2 \\
    \label{eq:xi_irr_l2}
    &\leq
    4
    e^{2 \|f\|_\infty}
    H_{\infty,\gamma}^2
    =\vcentcolon M.
  \end{align}
  Therefore with $\mathcal I \vcentcolon= \{(i,j): i,j \in [n], i < j\}$,
  by Lemmas~\ref{lem:derivatives_Sn} and~\ref{lem:derivatives_S},
  \begin{align}
    \nonumber
    &\sup_{t \in [0, 1]}
    \sum_{r \in \mathcal{R}} \sum_{r' \in \mathcal{R}}
    a_r a_{r'}
    \big( D^2 S_n(f, t)(e_r, e_{r'}) - D^2 S(f, t)(e_r, e_{r'}) \big)^2 \\
    \nonumber
    &\quad=
    \sup_{t \in [0, 1]}
    \frac{1}{n^2}
    \sum_{i=1}^n
    \sum_{j=1}^n
    \sum_{r \in \mathcal{R}} \sum_{r' \in \mathcal{R}}
    a_r a_{r'}
    \xi_{i r r'}(t)
    \xi_{j r r'}(t) \\
    \nonumber
    &\quad\leq
    \sup_{t \in [0, 1]}
    \frac{1}{n^2}
    \sum_{i=1}^n
    \sum_{r \in \mathcal{R}} \sum_{r' \in \mathcal{R}}
    a_r a_{r'}
    \xi_{i r r'}(t)^2
    + \sup_{t \in [0, 1]}
    \frac{2}{n^2}
    \sum_{(i, j) \in \mathcal I}
    \sum_{r \in \mathcal{R}} \sum_{r' \in \mathcal{R}}
    a_r a_{r'}
    \xi_{i r r'}(t)
    \xi_{j r r'}(t) \\
    &\quad\leq
    \label{eq:D2_second_derivative_term_U_stat}
    \frac{4 e^{2 \|f\|_\infty} H_{\infty,\gamma}^2}{n}
    + \sup_{t \in [0, 1]}
    \biggl|
    \frac{2}{n (n-1)}
    \sum_{(i, j) \in \mathcal I}
    \sum_{r \in \mathcal{R}} \sum_{r' \in \mathcal{R}}
    a_r a_{r'}
    \xi_{i r r'}(t)
    \xi_{j r r'}(t)
    \biggr|.
  \end{align}
  The second term above is a degenerate second-order $U$-process,
  since by \eqref{eq:xi_irr_l2}, the Cauchy--Schwarz inequality
  and the dominated convergence theorem,
  $(i, j) \in \mathcal I$ and $t \in [0,1]$,
  \begin{align*}
    \E \biggl(
      \sum_{r \in \mathcal{R}}
      \sum_{r' \in \mathcal{R}}
      a_r a_{r'} \xi_{i r r'}(t) \xi_{j r r'}(t)
      \Bigm|
      X_i,T_i,I_i
    \biggr)
    &=
    \sum_{r \in \mathcal{R}}
    \sum_{r' \in \mathcal{R}}
    a_r a_{r'}
    \xi_{i r r'}(t)
    \, \E \bigl\{
      \xi_{j r r'}(t)
    \bigr\}
    = 0.
  \end{align*}
  We now let $\mathcal{Q}$ denote the set of finitely-supported
  probability measures on $(\mathcal{X} \times [0,1] \times \{0,1\})^2$,
  and for $Q \in \mathcal{Q}$, seek to bound the $L_2(Q)$-covering
  number of the function class
  \begin{align*}
    \mathcal G
    &\vcentcolon=
    \biggl\{
      \big((X_i, T_i, I_i), (X_j, T_j, I_j)\big)
      \mapsto
      \sum_{r \in \mathcal{R}}
      \sum_{r' \in \mathcal{R}}
      a_r a_{r'} \xi_{i r r'}(t) \xi_{j r r'}(t) :
      t \in [0, 1]
    \biggr\}.
  \end{align*}
  For $0 \leq t < t' \leq 1$ and $i, j \in [n]$,
  by the Cauchy--Schwarz inequality,
  \begin{align*}
    &\sum_{r \in \mathcal R}
    \sum_{r' \in \mathcal R}
    a_r a_{r'}
    |\xi_{i r r'}(t)|
    | \xi_{j r r'}(t) - \xi_{j r r'}(t')| \\
    &\quad\leq
    e^{2\|f\|_\infty}
    \sum_{r \in \mathcal R}
    \sum_{r' \in \mathcal R}
    a_r a_{r'}
    \biggl(
      |e_r(X_i) e_{r'}(X_i)| + \int_\cX |e_r(x) e_{r'}(x)| \diffi P_X(x)
    \biggr) \\
    &\qquad\times
    \biggl(
      |e_r(X_j) e_{r'}(X_j)|
      |R_j(t) - R_j(t')|
      + \int_\cX
      |e_r(x) e_{r'}(x)|
      |q(x, t) - q(x, t')|
      \diffi P_X(x)
    \biggr) \\
    &\quad\leq
    2 e^{2\|f\|_\infty}
    H_{\infty,\gamma}^2
    \bigl\{
      \mathbbm 1_{\{t \leq T_j < t'\}}
      + \P(t \leq T < t')
    \bigr\}.
  \end{align*}
  As such, we deduce that
  \begin{align*}
    \biggl|
    \sum_{r \in \mathcal R}
    &\sum_{r' \in \mathcal R}
    a_r a_{r'}
    \xi_{i r r'}(t) \xi_{j r r'}(t)
    - \sum_{r \in \mathcal R}
    \sum_{r' \in \mathcal R}
    a_r a_{r'}
    \xi_{i r r'}(t') \xi_{j r r'}(t')
    \biggr| \\
    &\quad\leq
    \sum_{r \in \mathcal R}
    \sum_{r' \in \mathcal R}
    a_r a_{r'}
    |\xi_{i r r'}(t)|
    | \xi_{j r r'}(t) - \xi_{j r r'}(t')|
    + \sum_{r \in \mathcal R}
    \sum_{r' \in \mathcal R}
    a_r a_{r'}
    |\xi_{j r r'}(t')|
    | \xi_{i r r'}(t) - \xi_{i r r'}(t')| \\
    &\quad\leq
    2 e^{2\|f\|_\infty}
    H_{\infty,\gamma}^2
    \bigl\{
      \mathbbm 1_{\{t \leq T_i < t'\}}
      + \mathbbm 1_{\{t \leq T_j < t'\}}
      + 2 \P(t \leq T < t')
    \bigr\}.
  \end{align*}
  It follows that
  \begin{align*}
    \E_Q
    \biggl\{
      \biggl(&
        \sum_{r \in \mathcal{R}}
        \sum_{r' \in \mathcal{R}}
        a_r a_{r'}
        \bigl\{\xi_{i r r'}(t) \xi_{j r r'}(t)
        - \xi_{i r r'}(t') \xi_{j r r'}(t')\bigr\}
      \biggr)^2
    \biggr\} \\
    &\quad\leq
    12 e^{4 \|f\|_\infty}
    H_{\infty,\gamma}^4
    \bigl\{
      \P_Q(t \leq T_i < t')
      + \P_Q(t \leq T_j < t')
      + 4 \P(t \leq T < t')^2
    \bigr\}.
  \end{align*}
  For $\varepsilon \in (0,1)$ and $m \in \bigl[\lceil 1/\varepsilon^2
  \rceil - 1\bigr]$, define
  \[
    t_m := \inf\bigl\{t \in [0,1]: \P_Q(T_i \leq t) \geq m\varepsilon^2\bigr\}
  \]
  and set $t_{\lceil 1/\varepsilon^2 \rceil} := 1$ and $\mathcal{T}_1 :=
  \bigl\{t_1,\ldots,t_{\lceil 1/\varepsilon^2 \rceil}\bigr\}$. Thus,
  given any $t \in [0,1]$,
  there exists $m \in \bigl[\lceil 1/\varepsilon^2 \rceil\bigr]$ such that
  $\P_Q(t \leq T_i < t_m) \leq \varepsilon^2$.
  Let $\cT_2$ and $\cT_3$ satisfy the same properties but for
  $T_j$ under $Q$ and $T$ under $P_T$ respectively.
  Then the set $\cT \vcentcolon= \cT_1 \cup \cT_2 \cup \cT_3$ is a
  $\sqrt{72} e^{2 \|f\|_\infty} H_{\infty,\gamma}^2 \varepsilon$-cover
  of $\mathcal G$ under $L_2(Q)$, of cardinality at most $6/\varepsilon^2$.
  Defining $\varepsilon' \vcentcolon = 3\varepsilon/ \sqrt{2} \in
  (0,3/\sqrt{2})$ and with $M$ as in~\eqref{eq:xi_irr_l2},
  we deduce that
  \begin{align*}
    N\big(\varepsilon' M,\mathcal G, \|\cdot\|_{Q,2}\big)
    \leq \frac{27}{\varepsilon^{\prime 2}}.
  \end{align*}
  Hence
  \begin{align*}
    J
    &\vcentcolon=
    \sup_{Q \in \mathcal{Q}}
    \int_0^{1}
    \log N\bigl(1 + \varepsilon M, \mathcal{G}, \|\cdot\|_{Q,2}\bigr)
    \diffi \varepsilon
    \leq
    1 + \int_0^{1}
    \log(27 / \varepsilon^2)
    \diffi \varepsilon = 1 + \log(27) + 2
    \leq 7.
  \end{align*}
  By Lemma~\ref{lem:orliczUStat}, there exists a
  universal constant $C_1 > 0$ such that
  for all $s \geq 1$,
  \[
    \P \biggl(
      \sup_{t \in [0, 1]}
      \biggl|
      \frac{2}{n(n-1)}
      \sum_{(i,j) \in \mathcal{I}}
      \sum_{r \in \mathcal{R}}
      \sum_{r' \in \mathcal{R}}
      a_r a_{r'}
      \xi_{i r r'}(t) \xi_{j r r'}(t)
      \biggr|
      \geq \frac{C_1 s^2 e^{2 \|f\|_\infty} H_{\infty,\gamma}^2}{n}
    \biggr)
    \leq e^{-s^2}.
  \]
  By combining this with the bound
  \eqref{eq:D2_second_derivative_term_U_stat} on the
  diagonal elements, there exists a universal constant $C_2 > 0$ such that
  for all $s \geq 1$,
  with probability at least $1 - e^{-s^2}$,
  \begin{align*}
    \sup_{t \in [0, 1]}
    \sum_{r \in \mathcal{R}} \sum_{r' \in \mathcal{R}}
    a_r a_{r'}
    \big( D^2 S_n(f, t)(e_r, e_{r'}) - D^2 S(f, t)(e_r, e_{r'}) \big)^2
    &\leq
    \frac{C_2 s^2 e^{2 \|f\|_\infty} H_{\infty,\gamma}^2}{n}.
    \qedhere
  \end{align*}
\end{proof}

\subsubsection{Bounding first-order terms from the basic inequality}

By the Cauchy--Schwarz inequality,
the inner product term from \eqref{eq:basic_inequality}
can be bounded using
\begin{align*}
  &\bigl|
  \bigl\langle \tilde f_{n,\gamma} - f_0 + P_n(f_0) 1_\cX,
  f_0 - P_n(f_0) 1_\cX \bigr\rangle_\cH
  \bigr|
  \leq
  \bigl\|f_0 - P_n(f_0) 1_\cX\bigr\|_\cH
  \bigl\| \tilde f_{n,\gamma} - f_0 + P_n(f_0) 1_\cX \bigr\|_\cH.
\end{align*}
For the term in \eqref{eq:basic_inequality} involving
$D \ell_n$, the definition of the operator norm gives
\begin{align*}
  \bigl| D \ell_n(f_0) \bigl( \tilde f_{n,\gamma} - f_0 \bigr) \bigr|
  &\leq
  \|D \ell_n(f_0)\|_{\cH_\gamma,\op}
  \, \bigl\| \tilde f_{n,\gamma} - f_0 \bigr\|_{\cH_\gamma}.
\end{align*}
We now present a high-probability bound for
$\|D \ell_n(f_0)\|_{\cH_\gamma,\op}$.
\begin{proposition}[Convergence of $D \ell_n$]
  \label{prop:convergence_Dln}
  There exists a universal constant $C > 0$ such that for all $s \geq 1$,
  with probability at least $1 - e^{-s^2}$,
  \begin{align*}
    \|D \ell_n(f_0)\|_{\cH_\gamma,\op}
    \leq
    \frac{C e^{2 \|f_0\|_\infty}}{q_1^2}
    \biggl(
      \frac{\sqrt{H_{2,\gamma}} + s}{\sqrt n}
      + s^2 \frac{\sqrt{H_{\infty,\gamma}}}{n}
    \biggr),
  \end{align*}
\end{proposition}

\begin{proof}[Proof of Proposition~\ref{prop:convergence_Dln}]
  For any $f_1 \in \cH$, we have
  by Lemma~\ref{lem:derivatives_ln} that
  \begin{align}
    \nonumber
    D \ell_n(f_0)(f_1)
    &=
    \int_{0}^{1}
    \frac{D S_n(f_0, t)(f_1)}{S_n(f_0, t)}
    \diffi {N(t)}
    - \frac{1}{n} \sum_{i=1}^n
    f_1(X_i) N_i(1) \\
    \label{eq:convergence_Dln_1}
    &=
    \int_{0}^{1}
    \frac{D S_n(f_0, t)(f_1) - D S(f_0, t)(f_1)}{S_n(f_0, t)}
    \diffi {N(t)} \\
    \label{eq:convergence_Dln_2}
    &\qquad+
    \int_{0}^{1}
    \frac{D S(f_0, t)(f_1)\{S(f_0, t) - S_n(f_0, t)\}}
    {S(f_0, t) S_n(f_0, t)}
    \diffi {N(t)} \\
    \label{eq:convergence_Dln_3}
    &\qquad+
    \int_{0}^{1}
    \frac{D S(f_0, t)(f_1)}{S(f_0, t)}
    \bigl(\mathrm{d} {N(t)} - S(f_0, t) \lambda_0(t) \, \mathrm{d} t \bigr) \\
    \label{eq:convergence_Dln_4}
    &\qquad+
    \int_{0}^{1}
    D S(f_0, t)(f_1) \lambda_0(t) \diffi t
    - \frac{1}{n} \sum_{i=1}^n
    f_1(X_i) N_i(1).
  \end{align}
  We begin by working on the event
  $\cA_n \coloneqq \{\inf_{t \in [0, 1]} S_n(f_0, t) \geq q_1 / 2\}$,
  so that
  $\P(\cA_n) \geq 1 - \exp\bigl(1- n q_1^2 e^{-2\|f_0\|_\infty} / C\bigr)$
  for a universal constant $C > 0$,
  by Lemma~\ref{lem:convergence_Sn},
  as $S(f_0, t) = \E\bigl\{q(X_1, t) e^{f_0(X_1)}\bigr\}
  \geq q_1 e^{\E\{f_0(X_1)\}} = q_1$ by Jensen's inequality.  By
  Lemma~\ref{lem:convergence_DSn},
  we may bound \eqref{eq:convergence_Dln_1} as
  \begin{align*}
    \sup_{f_1 \in \cH:\, \|f_1\|_{\cH_\gamma} \leq 1} &
    \biggl|
    \int_{0}^{1}
    \frac{D S_n(f_0, t)(f_1) - D S(f_0, t)(f_1)}{S_n(f_0, t)}
    \diffi {N(t)}
    \mathbbm{1}_{\cA_n}
    \biggr| \\
    &\qquad\leq
    \frac{2}{q_1}
    \sup_{f_1 \in \cH:\, \|f_1\|_{\cH_\gamma} \leq 1}
    \sup_{t \in [0, 1]}
    \bigl|
    D S_n(f_0, t)(f_1) - D S(f_0, t)(f_1)
    \bigr|
    \mathbbm{1}_{\cA_n}
    \\
    &\qquad\leq
    \frac{C e^{\|f_0\|_\infty}}{q_1}
    \biggl(
      \frac{\sqrt{H_{2,\gamma}} + s}{\sqrt n}
      + \frac{s^2 \sqrt{H_{\infty,\gamma}}}{n}
    \biggr),
  \end{align*}
  with probability at least $1 - e^{-s^2}$, after increasing the
  universal constant $C > 0$ if necessary.
  For all $f_1 \in \cH$ and $t \in [0, 1]$,
  by Lemma~\ref{lem:derivatives_S} and
  Jensen's inequality,
  \begin{align}
    \nonumber
    D S(f_0, t)(f_1)^2
    &=
    \biggl(
      \int_\cX
      q(x, t)
      e^{f_0(x)}
      f_1(x)
      \diffi P_X(x)
    \biggr)^2 \\
    \label{eq:DS_new}
    &\leq
    \int_\cX
    q(x, t)^2
    e^{2f_0(x)}
    f_1(x)^2
    \diffi P_X(x)
    \leq
    e^{2 \|f_0\|_\infty}
    \|f_1\|_{L_2}^2.
  \end{align}
  Therefore,
  by Lemma~\ref{lem:convergence_Sn},
  we may bound \eqref{eq:convergence_Dln_2} as
  \begin{align*}
    \sup_{f_1 \in \cH:\, \|f_1\|_{\cH_\gamma} \leq 1} &
    \biggl|
    \int_{0}^{1}
    \frac{D S(f_0, t)(f_1)\{S(f_0, t) - S_n(f_0, t)\}}
    {S(f_0, t) S_n(f_0, t)}
    \diffi {N(t)}
    \biggr|
    \mathbbm{1}_{\cA_n}
    \\
    &\qquad\leq
    \frac{2 e^{\|f_0\|_\infty}}{q_1^2}
    \sup_{t \in [0, 1]}
    \bigl| S_n(f_0, t) - S(f_0, t) \bigr|
    \mathbbm{1}_{\cA_n}
    \leq
    \frac{C s e^{2 \|f_0\|_\infty}}{\sqrt n q_1^2},
  \end{align*}
  with probability at least $1 - e^{-s^2}$, again after increasing
  the universal constant $C > 0$ if necessary.
  We now control the term in \eqref{eq:convergence_Dln_3}.
  By Lemma~\ref{lem:operator_norm_bounds},
  \begin{align*}
    &\sup_{f_1 \in \cH:\, \|f_1\|_{\cH_\gamma} \leq 1}
    \biggl(
      \int_{0}^{1}
      \frac{D S(f_0, t)(f_1)}{S(f_0, t)}
      \bigl(\mathrm{d} {N(t)} - S(f_0, t) \lambda_0(t) \mathrm{d} t \bigr)
    \biggr)^2
    \leq
    \sum_{r \in \cR} a_r
    \biggl(
      \frac{1}{n}
      \sum_{i=1}^n
      \xi_{i r}
    \biggr)^2,
  \end{align*}
  where
  \begin{align}
    \label{eq:xi_ir_def_new}
    \xi_{i r}
    &\vcentcolon=
    \int_{0}^{1}
    \frac{D S(f_0, t)(e_r)}{S(f_0, t)}
    \big(
      \mathrm d N_i(t) - S(f_0, t) \lambda_0(t)\,\mathrm d t
    \big).
  \end{align}
  Note that $\xi_{i r}$ is $(X_i, T_i, I_i)$-measurable,
  and $\E(\xi_{i r}) = 0$ because
  the integrand in~\eqref{eq:xi_ir_def_new} is non-random and bounded, and
  $1 \geq \E\{N_i(t)\}
  = \int_0^t \lambda_0(s) S(f_0, s) \diffi s$
  for all $t \in [0, 1]$ by \citet[Theorem~1.3.1]{fleming2013counting}.
  Further, by Lemma~\ref{lem:bv_cs} and \eqref{eq:DS_new},
  \begin{align}
    \nonumber
    \sum_{r \in \mathcal R}
    a_r \xi_{i r}^2
    &\leq
    \sup_{t \in [0, 1]}
    \sum_{r \in \mathcal R}
    a_r
    \frac{D S(f_0, t)(e_r)^2}{S(f_0, t)^2}
    \biggl(
      N_i(1)
      + \int_{0}^{1} S(f_0, t) \lambda_0(t) \diffi t
    \biggr)^2 \\
    \label{eq:xi_ir_bound_new}
    &\leq
    \frac{4}{q_1^2}
    \sup_{t \in [0, 1]}
    \sum_{r \in \mathcal R}
    a_r
    D S(f_0, t)(e_r)^2
    \leq
    \frac{4 e^{2 \|f_0\|_\infty} H_{2,\gamma}}{q_1^2}.
  \end{align}
  Therefore, by Cauchy--Schwarz,
  with $\mathcal{I} := \{(i,j):i,j \in [n],i<j\}$, for $n \geq 2$,
  \begin{align*}
    \sum_{r \in \mathcal R}
    a_r
    \biggl( \frac{1}{n} \sum_{i=1}^n \xi_{i r} \biggr)^2 &\leq
    \frac{1}{n^2}
    \sum_{i=1}^n
    \sum_{r \in \mathcal R}
    a_r
    \xi_{i r}^{2}
    + \frac{2}{n^2}
    \sum_{(i, j) \in \mathcal I}
    \sum_{r \in \mathcal R}
    a_r
    \xi_{i r}
    \xi_{j r} \\
    &\leq
    \frac{4 e^{2 \|f_0\|_\infty} H_{2,\gamma}}{n q_1^2}
    + \biggl|
    \frac{2}{n (n-1)}
    \sum_{(i, j) \in \mathcal I}
    \sum_{r \in \mathcal R}
    a_r \xi_{i r} \xi_{j r}
    \biggr|.
  \end{align*}
  In the case $n=1$, the second term in the bound is not needed.
  Otherwise,
  since $\xi_{i r}$ and $\xi_{j r}$ are independent when
  $i < j$,
  this term is a degenerate second-order $U$-statistic
  so by Cauchy--Schwarz, \eqref{eq:xi_ir_bound_new}
  and Lemma~\ref{lem:orliczUStat},
  for all $s \geq 1$,
  \begin{align*}
    \biggl|
    \frac{2}{n (n-1)}
    \sum_{(i, j) \in \mathcal I}
    \sum_{r \in \mathcal R}
    a_r
    \xi_{i r}
    \xi_{j r}
    \biggr|
    \leq
    \frac{C s^2 e^{2 \|f_0\|_\infty} H_{2,\gamma}}{n q_1^2}
  \end{align*}
  with probability at least $1 - e^{-s^2}$, after
  increasing the universal constant $C > 0$ if necessary.
  As such, for \eqref{eq:convergence_Dln_3},
  \begin{align*}
    &\sup_{f_1 \in \cH:\, \|f_1\|_{\cH_\gamma} \leq 1}
    \biggl(
      \int_{0}^{1}
      \frac{D S(f_0, t)(f_1)}{S(f_0, t)}
      \bigl(\mathrm{d} {N(t)} - S(f_0, t) \lambda_0(t) \mathrm{d} t \bigr)
    \biggr)^2
    \leq
    \frac{C s^2 e^{2 \|f_0\|_\infty} H_{2,\gamma}}{n q_1^2}
  \end{align*}
  with probability at least $1 - e^{-s^2}$
  for all $s \geq 1$, once more
  increasing $C > 0$ if necessary.

  For \eqref{eq:convergence_Dln_4}, by
  \citet[Theorem~1.3.1]{fleming2013counting},
  symmetrisation
  \citep[Lemma~2.3.1]{van1996weak}
  and a contraction principle
  \citep[Proposition~A.1.10]{van1996weak},
  with independent Rademacher variables $\varepsilon_1, \ldots, \varepsilon_n$
  that are independent of the data,
  and by the proof of Lemma~\ref{lem:uniform_convergence_Pn},
  \begin{align*}
    &\E\Biggl\{
      \sup_{f_1 \in \cH:\, \|f_1\|_{\cH_\gamma} \leq 1}
      \biggl|
      \frac{1}{n} \sum_{i=1}^n f_1(X_i) N_i(1)
      - \int_{0}^{1} D S(f_0, t)(f_1) \lambda_0(t) \diffi t
      \biggr|
    \Biggr\} \\
    &\qquad\leq
    2\, \E\Biggl\{
      \sup_{f_1 \in \cH:\, \|f_1\|_{\cH_\gamma} \leq 1}
      \frac{1}{n} \sum_{i=1}^n
      \varepsilon_i
      f_1(X_i)
    \Biggr\}
    \leq
    2 \sqrt{\frac{H_{2,\gamma}}{n}}.
  \end{align*}
  By Lemma~\ref{lem:infty_bounds}, we have
  \begin{align*}
    \sigma^2
    &\coloneqq
    \sup_{f_1 \in \cH:\, \|f_1\|_{\cH_\gamma} \leq 1}
    \sum_{i=1}^n
    \E \biggl\{
      \biggl(
        \frac{1}{n}
        f_1(X_i)
        N_i(1)
      \biggr)^2
    \biggr\}
    \leq
    \sup_{f_1 \in \cH:\, \|f_1\|_{\cH_\gamma} \leq 1}
    \frac{\|f_1\|_{L_2}^2 }{n}
    \leq
    \frac{1}{n}, \\
    M
    &\coloneqq
    \sup_{x \in \cX}
    \sup_{f_1 \in \cH:\, \|f_1\|_{\cH_\gamma} \leq 1}
    \biggl|
    \frac{1}{n}
    f_1(x)
    N_i(1)
    \biggr|
    \leq
    \sup_{f_1 \in \cH:\, \|f_1\|_{\cH_\gamma} \leq 1}
    \frac{\|f_1\|_\infty}{n}
    \leq
    \frac{\sqrt{H_{\infty,\gamma}}}{n}.
  \end{align*}
  By \citet[Theorem~12.5]{boucheron2013concentration}, by increasing
  the universal constant $C > 0$ if necessary,
  with probability at least $1 - e^{-s^2}$,
  \begin{align*}
    \sup_{f_1 \in \cH:\, \|f_1\|_{\cH_\gamma} \leq 1}
    \biggl|
    \frac{1}{n} \sum_{i=1}^n f_1(X_i) N_i(1)
    - \int_{0}^{1} D S(f_0, t)(f_1) \lambda_0(t) \diffi t
    \biggr|
    &\leq
    C \biggl(
      \frac{\sqrt{H_{2,\gamma}} + s}{\sqrt n}
      + \frac{s^2 \sqrt{H_{\infty,\gamma}}}{n}
    \biggr).
  \end{align*}
  We deduce that
  \begin{align}
    \label{eq:Dln_bound}
    \|D \ell_n(f_0)\|_{\cH_\gamma,\op}
    \leq
    \frac{C e^{2 \|f_0\|_\infty}}{q_1^2}
    \biggl(
      \frac{\sqrt{H_{2,\gamma}} + s}{\sqrt n}
      + s^2 \frac{\sqrt{H_{\infty,\gamma}}}{n}
    \biggr),
  \end{align}
  with probability at least
  $1 - 4e^{-s^2} - \P(\cA_n) \geq
  1 - 4 e^{-s^2} - \exp\bigl(- n q_1^2 e^{-2\|f_0\|_\infty} / C\bigr)$,
  for all $s \geq 1$.
  By increasing the universal constant $C > 0$ again, we conclude
  that \eqref{eq:Dln_bound} holds
  with probability at least
  $1 - e^{-s^2}$ whenever
  $1 \leq s^2 \leq n q_1^2 e^{-2\|f_0\|_\infty} / C$.  Now suppose
  that $s^2 > n q_1^2 e^{-2\|f_0\|_\infty} / C$.
  Since
  \begin{align*}
    \bigl| D S_n(f_0, t)(f_1) \bigr|
    &=
    \biggl|
    \frac{1}{n}
    \sum_{i=1}^n
    R_i(t) e^{f_0(X_i)}
    f_1(X_i)
    \biggr|
    \leq
    \|f_1\|_\infty S_n(f_0, t),
  \end{align*}
  we have by
  Lemma~\ref{lem:derivatives_ln} and
  Lemma~\ref{lem:infty_bounds} that
  \begin{align*}
    \| D \ell_n(f_0) \|_{\cH_\gamma, \op}
    &\leq
    \sup_{f_1 \in \cH:\, \|f_1\|_{\cH_\gamma} \leq 1}
    \biggl\{
      \sup_{t \in [0, 1]}
      \biggl|
      \frac{D S_n(f_0, t)(f_1)}{S_n(f_0, t)}
      \biggr|
      + \|f_1\|_\infty
    \biggr\} \\
    &\leq
    2 \sup_{f_1 \in \cH:\, \|f_1\|_{\cH_\gamma} \leq 1}
    \|f_1\|_\infty
    \leq
    2 \sqrt{H_{\infty,\gamma}}.
  \end{align*}
  By increasing $C$ again,
  with $s^2 > n q_1^2 e^{-2\|f_0\|_\infty} / C$,
  we have
  \begin{align*}
    \frac{C^2 e^{2 \|f_0\|_\infty}}{q_1^2}
    \biggl(
      \frac{\sqrt{H_{2,\gamma}} + s}{\sqrt n}
      + s^2 \frac{\sqrt{H_{\infty,\gamma}}}{n}
    \biggr)
    \geq
    C \sqrt{H_{\infty,\gamma}},
  \end{align*}
  and so \eqref{eq:Dln_bound} holds with probability one,
  concluding the proof.
\end{proof}

\subsubsection{Bounding second-order terms from the basic inequality}

We establish a lower bound for the term in \eqref{eq:basic_inequality}
involving $D^2 \ell_n$. Observe that
\begin{align*}
  D^2 \ell_{n}(f_0)
  \bigl( \tilde f_{n,\gamma} - f_0 \bigr)^{\otimes 2}
  &\geq
  D^2 \ell_{\star}(f_0)
  \bigl( \tilde f_{n,\gamma} - f_0 \bigr)^{\otimes 2}
  - \bigl\|D^2 \ell_{n}(f_0) - D^2 \ell_{\star}(f_0)\bigr\|_{\cH_\gamma,\op}
  \, \bigl\| \tilde f_{n,\gamma} - f_0 \bigr\|_{\cH_\gamma}^2.
\end{align*}

Lemma~\ref{lem:strong_convexity} already provides a
lower bound based on restricted strong convexity
for $D^2 \ell_\star(f_0)$;
we now show convergence of
$D^2 \ell_{n}(f_0)$ to $D^2 \ell_{\star}(f_0)$
in Proposition~\ref{prop:convergence_D2ln}.

\begin{proposition}[Convergence of $D^2 \ell_n$]
  \label{prop:convergence_D2ln}
  There exists a universal constant $C > 0$ such that for all $s \geq 1$,
  with probability at least $1 - e^{-s^2}$,
  \begin{align*}
    \bigl\|D^2 \ell_{n}(f_0)
    - D^2 \ell_{\star}(f_0)\bigr\|_{\cH_\gamma,\op}
    \leq
    \frac{C s e^{2 \|f_0\|_\infty}}{q_1^2}
    \sqrt{\frac{H_{\infty,\gamma}^2}{n}}.
  \end{align*}
\end{proposition}

\begin{proof}[Proof of Proposition~\ref{prop:convergence_D2ln}]
  Since the second-order Gateaux derivative is a bilinear operator,
  we first state a general result allowing us to bound
  the bilinear $\|\cdot\|_{\cH_\gamma,\op}$-operator norm
  in terms of actions on the eigenfunctions.
  To this end, let $G: \cH \times \cH \to \R$ be bilinear.
  We have by the Cauchy--Schwarz inequality that
  \begin{align}
    \nonumber
    \|G\|_{\cH_\gamma,\op}^2
    &\vcentcolon=
    \sup_{f, g \in \mathcal{H}:\|f\|_{\cH_\gamma} = 1, \|g\|_{\cH_\gamma}=1}
    G(f, g)^2 \\
    \nonumber
    &\phantom{\vcentcolon}=
    \sup_{f, g \in \mathcal{H}:\|f\|_{\cH_\gamma} = 1, \|g\|_{\cH_\gamma}=1}
    \biggl(
      \sum_{r \in \mathcal{R}}
      \sum_{r' \in \mathcal{R}}
      \langle f, e_r \rangle_{L_2}
      \langle g, e_{r'} \rangle_{L_2}
      G(e_r, e_{r'})
    \biggr)^2 \\
    \nonumber
    &\phantom{\vcentcolon}\leq
    \sup_{f, g \in \mathcal{H}:\|f\|_{\cH_\gamma} = 1, \|g\|_{\cH_\gamma}=1}
    \biggl(
      \sum_{r \in \mathcal{R}}
      \frac{\langle f, e_r \rangle_{L_2}^2}{a_r}
    \biggr)
    \biggl(
      \sum_{r' \in \mathcal{R}}
      \frac{\langle g, e_{r'} \rangle_{L_2}^2}{a_{r'}}
    \biggr)\sum_{r \in \mathcal{R}}
    \sum_{r' \in \mathcal{R}}
    a_r a_{r'}
    G(e_r, e_{r'})^2 \\
    \label{eq:bilinear_operator_norm}
    &\phantom{\vcentcolon}=
    \sum_{r \in \mathcal{R}}
    \sum_{r' \in \mathcal{R}}
    a_r a_{r'}
    G(e_r, e_{r'})^2.
  \end{align}
  We see from this that
  \begin{align}
    \label{eq:D2_operator_bound}
    \big\|
    D^2 \ell_{n}(f_0) - D^2 \ell_\star(f_0)
    \big\|_{\cH_\gamma,\op}^2
    &\leq
    \sum_{r \in \mathcal{R}} \sum_{r' \in \mathcal{R}}
    a_r a_{r'}
    \big(
      D^2 \ell_{n}(f_0)(e_r, e_{r'}) - D^2 \ell_\star(f_0)(e_r, e_{r'})
    \big)^2.
  \end{align}
  Now, for any $r,r' \in \mathcal{R}$, by
  Lemmas~\ref{lem:derivatives_ln} and~\ref{lem:derivatives_lstar},
  \begin{align}
    \nonumber
    &D^2 \ell_n(f_0)(e_r, e_{r'})
    - D^2 \ell_\star(f_0)(e_r, e_{r'}) \\
    \label{eq:D2_second_derivative_term}
    &=
    \int_{0}^{1}
    \biggl(
      \frac{D^2 S_n(f_0, t)(e_r, e_{r'})}{S_n(f_0, t)}
      - \frac{D^2 S(f_0, t)(e_r, e_{r'})}{S(f_0, t)}
    \biggr)
    \diffi {N(t)} \\
    \label{eq:D2_first_derivative_term}
    &\hspace{0.7cm}-
    \int_{0}^{1}
    \biggl(
      \frac{D S_n(f_0, t)(e_r)}{S_n(f_0, t)}
      \frac{D S_n(f_0, t)(e_{r'})}{S_n(f_0, t)}
      - \frac{D S(f_0, t)(e_r)}{S(f_0, t)}
      \frac{D S(f_0, t)(e_{r'})}{S(f_0, t)}
    \biggr)
    \diffi N(t)
    \\
    \label{eq:D2_integrator_term}
    &\hspace{0.7cm}+
    \int_{0}^{1}
    \biggl(
      \frac{D^2 S(f_0, t)(e_r, e_{r'})}{S(f_0, t)}
      - \frac{D S(f_0, t)(e_r)}{S(f_0, t)}
      \frac{D S(f_0, t)(e_{r'})}{S(f_0, t)}
    \biggr)
    \big(
      \mathrm d N(t) - S(f_0, t) \lambda_0(t)\,\mathrm d t
    \big).
  \end{align}
  For \eqref{eq:D2_second_derivative_term}, observe that
  \begin{align*}
    &\frac{D^2 S_n(f_0, t)(e_r, e_{r'})}{S_n(f_0, t)}
    - \frac{D^2 S(f_0, t)(e_r, e_{r'})}{S(f_0, t)} \\
    &\quad=
    \frac{D^2 S_n(f_0, t)(e_r, e_{r'})\big(S(f_0, t) - S_n(f_0, t)\big)}
    {S(f_0, t) S_n(f_0, t)} + \frac{D^2 S_n(f_0, t)(e_r, e_{r'}) -
    D^2 S(f_0, t)(e_r, e_{r'})}{S(f_0, t)}.
  \end{align*}
  Now, by the Cauchy--Schwarz inequality, for $t \in [0, 1]$,
  \begin{align}
    \nonumber
    &\sum_{r \in \mathcal R}
    \sum_{r' \in \mathcal R}
    a_r a_{r'}
    D^2 S_n(f_0, t)(e_r, e_{r'})^2 \\
    \nonumber
    &\quad=
    \frac{1}{n^2}
    \sum_{i=1}^n
    \sum_{j=1}^n
    R_i(t) R_j(t)
    e^{f_0(X_i)}
    e^{f_0(X_j)}
    \biggl(
      \sum_{r \in \mathcal R} a_r e_r(X_i) e_r(X_j)
    \biggr)^2 \\
    \nonumber
    &\quad\leq
    \frac{1}{n^2}
    \sum_{i=1}^n
    \sum_{j=1}^n
    R_i(t) R_j(t)
    e^{f_0(X_i)}
    e^{f_0(X_j)}
    \biggl(
      \sum_{r \in \mathcal R}
      a_r e_r(X_i)^2
    \biggr)
    \biggl(
      \sum_{r \in \mathcal R}
      a_r e_r(X_j)^2
    \biggr) \\
    \label{eq:second_deriv_Hgamma_bound}
    &\quad\leq
    H_{\infty,\gamma}^2
    S_n(f_0, t)^2.
  \end{align}
  For all $t \in [0, 1]$, by Jensen's inequality
  and as $\int_\cX f_0(x) \diffi P_X(x) = 0$,
  \begin{align}
    \label{eq:lower_bound_S0}
    S(f_0, t) = \int_\cX q(x, t) e^{f_0(x)} \diffi P_X(x)
    \geq q_1 \int_\cX e^{f_0(x)} \diffi P_X(x)
    \geq q_1 e^{\int_\cX f_0(x) \diffi P_X(x)}
    = q_1.
  \end{align}
  Hence, since $\int_0^1 \mathrm d N(t) = N(1) \leq 1$,
  it follows from
  Lemma~\ref{lem:bv_cs}, \eqref{eq:second_deriv_Hgamma_bound} and
  \eqref{eq:lower_bound_S0} that
  \begin{align}
    \nonumber
    &\sum_{r \in \mathcal{R}} \sum_{r' \in \mathcal{R}}
    a_r a_{r'}
    \biggl\{
      \int_{0}^{1}
      \biggl(
        \frac{D^2 S_n(f_0, t)(e_r, e_{r'})}{S_n(f_0, t)}
        - \frac{D^2 S(f_0, t)(e_r, e_{r'})}{S(f_0, t)}
      \biggr)
      \diffi {N(t)}
    \biggr\}^2 \\
    \nonumber
    &\quad \leq \sum_{r \in \mathcal{R}} \sum_{r' \in \mathcal{R}}
    2 a_r a_{r'}
    \biggl(
      \int_{0}^{1}
      \frac{D^2 S_n(f_0, t)(e_r, e_{r'})\big(S(f_0, t) - S_n(f_0, t)\big)}
      {S(f_0, t) S_n(f_0, t)}
      \diffi {N(t)}
    \biggr)^2 \\
    \nonumber
    &\qquad + \sum_{r \in \mathcal{R}} \sum_{r' \in \mathcal{R}}
    2 a_r a_{r'}
    \biggl(
      \int_{0}^{1}
      \frac{D^2 S_n(f_0, t)(e_r, e_{r'}) - D^2 S(f_0, t)(e_r,
      e_{r'})}{S(f_0, t)}
      \diffi {N(t)}
    \biggr)^2 \\
    \nonumber
    &\quad\leq
    \sup_{t \in [0, 1]}
    \sum_{r \in \mathcal{R}} \sum_{r' \in \mathcal{R}}
    2 a_r a_{r'}
    \biggl(
      \frac{D^2 S_n(f_0, t)(e_r, e_{r'})\big(S(f_0, t) - S_n(f_0, t)\big)}
      {S(f_0, t) S_n(f_0, t)}
    \biggr)^2 \\
    \nonumber
    &\qquad+
    \sup_{t \in [0, 1]}
    \sum_{r \in \mathcal{R}} \sum_{r' \in \mathcal{R}}
    2 a_r a_{r'}
    \biggl(
      \frac{D^2 S_n(f_0, t)(e_r, e_{r'}) - D^2 S(f_0, t)(e_r,
      e_{r'})}{S(f_0, t)}
    \biggr)^2 \\
    \label{eq:second_derivative_term_1}
    &\quad\leq
    \frac{2 H_{\infty,\gamma}^2}{q_1^2}
    \sup_{t \in [0, 1]}
    \big(S_n(f_0, t) - S(f_0, t)\big)^2 \\
    \label{eq:second_derivative_term_2}
    &\qquad+
    \frac{2}{q_1^2}
    \sup_{t \in [0, 1]}
    \sum_{r \in \mathcal{R}} \sum_{r' \in \mathcal{R}}
    a_r a_{r'}
    \big( D^2 S_n(f_0, t)(e_r, e_{r'}) - D^2 S(f_0, t)(e_r, e_{r'}) \big)^2.
  \end{align}
  By Lemma~\ref{lem:convergence_Sn} applied to
  \eqref{eq:second_derivative_term_1}
  and Lemma~\ref{lem:convergence_D2Sn} applied to
  \eqref{eq:second_derivative_term_2},
  there exists a universal constant
  $C_2 > 0$ such that for all $s \geq 1$,
  with probability at least $1 - e^{-s^2}$,
  \begin{align}
    \label{eq:D2_second_derivative_term_bound}
    \sum_{r \in \mathcal{R}}\sum_{r' \in \mathcal{R}}
    a_r a_{r'}
    \biggl\{
      \int_{0}^{1}
      \!
      \biggl(
        \frac{D^2 S_n(f_0, t)(e_r, e_{r'})}{S_n(f_0, t)}
        - \frac{D^2 S(f_0, t)(e_r, e_{r'})}{S(f_0, t)}
      \biggr)
    \diffi {N(t)}\biggr\}^2
    \leq \frac{C_2 s^2 e^{2 \|f_0\|_\infty}
    H_{\infty,\gamma}^2}{n q_1^2}.
  \end{align}
  This concludes the analysis for
  \eqref{eq:D2_second_derivative_term},
  so we now turn to \eqref{eq:D2_first_derivative_term}.
  For $t \in [0, 1]$ and $r, r' \in \mathcal R$,
  \begin{align*}
    \frac{D S_n(f_0, t)(e_r)}{S_n(f_0, t)}&
    \frac{D S_n(f_0, t)(e_{r'})}{S_n(f_0, t)}
    - \frac{D S(f_0, t)(e_r)}{S(f_0, t)}
    \frac{D S(f_0, t)(e_{r'})}{S(f_0, t)} \\
    &=
    \frac{D S_n(f_0, t)(e_r) D S_n(f_0, t)(e_{r'})
      \big(S(f_0, t) - S_n(f_0, t)\big)
      \big(S(f_0, t) + S_n(f_0, t)\big)
    }{S(f_0, t)^2 S_n(f_0, t)^2} \\
    &\hspace{2cm}
    +\frac{D S_n(f_0, t)(e_r)
    \big(D S_n(f_0, t)(e_{r'}) - D S(f_0, t)(e_{r'})\big)}{S(f_0, t)^2} \\
    &\hspace{2cm}
    +\frac{D S(f_0, t)(e_{r'})
    \big(D S_n(f_0, t)(e_r) - D S(f_0, t)(e_r)\big)}{S(f_0, t)^2}.
  \end{align*}
  Therefore by Lemma~\ref{lem:bv_cs}, since $\int_0^1 \mathrm d N(t) \leq 1$,
  \begin{align*}
    &\sum_{r \in \mathcal{R}} \sum_{r' \in \mathcal{R}}
    a_r a_{r'}
    \biggl\{
      \int_{0}^{1}
      \biggl(
        \frac{D S_n(f_0, t)(e_r)}{S_n(f_0, t)}
        \frac{D S_n(f_0, t)(e_{r'})}{S_n(f_0, t)}
        - \frac{D S(f_0, t)(e_r)}{S(f_0, t)}
        \frac{D S(f_0, t)(e_{r'})}{S(f_0, t)}
      \biggr)
      \diffi N(t)
    \biggr\}^2 \\
    &\quad\leq
    \sup_{t \in [0, 1]}
    \sum_{r \in \mathcal{R}} \sum_{r' \in \mathcal{R}}
    a_r a_{r'}
    \biggl(
      \frac{D S_n(f_0, t)(e_r)}{S_n(f_0, t)}
      \frac{D S_n(f_0, t)(e_{r'})}{S_n(f_0, t)}
      - \frac{D S(f_0, t)(e_r)}{S(f_0, t)}
      \frac{D S(f_0, t)(e_{r'})}{S(f_0, t)}
    \biggr)^2 \\
    &\quad\leq
    3 \sup_{t \in [0, 1]}
    \sum_{r \in \mathcal{R}} \sum_{r' \in \mathcal{R}}
    a_r a_{r'}
    \biggl(
      \frac{D S_n(f_0, t)(e_r) D S_n(f_0, t)(e_{r'})
      }{S(f_0, t)^2 S_n(f_0, t)^2}
    \biggr)^2 \\
    &\qquad\qquad\qquad\times
    \big(S(f_0, t) - S_n(f_0, t)\big)^2
    \big(S(f_0, t) + S_n(f_0, t)\big)^2 \\
    &\qquad+
    3 \sup_{t \in [0, 1]}
    \sum_{r \in \mathcal{R}} \sum_{r' \in \mathcal{R}}
    a_r a_{r'}
    \biggl(
      \frac{D S_n(f_0, t)(e_r)
      \big(D S_n(f_0, t)(e_{r'}) - D S(f_0, t)(e_{r'})\big)}{S(f_0, t)^2}
    \biggr)^2 \\
    &\qquad+
    3 \sup_{t \in [0, 1]}
    \sum_{r \in \mathcal{R}} \sum_{r' \in \mathcal{R}}
    a_r a_{r'}
    \biggl(
      \frac{D S(f_0, t)(e_{r'})
      \big(D S_n(f_0, t)(e_r) - D S(f_0, t)(e_r)\big)}{S(f_0, t)^2}
    \biggr)^2 \\
    &\quad\leq
    6 \sup_{t \in [0, 1]}
    \frac{S(f_0, t)^2 + S_n(f_0, t)^2}{S(f_0, t)^4}
    \big(S_n(f_0, t) - S(f_0, t)\big)^2
    \biggl(
      \sum_{r \in \mathcal{R}}
      a_r
      \frac{D S_n(f_0, t)(e_r)^2}{S_n(f_0, t)^2}
    \biggr)^2
    \\
    &\qquad+
    6 \sup_{t \in [0, 1]}
    \frac{1}{S(f_0, t)^4}
    \sum_{r \in \mathcal{R}}
    a_r
    D S_n(f_0, t)(e_r)^2
    \sum_{r' \in \mathcal{R}}
    a_{r'}
    \big(D S_n(f_0, t)(e_{r'}) - D S(f_0, t)(e_{r'})\big)^2.
  \end{align*}
  By Lemma~\ref{lem:derivatives_S},
  the dominated convergence theorem
  and the Cauchy--Schwarz inequality,
  for each $t \in [0, 1]$,
  \begin{align}
    \nonumber
    &\sum_{r \in \mathcal{R}} a_r D S(f_0, t)(e_r)^2
    =
    \int_\cX
    \int_\cX
    q(x, t)
    q(x', t)
    e^{f_0(x)}
    e^{f_0(x')}
    \sum_{r \in \mathcal{R}}
    a_r e_r(x) e_r(x')
    \diffi P_X(x)
    \diffi P_X(x') \\
    \nonumber
    &\quad\leq
    \int_\cX
    \int_\cX
    q(x, t)
    q(x', t)
    e^{f_0(x)}
    e^{f_0(x')}
    \biggl\{
      \sum_{r \in \mathcal{R}} a_r e_r(x)^2
    \biggr\}^{1/2}
    \biggl\{
      \sum_{r \in \mathcal{R}} a_r e_r(x')^2
    \biggr\}^{1/2}
    \diffi P_X(x)
    \diffi P_X(x') \\
    \label{eq:DS_over_S}
    &\quad\leq
    H_{\infty,\gamma}
    S(f_0, t)^2.
  \end{align}
  Moreover, by Lemma~\ref{lem:derivatives_Sn} and the Cauchy--Schwarz
  inequality,
  \begin{align}
    \nonumber
    &\sum_{r \in \mathcal{R}} a_r D S_n(f_0, t)(e_r)^2
    =
    \frac{1}{n^2}
    \sum_{i=1}^n
    \sum_{j=1}^n
    R_i(t)
    R_j(t)
    e^{f_0(X_i)}
    e^{f_0(X_j)}
    \sum_{r \in \mathcal{R}}
    a_r e_r(X_i) e_r(X_j)
    \\
    \nonumber
    &\quad\leq
    \frac{1}{n^2}
    \sum_{i=1}^n
    \sum_{j=1}^n
    R_i(t)
    R_j(t)
    e^{f_0(X_i)}
    e^{f_0(X_j)}
    \biggl\{
      \sum_{r \in \mathcal{R}} a_r e_r(X_i)^2
    \biggr\}^{1/2}
    \biggl\{
      \sum_{r \in \mathcal{R}} a_r e_r(X_j)^2
    \biggr\}^{1/2} \\
    \label{eq:DSn_over_Sn}
    &\quad\leq
    H_{\infty,\gamma}
    S_n(f_0, t)^2.
  \end{align}
  Applying \eqref{eq:lower_bound_S0}, \eqref{eq:DSn_over_Sn}
  and \eqref{eq:DS_over_S},
  and noting that $|S(f_0, t)| \leq e^{\|f_0\|_\infty}$
  and $|S_n(f_0, t)| \leq e^{\|f_0\|_\infty}$,
  \begin{align*}
    \sum_{r \in \mathcal{R}} \sum_{r' \in \mathcal{R}}
    a_r a_{r'}
    \biggl\{
      &\int_{0}^{1}
      \biggl(
        \frac{D S_n(f_0, t)(e_r)}{S_n(f_0, t)}
        \frac{D S_n(f_0, t)(e_{r'})}{S_n(f_0, t)}
        - \frac{D S(f_0, t)(e_r)}{S(f_0, t)}
        \frac{D S(f_0, t)(e_{r'})}{S(f_0, t)}
      \biggr)
      \diffi N(t)
    \biggr\}^2 \\
    &\leq
    \frac{12 e^{2 \|f_0\|_\infty} H_{\infty,\gamma}^2}{q_1^4}
    \sup_{t \in [0, 1]}
    \big(S_n(f_0, t) - S(f_0, t)\big)^2 \\
    &\hspace{2cm}+
    \frac{6 e^{2 \|f_0\|_\infty} H_{\infty,\gamma}}{q_1^4}
    \sup_{t \in [0, 1]}
    \sum_{r \in \mathcal{R}}
    a_{r}
    \big(D S_n(f_0, t)(e_{r}) - D S(f_0, t)(e_{r})\big)^2.
  \end{align*}
  By Lemma~\ref{lem:convergence_Sn} and
  Lemma~\ref{lem:convergence_DSn},
  there is a universal constant $C_3 > 0$ such that for all $s \geq 1$,
  \begin{align}
    \nonumber
    \sum_{r \in \mathcal{R}} \sum_{r' \in \mathcal{R}}
    a_r a_{r'}
    \biggl\{
      \int_{0}^{1}
      \biggl(
        \frac{D S_n(f_0, t)(e_r)}{S_n(f_0, t)}
        \frac{D S_n(f_0, t)(e_{r'})}{S_n(f_0, t)}
        &- \frac{D S(f_0, t)(e_r)}{S(f_0, t)}
        \frac{D S(f_0, t)(e_{r'})}{S(f_0, t)}
      \biggr)
      \diffi N(t)
    \biggr\}^2 \\
    \label{eq:D2_first_derivative_term_bound}
    &\quad\leq
    \frac{C_3 s^2 e^{4 \|f_0\|_\infty} H_{\infty,\gamma}^2}{n q_1^4}
  \end{align}
  with probability at least $1 - e^{-s^2}$.
  We finally consider \eqref{eq:D2_integrator_term}.
  For $i \in [n]$ and $r, r' \in \mathcal R$, define
  \begin{align}
    \label{eq:xi_irr_def}
    \xi'_{i r r'}
    &\vcentcolon=
    \int_{0}^{1}
    \biggl(
      \frac{D^2 S(f_0, t)(e_r, e_{r'})}{S(f_0, t)}
      - \frac{D S(f_0, t)(e_r)}{S(f_0, t)}
      \frac{D S(f_0, t)(e_{r'})}{S(f_0, t)}
    \biggr)
    \big(
      \mathrm d N_i(t) - S(f_0, t) \lambda_0(t)\,\mathrm d t
    \big).
  \end{align}
  By
  Lemma~\ref{lem:derivatives_S} and
  Cauchy--Schwarz, for $t \in [0,1]$,
  \begin{align*}
    &\sum_{r \in \mathcal R}
    \sum_{r' \in \mathcal R}
    a_r a_{r'}
    D^2 S(f_0, t)(e_r, e_{r'})^2 \\
    &\quad=
    \sum_{r \in \mathcal R}
    \sum_{r' \in \mathcal R}
    a_r a_{r'}
    \left(
      \int_{\cX}
      q(x, t)
      e^{f_0(x)}
      e_r(x) e_{r'}(x)
      \diffi P_X(x)
    \right)^2 \\
    &\quad=
    \int_{\cX}
    \int_{\cX}
    q(x, t)
    q(x', t)
    e^{f_0(x)}
    e^{f_0(x')}
    \sum_{r \in \mathcal R}
    a_r
    e_r(x)
    e_r(x')
    \sum_{r' \in \mathcal R}
    a_{r'}
    e_{r'}(x)
    e_{r'}(x')
    \diffi P_X(x)
    \diffi P_X(x') \\
    &\quad\leq
    H_{\infty,\gamma}^2
    \biggl(
      \int_{\cX}
      q(x, t)
      e^{f_0(x)}
      \diffi P_X(x)
    \biggr)^2
    = H_{\infty,\gamma}^2
    S(f_0, t)^2.
  \end{align*}
  Note that $\xi'_{i r r'}$ is $(X_i, T_i, I_i)$-measurable,
  and $\E(\xi'_{i r r'}) = 0$ because
  the integrand in~\eqref{eq:xi_irr_def} is non-random and bounded, and
  $\E\{N(t)\} = \int_0^t \lambda_0(s) S(f_0, s) \diffi s \leq 1$
  for all $t \in [0, 1]$.
  Thus, by Lemma~\ref{lem:bv_cs} and \eqref{eq:DS_over_S},
  since $|S(f_0, t)| \leq e^{\|f_0\|_\infty}$,
  \begin{align}
    \nonumber
    \sum_{r \in \mathcal R}
    \sum_{r' \in \mathcal R}
    a_r a_{r'}
    \xi_{i r r'}^{\prime 2}
    &\leq
    \biggl(
      N_i(1)
      + \int_{0}^{1} S(f_0, t) \lambda_0(t) \diffi t
    \biggr)^2 \\
    \nonumber
    &\ \ \times
    \sup_{t \in [0, 1]}
    \sum_{r \in \mathcal R}
    \sum_{r' \in \mathcal R}
    a_r a_{r'}
    \biggl(
      \frac{D^2 S(f_0, t)(e_r, e_{r'})}{S(f_0, t)}
      - \frac{D S(f_0, t)(e_r)}{S(f_0, t)}
      \frac{D S(f_0, t)(e_{r'})}{S(f_0, t)}
    \biggr)^2 \\
    \label{eq:xi_irr_bound}
    &\leq
    16 H_{\infty,\gamma}^2.
  \end{align}
  Therefore by Cauchy--Schwarz,
  with $\mathcal{I} := \{(i,j):i,j \in [n],i<j\}$,
  \begin{align*}
    \sum_{r \in \mathcal R}
    \sum_{r' \in \mathcal R}
    a_r a_{r'}
    \biggl( \frac{1}{n} \sum_{i=1}^n \xi'_{i r r'} \biggr)^2 &\leq
    \frac{1}{n^2}
    \sum_{i=1}^n
    \sum_{r \in \mathcal R}
    \sum_{r' \in \mathcal R}
    a_r a_{r'}
    \xi_{i r r'}^{\prime 2}
    + \frac{2}{n^2}
    \sum_{(i, j) \in \mathcal I}
    \sum_{r \in \mathcal R}
    \sum_{r' \in \mathcal R}
    a_r a_{r'}
    \xi'_{i r r'}
    \xi'_{j r r'} \\
    &\leq
    \frac{16 H_{\infty,\gamma}^2}{n}
    +
    \biggl|
    \frac{2}{n (n-1)}
    \sum_{(i, j) \in \mathcal I}
    \sum_{r \in \mathcal R}
    \sum_{r' \in \mathcal R}
    a_r a_{r'}
    \xi'_{i r r'}
    \xi'_{j r r'}
    \biggr|.
  \end{align*}
  The last term above is a degenerate second-order $U$-statistic
  so by Cauchy--Schwarz, \eqref{eq:xi_irr_bound}
  and Lemma~\ref{lem:orliczUStat},
  there exists a universal constant $C_4 > 0$ such that for all
  $s \geq 1$,
  \begin{align*}
    \biggl|
    \frac{2}{n (n-1)}
    \sum_{(i, j) \in \mathcal I}
    \sum_{r \in \mathcal R}
    \sum_{r' \in \mathcal R}
    a_r a_{r'}
    \xi'_{i r r'}
    \xi'_{j r r'}
    \biggr|
    \leq
    \frac{C_4 s^2 H_{\infty,\gamma}^2}{n}
  \end{align*}
  with probability at least $1 - e^{-s^2}$.
  Thus, there is a universal constant $C_5 > 0$ such that
  \begin{align}
    \nonumber
    \sum_{r \in \mathcal R}
    \sum_{r' \in \mathcal R}
    a_r a_{r'}
    \biggl\{
      \int_{0}^{1}
      \biggl(
        \frac{D^2 S(f_0, t)(e_r, e_{r'})}{S(f_0, t)}
        &- \frac{D S(f_0, t)(e_r)}{S(f_0, t)}
        \frac{D S(f_0, t)(e_{r'})}{S(f_0, t)}
      \biggr) \\
      \nonumber
      &\hspace{2cm}\times
      \big(
        \mathrm d N(t) - S(f_0, t) \lambda_0(t)\,\mathrm d t
      \big)
    \biggr\}^2 \\
    \label{eq:D2_integrator_term_bound}
    &\leq
    \frac{C_5 s^2 H_{\infty,\gamma}^2}{n}
  \end{align}
  with probability at least $1 - e^{-s^2}$.
  Combining \eqref{eq:D2_operator_bound},
  \eqref{eq:D2_second_derivative_term},
  \eqref{eq:D2_first_derivative_term},
  \eqref{eq:D2_integrator_term},
  \eqref{eq:D2_second_derivative_term_bound},
  \eqref{eq:D2_first_derivative_term_bound}
  and \eqref{eq:D2_integrator_term_bound},
  it follows that there exists a universal constant $C > 0$ such that
  \begin{align*}
    \|D^2 \ell_{n}(f_0)
    - D^2 \ell_{\star}(f_0)\|_{\cH_\gamma,\op}^2
    \leq
    \frac{C^2 s^2 H_{\infty,\gamma}^2 e^{4\|f_0\|_\infty}}
    {n q_1^4}
  \end{align*}
  with probability at least $1 - e^{-s^2}$.
\end{proof}

\subsubsection{Bounding third-order terms from the basic inequality}

We bound the term in \eqref{eq:basic_inequality} involving $D^3 \ell_n$.
For $f, f_1, f_2 \in \cB(\cX)$, define
\begin{align*}
  V_n(f)(f_1, f_2)
  &\vcentcolon=
  \int_{0}^{1}
  \frac{D^2 S_n(f, t)(f_1, f_2)}{S_n(f, t)}
  \diffi {N(t)}, \\
  V(f)(f_1, f_2)
  &\vcentcolon=
  \int_{0}^{1}
  \frac{D^2 S(f, t)(f_1, f_2)}{S(f, t)}
  S(f_0, t)
  \lambda_0(t) \diffi t.
\end{align*}
In Lemma~\ref{lem:convergence_D3ln} we provide an upper bound
for $D^3 \ell_n$ in terms of $V_n$.
Lemma~\ref{lem:convergence_Vn} then establishes convergence of
$V_n$ to $V$, while
Lemma~\ref{lem:bound_V} presents an upper bound for $V$.

\begin{lemma}[Convergence of $D^3 \ell_n$]%
  \label{lem:convergence_D3ln}
  Let $f, f_1 \in \cB(\cX)$. Then
  \begin{align*}
    \big|D^3 \ell_{n}(f)(f_1)^{\otimes 3}\big|
    &\leq
    6 \|f_1\|_\infty e^{2\|f - f_0\|_\infty}
    V_n(f_0)(f_1, f_1).
  \end{align*}
\end{lemma}

\begin{proof}[Proof of Lemma~\ref{lem:convergence_D3ln}]
  By Lemma~\ref{lem:derivatives_Sn} and Cauchy--Schwarz,
  \begin{align*}
    |D^3 S_n(f, t)(f_1, f_1, f_1)|
    &\leq
    \frac{1}{n} \sum_{i=1}^n
    \bigl|
    R_i(t) f_1(X_i)^3 e^{f(X_i)}
    \bigr|
    \leq
    \|f_1\|_\infty
    D^2 S_n(f, t)(f_1, f_1), \\
    |D S_n(f, t)(f_1)|
    &\leq
    \frac{1}{n} \sum_{i=1}^n
    \bigl|
    R_i(t) f_1(X_i) e^{f(X_i)}
    \bigr|
    \leq
    \|f_1\|_\infty
    S_n(f, t), \\
    D S_n(f, t)(f_1)^2
    &=
    \biggl(
      \frac{1}{n} \sum_{i=1}^n
      R_i(t) f_1(X_i) e^{f(X_i)}
    \biggr)^2 \\
    &\leq
    \biggl(
      \frac{1}{n} \sum_{i=1}^n
      R_i(t) e^{f(X_i)}
    \biggr)
    \biggl(
      \frac{1}{n} \sum_{i=1}^n
      R_i(t) f_1(X_i)^2 e^{f(X_i)}
    \biggr) \\
    &= S_n(f, t) D^2 S_n(f, t)(f_1, f_1).
  \end{align*}
  Further,
  \begin{align*}
    S_n(f, t)
    &=
    \frac{1}{n} \sum_{i=1}^n
    R_i(t) e^{f_0(X_i)}
    e^{f(X_i) - f_0(X_i)}
    \geq
    e^{-\|f - f_0\|_\infty}
    S_n(f_0, t)
  \end{align*}
  and
  \begin{align*}
    D^2 S_n(f, t)(f_1, f_1)
    &=
    \frac{1}{n} \sum_{i=1}^n
    R_i(t) f_1(X_i)^2 e^{f_0(X_i)}
    e^{f(X_i) - f_0(X_i)}
    \leq
    e^{\|f - f_0\|_\infty}
    D^2 S_n(f_0, t)(f_1, f_1).
  \end{align*}
  Applying these bounds to \eqref{eq:third_derivative_ln}
  and using Lemma~\ref{lem:derivatives_ln}, we obtain
  \begin{align*}
    \big|D^3 \ell_{n}(f)(f_1, f_1, f_1)\big|
    &\leq
    6 \|f_1\|_\infty
    \int_{0}^{1}
    \frac{D^2 S_n(f, t)(f_1, f_1)}{S_n(f, t)}
    \diffi {N(t)} \\
    &\leq
    6 \|f_1\|_\infty e^{2\|f - f_0\|_\infty}
    \int_{0}^{1}
    \frac{D^2 S_n(f_0, t)(f_1, f_1)}{S_n(f_0, t)}
    \diffi {N(t)}. \qedhere
  \end{align*}
\end{proof}

\begin{lemma}[Convergence of $V_n$]%
  \label{lem:convergence_Vn}
  Consider $V_n(f_0)$ and $V(f_0)$ as bilinear functionals
  from $\cH \times \cH$ to $\R$.
  There is a universal constant $C > 0$ such that for all $s \geq 1$,
  with probability at least $1 - e^{-s^2}$,
  \begin{align*}
    \|V_n(f_0) - V(f_0)\|_{\cH_\gamma,\op}
    \leq
    \frac{C s e^{2 \|f_0\|_\infty}}{q_1^2}
    \sqrt{\frac{H_{\infty,\gamma}^2}{n}}.
  \end{align*}
\end{lemma}

\begin{proof}[Proof of Lemma~\ref{lem:convergence_Vn}]
  By~\eqref{eq:bilinear_operator_norm},
  \begin{align}
    \nonumber
    \|V_n(f_0) &- V(f_0)\|_{\cH_\gamma,\op}^2
    \leq
    \sum_{r \in \mathcal{R}}
    \sum_{r' \in \mathcal{R}}
    a_r a_{r'}
    \bigl( V_n(f_0)(e_r, e_{r'}) - V(f_0)(e_r, e_{r'}) \bigr)^2 \\
    \nonumber
    &=
    \sum_{r \in \mathcal{R}}
    \sum_{r' \in \mathcal{R}}
    a_r a_{r'}
    \biggl(
      \int_{0}^{1}
      \frac{D^2 S_n(f_0, t)(e_r, e_{r'})}{S_n(f_0, t)}
      \diffi {N(t)}
      - \int_{0}^{1}
      D^2 S(f_0, t)(e_r, e_{r'})
      \lambda_0(t) \diffi t
    \biggr)^2 \\
    \label{eq:part_D2l_1}
    &\leq
    2\sum_{r \in \mathcal{R}}
    \sum_{r' \in \mathcal{R}}
    a_r a_{r'}
    \biggl\{
      \int_{0}^{1}
      \biggl(
        \frac{D^2 S_n(f_0, t)(e_r, e_{r'})}{S_n(f_0, t)}
        - \frac{D^2 S(f_0, t)(e_r, e_{r'})}{S(f_0, t)}
      \biggr)
      \diffi N(t)
    \biggr\}^2 \\
    \label{eq:part_D2l_2}
    &\qquad +
    2\sum_{r \in \mathcal{R}}
    \sum_{r' \in \mathcal{R}}
    a_r a_{r'}
    \biggl\{
      \int_{0}^{1}
      \frac{D^2 S(f_0, t)(e_r, e_{r'})}{S(f_0, t)}
      \big( \mathrm d N(t) - S(f_0, t) \lambda_0(t)\,\mathrm d t \big)
    \biggr\}^2.
  \end{align}
  A probabilistic bound for~\eqref{eq:part_D2l_1} is given
  in~\eqref{eq:D2_second_derivative_term_bound},
  while \eqref{eq:part_D2l_2} is controlled by a slightly simpler
  argument than that used to control~\eqref{eq:D2_integrator_term}.
  We therefore recover the same probabilistic bound as that
  given in Proposition~\ref{prop:convergence_D2ln}.
\end{proof}

\begin{lemma}[Upper bound for $V$]%
  \label{lem:bound_V}
  For $f_1 \in L_2(P_X)$, we have
  \begin{align*}
    V(f_0)(f_1, f_1)
    &\leq
    e^{\|f_0\|_\infty}
    \Lambda
    \|f_1\|_{L_2}^2.
  \end{align*}
\end{lemma}

\begin{proof}[Proof of Lemma~\ref{lem:bound_V}]
  We have
  \begin{align*}
    V(f_0)(f_1, f_1)
    &=
    \int_{0}^{1}
    \int_\cX
    q(x, t)
    e^{f_0(x)}
    f_1(x)^2
    \diffi P_X(x)
    \lambda_0(t) \diffi t \\
    &\leq
    e^{\|f_0\|_\infty}
    \int_{0}^{1} \lambda_0(t) \diffi t
    \int_\cX f_1^2 \diffi P_X
    =
    e^{\|f_0\|_\infty}
    \Lambda
    \|f_1\|_{L_2}^2.
    \qedhere
  \end{align*}
\end{proof}

\subsubsection{Convergence of \texorpdfstring{$P_n$}{Pn}}

In Lemma~\ref{lem:convergence_Pn} we bound the difference
between $P_n(f_0)$ and $P_X(f_0)$,
while in Lemma~\ref{lem:uniform_convergence_Pn}
we give a bound for the supremum
of the empirical process $P_n - P_X$ over $\cH_\gamma$-balls in $\cH$.

\begin{lemma}[Convergence of $P_n$]%
  \label{lem:convergence_Pn}
  For all $s \geq 1$, we have
  $|P_n(f_0)| \leq 2 s \|f_0\|_\infty / \sqrt n$
  with probability at least $1 - e^{-s^2}$.
\end{lemma}

\begin{proof}[Proof of Lemma~\ref{lem:convergence_Pn}]
  This follows by Hoeffding's inequality, noting that $P_X(f_0) = 0$.
\end{proof}

\begin{lemma}[Uniform convergence of $P_n$ on $\cH_\gamma$-balls]%
  \label{lem:uniform_convergence_Pn}
  Let $M > 0$. Then for all $s \geq 1$,
  \begin{align*}
    \P \biggl(
      \sup_{f \in \cH: \|f\|_{\cH_\gamma} \leq M}
      \bigl|
      P_n(f) - P_X(f)
      \bigr|
      > 4 M s \sqrt{\frac{H_{\infty,\gamma}}{n}}
    \biggr)
    &\leq e^{-s^2}.
  \end{align*}
\end{lemma}

\begin{proof}[Proof of Lemma~\ref{lem:uniform_convergence_Pn}]
  Recall that $f \in \cH$ with $\|f\|_{\cH_\gamma} \leq M$ if and only if
  $f(\cdot) = \sum_{r \in \mathcal R} \alpha_r e_r(\cdot)$
  with $\sum_{r \in \mathcal R} \alpha_r^2 / a_r \leq M^2$,
  and define the set
  \begin{align*}
    \cH_{\gamma,M,\Q} \vcentcolon=
    \biggl\{
      f: \cX \to \R:
      f(\cdot) = \sum_{r \in \mathcal R} \tilde\alpha_r e_r(\cdot),
      (\tilde\alpha_r)_{r \in \mathcal R} \in \Q^{\mathcal R},
      \sum_{r \in \mathcal R} \frac{\tilde\alpha_r^2}{a_r} \leq M^2
    \biggr\}.
  \end{align*}
  Suppose that
  $f(\cdot) = \sum_{r \in \mathcal R} \alpha_r e_r(\cdot) \in \cH_{\gamma,M}$
  and let $N \in \N$. For each $r \in \mathcal R$, let
  $\tilde \alpha_{r,N} \in \Q$ satisfy
  $|\tilde \alpha_{r,N}| \leq \alpha_r$ and
  $|\tilde \alpha_{r,N} - \alpha_r| \leq \nu_r/N$,
  and define
  $\tilde f_N(\cdot) \vcentcolon= \sum_{r \in \mathcal R}
  \tilde \alpha_{r,N} e_r(\cdot) \in \cH_{\gamma,M,\Q}$.
  Then for each $x \in \cX$, by Cauchy--Schwarz and Mercer's theorem,
  \begin{align*}
    \bigl|
    \tilde f_N(x) - f(x)
    \bigr|
    &\leq
    \sum_{r \in \mathcal R}
    \bigl| \tilde \alpha_{r,N} - \alpha_r \bigr|
    |e_r(x)|
    \leq
    \frac{1}{N}
    \sum_{r \in \mathcal R}
    \nu_r
    |e_r(x)| \\
    &\leq
    \frac{1}{N}
    \biggl(
      \sum_{r \in \mathcal R}
      \nu_r
    \biggr)^{1/2}
    \biggl(
      \sum_{r \in \mathcal R}
      \nu_r
      e_r(x)^2
    \biggr)^{1/2}
    \leq
    \frac{\sqrt K}{N}
    \biggl(
      \sum_{r \in \mathcal R}
      \nu_r
    \biggr)^{1/2} \rightarrow 0
  \end{align*}
  as $N \rightarrow \infty$. Thus $\cH_{\gamma,M,\Q}$ is countable
  and dense in $\cH_{\gamma,M}$ under
  the topology of pointwise convergence; it follows that
  $\cH_{\gamma,M}$ is pointwise measurable in the
  terminology of \citet[Definition~2.3.3]{van1996weak}.
  Let $\varepsilon_1, \ldots, \varepsilon_n$ be independent Rademacher
  variables (taking values $\pm 1$ with equal probability) that are
  independent of the data.
  By a symmetrisation inequality \citep[Lemma~2.3.1]{van1996weak},
  Jensen's inequality, Cauchy--Schwarz, since
  $\varepsilon_1, \ldots, \varepsilon_n$ are conditionally independent
  given the data, and as $\varepsilon_i^2 = 1$ and
  $\E(\varepsilon_i) = 0$ for each $i \in [n]$,
  \begin{align*}
    \E \biggl(
      \sup_{f \in \cH_{\gamma,M}}
      \bigl|
      P_n(f) &- P_X(f)
      \bigr|
    \biggr)
    \leq
    2 \, \E \biggl(
      \sup_{f \in \cH_{\gamma,M}}
      \biggl|
      \frac{1}{n}
      \sum_{i=1}^n
      \varepsilon_i f(X_i)
      \biggr|
    \biggr) \\
    &=
    2 \, \E \biggl(
      \sup_{(\alpha_r)_{r \in \mathcal R}: \,
      \sum_{r \in \mathcal R} \alpha_r^2 / a_r \leq M^2}
      \frac{1}{n}
      \sum_{i=1}^n
      \varepsilon_i
      \sum_{r \in \mathcal R} \alpha_r e_r(X_i)
    \biggr) \\
    &\leq
    2 \biggl[ \E \biggl\{
        \sup_{(\alpha_r)_{r \in \mathcal R}: \,
        \sum_{r \in \mathcal R} \alpha_r^2 / a_r \leq M^2}
        \biggl(
          \sum_{r \in \mathcal R}
          \alpha_r
          \frac{1}{n}
          \sum_{i=1}^n
          \varepsilon_i
          e_r(X_i)
        \biggr)^2
      \biggr\}
    \biggr]^{1/2} \\
    &\leq
    2 \biggl[ \E \biggl\{
        \sup_{(\alpha_r)_{r \in \mathcal R}: \,
        \sum_{r \in \mathcal R} \alpha_r^2 / a_r \leq M^2}
        \sum_{r \in \mathcal R}
        a_r
        \biggl(
          \frac{1}{n}
          \sum_{i=1}^n
          \varepsilon_i
          e_r(X_i)
        \biggr)^2
        \sum_{r' \in \mathcal R}
        \frac{\alpha_{r'}^2}{a_{r'}}
      \biggr\}
    \biggr]^{1/2} \\
    &\leq
    2 M \biggl\{
      \sum_{r \in \mathcal R}
      \frac{a_r}{n^2}
      \sum_{i=1}^n
      \sum_{j=1}^n
      \E \bigl(
        \varepsilon_i
        \varepsilon_j
        e_r(X_i)
        e_r(X_j)
      \bigr)
    \biggr\}^{1/2} \\
    &=
    2 M \biggl\{
      \frac{1}{n^2}
      \sum_{i=1}^n
      \E \biggl(
        \sum_{r \in \mathcal R}
        a_r
        e_r(X_i)^2
      \biggr)
    \biggr\}^{1/2}
    \leq
    2 M \sqrt{\frac{H_{2,\gamma}}{n}}
    \leq
    2 M \sqrt{\frac{H_{\infty,\gamma}}{n}}.
  \end{align*}
  For each $f \in \cH_{\gamma,M}$,
  Lemma~\ref{lem:infty_bounds} gives
  $\|f\|_\infty \leq \sqrt {H_{\infty,\gamma}} \|f\|_{\cH_\gamma}
  \leq M \sqrt{H_{\infty,\gamma}}$,
  so by the bounded differences inequality
  \citep[][Theorem~6.2]{boucheron2013concentration},
  for $s \geq 1$,
  \begin{align*}
    \P \biggl(
      \sup_{f \in \cH_{\gamma,M}}
      \bigl|
      P_n(f) - P_X(f)
      \bigr|
      > 4 M s \sqrt{\frac{H_{\infty,\gamma}}{n}}
    \biggr)
    &\leq e^{-s^2}.
    \qedhere
  \end{align*}
\end{proof}

\subsubsection{Proof of Theorem~\ref{thm:rate}}

We prove our convergence result
for $\hat f_{n,\gamma}$, bounding each term in the
basic inequality \eqref{eq:basic_inequality}.

\begin{proof}[Proof of Theorem~\ref{thm:rate}]
  We write $C^0_1, C^0_2, \ldots$ for positive quantities
  depending only on
  $\|f_0\|_\cH$, $\Lambda$, $q_1$,
  $\|1_\cX\|_\cH$ and $K$,
  noting that this allows us to absorb increasing functions of
  $\|f_0\|_\infty$ by Lemma~\ref{lem:infty_bounds}.
  We begin by deriving a lower bound for $H_{\infty,\gamma}$.
  If $\gamma \geq \nu_1$, then
  \begin{align*}
    H_{\infty,\gamma}
    &=
    \sup_{x \in \cX} \sum_{r \in \mathcal R}
    \frac{e_r(x)^2}{1 + \gamma / \nu_r}
    \geq \frac{1}{2\gamma} \sup_{x \in \cX}
    \sum_{r \in \mathcal R} \nu_r e_r(x)^2
    = \frac{1}{2\gamma} \sup_{x \in \cX} k(x, x)
    = \frac{K}{2 \gamma},
  \end{align*}
  while if $\gamma < \nu_1$, then
  \begin{align*}
    H_{\infty,\gamma}
    > \frac{1}{2} \sup_{x \in \cX} e_1(x)^2
    \geq
    \frac{\|e_1\|_{L_2}^2}{2}
    = \frac{1}{2}.
  \end{align*}
  Therefore, $2 H_{\infty,\gamma} \geq 1 \land (K / \gamma)$, so
  $1 \leq 2 H_{\infty,\gamma} + 2 \gamma H_{\infty,\gamma} / K$.
  Taking $(C_0)^2 \geq 4/K$, we have for
  $\gamma H_{\infty,\gamma} \leq 1 / (C_0)^2$ that
  $1 \leq 2 H_{\infty,\gamma} + 2 / \{(C_0)^2 K\}
  \leq 2 H_{\infty,\gamma} + 1/2$.
  Thus, $H_{\infty,\gamma} \geq 1/4$.

  We now use the previous results to control each term
  in~\eqref{eq:basic_inequality}.
  By Proposition~\ref{prop:convergence_Dln},
  there exists $C^0_1 > 0$ such that for all $s \geq 1$,
  with probability at least $1 - e^{-s^2}$,
  \begin{align*}
    \bigl|D \ell_n(f_0)(\tilde f_{n,\gamma} - f_0)\bigr|
    \leq
    C_1^0
    \biggl(
      \frac{\sqrt{H_{2,\gamma}} + s}{\sqrt n}
      + s^2 \frac{\sqrt{H_{\infty,\gamma}}}{n}
    \biggr)
    \bigl\| \tilde f_{n,\gamma} - f_0 \bigr\|_{\cH_\gamma}.
  \end{align*}
  By Cauchy--Schwarz and
  Lemma~\ref{lem:convergence_Pn},
  for each $s \geq 1$, with probability at least $1 - e^{-s^2}$,
  \begin{align*}
    &\bigl|
    \gamma
    \bigl\langle \tilde f_{n,\gamma} - f_0 + P_n(f_0) 1_\cX,
    f_0 - P_n(f_0) 1_\cX \bigr\rangle_\cH
    \bigr| \\
    &\quad\leq
    \gamma \|f_0 - P_n(f_0) 1_\cX\|_\cH
    \bigl\| \tilde f_{n,\gamma} - f_0 + P_n(f_0) 1_\cX \bigr\|_\cH \\
    &\quad\leq
    \gamma
    \bigl(
      \|f_0\|_\cH + \|f_0\|_\infty \|1_\cX\|_\cH
    \bigr)
    \bigl(
      \bigl\| \tilde f_{n,\gamma} - f_0 \bigr\|_\cH
      + |P_n(f_0)| \| 1_\cX \|_\cH
    \bigr) \\
    &\quad\leq
    \gamma
    \bigl(
      \|f_0\|_\cH + \|f_0\|_\infty \|1_\cX\|_\cH
    \bigr)
    \biggl\{
      \bigl\| \tilde f_{n,\gamma} - f_0 \bigr\|_\cH
      + \frac{2 s \|f_0\|_\infty \| 1_\cX \|_\cH }{\sqrt n}
    \biggr\}.
  \end{align*}
  Since $P_X(f_0) = 0$, by Lemma~\ref{lem:strong_convexity} and
  Proposition~\ref{prop:convergence_D2ln},
  there exists $C^0_2 > 0$ such that for each $s \geq 1$,
  with probability at least $1 - e^{-s^2}$,
  \begin{align*}
    D^2 \ell_{n}(f_0)
    \bigl( \tilde f_{n,\gamma} - f_0 \bigr)^{\otimes 2}
    &\geq
    D^2 \ell_{\star}(f_0)
    \bigl( \tilde f_{n,\gamma} - f_0 \bigr)^{\otimes 2}
    - \bigl\|D^2 \ell_{n}(f_0) - D^2 \ell_{\star}(f_0)\bigr\|_{\cH_\gamma,\op}
    \, \bigl\| \tilde f_{n,\gamma} - f_0 \bigr\|_{\cH_\gamma}^2 \\
    &\geq
    \frac{1}{C^0_2}
    \bigl\|
    \tilde f_{n,\gamma} - P_X\bigl(\tilde f_{n,\gamma}\bigr) 1_\cX - f_0
    \bigr\|_{L_2}^2
    - C^0_2 s
    \sqrt{\frac{H_{\infty,\gamma}^2}{n}}
    \bigl\| \tilde f_{n,\gamma} - f_0 \bigr\|_{\cH_\gamma}^2.
  \end{align*}
  Note that $P_n(\tilde f_{n,\gamma}) = 0$ and $P_X(f_0) = 0$,
  and recall that
  $\bigl\|\tilde f_{n,\gamma} - f_0 + P_n(f_0) 1_\cX\bigr\|_{\cH_\gamma}
  \leq M$, where $M^2 \vcentcolon= 1 / (b^2 H_{\infty,\gamma})
  \leq 1 / H_{\infty,\gamma}$.
  Therefore, with
  $\cH_{\gamma,M} \vcentcolon= \{f \in \cH: \|f\|_{\cH_\gamma} \leq M\}$,
  by Lemmas~\ref{lem:convergence_Pn}
  and~\ref{lem:uniform_convergence_Pn},
  there exists $C^0_3 > 0$
  such that for $s \geq 1$, with probability at least $1 - e^{-s^2}$,
  \begin{align*}
    \bigl|P_X\bigl(\tilde f_{n,\gamma}\bigr)\bigr|
    &=
    \bigl|
    (P_X - P_n)\bigl(\tilde f_{n,\gamma} - f_0 + P_n(f_0) 1_\cX\bigr)
    - P_n(f_0)
    \bigr| \\
    &\leq
    \bigl| P_n(f_0) \bigr|
    + \sup_{f \in \cH_{\gamma,M}}
    \bigl| P_n(f) - P_X(f) \bigr|
    \leq
    \frac{2 s \|f_0\|_\infty}{\sqrt n}
    + C M s \sqrt{\frac{H_{\infty,\gamma}}{n}}
    \leq
    \frac{C_3^0 s}{\sqrt n}.
  \end{align*}
  Thus for $s \geq 1$, with probability at least $1 - e^{-s^2}$,
  \begin{align*}
    \bigl\| \tilde f_{n,\gamma} - f_0 \bigr\|_{L_2}^2
    &=
    \bigl\|
    \tilde f_{n,\gamma} - P_X\bigl(\tilde f_{n,\gamma}\bigr) 1_\cX - f_0
    \bigr\|_{L_2}^2
    + P_X\bigl(\tilde f_{n,\gamma}\bigr)^2 \\
    &\qquad+
    2 P_X\bigl(\tilde f_{n,\gamma}\bigr)
    \bigl\langle
    \tilde f_{n,\gamma} - P_X\bigl(\tilde f_{n,\gamma}\bigr) 1_\cX - f_0,
    1_\cX
    \bigr\rangle_{L_2} \\
    &=
    \bigl\|
    \tilde f_{n,\gamma} - P_X\bigl(\tilde f_{n,\gamma}\bigr) 1_\cX - f_0
    \bigr\|_{L_2}^2
    + P_X\bigl(\tilde f_{n,\gamma}\bigr)^2 \\
    &\leq
    \bigl\|
    \tilde f_{n,\gamma} - P_X\bigl(\tilde f_{n,\gamma}\bigr) 1_\cX - f_0
    \bigr\|_{L_2}^2
    + \frac{(C_3^0)^2 s^2}{n}.
  \end{align*}
  Therefore, there is $C^0_4 > 0$ such that for
  $1 \leq s \leq \sqrt{n} / (C^0_4 H_{\infty,\gamma})$,
  with probability at least $1 - e^{-s^2}$,
  \begin{align*}
    D^2 \ell_{n}(f_0)
    \bigl( \tilde f_{n,\gamma} - f_0 \bigr)^{\otimes 2}
    &\geq
    \frac{1}{C^0_4}
    \bigl\| \tilde f_{n,\gamma} - f_0 \bigr\|_{L_2}^2
    - \frac{\gamma}{4}
    \bigl\| \tilde f_{n,\gamma} - f_0 \bigr\|_{\cH}^2
    - \frac{C^0_4 s^2}{n}.
  \end{align*}
  For $f_1, f_2 \in \cH$ we have
  $\|f_1 + f_2\|_\cH^2 \geq \|f_1\|_\cH^2/2 - \|f_2\|_\cH^2$
  by the triangle inequality, so
  by Lemma~\ref{lem:convergence_Pn},
  with probability at least $1 - e^{-s^2}$,
  \begin{align*}
    \gamma \bigl\| \tilde f_{n,\gamma} - f_0 + P_n(f_0) 1_\cX\bigr\|_\cH^2
    &\geq
    \frac{\gamma}{2}
    \bigl\| \tilde f_{n,\gamma} - f_0\bigr\|_\cH^2
    - \gamma P_n(f_0)^2 \|1_\cX\|_\cH^2 \\
    &\geq
    \frac{\gamma}{2}
    \bigl\| \tilde f_{n,\gamma} - f_0\bigr\|_\cH^2
    - \frac{4 \gamma s^2 \|f_0\|_\infty^2 \|1_\cX\|_\cH^2}{n}.
  \end{align*}
  By Lemmas~\ref{lem:convergence_D3ln},~\ref{lem:convergence_Vn}
  and~\ref{lem:bound_V}, there exists a universal constant $C > 0$
  such that for all $f, f_1 \in \cH$ and $s \geq 1$,
  with probability at least $1 - e^{-s^2}$,
  \begin{align}
    \label{eq:third-derivative}
    &\big|D^3 \ell_{n}(f)(f_1)^{\otimes 3}\big|
    \leq
    6 \|f_1\|_\infty e^{2\|f - f_0\|_\infty}
    \bigl\{
      V(f_0)(f_1, f_1)
      + \|V_n(f_0) - V(f_0)\|_{\cH_\gamma,\op}
      \|f_1\|_{\cH_\gamma}^2
    \bigr\} \\
    &\qquad\leq
    6 e^{\|f_0\|_\infty}
    e^{2\|f - f_0\|_\infty}
    \Lambda
    \|f_1\|_\infty
    \|f_1\|_{L_2}^2
    + \frac{C s
      e^{2 \|f_0\|_\infty}
      e^{2\|f - f_0\|_\infty}
    }{q_1^2}
    \sqrt{\frac{H_{\infty,\gamma}^2}{n}}
    \|f_1\|_\infty
    \|f_1\|_{\cH_\gamma}^2. \notag
  \end{align}
  As $\ell_n$ is invariant under addition of constants, we have
  \begin{align*}
    \bigl|
    D^3 \ell_n\bigl(\check f_{n,\gamma}\bigr)
    \bigl( \tilde f_{n,\gamma} - f_0 \bigr)^{\otimes 3}
    \bigr|
    &=
    \bigl|
    D^3 \ell_n\bigl(\check f_{n,\gamma} + P_n(f_0) 1_\cX\bigr)
    \bigl( \tilde f_{n,\gamma} - f_0 + P_n(f_0) 1_\cX \bigr)^{\otimes 3}
    \bigr|.
  \end{align*}
  Since $\check f_{n,\gamma}$ is on the segment between
  $f_0 - P_n(f_0) 1_\cX$ and $\tilde f_{n, \gamma}$, we have
  $\|\check f_{n,\gamma} - f_0 + P_n(f_0) 1_\cX\|_{\infty} \leq 1/b$.
  Thus, by~\eqref{eq:third-derivative},
  and Lemma~\ref{lem:convergence_Pn}, there exists $C^0_5 > 0$
  such that for $1 \leq s \leq \sqrt{n} / (C^0_5 H_{\infty,\gamma})$,
  with probability at least $1 - e^{-s^2}$,
  \begin{align*}
    \bigl|
    D^3 \ell_n\bigl(\check f_{n,\gamma}\bigr)
    \bigl( \tilde f_{n,\gamma} - f_0 \bigr)^{\otimes 3}
    \bigr|
    &\leq
    \frac{C^0_5}{b}
    \bigl\| \tilde f_{n,\gamma} - f_0 + P_n(f_0) 1_\cX \bigr\|_{L_2}^2
    + \frac{C^0_5}{b} \gamma
    \bigl\| \tilde f_{n,\gamma} - f_0 + P_n(f_0) 1_\cX \bigr\|_{\cH}^2 \\
    &\leq
    \frac{2 C^0_5}{b}
    \bigl\| \tilde f_{n,\gamma} - f_0 \bigr\|_{\cH_\gamma}^2
    + \frac{2 C^0_5}{b}
    \bigl( 1 + \|1_\cX\|_{\cH}^2 \gamma \bigr)
    P_n(f_0)^2 \\
    &\leq
    \frac{2 C^0_5}{b}
    \bigl\| \tilde f_{n,\gamma} - f_0 \bigr\|_{\cH_\gamma}^2
    + \frac{16 C^0_5 s^2 \|f_0\|_\infty^2 }{b n}
    \bigl( 1 + \|1_\cX\|_{\cH}^2 \gamma \bigr).
  \end{align*}
  Therefore, by substituting these bounds into
  \eqref{eq:basic_inequality}, there exists $C^0_6 > 0$
  such that for $\gamma \leq 1/C^0_6$ and
  $1 \leq s \leq \sqrt{n} / (C^0_6 H_{\infty,\gamma})$,
  with probability at least $1 - e^{-s^2}$,
  as $b \geq 1$,
  \begin{align*}
    &
    \frac{1}{2}
    D^2 \ell_n(f_0) \bigl( \tilde f_{n,\gamma} - f_0 \bigr)^{\otimes 2}
    + \gamma \bigl\| \tilde f_{n,\gamma} - f_0 + P_n(f_0) 1_\cX\bigr\|_\cH^2 \\
    &\quad\leq
    \bigl|D \ell_n(f_0) \bigl( \tilde f_{n,\gamma} - f_0 \bigr)\bigr|
    + \frac{1}{6} \Bigl| D^3 \ell_n\bigl(\check f_{n,\gamma}\bigr)
    \bigl( \tilde f_{n,\gamma} - f_0 \bigr)^{\otimes 3} \Bigr| \\
    &\qquad+
    2 \gamma \bigl|
    \bigl\langle \tilde f_{n,\gamma} - f_0 + P_n(f_0) 1_\cX,
    f_0 - P_n(f_0) 1_\cX \bigr\rangle_\cH \bigr|, \\
    &
    \frac{1}{C^0_6}
    \bigl\| \tilde f_{n,\gamma} - f_0 \bigr\|_{L_2}^2
    + \frac{3\gamma}{8}
    \bigl\| \tilde f_{n,\gamma} - f_0\bigr\|_\cH^2
    \\
    &\quad\leq
    C_6^0
    \biggl(
      \frac{\sqrt{H_{2,\gamma}} + s}{\sqrt n}
      + s^2 \frac{\sqrt{H_{\infty,\gamma}}}{n}
    \biggr)
    \bigl\| \tilde f_{n,\gamma} - f_0 \bigr\|_{\cH_\gamma} \\
    &\qquad+
    \frac{C^0_6}{b}
    \bigl\| \tilde f_{n,\gamma} - f_0 \bigr\|_{\cH_\gamma}^2
    + C^0_6 \gamma
    \bigl\| \tilde f_{n,\gamma} - f_0 \bigr\|_\cH
    + C_6^0 \biggl( \frac{s}{\sqrt n} + \gamma \biggr)^2.
  \end{align*}
  By taking $C_6^0 \geq 8/3$
  and setting $b \coloneqq 2 (C_6^0)^2$,
  with probability at least $1 - e^{-s^2}$,
  \begin{align*}
    \frac{1}{2 C^0_6}
    \bigl\| \tilde f_{n,\gamma} - f_0 \bigr\|_{\cH_\gamma}^2
    &\leq
    C_6^0
    \biggl(
      \frac{\sqrt{H_{2,\gamma}} + s}{\sqrt n}
      + s^2 \frac{\sqrt{H_{\infty,\gamma}}}{n}
      + \sqrt\gamma
    \biggr)
    \bigl\| \tilde f_{n,\gamma} - f_0 \bigr\|_{\cH_\gamma}
    + C_6^0 \biggl( \frac{s}{\sqrt n} + \gamma \biggr)^2.
  \end{align*}
  Now if $x^2 \leq a x + c$ for some $x, a, c \geq 0$,
  then $x \leq a + \sqrt{c}$.
  It follows that there exists $C^0_7 > 0$
  such that for $\gamma \leq 1/C^0_7$ and
  $1 \leq s \leq \sqrt{n} / (C^0_7 H_{\infty,\gamma})$,
  with probability at least $1 - e^{-s^2}$,
  \begin{align}
    \label{eq:rate_main_bound}
    \bigl\|
    \tilde f_{n,\gamma} - f_0
    \bigr\|_{\cH_\gamma}
    &\leq
    C^0_7
    \biggl(
      \frac{\sqrt{H_{2,\gamma}} + s}{\sqrt n}
      + s^2 \frac{\sqrt{H_{\infty,\gamma}}}{n}
      + \sqrt{\gamma}
    \biggr).
  \end{align}
  Next, we show that $\hat f_{n,\gamma} = \tilde f_{n,\gamma}$
  with high probability.
  By definition of $\tilde f_{n,\gamma}$, we have
  $\hat f_{n,\gamma} = \tilde f_{n,\gamma}$ whenever
  $b \sqrt{H_{\infty,\gamma}} \,
  \bigl\| \tilde f_{n,\gamma} - f_0 + P_n(f_0) 1_\cX
  \bigr\|_{\cH_\gamma} < 1$.
  By definition of $b$, this holds if both
  $4 (C^0_6)^2 \sqrt{H_{\infty,\gamma}} \,
  \bigl\| \tilde f_{n,\gamma} - f_0 \bigr\|_{\cH_\gamma} < 1$
  and $4 (C^0_6)^2 \sqrt{H_{\infty,\gamma}} |P_n(f_0)| \| 1_\cX
  \|_{\cH_\gamma} < 1$.
  The first of these holds
  by \eqref{eq:rate_main_bound}
  with probability at least $1 - e^{-s^2}$ if
  \begin{align}
    \label{eq:second_final_display}
    4 (C^0_6)^2 C^0_7
    \biggl(
      \sqrt{\frac{H_{2,\gamma} H_{\infty,\gamma}}{n}}
      + s \sqrt{\frac{H_{\infty,\gamma}}{n}}
      + s^2 \frac{H_{\infty,\gamma}}{n}
      + \sqrt{\gamma H_{\infty,\gamma}}
    \biggr) < 1,
  \end{align}
  while, by Lemma~\ref{lem:convergence_Pn},
  the second holds with probability at least $1 - e^{-s^2}$ if
  \begin{align}
    \label{eq:final_display}
    8 (C^0_6)^2 s \|f_0\|_\infty \|1_\cX\|_\cH
    \sqrt{\frac{H_{\infty,\gamma}}{n}} < 1.
  \end{align}
  For sufficiently large $C^0_8$,
  \eqref{eq:second_final_display}
  and~\eqref{eq:final_display} both hold simultaneously
  when $\gamma H_{\infty,\gamma} \leq 1/C^0_8$
  and $1 \leq s \leq \sqrt n /
  (C^0_8 H_{\infty,\gamma})$.

  Finally, we remove the $H_{\infty,\gamma}$
  term from the right-hand side of
  \eqref{eq:rate_main_bound}.
  Under these conditions above,
  taking $C^0_8 \geq 2$, we have that
  $s^2 \sqrt{H_{\infty,\gamma}} / n
  \leq s / \bigl(2 \sqrt {n H_{\infty,\gamma}}\bigr)
  \leq s / \sqrt {n}$.
  Thus, with probability at least $1 - e^{-s^2}$,
  \begin{align*}
    \bigl\|
    \hat f_{n,\gamma} - f_0
    \bigr\|_{\cH_\gamma}
    &\leq
    C^0_9
    \biggl(
      \frac{\sqrt{H_{2,\gamma}} + s}{\sqrt n}
      + \sqrt{\gamma}
    \biggr),
  \end{align*}
  whenever $\gamma H_{\infty,\gamma} \leq 1/C^0_9$
  and $1 \leq s \leq \sqrt n /
  (C^0_9 H_{\infty,\gamma})$, as required.
\end{proof}

\subsubsection{Details for Examples~\ref{ex:polynomial},
\ref{ex:sobolev} and \ref{ex:gaussian}}

We prove our results on $H_{\infty,\gamma}$
and $\|1_\cX\|_\cH$ for polynomial, Sobolev and Gaussian kernels
respectively.

\begin{proof}[Details for Example~\ref{ex:polynomial}]
  Let $\{q_r(\cdot)\}_{r \in [R]}$ be an orthonormal
  basis under $L_2(P_X)$ for the space $\mathrm{poly}_p(\mathcal{X})$
  of $d$-variate
  polynomials on $\cX$ of degree at most $p$, where $R \leq \binom{d + p}{p}$.
  For $\alpha = (\alpha_1,\ldots,\alpha_d)^\T \in \N_0^d$ with
  $|\alpha| := \sum_{j=1}^d \alpha_j \leq p$, and for $r \in [R]$,
  let $b_{\alpha,r} \in \R$ be such that
  $\prod_{j=1}^d x_j^{\alpha_j} = \sum_{r=1}^R b_{\alpha,r} q_r(x)$
  for all $x \in \cX$.
  Let $\bB = (B_{r r'})_{r,r' \in [R]} \in \R^{R \times R}$
  be the matrix with entries
  \begin{align*}
    B_{r r'}
    \vcentcolon=
    \sum_{s=0}^p \binom{p}{s} a^{p-s}
    \sum_{\alpha:|\alpha| = s}
    \binom{s}{\alpha_1,\ldots,\alpha_d}
    b_{\alpha,r}
    b_{\alpha,r'},
  \end{align*}
  for $r,r' \in [R]$. Then $\bB$ is a non-negative linear
  combination of positive semi-definite rank-one matrices, so is
  positive semi-definite. It therefore has
  an eigendecomposition $\bB = \bU^\T \bD \bU$
  where $\bU \in \R^{R \times R}$ is orthogonal and
  $\bD \in \R^{R \times R}$ is diagonal
  with non-negative diagonal entries given by $(\nu_r)_{r \in [R]}$.
  Define $e : \cX \to \R^d$ by $e(x) \vcentcolon= \bU q(x)$,
  where $q(x) \vcentcolon= \bigl(q_1(x), \ldots, q_R(x)\bigr)^\T$.
  By the binomial and multinomial theorems,
  \begin{align*}
    k^\poly_{p,a}(x, y)
    &=
    \sum_{s=0}^p \binom{p}{s} a^{p-s}
    \sum_{\alpha:|\alpha| = s}
    \binom{s}{\alpha_1,\ldots,\alpha_d}
    \biggl(\prod_{j=1}^d x_j^{\alpha_j}\biggr)
    \biggl(\prod_{j=1}^d y_j^{\alpha_j}\biggr) \\
    &=
    \sum_{s=0}^p \binom{p}{s} a^{p-s}
    \sum_{\alpha:|\alpha| = s}
    \binom{s}{\alpha_1,\ldots,\alpha_d}
    \sum_{r=1}^R b_{\alpha,r} q_r(x)
    \sum_{r'=1}^R b_{\alpha,r'} q_{r'}(y) \\
    &=
    \sum_{r=1}^R
    \sum_{r'=1}^R
    q_r(x)
    B_{r r'}
    q_{r'}(y)
    =
    q(x)^\T \bB q(y)
    =
    e(x)^\T \bD e(y)
    =
    \sum_{r = 1}^R
    \nu_r
    e_r(x) e_r(y).
  \end{align*}
  Further, $\{e_r(\cdot)\}_{r \in [R]}$ forms an $L_2(P_X)$-orthonormal
  basis of $\mathrm{poly}_p(\mathcal{X})$ as
  \begin{align*}
    \int_\cX e(x) e(x)^\T \diffi P_X(x)
    &= \int_\cX \bU q(x) q(x)^\T \bU^\T \diffi P_X(x)
    = {\bI}_R.
  \end{align*}
  Note that
  \[
    \max_{r \in [R]} \sup_{x \in \cX} |e_r(x)|
    \leq \sup_{x \in \cX} \|e(x)\|_2
    = \sup_{x \in \cX} \|\bU q(x)\|_2
    = \sup_{x \in \cX} \|q(x)\|_2,
  \]
  which depends only on $p$, $d$ and $P_X$,
  and is finite because each $q_r(\cdot)$ is a polynomial
  and $\cX$ is bounded.
  Define $\cR := \{r \in [R] : \nu_r > 0\}$, and endow it with
  the ordering induced by the reverse order of $(\nu_r)_{r \in
  \cR}$, so that the representation $k^\poly_{p,a}(x,y) = \sum_{r \in \cR}
  \nu_r e_r(x) e_r(y)$ of the kernel $k^\poly_{p,a}$ aligns with that in
  Mercer's theorem \citep[Theorem~A.8.8]{samworth2026modern}.
  Since $a_r = 1/(1 + \gamma / \nu_r) \leq 1$ for $r \in \cR$,
  there exists $C(p, d, P_X) > 0$ such that
  \begin{align*}
    H_{\infty,\gamma}
    &\leq
    R \, \max_{r \in [R]} \sup_{x \in \cX} e_r(x)^2
    \leq C(p, d, P_X).
  \end{align*}

  Finally, we compute $\|1_\cX\|_\cH$. By the binomial theorem,
  \begin{align*}
    k^\poly_{p,a}(x, y)
    = \sum_{r=0}^p \binom{p}{r} (x^\T y)^r a^{p-r}.
  \end{align*}
  Since $(x, y) \mapsto x^\T y$ is a kernel,
  and since kernels are preserved by non-negative scaling,
  sums and products
  \citep[Proposition~6.4\emph{(a)} and
  \emph{(c)}]{samworth2026modern}, we conclude that
  $k^\poly_{p, a}(\cdot,\cdot) - a^p$ is a kernel on $\cX$.
  By considering $x = y = 0_d$, we see that
  $k^\poly_{p,a}(\cdot, \cdot) - a^p - \delta$ is not a kernel
  if $\delta > 0$. We conclude by Lemma~\ref{lem:paulsen} that
  $1_\cX \in \cH$ if and only if
  $a > 0$, in which case $\|1_\cX\|_\cH^2 = 1/a^p$.
\end{proof}

\begin{proof}[Details for Example~\ref{ex:sobolev}]%
  We argue similarly to \citet[Example~6.22]{samworth2026modern}.
  Suppose that $e_\nu : [0, 1] \to \R$ is a unit $L_2(P_X)$-norm eigenfunction
  of $\cT_{k^\Sob_{1,a}}$ with eigenvalue $\nu > 0$. For $x \in [0, 1]$,
  \begin{align}
    \nonumber
    \nu e_\nu(x) &= \cT_{k^\Sob_{1,a}}(e_\nu)(x)
    =
    \int_\cX k^\Sob_{1,a}(x, y) e_\nu(y) \diffi P_X(y)
    = \int_0^1 \bigl(a + (x \land y)\bigr) e_\nu(y) \diffi y \\
    \label{eq:sobolev_eig}
    &=
    a \int_0^1 e_\nu(y) \diffi y
    + \int_0^x y e_\nu(y) \diffi y
    + x \int_x^1 e_\nu(y) \diffi y.
  \end{align}
  Differentiating twice with respect to $x$, we obtain that
  for all $x \in [0, 1]$,
  \begin{align}
    \label{eq:sobolev_first_deriv}
    \int_x^1 e_\nu(y) \diffi y
    &= \nu e_\nu'(x), \\
    \label{eq:sobolev_second_deriv}
    - e_\nu(x)
    &= \nu e_\nu''(x).
  \end{align}
  By \eqref{eq:sobolev_second_deriv},
  there exist $A_\nu, B_\nu \in \R$ with
  $e_\nu(x) = A_\nu \sin(x/\sqrt \nu) + B_\nu \cos(x/\sqrt \nu)$
  for all $x \in [0, 1]$.
  Substituting into \eqref{eq:sobolev_first_deriv} and setting
  $x = 1$, we have
  \begin{align}
    \label{eq:sobolev_first_eq}
    A_\nu \cos(1/\sqrt \nu) = B_\nu \sin(1/\sqrt \nu).
  \end{align}
  Returning to \eqref{eq:sobolev_eig}, with $x = 0$ we obtain
  $a \bigl\{A_\nu \bigl(1 - \cos(1/\sqrt \nu)\bigr)
  + B_\nu \sin(1/\sqrt \nu)\bigr\} = \sqrt \nu B_\nu$.
  Combining this with \eqref{eq:sobolev_first_eq},
  we obtain
  \begin{align}
    \label{eq:sobolev_second_eq}
    a A_\nu = \sqrt \nu B_\nu.
  \end{align}
  Now
  \begin{align*}
    4 &= 4\int_0^1 e_\nu(x)^2 \diffi x \\
    &= A_\nu^2 \{2 - \sqrt \nu \sin(2/\sqrt \nu)\}
    + \sqrt \nu A_\nu B_\nu \{2 - 2 \cos(2/\sqrt \nu)\}
    + B_\nu^2 \{2 + \sqrt \nu \sin(2/\sqrt \nu)\}.
  \end{align*}
  Substituting \eqref{eq:sobolev_second_eq} into \eqref{eq:sobolev_first_eq},
  we see that $\tan(1/\sqrt \nu) = \sqrt \nu/a$,
  interpreted as $\tan(1/\sqrt{\nu}) = \infty$ if $a = 0$.
  Thus, by the half-angle identities,
  $\sin(2/\sqrt \nu) = 2 \sqrt \nu a / (a^2 + \nu)$
  and $\cos(2/\sqrt \nu) = (a^2 - \nu) / (a^2 + \nu)$. Thus,
  since $e_\nu$ is defined only up to sign, we may take
  \begin{align*}
    A_\nu = \sqrt{\frac{2 \nu}{a^2 + \nu + a \nu}} \quad\text{and}\quad
    B_\nu = \sqrt{\frac{2 a^2}{a^2 + \nu + a \nu}}.
  \end{align*}
  Further, for every $r \in \N$, there exists
  $\nu_r \in \bigl[(r - 1/2)^{-2} \pi^{-2}, (r - 1)^{-2} \pi^{-2}\bigr)$
  such that $\tan(1/\sqrt{\nu_r}) = \sqrt{\nu_r}/a$,
  and these are the only positive values of $\nu$ to solve this equation.
  Hence, the eigenvalue--eigenfunction pairs for $\cT_{k^\Sob_{1,a}}$
  are, for $r \in \N$,
  \begin{align*}
    \nu_r
    &\in
    \biggl[
      \frac{1}{(r - 1/2)^2 \pi^2},
      \frac{1}{(r - 1)^2 \pi^2}
    \biggr), \\
    e_r(\cdot)
    &=
    \sqrt{\frac{2 \nu_r}{a^2 + \nu_r + a \nu_r}}
    \sin\biggl(\frac{\cdot}{\sqrt{\nu_r}}\biggr)
    + \sqrt{\frac{2 a^2}{a^2 + \nu_r + a \nu_r}}
    \cos\biggl(\frac{\cdot}{\sqrt{\nu_r}}\biggr).
  \end{align*}
  By Cauchy--Schwarz,
  $\|e_r\|_\infty^2 \leq 2 (\nu_r + a^2) / (a^2 + \nu_r + a \nu_r) \leq 2$.
  For $r \geq 2$, we have
  $a_r = 1/(1 + \gamma / \nu_r) \leq 1 / \{1 + \gamma (r-1)^2 \pi^2\}$,
  while $a_1 \leq 1$. Therefore,
  \begin{align*}
    H_{\infty,\gamma}
    &\leq
    \sum_{r=1}^\infty
    \frac{2}{1 + \gamma / \nu_r}
    \leq
    2 + \sum_{r=1}^\infty
    \frac{2}{1 + \gamma r^2 \pi^2 }
    \leq
    2 + \int_0^\infty
    \frac{2}{1 + \gamma x^2 \pi^2 }
    \diffi x
    =
    2 + \frac{1}{\sqrt{\gamma}}.
  \end{align*}

  Finally, we compute $\|1_\cX\|_\cH$.
  The functions $k^\Sob_{1,a}$ and $k^\Sob_{2,a}$ are kernels by
  \citet[Examples~6.8 and~6.9 and
  Proposition~6.4(a)]{samworth2026modern}.
  By considering $x = y = 0$, we see that
  $k^\Sob_{1,a} - a - \delta$ and
  $k^\Sob_{2,a} - a - \delta$ are not kernels
  if $\delta > 0$.
  The conclusions now follow by Lemma~\ref{lem:paulsen}.
\end{proof}

\begin{proof}[Details for Example~\ref{ex:gaussian}]
  Define $\tilde{k}^\Gauss_{\Sigma,a}:\mathcal{X} \times \mathcal{X}
  \rightarrow [0,\infty)$ by
  \begin{align*}
    \tilde k^\Gauss_{\Sigma,a}(x, y)
    &=
    a + \exp\bigl(-(\tilde x - \tilde y)^\T (\tilde x - \tilde y)\bigr),
  \end{align*}
  and set $\tilde x \vcentcolon= \Sigma^{1/2} x$
  and $\tilde y \vcentcolon= \Sigma^{1/2} y$.
  Since
  the constant map $(x, y) \mapsto a$ is a kernel and the standard Gaussian
  kernel is a kernel \citep[Example~6.7]{samworth2026modern}, and as
  addition preserves kernels
  \cite[Proposition~6.4]{samworth2026modern},
  we deduce that $k^\Gauss_{\Sigma,a}$ is a kernel on $\cX$.

  Now suppose that $a = 0$.
  Let $n \in \N$, let $\delta \in \bigl(0,1/(2n^2)\bigr)$ and,
  recalling that $\cX$ is unbounded, choose
  $x_1, \ldots, x_n \in \cX$ so that for each $i, j \in [n]$
  with $i \neq j$, we have
  $\|x_i - x_j\|_2^2 \geq \lambda_{\max}(\Sigma) \sqrt{\log(1/\delta)}$.
  Then
  $(x_i - x_j)^\T \Sigma^{-1} (x_i - x_j) \geq
  \lambda_{\min}(\Sigma^{-1})\|x_i - x_j\|_2^2 \geq \log(1/\delta)$,
  so $k^\Gauss_{\Sigma,0}(x_i, x_j) \leq \delta$.
  Note also that $k^\Gauss_{\Sigma,0}(x_i, x_i) = 1$ for each $i \in [n]$.
  Now define $\bK_n \vcentcolon= \bK(x_1, \ldots, x_n)$
  in the notation of Proposition~\ref{prop:constant_norm}, so that
  $\|\bI_n - \bK_n\|_{\mathrm{op}}
  \leq n \|\bI_n - \bK_n\|_{\infty}
  \leq n \delta < 1$.
  Hence $\bK_n$ is invertible, and since
  $\bI_n - \bK_n^{-1}
  = (\bI_n - \bK_n^{-1})(\bI_n - \bK_n) - (\bI_n - \bK_n)$, we have
  \begin{align*}
    \bigl\|
    \bI_n
    - \bK_n^{-1}
    \bigr\|_{\op}
    &\leq
    \frac{\bigl\|\bI_n - \bK_n\bigr\|_{\op}}
    {1 - \bigl\|\bI_n - \bK_n\bigr\|_{\op}}
    \leq
    \frac{n \delta}{1 - n \delta}
    \leq 2 n \delta.
  \end{align*}
  It follows that
  \begin{align*}
    \bigl|
    1_n^\T 1_n - 1_n^\T \bK_n^{-1} 1_n
    \bigr|
    \leq 2 n \delta \|1_n\|_2^2
    = 2 n^2 \delta < 1,
  \end{align*}
  so that
  \begin{align*}
    \kappa(\bK_n)
    &=
    \frac{1}{1_n^\T \bK_n^{-1} 1_n}
    \leq
    \frac{1}{1_n^\T 1_n - 1}
    =
    \frac{1}{n-1}
    \to 0
  \end{align*}
  as $n \to \infty$. We deduce from
  Proposition~\ref{prop:constant_norm} that $\kappa_{k^\Gauss_{\Sigma,0}} = 0$
  and so $1_\cX \not\in \cH$.
  Moreover, by Lemma~\ref{lem:paulsen},
  if $a > 0$ then $1_\cX \in \cH$ and $\|1_\cX\|_\cH^2 = 1/a$.
\end{proof}

\subsubsection{Proof of the minimax lower bound}

\begin{proof}[Proof of Theorem~\ref{thm:minimax}]
  Since $\lambda_0(t) = 1$ for all $t \in [0, 1]$
  and no censoring is present,
  the conditional hazard function
  \eqref{eq:hazard} at $f_0 \in \cH$ is
  \begin{equation*}
    h(t \mid x)
    = - \frac{\diffd}{\diffd t}
    \log \P (T > t \mid X=x)
    = e^{f_0(x)}.
  \end{equation*}
  It follows that $T \mid X \sim \Exp\bigl(e^{f_0(X)}\bigr)$
  and so the log-likelihood ratio between
  $f_0, g_0 \in \cH$ for a single observation $(X, T)$, is
  \[
    \hat{\Lambda}(f_0,g_0;X,T)
    = f_0(X) - g_0(X) + T\bigl(e^{g_0(X)} - e^{f_0(X)}\bigr).
  \]
  Writing $P_{f_0}^{\otimes n}$ for the joint distributions of $n$
  independent copies of $(X,T)$ when $X \sim P_X$
  and $T \mid X \sim \Exp\bigl(e^{f_0(X)}\bigr)$ and defining
  $P_{g_0}^{\otimes n}$ similarly, it follows that the
  Kullback--Leibler divergence from $P_{g_0}^{\otimes n}$ to
  $P_{f_0}^{\otimes n}$ satisfies
  \begin{align*}
    \KL\bigl(P_{f_0}^{\otimes n}
    \dbar P_{g_0}^{\otimes n}\bigr)
    &=
    \sum_{i=1}^n \E_{P_{f_0}^{\otimes n}}
    \bigl\{
      \hat{\Lambda}(f_0,g_0;X_i,T_i)
    \bigr\} \\
    &=
    n\,
    \E_{P_{f_0}} \bigl\{
      f_0(X) - g_0(X) + T \bigl(e^{g_0(X)} - e^{f_0(X)}\bigr)
    \bigr\} \\
    &=
    n\,
    \E_{P_X} \bigl\{
      e^{g_0(X) - f_0(X)} - 1
      + f_0(X) - g_0(X)
    \bigr\}.
  \end{align*}
  Since $f \in \cH_1$, we have
  $\|f\|_\infty^2 \leq K$ by Lemma~\ref{lem:infty_bounds}.
  Since $e^x - 1 - x \leq 2 x^2$
  for $x \in [-2, 2]$,
  we have for $f_0, g_0 \in \cH_1$ that
  \begin{align*}
    \KL\bigl(P_{f_0}^{\otimes n} \dbar P_{g_0}^{\otimes n}\bigr)
    \leq
    2 n \|f_0 - g_0\|_{L_2}^2.
  \end{align*}
  We now proceed by cases, depending on whether
  (i) $\cR = [R]$ for some $R \in \N$; or (ii) $\cR = \N$.
  In case~(i), consider the distributions induced by
  $f_0 \in \{0, \sqrt{\gamma} e_1\}$,
  where $\gamma \in (0, 1]$ is to be chosen later
  and $e_1 \in \mathcal{H}$ is an
  eigenfunction with $\|e_1\|_{L_2} = 1$
  associated with the largest
  eigenvalue $\nu_1$ of the kernel integral operator $T_k$.
  We have
  $\sqrt{\gamma} e_1 \in \cH_1$
  whenever $\gamma \leq 1/\|e_1\|_\cH^2 = \nu_1$.
  Applying Le Cam's two-point lemma
  \citep[e.g.][Lemma~8.4]{samworth2026modern}
  and Pinsker's inequality
  \citep[Exercise 8.2\textit{(b)}]{samworth2026modern}
  then yields
  \begin{align*}
    \inf_{\hat f_n \in \cF_n}
    \sup_{f_0 \in \cH_1}
    \E\bigl(
      \|\hat f_n - f_0\|_{L_2}^2
    \bigr)
    &\geq
    \frac{\gamma \|e_1\|_{L_2}^2}{2}
    \Bigl\{
      1 -
      \TV\bigl(P_0^{\otimes n}, P_{\sqrt{\gamma} e_1}^{\otimes n}\bigr)
    \Bigr\} \\
    &\geq
    \frac{\gamma}{2}
    \Biggl\{
      1 -
      \frac{\KL\bigl(P_0^{\otimes n} \dbar
      P_{\sqrt{\gamma} e_1}^{\otimes n}\bigr)^{1/2}}{2^{1/2}}
    \Biggr\}
    \geq
    \frac{\gamma}{2}\bigl( 1 - \sqrt{n \gamma}\bigr).
  \end{align*}
  Setting $\gamma = \gamma^\star \coloneqq
  \nu_1 \land \{1 / (4 n)\}$ gives,
  for all $n \in \N$,
  \begin{align*}
    \inf_{\hat f_n \in \cF_n}
    \sup_{f_0 \in \cH_1}
    \E\bigl( \|\hat f_n - f_0\|_{L_2}^2 \bigr)
    &\geq
    \frac{\gamma^\star}{4}
    =
    \frac{\nu_1}{4}
    \land
    \frac{1}{16 n}
    \geq
    \frac{1}{C_1 n}
  \end{align*}
  where $C_1 > 0$ depends only on $\nu_1$.
  As $\gamma \searrow 0$ we have $H_{2,\gamma} \nearrow R$, so
  $\inf_{\gamma > 0} \bigl\{ \gamma + H_{2,\gamma} / n \bigr\} \leq R/n$.
  Therefore, there exists $C_2 > 0$ depending only on
  $R$ and $\nu_1$ such that
  \begin{align*}
    \inf_{\hat f_n \in \cF_n}
    \sup_{f_0 \in \cH_1}
    \E\bigl( \|\hat f_n - f_0\|_{L_2}^2 \bigr)
    &\geq
    \frac{1}{C_2}
    \inf_{\gamma > 0}
    \biggl\{
      \gamma
      + \frac{H_{2,\gamma}}{n}
    \biggr\},
  \end{align*}
  concluding the analysis of the setting where $\cR = [R]$.

  We now consider the case where $\cR = \N$,
  following ideas of \citet{yang2017randomized}.
  For $\gamma \in (0, 1]$, define
  $\tilde \cH_\gamma \coloneqq
  \{f \in \cH: \|f\|_{\cH_\gamma}^2 \leq \gamma\}$,
  so that $\tilde \cH_\gamma \subseteq \tilde \cH_1 \subseteq \cH_1$.
  For $d \in \cR$ and $b > 0$ to be chosen later,
  let $\beta \coloneqq b Z$,
  where $Z$ is a standard Gaussian random vector in $\R^d$.
  Taking $\gamma a_d / b^2 \geq 2 d$,
  by a chi-squared upper tail bound \citep[Lemma~1]{laurent2000adaptive},
  we have
  \begin{align*}
    \P\biggl(
      \sum_{r=1}^d
      \frac{\beta_r^2}{a_r}
      > \gamma
    \biggr)
    \leq
    \P\biggl(
      \|Z\|_2^2
      > \frac{\gamma a_d}{b^2}
    \biggr)
    \leq
    \P\bigl(
      \|Z\|_2^2
      > 2 d
    \bigr)
    \leq
    e^{-d/16}.
  \end{align*}
  For $M \in \N$, let $\beta^{(1)}, \ldots, \beta^{(M)}$
  be independent copies of $\beta$ and define
  the random function
  $f_\beta \coloneqq \sum_{r=1}^d \beta_r e_r$.
  Since $\|f_\beta\|_{\cH_\gamma}^2 = \sum_{r=1}^d \beta_r^2 / a_r$,
  we have by a union bound that
  \begin{align*}
    \P\biggl(
      \bigcap_{m=1}^M
      \bigl\{
        f_{\beta^{(m)}}
        \in \tilde\cH_\gamma
      \bigr\}
    \biggr)
    \geq
    1 - M e^{-d/16}.
  \end{align*}
  For distinct $m,m' \in [M]$ and independent standard Gaussian
  vectors $Z,Z'$ in $\R^d$, and assuming that $\gamma / (8 b^2) \leq
  d/2$, we have
  \begin{align*}
    \P\biggl(
      \bigl\|
      f_{\beta^{(m)}}
      - f_{\beta^{(m')}}
      \bigr\|_{L_2}^2
      \leq \frac{\gamma}{4}
    \biggr)
    &=
    \P\biggl(
      \sum_{r=1}^d
      \bigl(
        \beta^{(m)}_r
        - \beta^{(m')}_r
      \bigr)^2
      \leq \frac{\gamma}{4}
    \biggr)
    = \P\biggl(
      b^2
      \|Z - Z'\|_2^2
      \leq \frac{\gamma}{4}
    \biggr) \\
    &=
    \P\biggl(
      \|Z\|_2^2
      \leq \frac{\gamma}{8 b^2}
    \biggr)
    \leq
    \P\biggl(
      \|Z\|_2^2
      \leq \frac{d}{2}
    \biggr)
    \leq e^{-d/16},
  \end{align*}
  by a chi-squared lower tail bound
  \citep[Lemma~1]{laurent2000adaptive}.
  Thus, by another union bound,
  as $\binom{M}{2} + M \leq M^2$,
  the collection
  $\bigl\{f_{\beta^{(m)}}: m \in [M]\bigr\}$
  forms a $\sqrt{\gamma/4}$-packing of~$\tilde\cH_\gamma$
  under $\|\cdot\|_{L_2}$
  with probability at least $1 - M^2 e^{-d/16}$.
  Taking $\log M \leq d/64$ makes this probability positive,
  and so a packing $\{f_1, \ldots, f_{M^\star}\}$ of cardinality
  $M^\star \coloneqq \lfloor e^{d/64} \rfloor$
  exists whenever both $\gamma a_d / b^2 \geq 2 d$ and
  $\gamma / (8 b^2) \leq d/2$.
  We may satisfy these inequalities by choosing
  $b^2 \coloneqq \gamma / (4 d)$
  and when $a_d \geq 1/2$; that is, $\gamma \leq \nu_d$.
  Note that $f \in \tilde\cH_\gamma$
  implies $\|f\|_{L_2}^2 \leq \gamma$.
  Therefore, by Fano's lemma \citep[e.g.,][Corollary~8.12]{samworth2026modern},
  taking $d \geq 64 \log 2$,
  and as $\log(2-1/x) / \log x \leq \log(3/2) / \log 2 \leq 3/4$
  for $x \geq 2$,
  \begin{align*}
    \inf_{\hat f_n \in \cF_n}
    \sup_{f_0 \in \cH_1}
    \E\bigl(
      \|\hat f_n - f_0\|_{L_2}^2
    \bigr)
    &\geq
    \frac{\gamma}{4}
    \biggl(
      1 -
      \frac{\sum_{m=1}^{M^\star}
        \KL\bigl(
          P_0^{\otimes n}
          \dbar
          P_{f_m}^{\otimes n}
        \bigr)
        / {M^\star}
      + \log (2 - 1/{M^\star})}{\log {M^\star}}
    \biggr) \\
    &\geq
    \frac{\gamma}{4}
    \biggl(
      \frac{1}{4} -
      \frac{2 n \sum_{m=1}^{M^\star} \|f_m\|_{L_2}^2 / M^\star}{d/128}
    \biggr)
    =
    \frac{\gamma}{16}
    \biggl\{
      1 -
      \frac{1024 n \gamma}{d}
    \biggr\},
  \end{align*}
  where we used the fact that if
  $d \geq 64 \log 2$ then
  $M^\star \geq 2 \lor \lfloor e^{d/64} \rfloor$
  so $M^\star \geq e^{d/128}$.  Thus, if $n \gamma / d \leq 1/2048$,
  $\gamma \leq \nu_d$ and $d \geq 64 \log 2$, then
  \begin{align*}
    \inf_{\hat f_n \in \cF_n}
    \sup_{f_0 \in \cH_1}
    \E\bigl(
      \|\hat f_n - f_0\|_{L_2}^2
    \bigr)
    &\geq
    \frac{\gamma}{32}.
  \end{align*}
  Now, by condition (i) on the kernel,
  for $\gamma \geq \nu_d$, we have
  \begin{align*}
    H_{2,\gamma}
    &=
    \sum_{r \in \cR}
    \frac{1}{1 + \gamma / \nu_r}
    \leq
    \sum_{r \in \cR,\, r \leq d}
    \frac{1}{1 + \nu_d / \nu_r}
    + \sum_{r \in \cR,\, r > d}
    \frac{1}{1 + \nu_d / \nu_r} \\
    &\leq
    d +
    \frac{1}{\nu_d}
    \sum_{r \in \cR, r > d}
    \nu_r
    \leq
    (C_K + 1) d.
  \end{align*}
  Taking $\gamma \coloneqq \nu_d$ and
  $C_3 \coloneqq 2048 (C_K + 1)$, and writing
  $\bar d \coloneqq \lceil 64 \log 2 \rceil = 45$,
  we deduce that
  \begin{align*}
    \inf_{\hat f_n \in \cF_n}
    \sup_{f_0 \in \cH_1}
    \E\bigl(
      \|\hat f_n - f_0\|_{L_2}^2
    \bigr)
    &\geq
    \frac{1}{32}
    \max
    \biggl\{
      \nu_d:
      d \in \cR,\,
      \nu_d  \leq \frac{H_{2,\nu_d}}{C_3 n}
      \text{ and }
      d \geq \bar d
    \biggr\}.
  \end{align*}
  Taking $n > H_{2,\nu_{\bar d}}/(C_3 \nu_{\bar d}) \eqqcolon n_0$
  ensures that
  $\nu_d > H_{2,\nu_d}/(C_3 n)$
  for all $d \leq \bar d$.
  Therefore, for $n > n_0$,
  \begin{align*}
    \inf_{\hat f_n \in \cF_n}
    \sup_{f_0 \in \cH_1}
    \E\bigl(
      \|\hat f_n - f_0\|_{L_2}^2
    \bigr)
    &\geq
    \frac{1}{32}
    \max
    \biggl\{
      \nu_d :
      d \in \cR,\,
      \nu_d \leq \frac{H_{2,\nu_d}}{C_3 n}
    \biggr\}.
  \end{align*}
  This maximum is attained at
  $d^\star \coloneqq \min\bigl\{
  d \in \cR: \nu_d \leq H_{2,\nu_d} / (C_3 n)\bigr\}$,
  since $\cR = \N$ and $d \mapsto \nu_d$ is decreasing with $\lim_{d
  \rightarrow \infty} \nu_d = 0$ and $d \mapsto H_{2,\nu_d}$ is
  increasing with $\lim_{d \rightarrow \infty} H_{2,\nu_d} = \infty$.
  Further, we have $d^\star \geq \bar d$.
  By minimality of $d^\star$ and condition (ii) on the kernel,
  $\nu_{d^\star - 1} > H_{2,\nu_{d^\star - 1}} / (C_3 n)$ and so
  $\nu_{d^\star - 1} + H_{2,\nu_{d^\star - 1}} / (C_3 n)
  < 2 \nu_{d^\star - 1} \leq 2 C_K \nu_{d^\star}$.
  Therefore, for $n > n_0$,
  \begin{align*}
    \inf_{\hat f_n \in \cF_n}
    \sup_{f_0 \in \cH_1}
    \E\bigl(
      \|\hat f_n - f_0\|_{L_2}^2
    \bigr)
    &\geq
    \frac{\nu_{d^\star}}{32}
    \geq
    \frac{1}{64 C_K}
    \biggl(
      \nu_{d^\star - 1}
      + \frac{H_{2,\nu_{d^\star - 1}}}{C_3 n}
    \biggr).
  \end{align*}
  Since the minimax risk is decreasing in $n$,
  we have for all $n \in \N$ that
  \begin{align*}
    \inf_{\hat f_n \in \cF_n}
    \sup_{f_0 \in \cH_1}
    \E\bigl(
      \|\hat f_n - f_0\|_{L_2}^2
    \bigr)
    &\geq
    \frac{1}{C_4}
    \inf_{\gamma > 0}
    \biggl(
      \gamma
      + \frac{H_{2,\gamma}}{n}
    \biggr),
  \end{align*}
  where $C_4 > 0$ depends only on $C_K$ and $n_0$.
  Finally,
  by condition (ii) on the kernel again,
  \begin{align*}
    n_0
    =
    \frac{H_{2,\nu_{\bar d}}}{C_3 \nu_{\bar d}}
    \leq
    \frac{(C_K + 1) \bar d}{C_3 \nu_{\bar d}}
    \leq
    \frac{(C_K + 1) \bar d C_K^{\bar d - 1}}{C_3 \nu_1}
    =
    \frac{45 (C_K + 1) C_K^{44}}{C_3 \nu_1},
  \end{align*}
  and depends only on $C_K$ and $\nu_1$, as required.
\end{proof}

\subsubsection{Uniform convergence of \texorpdfstring{$S_n$}{Sn}
and \texorpdfstring{$D S_n$}{DSn}}

In Lemmas~\ref{lem:uniform_convergence_Sn}
and~\ref{lem:uniform_convergence_DSn} we present uniform convergence
results for $S_n(f, t)$ and $D S_n(f, t)$ respectively,
uniformly over $f$ lying on line segments. These results
may be used for bounding terms appearing in mean-value forms of
Taylor's theorem; the previous Lemmas~\ref{lem:convergence_Sn}
and~\ref{lem:convergence_DSn}, while similar, are unsuited for this task.

\begin{lemma}[Uniform convergence of $S_n$]%
  \label{lem:uniform_convergence_Sn}
  Let $f_1, f_2 \in \cB(\cX)$
  and for $a \in [0, 1]$ define $\tilde f_a \vcentcolon= a f_1 + (1-a) f_2$.
  There exists a universal constant $C > 0$ such that
  for all $s \geq 1$, with probability at least $1 - e^{-s^2}$,
  \begin{align*}
    \sup_{a \in [0, 1]}
    \sup_{t \in [0, 1]}
    \bigl|S_n(\tilde f_a, t) - S(\tilde f_a, t) \bigr|
    &\leq
    \frac{C s e^{2(\|f_1\|_\infty \lor \|f_2\|_\infty)}}{\sqrt n}.
  \end{align*}
\end{lemma}

\begin{proof}[Proof of Lemma~\ref{lem:uniform_convergence_Sn}]
  For $i \in [n]$, $a \in [0, 1]$ and $t \in [0, 1]$, define
  \begin{align*}
    \xi_{i a t}
    &\vcentcolon=
    R_i(t) e^{\tilde f_a(X_i)}
    - \int_\cX q(x, t) e^{\tilde f_a(x)} \diffi P_X(x),
  \end{align*}
  noting that $|\xi_{i a t}| \leq 2 e^{\|f_1\|_\infty \lor \|f_2\|_\infty}$.
  It follows that with $\mathcal{I} := \{(i,j):i,j \in [n],i<j\}$,
  \begin{align}
    \nonumber
    &\sup_{a \in [0, 1]}
    \sup_{t \in [0, 1]}
    \bigl(S_n(\tilde f_a, t) - S(\tilde f_a, t) \bigr)^2
    =
    \sup_{a \in [0, 1]}
    \sup_{t \in [0, 1]}
    \frac{1}{n^2}
    \sum_{i=1}^n
    \sum_{j=1}^n
    \xi_{i a t}
    \xi_{j a t} \\
    \nonumber
    &\qquad\leq
    \sup_{a \in [0, 1]}
    \sup_{t \in [0, 1]}
    \frac{1}{n^2}
    \sum_{i=1}^n
    \xi_{i a t}^2
    + \sup_{a \in [0, 1]}
    \sup_{t \in [0, 1]}
    \frac{2}{n^2}
    \sum_{(i,j) \in \mathcal{I}}
    \xi_{i a t}
    \xi_{j a t} \\
    \label{eq:diagonal_bound_Sn}
    &\qquad\leq
    \frac{4 e^{2 (\|f_1\|_\infty \lor \|f_2\|_\infty)}}{n}
    + \sup_{a \in [0, 1]}
    \sup_{t \in [0, 1]}
    \frac{2}{n(n-1)}
    \biggl|
    \sum_{(i,j) \in \mathcal{I}}
    \xi_{i a t}
    \xi_{j a t}
    \biggr|.
  \end{align}
  The second term in the above display is a degenerate $U$-process
  \citep[Chapter~5]{de1999decoupling} as
  $\E \bigl\{ \xi_{i a t} \xi_{j a t} \mid X_i, T_i, I_i \bigr\}
  = \xi_{i a t} \, \E (\xi_{j a t}) = 0$ for $(i,j) \in
  \mathcal{I}$, $a \in [0, 1]$ and $t \in [0, 1]$.
  We now let $\mathcal{Q}$ denote the set of finitely-supported
  probability measures on $(\mathcal{X} \times [0,1] \times \{0,1\})^2$,
  and for $Q \in \mathcal{Q}$, seek to bound the $L_2(Q)$-covering
  number of the function class
  \begin{align*}
    \mathcal G
    &=
    \biggl\{
      \big((X_i, T_i, I_i), (X_j, T_j, I_j)\big)
      \mapsto \xi_{i a t} \xi_{j a t} :
      a \in [0, 1], t \in [0, 1]
    \biggr\}.
  \end{align*}
  For $0 \leq a < a' \leq 1$, $0 \leq t < t' \leq 1$ and $i \in [n]$,
  \begin{align*}
    |\xi_{i a t} - \xi_{i a' t'}|
    &\leq
    \bigl|
    R_i(t) e^{\tilde f_a(X_i)}
    - R_i(t') e^{\tilde f_{a'}(X_i)}
    \bigr|
    + \int_\cX
    \Bigl|
    q(x, t) e^{\tilde f_a(x)} \diffi P_X(x)
    - q(x, t') e^{\tilde f_{a'}(x)}
    \Bigr|
    \diffi P_X(x) \\
    &\leq
    e^{\|f_1\|_\infty \lor \|f_2\|_\infty}
    \bigl(
      \mathbbm 1_{\{t \leq T_i < t'\}}
      + \P(t \leq T < t')
    \bigr)
    + 4 (a' - a)
    e^{2(\|f_1\|_\infty \lor \|f_2\|_\infty)},
  \end{align*}
  so
  \begin{align*}
    &|\xi_{i a t} \xi_{j a t} - \xi_{i a' t'} \xi_{j a' t'}|
    \leq
    |\xi_{i a t}| | \xi_{j a t} - \xi_{j a' t'}|
    + |\xi_{j a t'}| | \xi_{i a t} - \xi_{i a' t'}| \\
    &\quad\leq
    2 e^{2(\|f_1\|_\infty \lor \|f_2\|_\infty)}
    \bigl(
      \mathbbm 1_{\{t \leq T_i < t'\}}
      + \mathbbm 1_{\{t \leq T_j < t'\}}
      + 2 \P(t \leq T < t')
    \bigr)
    + 16 (a' - a)
    e^{3(\|f_1\|_\infty \lor \|f_2\|_\infty)} \\
    &\quad\leq
    16 e^{3(\|f_1\|_\infty \lor \|f_2\|_\infty)}
    \bigl(
      \mathbbm 1_{\{t \leq T_i < t'\}}
      + \mathbbm 1_{\{t \leq T_j < t'\}}
      + \P(t \leq T < t')
      + a' - a
    \bigr).
  \end{align*}
  It follows that there exists a universal constant $C_1 \geq 2$
  such that
  \begin{align*}
    &\E_Q
    \Bigl\{
      \bigl(
        \xi_{i a t} \xi_{j a t} - \xi_{i a' t'} \xi_{j a' t'}
      \bigr)^2
    \Bigr\} \\
    &\quad\leq
    C_1^2 e^{6(\|f_1\|_\infty \lor \|f_2\|_\infty)}
    \bigl\{
      \P_Q(t \leq T_i < t')
      + \P_Q(t \leq T_j < t')
      + \P(t \leq T < t')^2
      + (a'-a)^2
    \bigr\}.
  \end{align*}
  Given $\varepsilon \in (0,1)$, define $\mathcal{T}_1 :=
  \{t_1,\ldots,t_{\lceil 1/\varepsilon^2 \rceil}\}$, where
  \begin{align*}
    t_m
    &\vcentcolon=
    \inf\{t \in [0,1]: \P_Q(T_i \leq t) \geq m\varepsilon^2\}
  \end{align*}
  for $m \in [\lceil 1/\varepsilon^2 \rceil - 1]$ and $t_{\lceil
  1/\varepsilon^2 \rceil} := 1$, so that given any $t \in [0,1]$,
  there exists $m \in [\lceil 1/\varepsilon^2 \rceil]$ such that
  $\P_Q(t \leq T_i < t_m) \leq \varepsilon^2$.
  Let $\cT_2$ and $\cT_3$ satisfy the same properties but for
  $T_j$ under $Q$ and $T$ under $P_T$ respectively,
  and define $\cT \vcentcolon= \cT_1 \cup \cT_2 \cup \cT_3$.
  Likewise, let $\cA \vcentcolon=
  \{0, \varepsilon, 2 \varepsilon, \ldots,
  (\lceil 1/\varepsilon \rceil - 1) \varepsilon, 1\}$, so that given
  any $a \in [0, 1]$, there exists $a' \in \cA$ such that
  $|a' - a| \leq \varepsilon$.
  Then the set $\cT \times \cA$ induces a
  $\bigl(2 C_1 e^{3(\|f_1\|_\infty \lor \|f_2\|_\infty)}
  \varepsilon\bigr)$-cover
  of $\mathcal G$ under $L_2(Q)$,
  of cardinality at most
  $3\lceil 1/\varepsilon^2 \rceil \times \lceil 1/\varepsilon\rceil
  \leq 12 / \varepsilon^3$.
  Defining $M \vcentcolon = 4 e^{3(\|f_1\|_\infty \lor \|f_2\|_\infty)}$,
  there therefore exists a universal constant $C_2 > 0$ such that
  for all $\varepsilon \in (0, 1)$,
  \begin{align*}
    N\big(\varepsilon M,\mathcal G, \|\cdot\|_{Q,2}\big)
    \leq \frac{C_2}{\varepsilon^3}.
  \end{align*}
  As $\log(1+x) \leq 1 + \log x$ for $x \geq 1$,
  and $\int_0^1 \log(1/x) \diffi x = 1$,
  for a universal constant $C_3 > 0$,
  \begin{align*}
    J
    &\vcentcolon=
    \sup_{Q \in \mathcal{Q}}
    \int_0^{1}
    \log
    \bigl(
      1 +
      N\bigl(\varepsilon M, \mathcal{G}, \|\cdot\|_{Q,2}\bigr)
    \bigr)
    \diffi \varepsilon
    \leq
    1 + \int_0^{1}
    \log\bigl(C_2 / \varepsilon^3\bigr)
    \diffi \varepsilon
    \leq C_3.
  \end{align*}
  It follows by Lemma~\ref{lem:orliczUStat} that there exists a
  universal constant $C_4 > 0$ such that
  for all $s \geq 1$,
  \begin{align*}
    \P \biggl(
      \sup_{a \in [0, 1]}
      \sup_{t \in [0, 1]}
      \frac{2}{n(n-1)}
      \biggl|
      \sum_{(i,j) \in \mathcal{I}}
      \xi_{i}(t) \xi_{j}(t)
      \biggr|
      > \frac{C_4 s^2 e^{3(\|f_1\|_\infty \lor \|f_2\|_\infty)}}{n}
    \biggr)
    \leq e^{-s^2}.
  \end{align*}
  We deduce from~\eqref{eq:diagonal_bound_Sn} that there exists a
  universal constant $C_5 > 0$ such that
  for all $s \geq 1$,
  with probability at least $1 - e^{-s^2}$,
  \begin{align*}
    \sup_{a \in [0, 1]}
    \sup_{t \in [0, 1]}
    \big( S_n(f, t) - S(f, t) \big)^2
    &\leq
    \frac{C_5 s^2 e^{3(\|f_1\|_\infty \lor \|f_2\|_\infty)}}{n},
  \end{align*}
  as required.
\end{proof}

\begin{lemma}[Uniform convergence of $D S_n$]%
  \label{lem:uniform_convergence_DSn}
  Let $f, f_1, f_2 \in \cB(\cX)$
  and for $a \in [0, 1]$ define $\tilde f_a \vcentcolon= a f_1 + (1-a) f_2$.
  Then there is a universal constant $C > 0$ such that for all $s \geq 1$,
  with probability at least $1 - e^{-s^2}$,
  \begin{align*}
    \sup_{a \in [0, 1]}
    \sup_{t \in [0, 1]}
    \bigl| D S_n(\tilde f_a, t)(f) - D S(\tilde f_a, t)(f) \bigr|
    &\leq C e^{2(\|f_1\|_\infty \lor \|f_2\|_\infty)}
    \biggl(
      \frac{s \|f\|_{L_2}}{\sqrt n} + \frac{s^2 \|f\|_\infty}{n}
    \biggr).
  \end{align*}
\end{lemma}

\begin{proof}[Proof of Lemma~\ref{lem:uniform_convergence_DSn}]
  Define the class of functions
  \begin{align*}
    \mathcal G
    \vcentcolon=
    \biggl\{
      (X_i, T_i, I_i)
      \mapsto
      R_i(t) e^{\tilde f_a(X_i)} f(X_i) :
      a \in [0, 1], t \in [0, 1]
    \biggr\},
  \end{align*}
  with associated envelope function
  $F(X, T, I) \vcentcolon= e^{f_1(X) \lor f_2(X)} |f(X)|$.
  By Lemmas~\ref{lem:derivatives_Sn}
  and~\ref{lem:derivatives_S},
  \begin{align*}
    &\sup_{a \in [0, 1]}
    \sup_{t \in [0, 1]}
    \bigl|
    D S_n(\tilde f_a, t)(f) - D S(\tilde f_a, t)(f)
    \bigr| \\
    &\quad=
    \sup_{a \in [0, 1]}
    \sup_{t \in [0, 1]}
    \biggl|
    \frac{1}{n} \sum_{i=1}^n
    R_i(t) e^{\tilde f_a(X_i)} f(X_i)
    - \int_\cX q(x, t) e^{\tilde f_a(x)} f(x) \diffi P_X(x)
    \biggr| \\
    &\quad=
    \sup_{g \in \cG}
    \biggl|
    \frac{1}{n} \sum_{i=1}^n
    \Bigl\{
      g(X_i, T_i, I_i) - \E \bigl( g(X_i, T_i, I_i) \bigr)
    \Bigr\}
    \biggr|.
  \end{align*}
  Define
  $\sigma^2 \vcentcolon= \E\bigl(F(X, T, I)^2\bigr)
  \leq e^{2 (\|f_1\|_\infty \lor \|f_2\|_\infty)} \|f\|_{L_2}^2$ and
  $M \vcentcolon= \|F\|_\infty
  \leq e^{\|f_1\|_\infty \lor \|f_2\|_\infty} \|f\|_\infty$.
  Let $\cQ$ be the set of discrete measures on
  $\cX \times [0, 1] \times \{0, 1\}$
  and for $Q \in \cQ$
  write $\|\cdot\|_{Q,2}$ for the $L_2(Q)$-norm.
  Define the uniform entropy integral
  \begin{align*}
    J
    &\vcentcolon=
    \sup_{Q \in \mathcal{Q}}
    \int_0^{1}
    \sqrt{1 + \log
    N\bigl(\varepsilon \|F\|_{Q, 2}, \mathcal{G}, \|\cdot\|_{Q,2}\bigr)}
    \diffi \varepsilon.
  \end{align*}
  By \citet[Theorem~5.2]{chernozhukov2014gaussian},
  there is a universal constant $C_1 > 0$ such that
  \begin{align}
    \nonumber
    &\E \biggl[
      \sup_{g \in \cG}
      \biggl|
      \frac{1}{n} \sum_{i=1}^n
      \Bigl\{
        g(X_i, T_i, I_i) - \E \bigl( g(X_i, T_i, I_i) \bigr)
      \Bigr\}
      \biggr|
    \biggr]
    \leq
    \frac{C_1 J \sigma}{\sqrt n} + \frac{C_1 J^2 M}{n} \\
    \label{eq:chernozhukov_bernstein}
    &\qquad\leq
    \frac{C_1 J e^{\|f_1\|_\infty \lor \|f_2\|_\infty} \|f\|_{L_2}}{\sqrt n}
    + \frac{C_1 J^2 e^{\|f_1\|_\infty \lor \|f_2\|_\infty} \|f\|_\infty}{n}.
  \end{align}
  To compute the entropy integral, note that for any $i \in [n]$,
  $0 \leq a < a' \leq 1$ and $0 \leq t < t' \leq 1$,
  \begin{align*}
    &\bigl|
    R_i(t') e^{\tilde f_{a'}(X_i)} f(X_i)
    - R_i(t) e^{\tilde f_a(X_i)} f(X_i)
    \bigr| \\
    &\quad\leq
    e^{\tilde f_{a'}(X_i)}
    |f(X_i)|
    \bigl| R_i(t') - R_i(t) \bigr|
    + R_i(t) |f(X_i)|
    \bigl| e^{\tilde f_{a'}(X_i)} - e^{\tilde f_a(X_i)} \bigr| \\
    &\quad\leq
    e^{\|f_1\|_\infty \lor \|f_2\|_\infty}
    |f(X_i)|
    \mathbbm 1_{\{t \leq T_i < t'\}}
    + e^{\|f_1\|_\infty \lor \|f_2\|_\infty}
    |f(X_i)|
    \bigl| \tilde f_{a'}(X_i) - \tilde f_a(X_i) \bigr| \\
    &\quad\leq
    2 e^{2(\|f_1\|_\infty \lor \|f_2\|_\infty)} |f(X_i)|
    \bigl( \mathbbm 1_{\{t \leq T_i < t'\}} + a'-a \bigr).
  \end{align*}
  Therefore for each $Q \in \cQ$,
  \begin{align*}
    &\E_Q \Bigl\{
      \bigl(
        R_i(t') e^{\tilde f_{a'}(X_i)} f(X_i)
        - R_i(t) e^{\tilde f_a(X_i)} f(X_i)
      \bigr)^2
    \Bigr\} \\
    &\quad\leq
    8 e^{4(\|f_1\|_\infty \lor \|f_2\|_\infty)}
    \Bigl\{
      \E_Q \bigl(
        f(X_i)^2
        \mathbbm 1_{\{t \leq T_i < t'\}}
      \bigr)
      + (a' - a)^2 \|f\|_{Q,2}^2
    \Bigr\}.
  \end{align*}
  Given $\varepsilon \in (0,1)$, define
  $\mathcal{T} := \{t_1,\ldots,t_{\lceil 1/\varepsilon^2 \rceil}\}$,
  where
  \begin{align*}
    t_m := \inf\Bigl\{t \in [0,1]:
      \E_Q \bigl(
        f(X_i)^2
        \mathbbm 1_{\{T_i \leq t\}}
      \bigr)
      \geq m \varepsilon^2 \|f\|_{Q,2}^2
    \Bigr\}
  \end{align*}
  for $m \in [\lceil 1/\varepsilon^2 \rceil - 1]$ and $t_{\lceil
  1/\varepsilon^2 \rceil} := 1$, so that given any $t \in [0,1]$,
  there exists $m \in [\lceil 1/\varepsilon^2 \rceil]$ such that
  $\E_Q \bigl( f(X_i)^2 \mathbbm 1_{\{t \leq T_i < t_m\}} \bigr)
  \leq \varepsilon^2 \|f\|_{Q,2}^2$.
  Likewise, let $\cA \vcentcolon=
  \{0, \varepsilon, 2 \varepsilon, \ldots,
  (\lceil 1/\varepsilon \rceil - 1) \varepsilon, 1\}$, so that given
  any $a \in [0, 1]$, there exists $a' \in \cA$ such that
  $|a' - a| \leq \varepsilon/2$.
  Then $\cA \times \cT$ induces a
  $\bigl(4 e^{2(\|f_1\|_\infty \lor \|f_1\|_\infty)}
  \varepsilon \|f\|_{Q,2}\bigr)$-cover
  of $\mathcal G$ under $L_2(Q)$, of cardinality at most
  $\lceil 1/\varepsilon^2 \rceil \times \lceil 1/\varepsilon\rceil
  \leq 4 / \varepsilon^3$.
  Since
  $\|F\|_{Q,2} \geq e^{-(\|f_1\|_\infty \lor \|f_2\|_\infty)} \|f\|_{Q,2}$,
  we have for $\varepsilon \in (0,1)$ that
  \begin{align*}
    N\Big(
      \varepsilon \|F\|_{Q, 2},
      \mathcal G, \|\cdot\|_{Q,2}
    \Big)
    &\leq
    N\Big(
      \varepsilon
      e^{-(\|f_1\|_\infty \lor \|f_2\|_\infty)}
      \|f\|_{Q, 2},
      \mathcal G, \|\cdot\|_{Q,2}
    \Big)
    \leq
    \frac{256 e^{9(\|f_1\|_\infty \lor \|f_2\|_\infty)}}{\varepsilon^{3}}.
  \end{align*}
  Since $\int_0^1 \sqrt{\log(1/\varepsilon)} \diffi \varepsilon
  = \sqrt{\pi}/2$, it follows that
  \begin{align*}
    J
    &\leq
    \int_0^{1}
    \sqrt{1 + \log \bigl(256 e^{9 (\|f_1\|_\infty \lor \|f_2\|_\infty)}
    / \varepsilon^{3}\bigr)}
    \diffi \varepsilon \\
    &\leq
    \sqrt{1 + \log 256}
    + 3 \sqrt{\|f_1\|_\infty \lor \|f_2\|_\infty}
    + \sqrt{3 \pi} / 2 \\
    &\leq
    5 + 3 \sqrt{\|f_1\|_\infty \lor \|f_2\|_\infty}
    \leq
    6 e^{(\|f_1\|_\infty \lor \|f_2\|_\infty)/2}.
  \end{align*}
  Therefore, by \eqref{eq:chernozhukov_bernstein},
  there exists a universal constant $C_2 > 0$ with
  \begin{align}
    \label{eq:Dl_local_maximal}
    \E \biggl(
      \sup_{a \in [0, 1]}
      \sup_{t \in [0, 1]}
      \bigl|
      D S_n(\tilde f_a, t)(f) - D S(\tilde f_a, t)(f)
      \bigr|
    \biggr)
    &\leq
    C_2 e^{2 (\|f_1\|_\infty \lor \|f_2\|_\infty)}
    \biggl(
      \frac{\|f\|_{L_2}}{\sqrt n} + \frac{\|f\|_\infty}{n}
    \biggr).
  \end{align}
  We now apply a Bernstein-type inequality for empirical processes
  \citep[Chapter~12]{boucheron2013concentration}
  to control the tail behaviour of the supremum.
  Noting that the result is trivial if $\|f\|_\infty = 0$, assume
  that $\|f\|_\infty > 0$.
  For $i \in [n]$, $a \in [0, 1]$ and $t \in [0, 1]$, write
  \begin{align*}
    X_{i a t}
    &\vcentcolon=
    \frac{1}{2 e^{(\|f_1\|_\infty \lor \|f_2\|_\infty)} \|f\|_\infty}
    \biggl(
      R_i(t) e^{\tilde f_a(X_i)} f(X_i)
      - \int_\cX q(x, t) e^{\tilde f_a(x)} f(x) \diffi P_X(x)
    \biggr),
  \end{align*}
  so that $(X_{i a t})_{a \in [0, 1], t \in [0, 1]}$ are independent
  and identically distributed, with
  $\E(X_{i a t}) = 0$ and $|X_{i a t}| \leq 1$.
  Further,
  \begin{align*}
    \sup_{a \in [0, 1]}
    \sup_{t \in [0, 1]}
    \bigl| D S_n(\tilde f_a, t)(f) - D S(\tilde f_a, t)(f) \bigr|
    &=
    \frac{2 e^{(\|f_1\|_\infty \lor \|f_2\|_\infty)} \|f\|_\infty}{n}
    \sup_{a \in [0, 1]}
    \sup_{t \in [0, 1]}
    \biggl|
    \sum_{i=1}^n
    X_{i a t}
    \biggr|,
  \end{align*}
  and also
  \begin{align*}
    \E \biggl(
      \sup_{a \in [0, 1]}
      \sup_{t \in [0, 1]}
      \sum_{i=1}^n
      X_{i a t}^2
    \biggr)
    \leq
    \frac{1}
    {4 \|f\|_\infty^2}
    \biggl(
      2 \sum_{i=1}^n \E \bigl( f(X_i)^2 \bigr)
      + 2 n \int_\cX f(x)^2 \diffi P_X(x)
    \biggr) =
    \frac{n \|f\|_{L_2}^2} {\|f\|_\infty^2}.
  \end{align*}
  Therefore by \citet[Theorem~12.2]{boucheron2013concentration},
  for all $s' \geq 0$,
  \begin{align*}
    \P \biggl(
      \sup_{a \in [0, 1]}
      \sup_{t \in [0, 1]}
      \biggl| \sum_{i=1}^n X_{i a t} \biggr|
      \geq
      \E \biggl\{
        \sup_{a \in [0, 1]}
        \sup_{t \in [0, 1]}
        \biggl| \sum_{i=1}^n X_{i a t} \biggr|
      \biggr\}
      + s'
    \biggr)
    &\leq
    2\exp \biggl(
      - \frac{s'^2}{8 n \|f\|_{L_2}^2 / \|f\|_\infty^2 + 2 s'}
    \biggr).
  \end{align*}
  In combination with \eqref{eq:Dl_local_maximal}, setting
  $s' \vcentcolon=
  3s \sqrt n \|f\|_{L_2} / \|f\|_\infty + 2s^2$
  with $s \geq 1$, we deduce that there exists a universal constant
  $C > 0$ such that with probability at least $1 - e^{-s^2}$,
  \begin{align*}
    \sup_{a \in [0, 1]}
    \sup_{t \in [0, 1]}
    \bigl| D S_n(\tilde f_a, t)(f) - D S(\tilde f_a, t)(f) \bigr|
    &\leq C e^{2(\|f_1\|_\infty \lor \|f_2\|_\infty)}
    \biggl(
      \frac{s \|f\|_{L_2}}{\sqrt n} + \frac{s^2 \|f\|_\infty}{n}
    \biggr).
    \qedhere
  \end{align*}
\end{proof}

\subsubsection{Local convergence of \texorpdfstring{$\ell_n$}{ln}}

We now apply the uniform bounds given in
Lemmas~\ref{lem:uniform_convergence_Sn}
and~\ref{lem:uniform_convergence_DSn} to deduce a local convergence result
for $D \ell_n$ in Lemma~\ref{lem:local_convergence_Dln}.
Lemma~\ref{lem:lipschitz_ln_lstar} demonstrates that
$\ell_n$ and $\ell_\star$ are Lipschitz with respect to
the supremum norm.
Finally, Lemma~\ref{lem:local_convergence_ln} provides a bound on
$\ell_n(f)$ which is local in the sense that
the error reduces when $f$ is close to $f_0$ under the $L_2(P_X)$-norm.
This property is essential when used in combination with
the $L_2(P_X)$-strong convexity result for $\ell_\star$ from
Lemma~\ref{lem:strong_convexity},
and allows us to derive our main theorems on cross-validation.

\begin{lemma}[Local convergence of $D \ell_n$]%
  \label{lem:local_convergence_Dln}
  Let $f, f_1, f_2 \in \cB(\cX)$,
  and for $a \in [0, 1]$ define $\tilde f_a \vcentcolon= a f_1 + (1-a) f_2$.
  Then there is a universal constant $C > 0$ such that for all $s \geq 1$,
  with probability at least $1 - e^{-s^2}$,
  \begin{align*}
    \sup_{a \in [0, 1]}
    \bigl| D \ell_n(\tilde f_a)(f) - D \ell_\star(\tilde f_a)(f) \bigr|
    &\leq
    \frac{
      C e^{4\|f_0\|_\infty}
    e^{6(\|f_1 - f_0\|_\infty \lor \|f_2 - f_0\|_\infty)}}{q_1^2}
    \biggl(
      \frac{s \|f\|_{L_2}}{\sqrt n}
      + \frac{s^2 \|f\|_\infty}{n}
    \biggr).
  \end{align*}
\end{lemma}

\begin{proof}[Proof of Lemma~\ref{lem:local_convergence_Dln}]
  By Lemmas~\ref{lem:derivatives_ln} and~\ref{lem:derivatives_lstar},
  we have
  \begin{align}
    \nonumber
    D \ell_n(\tilde f_a)(f) - D \ell_\star(\tilde f_a)(f) &=
    \int_{0}^{1}
    \frac{D S_n(\tilde f_a, t)(f)}{S_n(\tilde f_a, t)}
    \diffi {N(t)}
    -
    \frac{1}{n}
    \sum_{i=1}^n
    f(X_i)
    N_i(1) \\
    \nonumber
    &\qquad-
    \int_{0}^{1}
    \frac{D S(\tilde f_a, t)(f)}{S(\tilde f_a, t)}
    S(f_0, t) \lambda_0(t) \diffi t
    + \int_{0}^{1}
    D S(f_0, t)(f)
    \lambda_0(t)
    \diffi t \\
    \label{eq:local_convergence_Dln_1}
    &=
    \int_{0}^{1}
    \frac{D S_n(\tilde f_a, t)(f) - D S(\tilde f_a, t)(f)}{S_n(\tilde f_a, t)}
    \diffi {N(t)} \\
    \label{eq:local_convergence_Dln_2}
    &\qquad+
    \int_{0}^{1}
    \frac{D S(\tilde f_a, t)(f)}{S(\tilde f_a, t)}
    \frac{S(\tilde f_a, t) - S_n(\tilde f_a, t)}{S_n(\tilde f_a, t)}
    \diffi {N(t)} \\
    \label{eq:local_convergence_Dln_3}
    &\qquad+
    \int_{0}^{1}
    \frac{D S(\tilde f_a, t)(f)}{S(\tilde f_a, t)}
    \bigl(\mathrm d N(t) - S(f_0, t) \lambda_0(t) \, \mathrm d t\bigr) \\
    \label{eq:local_convergence_Dln_4}
    &\qquad+
    \int_{0}^{1}
    D S(f_0, t)(f)
    \lambda_0(t)
    \diffi t
    - \frac{1}{n}
    \sum_{i=1}^n
    f(X_i)
    N_i(1).
  \end{align}
  We begin by bounding the term in \eqref{eq:local_convergence_Dln_1}.
  By Lemmas~\ref{lem:uniform_convergence_Sn}
  and~\ref{lem:uniform_convergence_DSn},
  and since
  $S(\tilde f_a, t) = \int_\cX q(x, t) e^{\tilde f_a(x)} \diffi P_X(x)
  \geq e^{-(\|f_1 - f_0\|_\infty \lor \|f_2 - f_0\|_\infty)} q_1$,
  there exists a universal constant $C_1 > 0$ such that whenever
  $1 \leq s \leq e^{-2 \|f_0\|_\infty}
  e^{-3 (\|f_1 - f_0\|_\infty \lor \|f_2 - f_0\|_\infty)}
  q_1 \sqrt n / C_1$, with probability at least $1 - e^{-s^2}$,
  \begin{align*}
    &\sup_{a \in [0, 1]}
    \biggl|
    \int_{0}^{1}
    \frac{D S_n(\tilde f_a, t)(f) - D S(\tilde f_a, t)(f)}{S_n(\tilde f_a, t)}
    \diffi {N(t)}
    \biggl| \\
    &\quad\leq
    \frac{\sup_{a \in [0, 1]} \sup_{t \in [0, 1]}
    \bigl| D S_n(\tilde f_a, t)(f) - D S(\tilde f_a, t)(f) \bigr|}
    {\inf_{a \in [0, 1]} \inf_{t \in [0, 1]} S(\tilde f_a, t)
      - \sup_{a \in [0, 1]} \sup_{t \in [0, 1]} |S_n(\tilde f_a, t) -
    S(\tilde f_a, t)|} \\
    &\quad\leq
    \frac{(C_1/2) e^{2 \|f_0\|_\infty}
    e^{2(\|f_1 - f_0\|_\infty \lor \|f_2 - f_0\|_\infty)}} {
      e^{-(\|f_1 - f_0\|_\infty \lor \|f_2 - f_0\|_\infty)} q_1
      - (C_1 / 2) s
      e^{2\|f_0\|_\infty}
      e^{2(\|f_1 - f_0\|_\infty \lor \|f_2 - f_0\|_\infty)}
    / \sqrt n } \\
    &\qquad\times
    \biggl(
      \frac{s \|f\|_{L_2}}{\sqrt n} + \frac{s^2 \|f\|_\infty}{n}
    \biggr) \\
    &\quad\leq
    C_1 e^{2 \|f_0\|_\infty}
    e^{3(\|f_1 - f_0\|_\infty \lor \|f_2 - f_0\|_\infty)}
    \biggl(
      \frac{s \|f\|_{L_2}}{\sqrt n q_1} + \frac{s^2 \|f\|_\infty}{n q_1}
    \biggr).
  \end{align*}
  Next we consider \eqref{eq:local_convergence_Dln_2}.
  By Lemma~\ref{lem:derivatives_S} and Cauchy--Schwarz, we have
  \begin{align*}
    \bigl|
    D S(\tilde f_a, t)(f)
    \bigr|
    &\leq
    \int_\cX q(x, t) e^{\tilde f_a(x)} |f(x)| \diffi P_X(x)
    \leq
    e^{\|f_0\|_\infty}
    e^{\|f_1 - f_0\|_\infty \lor \|f_2 - f_0\|_\infty}
    \|f\|_{L_2}.
  \end{align*}
  Thus, by Lemmas~\ref{lem:uniform_convergence_Sn}
  and~\ref{lem:uniform_convergence_DSn},
  there exists a universal constant $C_2 > 0$ such that whenever
  $1 \leq s \leq e^{-2 \|f_0\|_\infty}
  e^{-3 (\|f_1 - f_0\|_\infty \lor \|f_2 - f_0\|_\infty)}
  q_1 \sqrt n / C_2$, with probability at least $1 - e^{-s^2}$,
  \begin{align*}
    &\sup_{a \in [0, 1]}
    \biggl|
    \int_{0}^{1}
    \frac{D S(\tilde f_a, t)(f)}{S(\tilde f_a, t)}
    \frac{S(\tilde f_a, t) - S_n(\tilde f_a, t)}{S_n(\tilde f_a, t)}
    \diffi {N(t)}
    \biggr| \\
    &\quad\leq
    \frac{\sup_{a \in [0, 1]} \sup_{t \in [0, 1]} |D S(\tilde f_a, t)(f)|}
    {\inf_{a \in [0, 1]} \inf_{t \in [0, 1]} S(\tilde f_a, t)} \\
    &\qquad\times
    \frac{\sup_{a \in [0, 1]} \sup_{t \in [0, 1]}
    \bigl| S_n(\tilde f_a, t) - S(\tilde f_a, t) \bigr|}
    {\inf_{a \in [0, 1]} \inf_{t \in [0, 1]} S(\tilde f_a, t)
      - \sup_{a \in [0, 1]} \sup_{t \in [0, 1]} |S_n(\tilde f_a, t) -
    S(\tilde f_a, t)|} \\
    &\quad\leq
    C_2 e^{3\|f_0\|_\infty}
    e^{5(\|f_1 - f_0\|_\infty \lor \|f_2 - f_0\|_\infty)}
    \frac{s \|f\|_{L_2}}{\sqrt n q_1}.
  \end{align*}

  We now turn to \eqref{eq:local_convergence_Dln_3}.
  Define the function class
  \begin{align*}
    \mathcal G
    &=
    \biggl\{
      (X_i, T_i, I_i)
      \mapsto
      \int_{0}^{1}
      \frac{D S(\tilde f_a, t)(f)}{S(\tilde f_a, t)}
      \bigl(\mathrm d N_i(t) - S(f_0, t) \lambda_0(t) \, \mathrm d t\bigr)
      : a \in [0, 1]
    \biggr\},
  \end{align*}
  with associated envelope function
  $F \vcentcolon= 8 e^{2\|f_0\|_\infty}
  e^{5(\|f_1 - f_0\|_\infty \lor \|f_2 - f_0\|_\infty)} \|f\|_{L_2} / q_1^2$.
  For $a, a' \in [0, 1]$, by Lemma~\ref{lem:derivatives_S}
  and Cauchy--Schwarz,
  \begin{align}
    \nonumber
    \sup_{t \in [0, 1]}
    \bigl|
    D S(\tilde f_a, t)(f)
    - D S(\tilde f_{a'}, t)(f)
    \bigr|
    &=
    \sup_{t \in [0, 1]}
    \biggl|
    \int_\cX
    q(x, t)
    \bigl(
      e^{\tilde f_a(x)}
      - e^{\tilde f_{a'}(x)}
    \bigr)
    f(x)
    \diffi P_X(x)
    \biggr| \\
    \nonumber
    &\leq
    \sup_{x \in \cX}
    \bigl|
    e^{\tilde f_a(x)}
    - e^{\tilde f_{a'}(x)}
    \bigr|
    \|f\|_{L_2} \\
    \nonumber
    &\leq
    e^{\|f_0\|_\infty}
    e^{\|f_1 - f_0\|_\infty \lor \|f_2 - f_0\|_\infty}
    \bigl\| \tilde f_a - \tilde f_{a'} \bigr\|_\infty
    \|f\|_{L_2} \\
    \nonumber
    &\leq
    e^{\|f_0\|_\infty}
    e^{\|f_1 - f_0\|_\infty \lor \|f_2 - f_0\|_\infty}
    |a - a'|
    \| f_1 - f_2 \|_\infty
    \|f\|_{L_2} \\
    \label{eq:DS_lipschitz}
    &\leq
    2 e^{\|f_0\|_\infty}
    e^{2(\|f_1 - f_0\|_\infty \lor \|f_2 - f_0\|_\infty)}
    |a - a'|
    \|f\|_{L_2}.
  \end{align}
  By a similar argument, we have
  \begin{align}
    \label{eq:S_lipschitz}
    \sup_{t \in [0, 1]}
    \bigl|
    S(\tilde f_a, t)
    - S(\tilde f_{a'}, t)
    \bigr|
    &\leq
    2 e^{\|f_0\|_\infty}
    e^{2(\|f_1 - f_0\|_\infty \lor \|f_2 - f_0\|_\infty)}
    |a - a'|.
  \end{align}
  Hence, by \eqref{eq:DS_lipschitz} and \eqref{eq:S_lipschitz},
  since $N_i(1) \leq 1$
  for each $i \in [n]$
  and $\int_0^1 S(f_0, t) \lambda_0(t) \diffi t \leq 1$,
  \begin{align*}
    &\biggl|
    \int_{0}^{1}
    \biggl(
      \frac{D S(\tilde f_a, t)(f)}{S(\tilde f_a, t)}
      - \frac{D S(\tilde f_{a'}, t)(f)}{S(\tilde f_{a'}, t)}
    \biggr)
    \bigl(\mathrm d N_i(t) - S(f_0, t) \lambda_0(t) \, \mathrm d t\bigr)
    \biggl| \\
    &\quad\leq
    \frac{2 \sup_{t \in [0, 1]}
    \bigl|D S(\tilde f_a, t)(f) - D S(\tilde f_{a'}, t)(f)\bigr|}
    {\inf_{a \in [0, 1]} \inf_{t \in [0, 1]} S(\tilde f_a, t)} \\
    &\qquad+
    \frac{2 \sup_{a \in [0, 1]} \sup_{t \in [0, 1]} |D S(\tilde f_{a}, t)(f)|
    \sup_{t \in [0, 1]} \bigl|S(\tilde f_a, t) - S(\tilde f_{a'}, t)\bigr|}
    {\inf_{a \in [0, 1]} \inf_{t \in [0, 1]} S(\tilde f_a, t)^2} \\
    &\quad\leq
    \frac{8}{q_1^2}
    e^{2\|f_0\|_\infty}
    e^{5(\|f_1 - f_0\|_\infty \lor \|f_2 - f_0\|_\infty)}
    |a - a'|
    \|f\|_{L_2}
    = F |a - a'|.
  \end{align*}
  Thus by setting
  $\cA \vcentcolon= \{0, \varepsilon, 2 \varepsilon, \ldots,
  (\lceil 1/\varepsilon \rceil - 1) \varepsilon, 1\}$, given
  any $a \in [0, 1]$, there exists $a' \in \cA$ such that
  $|a' - a| \leq \varepsilon$. It follows that
  $N\bigl(\varepsilon F, \cG, \|\cdot\|_\infty\bigr) \leq 2/\varepsilon$,
  so as $\int_0^1 \sqrt{\log(1/\varepsilon)} \diffi \varepsilon
  = \sqrt{\pi}/2$,
  \begin{align*}
    J
    &\vcentcolon=
    \int_0^{1}
    \sqrt{1 + \log
    N\bigl(\varepsilon F, \mathcal{G}, \|\cdot\|_{\infty}\bigr)}
    \diffi \varepsilon
    \leq
    \int_0^{1}
    \sqrt{1 + \log (2/\varepsilon)}
    \diffi \varepsilon
    \leq \sqrt{1 + \log 2} + \frac{\sqrt \pi}{2}
    \leq 3.
  \end{align*}
  Therefore, by \citet[Theorem~5.2]{chernozhukov2014gaussian},
  there exists a universal constant $C_3 > 0$ such that
  \begin{align*}
    \E \biggl\{
      \sup_{a \in [0, 1]}
      \biggl|
      \int_{0}^{1}
      \frac{D S(\tilde f_a, t)(f)}{S(\tilde f_a, t)}
      \bigl(\mathrm d N(t) &- S(f_0, t) \lambda_0(t) \, \mathrm d t\bigr)
      \biggr|
    \biggr\} \\
    &\leq
    C_3 e^{2\|f_0\|_\infty}
    e^{5(\|f_1 - f_0\|_\infty \lor \|f_2 - f_0\|_\infty)}
    \frac{\|f\|_{L_2}}{\sqrt n q_1^2}.
  \end{align*}
  It follows by the bounded differences inequality
  \citep[][Theorem~6.2]{boucheron2013concentration} that
  there exists a universal constant $C_4 > 0$ such that
  for all $s \geq 1$, with probability at least $1 - e^{-s^2}$,
  \begin{align*}
    \sup_{a \in [0, 1]}
    \biggl|
    \int_{0}^{1}
    \frac{D S(\tilde f_a, t)(f)}{S(\tilde f_a, t)}
    \bigl(\mathrm d N(t) - S(f_0, t) \lambda_0(t) \, \mathrm d t\bigr)
    \biggr|
    &\leq
    C_4
    e^{2\|f_0\|_\infty}
    e^{5(\|f_1 - f_0\|_\infty \lor \|f_2 - f_0\|_\infty)}
    \frac{s \|f\|_{L_2}}{\sqrt n q_1^2}.
  \end{align*}

  Finally, we consider \eqref{eq:local_convergence_Dln_4}.
  By Fubini's theorem and Lemma~\ref{lem:derivatives_S}, for
  each $i \in [n]$,
  \begin{align*}
    \E\bigl( f(X_i) N_i(1) \bigr)
    &=
    \E\biggl(
      f(X_i)
      \int_0^1
      R_i(t) e^{f_0(X_i)}
      \lambda_0(t)
      \diffi t
    \biggr) \\
    &=
    \int_0^1
    \int_\cX
    q(x, t)
    e^{f_0(x)}
    f(x)
    \diffi P_X(x)
    \lambda_0(t)
    \diffi t
    =
    \int_{0}^{1}
    D S(f_0, t)(f)
    \lambda_0(t)
    \diffi t.
  \end{align*}
  Also, $|f(X_i) N_i(1)| \leq \|f\|_\infty$
  and $\E\bigl(f(X_i)^2 N_i(1)^2\bigr) \leq \|f\|_{L_2}^2$.
  Hence, by Bernstein's inequality,
  \begin{align*}
    \P \biggl(
      \biggl|
      \int_{0}^{1}
      D S(f_0, t)(f)
      \lambda_0(t)
      \diffi t
      - \frac{1}{n}
      \sum_{i=1}^n
      f(X_i)
      N_i(1)
      \biggr|
      > s'
    \biggr)
    \leq 2 \exp \biggl(
      - \frac{n s^{\prime 2}}{2 \|f\|_{L_2}^2 + \frac{2}{3} \|f\|_\infty s'}
    \biggr)
  \end{align*}
  for $s' > 0$. With
  $s' \vcentcolon= 2 s \|f\|_{L_2} / \sqrt n + 2 s^2 \|f\|_\infty / n$
  for $s \geq 1$ and a universal constant $C_5 > 0$,
  \begin{align*}
    \biggl|
    \int_{0}^{1}
    D S(f_0, t)(f)
    \lambda_0(t)
    \diffi t
    - \frac{1}{n}
    \sum_{i=1}^n
    f(X_i)
    N_i(1)
    \biggr|
    \leq
    C_5 \biggl(
      \frac{s \|f\|_{L_2}}{\sqrt n} + \frac{s^2 \|f\|_\infty}{n}
    \biggr),
  \end{align*}
  with probability at least $1 - e^{-s^2}$.

  We deduce that there exists a universal constant $C_6 > 0$
  such that for all $s$ with
  $1 \leq s \leq e^{-2 \|f_0\|_\infty}
  e^{-3 (\|f_1 - f_0\|_\infty \lor \|f_2 - f_0\|_\infty)}
  q_1 \sqrt n / C_2$, with probability at least $1 - e^{-s^2}$,
  \begin{align}
    \label{eq:first_derivative_new_small_s}
    \sup_{a \in [0, 1]}
    \bigl| D \ell_n(\tilde f_a)(f) - D \ell_\star(\tilde f_a)(f) \bigr|
    &\leq
    C_6
    e^{4\|f_0\|_\infty}
    e^{6(\|f_1 - f_0\|_\infty \lor \|f_2 - f_0\|_\infty)}
    \biggl(
      \frac{s \|f\|_{L_2}}{\sqrt n q_1^2}
      + \frac{s^2 \|f\|_\infty}{n q_1^2}
    \biggr).
  \end{align}
  Finally, we remove the restriction on $s$.
  By Lemma~\ref{lem:lipschitz_ln_lstar},
  \begin{align*}
    \sup_{a \in [0,1]} \bigl| D \ell_n(\tilde f_a)(f) - D
    \ell_\star(\tilde f_a)(f) \bigr|
    &\leq 4 \|f\|_\infty,
  \end{align*}
  so if $s > e^{-2\|f_0\|_\infty}
  e^{-3 (\|f_1 - f_0\|_\infty \lor \|f_2 - f_0\|_\infty)}
  q_1 \sqrt n / C_2$, then
  \eqref{eq:first_derivative_new_small_s} holds by
  increasing the universal constant $C_6$ if necessary.
\end{proof}

\begin{lemma}[Lipschitz property of $\ell_n$ and $\ell_\star$]%
  \label{lem:lipschitz_ln_lstar}
  Let $f_1, f_2 \in \cB(\cX)$. Then
  \begin{align*}
    | \ell_n(f_1) - \ell_n(f_2) |
    \leq
    2 \|f_1 - f_2\|_\infty
    \quad \text{and} \quad
    | \ell_\star(f_1) - \ell_\star(f_2) |
    \leq
    2 \|f_1 - f_2\|_\infty.
  \end{align*}
\end{lemma}

\begin{proof}[Proof of Lemma~\ref{lem:lipschitz_ln_lstar}]
  By Taylor's theorem and Lemmas~\ref{lem:derivatives_ln}
  and~\ref{lem:derivatives_Sn}, there exists
  $\tilde f$ on the line segment
  between $f_1$ and $f_2$ such that
  \begin{align*}
    \bigl|
    \ell_n(f_1) - \ell_n(f_2)
    \bigr|
    &=
    \bigl|
    D \ell_n(\tilde f)(f_1 - f_2)
    \bigr| \\
    &=
    \biggl|
    \frac{1}{n}
    \sum_{i=1}^n
    \frac{D S_n(\tilde f, T_i)(f_1 - f_2)}{S_n(\tilde f, T_i)} N_i(1)
    - \frac{1}{n} \sum_{i=1}^n \bigl(f_1(X_i) - f_2(X_i)\bigr) N_i(1)
    \biggr| \\
    &\leq
    \frac{1}{n}
    \sum_{i=1}^n
    \frac{|D S_n(\tilde f, T_i)(f_1 - f_2)|}{S_n(\tilde f, T_i)}
    + \|f_1 - f_2\|_\infty \\
    &\leq
    \frac{1}{n}
    \sum_{i=1}^n
    \frac{\bigl|\sum_{j=1}^n R_j(T_i)
    e^{\tilde f(X_j)} \bigl(f_1(X_j) - f_2(X_j)\bigr)\bigr|}
    {\sum_{j=1}^n R_j(T_i) e^{\tilde f(X_j)}}
    + \|f_1 - f_2\|_\infty \\
    &\leq
    2 \|f_1 - f_2\|_\infty.
  \end{align*}
  Moreover, again by Taylor's theorem and applying
  Lemmas~\ref{lem:derivatives_lstar}
  and~\ref{lem:derivatives_S}, there exists $\check f$ on the
  line segment between $f_1$ and $f_2$ such that
  \begin{align*}
    \bigl|
    \ell_\star(f_1) - \ell_\star(f_2)
    \bigr|
    &=
    \bigl|
    D \ell_\star(\check f)(f_1 - f_2)
    \bigr| \\
    &=
    \biggl|
    \int_{0}^{1}
    D S(\check f, t)(f_1 - f_2)
    \frac{S(f_0, t)}{S(\check f, t)}
    \lambda_0(t) \diffi t
    - \int_{0}^{1}
    D S(f_0, t)(f_1 - f_2)
    \lambda_0(t)
    \diffi t
    \biggr| \\
    &\leq
    \biggl|
    \int_{0}^{1}
    \frac{\int_\cX q(x, t) e^{\check f(x)}
    \bigl(f_1(x) - f_2(x)\bigr) \diffi P_X(x)}
    {\int_\cX q(x, t) e^{\check f(x)} \diffi P_X(x)}
    S(f_0, t)
    \lambda_0(t) \diffi t
    \biggr| \\
    &\qquad \qquad+
    \biggl|
    \int_{0}^{1}
    \int_\cX q(x, t) e^{f_0(x)} \bigl(f_1(x) - f_2(x)\bigr) \diffi P_X(x)
    \lambda_0(t)
    \diffi t
    \biggr| \\
    &\leq
    2 \|f_1 - f_2\|_\infty
    \int_{0}^{1}
    S(f_0, t)
    \lambda_0(t) \diffi t
    \leq
    2 \|f_1 - f_2\|_\infty
    \E\bigl(N(1)\bigr) \leq
    2 \|f_1 - f_2\|_\infty,
  \end{align*}
  as required.
\end{proof}

\begin{lemma}[Convergence of $\ell_n$]%
  \label{lem:local_convergence_ln}
  Let $f_1, f_2 \in \cB(\cX)$.
  There exists a universal constant $C > 0$
  such that for all $s \geq 1$, with
  probability at least $1 - e^{-s^2}$,
  \begin{align*}
    &\bigl|
    \ell_n(f_1) - \ell_\star(f_1)
    - \ell_n(f_2) + \ell_\star(f_2)
    \bigr| \\
    &\qquad\leq
    \frac{
      C e^{4\|f_0\|_\infty}
    e^{6(\|f_1 - f_0\|_\infty \lor \|f_2 - f_0\|_\infty)}}{q_1^2}
    \biggl(
      \frac{s \|f_1 - f_2\|_{L_2}}{\sqrt n}
      + \frac{s^2 \|f_1 - f_2\|_\infty}{n}
    \biggr).
  \end{align*}
\end{lemma}

\begin{proof}[Proof of Lemma~\ref{lem:local_convergence_ln}]
  By the mean value theorem applied to $\ell_n - \ell_\star$,
  there exists $\tilde a$ taking values in $[0, 1]$
  such that, defining $\tilde f_a \vcentcolon= a f_1 + (1 - a) f_2$,
  \begin{align*}
    \bigl|
    \ell_n(f_1) - \ell_\star(f_1)
    - \ell_n(f_2) + \ell_\star(f_2)
    \bigr|
    &=
    \bigl|
    D \ell_n(\tilde f_{\tilde a})(f_1 - f_2)
    - D \ell_\star(\tilde f_{\tilde a})(f_1 - f_2)
    \bigr| \\
    &\leq
    \sup_{a \in [0, 1]}
    \bigl|
    D \ell_n(\tilde f_{a})(f_1 - f_2)
    - D \ell_\star(\tilde f_{a})(f_1 - f_2)
    \bigr|.
  \end{align*}
  The result therefore follows
  by Lemma~\ref{lem:local_convergence_Dln}.
\end{proof}

\subsubsection{Proof of Theorem~\ref{thm:parameter_tuning}}

We use the result on local convergence of $\ell_n$ from
Lemma~\ref{lem:local_convergence_ln} to derive a general-purpose
cross-validation bound in Lemma~\ref{lem:general_cv}.
Next, Lemma~\ref{lem:lipschitz_fhat} shows that
$\gamma \mapsto \hat f_{n,\gamma}$ is Lipschitz with
respect to the supremum norm.
We then use these results to prove Theorem~\ref{thm:parameter_tuning}.

\begin{lemma}[Cross-validation on finite sets]%
  \label{lem:general_cv}
  Let $\cA$ be a non-empty finite set and suppose that for each
  $a \in \cA$, $f_a \in \cB(\cX)$
  with $P_X(f_a) = 0$ and $\|f_a - f_0\|_\infty \leq M$.
  Define $\hat a \vcentcolon= \argmin_{a \in \cA} \ell_n(f_a)$.
  Then there exists a universal constant $C > 0$ such that
  for all $s \geq 1$, with probability at least $1 - e^{-s^2}$,
  \begin{align*}
    \|f_{\hat a} - f_0\|_{L_2}
    &\leq
    \frac{C e^{5 \|f_0\|_\infty} e^{8 M} ( 1 + 1/\Lambda)}{q_1^3}
    \biggl(
      \min_{a \in \cA}
      \|f_{a} - f_0\|_{L_2}
      + \frac{s + \sqrt {\log |\cA|}}{\sqrt n}
    \biggr).
    \qedhere
  \end{align*}
\end{lemma}

\begin{proof}[Proof of Lemma~\ref{lem:general_cv}]
  Choose $a^\star \in \argmin_{a \in \cA} \|f_a - f_0\|_{L_2}$.
  By definition of $\hat a$, we have
  \begin{align}
    \nonumber
    0
    &\geq
    \ell_n(f_{\hat a}) - \ell_n(f_{a^\star}) \\
    \nonumber
    &=
    \bigl( \ell_n(f_{\hat a}) - \ell_\star(f_{\hat a})
    - \ell_n(f_0) + \ell_\star(f_0) \bigr)
    - \bigl( \ell_n(f_{a^\star}) - \ell_\star(f_{a^\star})
    - \ell_n(f_0) + \ell_\star(f_0) \bigr) \\
    \label{eq:cv_convex}
    &\quad+
    \bigl( \ell_\star(f_{\hat a}) - \ell_\star(f_0) \bigr)
    - \bigl( \ell_\star(f_{a^\star}) - \ell_\star(f_0) \bigr).
  \end{align}
  By Lemma~\ref{lem:local_convergence_ln} together with a union bound over
  $a \in \cA$, there exists a universal constant $C_1 > 0$
  such that for all $s \geq 1$, with probability at least
  $1 - e^{-s^2}$,
  \begin{align*}
    \bigl| \ell_n(f_{\hat a}) - \ell_\star(f_{\hat a})
    &- \ell_n(f_0) + \ell_\star(f_0) \bigr|
    + \bigl| \ell_n(f_{a^\star}) - \ell_\star(f_{a^\star})
    - \ell_n(f_0) + \ell_\star(f_0) \bigr| \\
    &\leq
    \frac{C_1 e^{4 \|f_0\|_\infty} e^{6 M}}{q_1^2}
    \biggl(
      \frac{s + \sqrt {\log |\cA|}}{\sqrt n}
      \|f_{\hat a} - f_0\|_{L_2}
      + \frac{\bigl(s^2 + \log |\cA|\bigr) M}{n}
    \biggr).
  \end{align*}
  For \eqref{eq:cv_convex}, we apply Taylor's theorem and
  Lemma~\ref{lem:strong_convexity} twice to obtain
  \begin{align*}
    \bigl( \ell_\star(f_{\hat a}) - \ell_\star(f_0) \bigr)
    &- \bigl( \ell_\star(f_{a^\star}) - \ell_\star(f_0) \bigr) \\
    &\geq
    e^{-\|f_0\|_\infty} e^{-2 M}
    \Lambda q_1
    \|f_{\hat a} - f_0\|_{L_2}^2
    - e^{\|f_0\|_\infty} e^{2 M}
    \Lambda
    \|f_{a^\star} - f_0\|_{L_2}^2.
  \end{align*}
  Therefore, simplifying, we have that with
  probability at least $1 - e^{-s^2}$,
  \begin{align*}
    \|f_{\hat a} - f_0\|_{L_2}^2
    &\leq
    \frac{C_1 e^{5 \|f_0\|_\infty} e^{8 M}}{q_1^3 \Lambda}
    \frac{s + \sqrt {\log |\cA|}}{\sqrt n}
    \|f_{\hat a} - f_0\|_{L_2} \\
    &\qquad+
    \frac{e^{2\|f_0\|_\infty} e^{4 M}}{q_1}
    \|f_{a^\star} - f_0\|_{L_2}^2
    + \frac{C_1 e^{5 \|f_0\|_\infty} e^{8 M}}{q_1^3 \Lambda}
    \frac{\bigl(s^2 + \log |\cA|\bigr) M}{n}.
  \end{align*}
  We now again use the fact that if $x^2 \leq a x + c$ for some $x,
  a, c \geq 0$, then $x \leq a + \sqrt{c}$. Thus, there exists a
  universal constant $C_2 > 0$ such that
  for all $s \geq 1$, with probability at least $1 - e^{-s^2}$,
  \begin{align*}
    \|f_{\hat a} - f_0\|_{L_2}
    &\leq
    \frac{C_2 e^{5 \|f_0\|_\infty} e^{8 M} ( 1 + 1/\Lambda)}{q_1^3}
    \biggl(
      \|f_{a^\star} - f_0\|_{L_2}
      + \frac{s + \sqrt {\log |\cA|}}{\sqrt n}
    \biggr).
    \qedhere
  \end{align*}
\end{proof}

\begin{lemma}[Lipschitz property of $\hat f_{n,\gamma}$]%
  \label{lem:lipschitz_fhat}
  Let $\gamma' \geq \gamma > 0$. Then
  \begin{align*}
    \bigl\|\hat f_{n,\gamma'} - \hat f_{n,\gamma}\bigr\|_{\infty}
    &\leq
    \frac{2 \sqrt K (\gamma' - \gamma) \sqrt{\log n}}{\gamma^{3/2}}.
  \end{align*}
\end{lemma}

\begin{proof}[Proof of Lemma~\ref{lem:lipschitz_fhat}]
  Since the result is clear if $\gamma' = \gamma$ or if $\hat
  f_{n,\gamma'} = \hat f_{n,\gamma}$, we may assume that $\gamma' >
  \gamma$ and $\|\hat f_{n,\gamma'} - \hat f_{n,\gamma}\|_{\cH} > 0$.
  By definition of $\hat f_{n,\gamma'}$ and Taylor's theorem,
  there exists $\tilde f$ on the line segment between
  $\hat f_{n,\gamma}$ and $\hat f_{n,\gamma'}$ such that
  \begin{align*}
    0
    &\geq
    \ell_{n,\gamma'}(\hat f_{n,\gamma'})
    - \ell_{n,\gamma'}(\hat f_{n,\gamma}) \\
    &=
    \ell_{n,\gamma}(\hat f_{n,\gamma'})
    - \ell_{n,\gamma}(\hat f_{n,\gamma})
    + (\gamma' - \gamma) \bigl(
      \|\hat f_{n,\gamma'}\|_\cH^2
      - \|\hat f_{n,\gamma}\|_\cH^2
    \bigr) \\
    &=
    D \ell_{n,\gamma}(\hat f_{n,\gamma})
    (\hat f_{n,\gamma'} - \hat f_{n,\gamma})
    + \frac{1}{2} D^2 \ell_{n,\gamma}(\tilde f)
    (\hat f_{n,\gamma'} - \hat f_{n,\gamma})^{\otimes 2}
    + (\gamma' - \gamma) \bigl(
      \|\hat f_{n,\gamma'}\|_\cH^2
      - \|\hat f_{n,\gamma}\|_\cH^2
    \bigr).
  \end{align*}
  Now $D \ell_{n,\gamma}(\hat f_{n,\gamma}) = 0$
  and $D^2 \ell_{n,\gamma}(\tilde f)(f_1,f_1) \geq 2 \gamma \|f_1\|_\cH^2$
  for $f_1 \in \cH$. Further,
  $\ell_{n,\gamma}(\hat f_{n,\gamma}) \leq \ell_{n,\gamma}(\hat f_{n,\gamma'})$
  and $\ell_{n,\gamma'}(\hat f_{n,\gamma'}) \leq
  \ell_{n,\gamma'}(\hat f_{n,\gamma})$,
  so
  $(\gamma' - \gamma) \|\hat f_{n,\gamma'}\|_\cH^2
  \leq (\gamma' - \gamma) \|\hat f_{n,\gamma}\|_\cH^2$ and hence
  $\|\hat f_{n,\gamma'}\|_\cH \leq \|\hat f_{n,\gamma}\|_\cH$.
  Therefore,
  \begin{align*}
    \gamma
    \|\hat f_{n,\gamma'} - \hat f_{n,\gamma}\|_\cH^2
    &\leq
    (\gamma' - \gamma) \bigl(
      \|\hat f_{n,\gamma}\|_\cH^2
      - \|\hat f_{n,\gamma'}\|_\cH^2
    \bigr) \\
    &\leq
    (\gamma' - \gamma)
    \|\hat f_{n,\gamma'} - \hat f_{n,\gamma}\|_\cH
    \bigl(
      \|\hat f_{n,\gamma}\|_\cH
      + \|\hat f_{n,\gamma'}\|_\cH
    \bigr) \\
    &\leq
    2 (\gamma' - \gamma)
    \|\hat f_{n,\gamma'} - \hat f_{n,\gamma}\|_\cH
    \|\hat f_{n,\gamma}\|_\cH,
  \end{align*}
  so
  \begin{align*}
    \|\hat f_{n,\gamma'} - \hat f_{n,\gamma}\|_\cH
    &\leq
    2 \frac{\gamma' - \gamma}{\gamma}
    \|\hat f_{n,\gamma}\|_\cH.
  \end{align*}
  Next, $S_n(f, T_i) \geq e^{f(X_i)}/n$ for $i \in [n]$
  and $f \in \cH$, so
  \begin{align*}
    \ell_n(f)
    &=
    \frac{1}{n} \sum_{i=1}^n \log\bigl(S_n(f, T_i)\bigr) N_i(1)
    - \frac{1}{n} \sum_{i=1}^n f(X_i) N_i(1)
    \geq - \log n.
  \end{align*}
  Similarly, $S_n(0_\cX, T_i) \leq 1$ so $\ell_n(0_\cX) \leq 0$.
  By definition of $\hat f_{n,\gamma}$, we have
  \[
    -\log n + \gamma \|\hat f_{n,\gamma}\|_\cH^2 \leq \ell_{n}(\hat
    f_{n,\gamma}) + \gamma \|\hat f_{n,\gamma}\|_\cH^2 =
    \ell_{n,\gamma}(\hat f_{n,\gamma}) \leq \ell_{n,\gamma}(0_\cX) =
    \ell_n(0_\cX) \leq 0,
  \]
  so
  $\|\hat f_{n,\gamma}\|_\cH^2 \leq \gamma^{-1}\log n$. Thus,
  by Lemma~\ref{lem:infty_bounds},
  \begin{align*}
    \|\hat f_{n,\gamma'} - \hat f_{n,\gamma}\|_\infty
    &\leq
    \sqrt K \|\hat f_{n,\gamma'} - \hat f_{n,\gamma}\|_\cH
    \leq
    \frac{2 \sqrt K (\gamma' - \gamma) \sqrt{\log n}}{\gamma^{3/2}}.
    \qedhere
  \end{align*}
\end{proof}

\begin{proof}[Proof of Theorem~\ref{thm:parameter_tuning}]
  We write $C^0_1, C^0_2, \ldots$ for positive quantities depending only on
  $\|f_0\|_\cH$, $\Lambda$, $q_1$,
  $\|1_\cX\|_\cH$, $K$ and $\xi$.
  Let $\Gamma'$ be an $n^{-3\xi}$-cover
  of $(\Gamma, |\cdot|)$ of cardinality at most $|\Gamma| \land n^{4\xi}$.
  For each $\gamma \in \Gamma$, by Lemma~\ref{lem:lipschitz_fhat} and
  since $\gamma^- \geq n^{-\xi}$,
  there exists $\gamma' \in \Gamma'$ with
  \begin{align}
    \label{eq:parameter_tuning_lipschitz}
    \bigl\|\hat f_{n,\gamma'} - \hat f_{n,\gamma}\bigr\|_\infty
    &\leq
    \frac{2 \sqrt K n^{-3\xi} \sqrt{\log n}}{n^{-3\xi/2}}
    \leq
    \frac{2 \sqrt K}{n}.
  \end{align}
  Recall that $\gamma \mapsto H_{\infty,\gamma}$ is decreasing on $(0,\infty)$
  and $\gamma \mapsto \gamma H_{\infty,\gamma}$ is increasing on $(0,\infty)$.
  Therefore, by Lemma~\ref{lem:infty_bounds}, Theorem~\ref{thm:rate},
  and a union bound over $\gamma \in \Gamma'$, there exists $C^0_1 > 0$
  such that if $\gamma^+ H_{\infty,\gamma^+} \leq 1 / C^0_1$
  and $1 \leq s \leq \sqrt{n} / (C^0_1 H_{\infty,\gamma^-})$,
  then with probability at least $1 - e^{-s^2}$,
  \begin{align*}
    \max_{\gamma \in \Gamma'}
    \bigl\| \hat f_{n,\gamma} - f_0 \bigr\|_{\infty}
    &\leq
    \max_{\gamma \in \Gamma'}
    \Bigl\{
      \sqrt{H_{\infty,\gamma}} \bigl\| \hat f_{n,\gamma} - f_0
      \bigr\|_{\cH_\gamma}
    \Bigr\} \\
    &\leq C^0_1
    \max_{\gamma \in \Gamma'}\biggl\{
      \bigl(s + \sqrt{\log |\Gamma'|}\bigr) \frac{H_{\infty,\gamma}}{\sqrt n}
      + \sqrt{\gamma H_{\infty,\gamma}}
    \biggr\}\\
    &\leq C^0_1
    \biggl\{
      \bigl(s + \sqrt{\log |\Gamma'|}\bigr) \frac{H_{\infty,\gamma^-}}{\sqrt n}
      + \sqrt{\gamma^+ H_{\infty,\gamma^+}}
    \biggr\}.
  \end{align*}
  We deduce that there exists $C^0_2 > 0$ such that if
  $\gamma^+ H_{\infty,\gamma^+} \leq 1 / C^0_2$,
  $H_{\infty,\gamma^-} \sqrt{\log |\Gamma'|} / \sqrt n \leq 1/C^0_2$ and
  $1 \leq s \leq \sqrt{n} / (C^0_2 H_{\infty,\gamma^-})$,
  then with probability at least $1 - e^{-s^2}$,
  \begin{align}
    \label{eq:parameter_tuning_infty}
    \max_{\gamma \in \Gamma'}
    \bigl\| \hat f_{n,\gamma} - f_0 \bigr\|_{\cH_\gamma} \leq
    \frac{1}{\sqrt{H_{\infty,\gamma^-}}} \quad \text{and} \quad
    \max_{\gamma \in \Gamma'}
    \bigl\| \hat f_{n,\gamma} - f_0 \bigr\|_{\infty}
    \leq 1.
  \end{align}
  Moreover, for every $\gamma \in \Gamma'$, we have
  \begin{align*}
    \bigl\|P_X(\hat f_{n,\gamma})1_{\mathcal{X}}\bigr\|_{\infty}
    &= \bigl|P_X(\hat f_{n,\gamma})\bigr|
    = \bigl|P_X(\hat f_{n,\gamma} - f_0)\bigr|
    \leq \bigl\|\hat{f}_{n,\gamma} - f_0\bigr\|_\infty.
  \end{align*}
  Applying Lemma~\ref{lem:general_cv} to $\tilde \ell_n$ and
  $\bigl\{\hat f_{n,\gamma} - P_X(\hat f_{n,\gamma}) 1_\cX: \gamma \in
  \Gamma'\bigr\}$
  conditionally on the training data, on the event where
  \eqref{eq:parameter_tuning_infty} holds, there exists
  $C^0_3 > 0$ such that with probability at least $1 - e^{-s^2}$,
  \begin{align}
    \label{eq:cv_centred_bound}
    \bigl\|\hat f_{n,\hat\gamma} - P_X(\hat f_{n,\hat\gamma}) 1_\cX -
    f_0\bigr\|_{L_2}
    &\leq
    C^0_3
    \biggl(
      \min_{\gamma \in \Gamma'}
      \bigl\|\hat f_{n,\gamma} - P_X(\hat f_{n,\gamma}) 1_\cX - f_0\bigr\|_{L_2}
      + \frac{s + \sqrt {\log |\Gamma'|}}{\sqrt n}
    \biggr).
  \end{align}
  By Lemmas~\ref{lem:uniform_convergence_Pn}
  and~\ref{lem:convergence_Pn},
  \eqref{eq:parameter_tuning_infty} and a union bound,
  there exists $C^0_4 > 0$ such that, for $s \geq 1$,
  with probability at least $1 - e^{-s^2}$,
  \begin{align}
    \label{eq:cv_PX_bound}
    \max_{\gamma \in \Gamma'}
    \bigl|P_X\bigl(\hat f_{n,\gamma}\bigr)\bigr|
    &\leq
    \max_{\gamma \in \Gamma'}
    \bigl|(P_n-P_X)\bigl(\hat f_{n,\gamma} - f_0\bigr)\bigr|
    + |(P_n - P_X)(f_0)|
    \leq \frac{C^0_4 s}{\sqrt n}.
  \end{align}
  Finally, by \eqref{eq:parameter_tuning_lipschitz},
  \begin{align}
    \label{eq:cv_min_bound}
    \min_{\gamma \in \Gamma'}
    \bigl\|\hat f_{n,\gamma} - f_0\bigr\|_{L_2}
    &\leq
    \min_{\gamma \in \Gamma}
    \bigl\|\hat f_{n,\gamma} - f_0\bigr\|_{L_2}
    + \frac{2 \sqrt K}{n}.
  \end{align}
  The conclusion follows from~\eqref{eq:cv_centred_bound},
  \eqref{eq:cv_PX_bound} and
  \eqref{eq:cv_min_bound}.
\end{proof}

\subsubsection{Proof of Theorem~\ref{thm:model_selection}}

In Lemma~\ref{lem:covering_simplex} we give a bound on the
$\ell_1$-covering number of the standard simplex.
We then present a Bernstein-type tail bound
for stochastic process suprema in Lemma~\ref{lem:bernstein}
and use it to establish a result on
cross-validation over the simplex in Lemma~\ref{lem:cv_simplex}.
These results are then used to prove Theorem~\ref{thm:model_selection}.

\begin{lemma}[Covering numbers of the standard simplex]%
  \label{lem:covering_simplex}
  Let $d \in \N$ and define the simplex
  $\Delta_d \vcentcolon=
  \bigl\{x \in [0, 1]^{d} : \|x\|_1 \leq 1\bigr\}$.
  Then for $\varepsilon > 0$, we have
  $N(\varepsilon, \Delta_d, \|\cdot\|_1) \leq (1 + 1/\varepsilon)^d$.
\end{lemma}

\begin{proof}[Proof of Lemma~\ref{lem:covering_simplex}]
  The volume ($d$-dimensional Lebesgue measure)
  of $\Delta_d$ is $\Vol_d(\Delta_d) = 1 / d!$,
  while the volume of
  $B \vcentcolon= \{x \in \R^d: \|x\|_1 \leq 1\}$
  is $\Vol_d(B) = 2^d / d!$.
  Writing $N_\pack$ for a packing number, we have
  \begin{align*}
    \frac{\varepsilon^d}{d!}
    N_\pack(\varepsilon, \Delta_d, \|\cdot\|_1)
    &\leq
    \Vol_d\bigl(\Delta_d + \varepsilon B / 2\bigr).
  \end{align*}
  If $u \in \Delta_d + \varepsilon B / 2$, then
  $u + \varepsilon 1_d/ 2 \geq 0_d$ and
  $\|u + \varepsilon 1_d/2\|_1 \leq 1 + \varepsilon$,
  so
  $\Vol_d(\Delta_d + \varepsilon B / 2)
  \leq (1 + \varepsilon)^d \Vol_d(\Delta_d)$.
  As covering numbers are bounded by their corresponding
  packing numbers,
  \begin{align*}
    N(\varepsilon, \Delta_d, \|\cdot\|_1)
    &\leq
    N_\pack(\varepsilon, \Delta_d, \|\cdot\|_1)
    \leq
    \Bigl(\frac{1}{\varepsilon} + 1\Bigr)^d.
    \qedhere
  \end{align*}
\end{proof}

\begin{lemma}[Stochastic process suprema with Bernstein tails]%
  \label{lem:bernstein}
  Let $\cT$ be a non-empty set equipped with two
  pseudo-metrics $d_1$ and $d_2$.
  Suppose that $\cT$ has finite diameter $D_1$ under~$d_1$
  and $D_2$ under $d_2$, respectively.
  Let $(X_t)_{t \in \cT}$ be a stochastic process that is
  separable in the sense of Lemma~\ref{lem:orlicz_chaining}.
  Suppose that for all $t, t' \in \cT$ and $s \geq 0$, we have
  \begin{align*}
    \P \bigl(
      |X_t - X_{t'}|
      > s
    \bigr)
    \leq
    2 \exp\biggl(- \frac{s}{d_1(t, t')}\biggr)
    + 2 \exp\biggl(- \frac{s^2}{d_2(t, t')^2}\biggr).
  \end{align*}
  Then there exists a universal constant $C > 0$ such that
  for any $t_0 \in \cT$ and all $s \geq 1$, we have
  \begin{align*}
    \sup_{t \in \cT}
    \big|
    X_t - X_{t_0}
    \big|
    \leq
    C \int_0^\infty
    \!\!
    \sqrt{\log N(\varepsilon, \cT, d_2)}
    \diffi \varepsilon
    + C \int_0^\infty
    \!
    \log N(\varepsilon, \cT, d_1)
    \diffi \varepsilon
    + C s D_2
    + C s^2 D_1,
  \end{align*}
  with probability at least $1 - e^{-s^2}$.
\end{lemma}

\begin{proof}[Proof of Lemma~\ref{lem:bernstein}]
  This follows by
  \citet[Theorem~3.5]{dirksen2015tail},
  using (2.3) therein to bound the
  Talagrand $\gamma_1$ and $\gamma_2$ functionals
  by the respective entropy integrals.
\end{proof}

\begin{lemma}[Cross-validation on the simplex]%
  \label{lem:cv_simplex}
  Let $f_0, \ldots, f_d \in \cB(\cX)$
  with $P_X(f_j) = 0$ and $\|f_j - f_0\|_\infty \leq M$ for each $j$.
  For $\theta \in \Delta_d$,
  write $\theta_0 \coloneqq 1 - \|\theta\|_1$ and let
  $f_\theta \coloneqq \sum_{j=0}^d \theta_j f_j$.
  Define $\hat \theta \vcentcolon= \argmin_{\theta \in \Delta_d}
  \ell_n(f_\theta)$.
  Then there exists a universal constant $C > 0$ such that
  for all $s \geq 1$, with probability at least $1 - e^{-s^2}$,
  \begin{align*}
    \|f_{\hat\theta} - f_0\|_{L_2}
    \leq
    \frac{C e^{5 \|f_0\|_\infty} e^{8M} (1 + 1/\Lambda) }{q_1^3}
    \biggl(
      \min_{\theta \in \Delta_d}
      \|f_\theta - f_0\|_{L_2}
      + \frac{s + \sqrt d}{\sqrt n}
    \biggr).
  \end{align*}
\end{lemma}

\begin{proof}[Proof of Lemma~\ref{lem:cv_simplex}]
  Let $\theta_\star \in \argmin_{\theta \in \Delta_d}
  \ell_\star(f_\theta)$
  and write $\delta \coloneqq \|f_{\theta_\star} - f_0\|_{L_2}$.
  For $\varepsilon > 0$, define
  $\Theta_\varepsilon \coloneqq
  \bigl\{
    \theta \in \Delta_d:
    \varepsilon < \|f_\theta - f_{\theta_\star}\|_{L_2} \leq 2 \varepsilon
  \bigr\}$.
  Write $R_n(\theta) \coloneqq \ell_n(f_\theta) - \ell_n(f_0)
  - \ell_\star(f_\theta) + \ell_\star(f_0)$ and
  observe that if $\hat\theta \in \Theta_\varepsilon$, then
  \begin{align}
    \nonumber
    0
    &\geq
    \inf_{\theta \in \Theta_\varepsilon}
    \bigl\{
      \ell_n(f_\theta) - \ell_n(f_{\theta_\star})
    \bigr\}
    =
    \inf_{\theta \in \Theta_\varepsilon}
    \bigl\{
      \ell_\star(f_\theta) - \ell_\star(f_{\theta_\star})
      + R_n(\theta)
      - R_n(\theta_\star)
    \bigr\} \\
    \label{eq:peeling}
    &\geq
    \inf_{\theta \in \Theta_\varepsilon}
    \bigl\{
      \ell_\star(f_\theta) - \ell_\star(f_{\theta_\star})
    \bigr\}
    - \sup_{\theta \in \Theta_\varepsilon} |R_n(\theta)|
    - |R_n(\theta_\star)|.
  \end{align}
  For $\theta \in \Theta_\varepsilon$,
  by Taylor's theorem,
  the definition of $f_{\theta_\star}$
  and Lemma~\ref{lem:strong_convexity},
  there exists $\tilde \theta$ on the line segment between
  $\theta$ and $\theta_\star$ with
  \begin{align}
    \nonumber
    \ell_\star(f_\theta) - \ell_\star(f_{\theta_\star})
    &=
    D \ell_\star(f_{\theta_\star})(f_\theta - f_{\theta_\star})
    +
    \frac{1}{2}
    D^2 \ell_\star(f_{\tilde \theta})
    (f_\theta - f_{\theta_\star})^{\otimes 2} \\
    \label{eq:peeling_curvature}
    &\geq
    \frac{1}{2}
    D^2 \ell_\star(f_{\tilde \theta})
    (f_\theta - f_{\theta_\star})^{\otimes 2}
    \geq
    \frac{e^{-2M} e^{-\|f_0\|_\infty}
    \Lambda q_1 \varepsilon^2}{2}.
  \end{align}
  By Lemma~\ref{lem:local_convergence_ln},
  there exists a universal constant $C > 0$
  such that for all $s \geq 1$,
  \begin{align}
    \label{eq:peeling_star}
    |R_n(\theta_\star)|
    &\leq
    \frac{C e^{4 \|f_0\|_\infty} e^{6 M}}{q_1^2}
    \biggl(
      \frac{s \delta}{\sqrt n}
      + \frac{s^2 M}{n}
    \biggr),
  \end{align}
  with probability at least $1 - e^{-s^2}$.
  Next, for $\theta, \theta' \in \Delta_d$, define
  $d_1(\theta, \theta') \coloneqq \|f_{\theta} - f_{\theta'}\|_\infty / n$
  so that
  $N(\xi, \Theta_\varepsilon, d_1) \leq
  N(\xi, \Delta_d, d_1)
  \leq N(n \xi / M, \Delta_d, \|\cdot\|_1)
  \leq \bigl(1 + M/(n \xi)\bigr)^d$.
  Since $d_1(\theta, \theta') \leq M \|\theta - \theta'\|_1 / n$,
  we have $N(\xi, \Theta_\varepsilon, d_1) = 1$
  for $\xi \geq M / n$, so it follows that
  \begin{align*}
    \int_0^\infty
    \log N(\xi, \Theta_\varepsilon, d_1)
    \diffi \xi
    \leq
    d
    \int_0^{M/n}
    \log \bigl(1 + M/(n \xi)\bigr)
    \diffi \xi
    &=
    \frac{d M}{n}
    \int_0^1
    \log (1 + 1/u)
    \diffi u \\
    &=
    \frac{d M \log 4}{n}.
  \end{align*}
  Similarly, define the metric
  $d_2(\theta, \theta') \coloneqq
  \|f_{\theta} - f_{\theta'}\|_{L_2} / \sqrt n$.
  Writing $A$ for the matrix with entries
  $A_{j k} \coloneqq \langle f_j, f_k \rangle_{L_2}$
  for $j, k \in \{0, \ldots, d\}$,
  we have that
  $d_2(\theta, \theta')^2
  = (\theta - \theta')^\T A (\theta - \theta') / n$
  and $\Theta_\varepsilon \subseteq
  \bigl\{\theta = (\theta_0, \ldots, \theta_d):
    (\theta - \theta_\star)^\T
    A (\theta - \theta_\star)
    \leq 4 \varepsilon^2
  \bigr\}$.
  Thus, by considering
  $\psi \coloneqq A^{1/2} (\theta - \theta_\star) / (2 \varepsilon)$,
  with $B_{d+1}$ the $\|\cdot\|_2$-unit ball in $\R^{d+1}$,
  we have that
  $N(\xi, \Theta_\varepsilon, d_2) \leq
  N(\sqrt n \xi / (2 \varepsilon), B_{d+1}, \|\cdot\|_2)
  \leq \bigl(1 + 4 \varepsilon / (\sqrt{n} \xi)\bigr)^{d+1}$
  by, e.g., \citet[Exercise~2.5.9]{talagrand2021upper}.
  Since $d_2(\theta, \theta') \leq 4 \varepsilon / \sqrt n$
  for all $\theta, \theta' \in \Theta_\varepsilon$, we have
  $N(\xi, \Theta_\varepsilon, d_2) = 1$
  whenever $\xi \geq 4 \varepsilon / \sqrt n$, so
  by Jensen's inequality,
  \begin{align*}
    \int_0^\infty
    \sqrt{\log N(\xi, \Theta_\varepsilon, d_2)}
    \diffi \xi
    &\leq
    \sqrt{d + 1}
    \int_0^{4 \varepsilon / \sqrt n}
    \sqrt{\log\bigl(1 + 4 \varepsilon / (\sqrt{n} \xi)\bigr)}
    \diffi \xi \\
    &\leq
    4 \varepsilon \sqrt{\frac{d + 1}{n}}
    \int_0^1
    \sqrt{\log(2/u)}
    \diffi u
    \leq
    6 \varepsilon \sqrt{\frac{d + 1}{n}}.
  \end{align*}
  By Lemma~\ref{lem:local_convergence_ln},
  for each $\theta, \theta' \in \Theta_\varepsilon$,
  with probability at least $1 - e^{-s^2}$,
  \begin{align*}
    | R_n(\theta) - R_n(\theta') |
    &\leq
    \frac{C e^{4 \|f_0\|_\infty} e^{6 M}}{q_1^2}
    \biggl(
      \frac{s \|f_{\theta} - f_{\theta'}\|_{L_2}}{\sqrt n}
      + \frac{s^2 \|f_\theta - f_{\theta'}\|_\infty}{n}
    \biggr) \\
    &=
    \frac{C e^{4 \|f_0\|_\infty} e^{6 M}}{q_1^2}
    \bigl(
      s d_2(\theta, \theta')
      + s^2 d_1(\theta, \theta')
    \bigr).
  \end{align*}
  Therefore, by Lemma~\ref{lem:bernstein},
  there exists a universal constant $C > 0$ such that for all $s \geq 1$,
  with $\theta' \in \Theta_\varepsilon$ fixed and
  with probability at least $1 - e^{-s^2}$,
  \begin{align*}
    \sup_{\theta \in \Theta_\varepsilon}
    |R_n(\theta) - R_n(\theta')|
    &\leq
    \frac{C e^{4 \|f_0\|_\infty} e^{6 M}}{q_1^2}
    \biggl(
      \frac{\varepsilon(s + \sqrt d)}{\sqrt n}
      + \frac{M (s^2 + d)}{n}
    \biggr).
  \end{align*}
  Applying Lemma~\ref{lem:local_convergence_ln}
  again to $R_n(\theta')$, using \eqref{eq:peeling_star} and
  by increasing $C$ if necessary, we have
  \begin{align*}
    \sup_{\theta \in \Theta_\varepsilon}
    |R_n(\theta)|
    + |R_n(\theta_\star)|
    &\leq
    \frac{C e^{4 \|f_0\|_\infty} e^{6 M}}{q_1^2}
    \biggl(
      \frac{
        \varepsilon (s + \sqrt d) + s \delta
      }{\sqrt n}
      + \frac{M (s^2 + d)}{n}
    \biggr),
  \end{align*}
  with probability at least $1 - e^{-s^2}$.
  By~\eqref{eq:peeling} and~\eqref{eq:peeling_curvature},
  writing
  $C_0 \coloneqq e^{8M} e^{5 \|f_0\|_\infty} q_1^{-3} (1 + 1/\Lambda)$
  so that $C_0 \geq 1 \lor M$ and taking $\delta_n \coloneqq C C_0 \sqrt{d/n}$
  and $s \coloneqq (\delta + \varepsilon) \sqrt n / (C C_0)$, by
  increasing $C$ if necessary,
  \begin{align*}
    \P\bigl(\hat\theta \in \Theta_{\delta + \delta_n + \varepsilon}\bigr)
    &\leq
    \P\biggl(
      \sup_{\theta \in \Theta_{\delta + \delta_n + \varepsilon}}
      |R_n(\theta)|
      + |R_n(\theta_\star)|
      \geq
      \frac{
        e^{-2M}
        e^{-\|f_0\|_\infty}
      \Lambda q_1}{2}
      \bigl(\delta + \delta_n + \varepsilon\bigr)^2
    \biggr) \\
    &\leq
    \P\biggl(
      \sup_{\theta \in \Theta_{\delta + \delta_n + \varepsilon}}
      |R_n(\theta)|
      + |R_n(\theta_\star)|
      \geq
      \frac{C e^{4\|f_0\|_\infty} e^{6M}}{q_1^2}
      \frac{(\delta + \delta_n + \varepsilon)^2}{2 C C_0}
    \biggr) \\
    &\leq
    \exp\biggl(
      -\frac{(\delta + \varepsilon)^2 n}{C^2 C_0^2}
    \biggr)
    \leq
    \exp\biggl(
      -\frac{\varepsilon^2 n}{C^2 C_0^2}
    \biggr).
  \end{align*}
  Thus,
  by the geometric series formula and as
  $\{x/(1-x)\} \land 1 \leq 2 x$ for all $x \in [0, 1)$,
  \begin{align*}
    &\P\bigl(
      \|f_{\hat\theta} - f_{\theta_\star}\|_{L_2}
      > \delta + \delta_n + \varepsilon
    \bigr)
    =
    \P\biggl(
      \hat\theta \in
      \bigcup_{j=0}^\infty
      \Theta_{2^j (\delta + \delta_n + \varepsilon)}
    \biggr)
    = \sum_{j=0}^\infty
    \P\bigl(
      \hat\theta \in
      \Theta_{2^j (\delta + \delta_n + \varepsilon)}
    \bigr) \\
    &\qquad=
    \sum_{j=0}^\infty
    \P\bigl(
      \hat\theta \in
      \Theta_{\delta + \delta_n + (2^j-1) (\delta + \delta_n) + 2^j \varepsilon}
    \bigr)
    \leq
    1 \land
    \sum_{j=0}^\infty
    \exp\biggl(
      -\frac{2^{2 j} \varepsilon^2 n}{C^2 C_0^2}
    \biggr) \\
    &\qquad\leq
    1 \land
    \sum_{j=1}^\infty
    \exp\biggl(
      - \frac{j \varepsilon^2 n}
      {C^2 C_0^2}
    \biggr)
    \leq
    2 \exp\biggl(
      - \frac{\varepsilon^2 n}
      {C^2 C_0^2}
    \biggr).
  \end{align*}
  Taking $\varepsilon \coloneqq C^2 C_0 s / \sqrt{n}$
  and then increasing $C$
  yields
  \begin{align}
    \label{eq:peeling_tail}
    \P\biggl(
      \|f_{\hat\theta} - f_{\theta_\star}\|_{L_2}
      >
      \delta
      + \frac{C^2 \bigl(s + \sqrt d\bigr)
        e^{8M} e^{5 \|f_0\|_\infty}
        (1 + 1/\Lambda)
      }{q_1^3 \sqrt n}
    \biggr)
    &\leq
    2 e^{- C^2 s^2}
    \leq e^{-s^2},
  \end{align}
  for all $s \geq 1$.  Writing $\check\theta \in \argmin_{\theta \in \Delta_d}
  \|f_\theta - f_0\|_{L_2}$,
  by definition of $\theta_\star$,
  Taylor's theorem and
  Lemma~\ref{lem:strong_convexity},
  for some $\tilde f$ on the line segment between
  $f_0$ and $f_{\theta_\star}$,
  \begin{align*}
    \ell_\star(f_{\check\theta})
    - \ell_\star(f_0)
    &\geq
    \ell_\star(f_{\theta_\star})
    - \ell_\star(f_0)
    =
    \frac{1}{2}
    D^2 \ell_\star(\tilde f)
    (f_{\theta_\star}- f_0)^{\otimes 2}
    \geq
    \frac{1}{2}
    e^{-2M}
    e^{-\|f_0\|_\infty}
    \Lambda q_1
    \delta^2.
  \end{align*}
  On the other hand, by Taylor's theorem and
  Lemma~\ref{lem:strong_convexity} again,
  for some $\check f$ on the line segment between
  $f_0$ and $f_{\check\theta}$,
  \begin{align*}
    \ell_\star(f_{\check\theta})
    - \ell_\star(f_0)
    =
    \frac{1}{2}
    D^2 \ell_\star(\check f)
    (f_{\check\theta}- f_0)^{\otimes 2}
    \leq
    \frac{1}{2}
    e^{2M}
    e^{\|f_0\|_\infty}
    \Lambda
    \min_{\theta \in \Delta_d}
    \|f_\theta - f_0\|_{L_2}^2.
  \end{align*}
  Combining the last two displays yields
  \begin{align}
    \label{eq:peeling_delta}
    \delta
    \leq
    \frac{e^{2 M} e^{\|f_0\|_\infty}}{\sqrt{q_1}}
    \min_{\theta \in \Delta_d}
    \|f_\theta - f_0\|_{L_2}.
  \end{align}
  After increasing the universal constant $C > 0$ if necessary,
  we conclude by \eqref{eq:peeling_tail}
  and~\eqref{eq:peeling_delta} that for all $s \geq 1$,
  \begin{align*}
    \|f_{\hat\theta} - f_0\|_{L_2}
    \leq
    \frac{C e^{5 \|f_0\|_\infty} e^{8M} (1 + 1/\Lambda) }{q_1^3}
    \biggl(
      \min_{\theta \in \Delta_d}
      \|f_\theta - f_0\|_{L_2}
      + \frac{s + \sqrt d}{\sqrt n}
    \biggr),
  \end{align*}
  with probability at least $1 - e^{-s^2}$, as required.
\end{proof}

\begin{proof}[Proof of Theorem~\ref{thm:model_selection}]
  We write $C^0_1, C^0_2, \ldots$ for positive quantities depending only on
  $\|f_0\|_\cH$, $\Lambda$, $q_1$, $\|1_\cX\|_\cH$,
  $M$, $K$ and $\xi$,
  recalling from Lemma~\ref{lem:infty_bounds}
  that $\|f_0\|_\infty^2 \leq K \|f_0\|_\cH^2$.
  For $\gamma \in \Gamma$, let
  $\check\theta(\gamma) \coloneqq \sargmin_{\theta \in \Theta}
  \tilde \ell_n(\check f_{n,\gamma,\theta})$.
  Recall that $\max_{m \in \cM} \|\tilde f_m\|_\infty \leq M$ almost surely
  and $\gamma^- \geq n^{-\xi}$.
  Further, by the proof of Theorem~\ref{thm:parameter_tuning},
  there exists $C^0_1 > 0$ such that
  if $\gamma^+ H_{\infty,\gamma^+} \leq 1 / C^0_1$,
  $H_{\infty,\gamma^-} \sqrt{(\log |\Gamma|) \land \log n} / \sqrt n
  \leq 1/C^0_1$ and
  $1 \leq s \leq \sqrt{n} / (C^0_1 H_{\infty,\gamma^-})$,
  then $\max_{\gamma \in \Gamma}
  \| \hat f_{n,\gamma} - f_0 \|_\infty \leq 1$ on an event $\cA_n$ with
  probability at least $1 - e^{-s^2}$.
  Applying Lemma~\ref{lem:cv_simplex} on $\cA_n$
  to $\tilde \ell_n$ and
  $\bigl\{\check f_{n,\gamma,\theta} - P_X(\check f_{n,\gamma,\theta}):
  \theta \in \Theta\bigr\}$,
  conditionally on the training data and on
  $(\hat f_m: m \in \cM)$,
  there exists $C^0_2 > 0$ such that
  for each $\gamma \in \Gamma$,
  with probability at least $1 - e^{-s^2}$,
  \begin{align*}
    \bigl\|\check f_{n,\gamma,\check\theta(\gamma)}
    - P_X(\check f_{n,\gamma,\check\theta(\gamma)}) - f_0\bigr\|_{L_2}
    &\leq
    C^0_2
    \biggl(
      \min_{\theta \in \Theta}
      \bigl\|\check f_{n,\gamma,\theta}
      - P_X(\check f_{n,\gamma,\theta}) - f_0\bigr\|_{L_2}
      + \frac{s + \sqrt{|\cM|}}{\sqrt n}
    \biggr).
  \end{align*}
  Let $\Gamma' \subseteq \Gamma$ be as defined in the proof of
  Theorem~\ref{thm:parameter_tuning}.
  By Lemma~\ref{lem:lipschitz_fhat},
  for any $\gamma \in \Gamma$, there exists $\gamma' \in \Gamma'$ with
  \begin{align*}
    \sup_{\theta \in \Theta}
    \bigl\|\check f_{n,\gamma',\theta}
    - \check f_{n,\gamma,\theta}\bigr\|_{\infty}
    &=
    \bigl\| \hat f_{n,\gamma'} -  \hat f_{n,\gamma} \bigl\|_\infty
    \leq
    \frac{2 \sqrt K (\gamma' - \gamma) \sqrt{\log n}}{\gamma^{3/2}}
    \leq
    \frac{2 \sqrt{K \log n}}{n^{3\xi/2}}
    \leq
    \frac{2 \sqrt{K}}{n}.
  \end{align*}
  By Lemma~\ref{lem:general_cv},
  as $|\Gamma'| \leq |\Gamma| \land n^{4\xi}$,
  with probability at least $1 - e^{-s^2}$,
  \begin{align}
    \nonumber
    &\bigl\|\check f_{n,\check\gamma,\check\theta}
    - P_X(\check f_{n,\check\gamma,\check\theta}) - f_0\bigr\|_{L_2} \\
    \label{eq:peeling_centred}
    &\qquad\leq
    C^0_3
    \biggl(
      \min_{(\gamma,\theta) \in \Gamma' \times \Theta}
      \bigl\|\check f_{n,\gamma,\theta}
      - P_X(\check f_{n,\gamma,\theta}) - f_0\bigr\|_{L_2}
      + \frac{s + \sqrt{|\cM| + \log n}}{\sqrt n}
    \biggr).
  \end{align}
  Observe that
  \[
    P_X(\tilde f_m) = P_X\bigl(\hat f_m - P_n(\hat f_m) 1_\cX\bigr)
    = (P_X - P_n)(\hat f_m)
  \]
  and $\|\hat f_m\|_\infty \leq M$
  almost surely for each $m \in \cM$.
  Therefore, applying \eqref{eq:cv_PX_bound},
  Hoeffding's inequality and a union bound,
  there exists $C^0_4 > 0$
  such that with probability at least $1 - e^{-s^2}$,
  \begin{align}
    \label{eq:peeling_PX}
    \max_{(\gamma,\theta) \in \Gamma' \times \Theta}
    \bigl|P_X\bigl(\check f_{n,\gamma,\theta}\bigr)\bigr|
    &\leq
    \max_{\gamma \in \Gamma'}
    \bigl|P_X\bigl(\hat f_{n,\gamma}\bigr)\bigr|
    + \max_{m \in \cM}
    \bigl|P_X\bigl(\tilde f_m\bigr)\bigr|
    \leq
    C^0_4
    \frac{s + \sqrt{\log |\cM|}}{\sqrt n}.
  \end{align}
  The result follows from~\eqref{eq:peeling_centred} and~\eqref{eq:peeling_PX}.
\end{proof}

\subsection{Proofs for Section~\ref{sec:implementation}}

For the reader's convenience, we state a useful result for
computing the Hilbert space norm of a constant function in
Lemma~\ref{lem:paulsen}.

\begin{lemma}[Conditions for $1_\cX \in \cH$]%
  \label{lem:paulsen}
  Let $\cH$ be an RKHS on $\cX$
  with kernel $k$. Then $1_\cX \in \cH$
  if and only if there exists $s > 0$ such that
  $(x, y) \mapsto k(x, y) - s$ is a kernel on $\cX$,
  in which case $1/\|1_\cX\|_\cH^2 = \sup\{s > 0:k(\cdot,\cdot) - s
  \text{ is a kernel on } \mathcal{X}\}$.
\end{lemma}

\begin{proof}[Proof of Lemma~\ref{lem:paulsen}]
  See \citet[Theorem~3.11]{paulsen2016introduction},
  setting $f = 1_\cX$ and $s = 1/c^2$ in their notation.
\end{proof}

\begin{proof}[Proof of Proposition~\ref{prop:constant_norm}]
  We begin by showing equivalence of the two expressions for $\kappa(\bA)$.
  Let $v_1, \ldots, v_n \in \R^n$ form an orthonormal basis of
  eigenvectors of $\bA$, with associated eigenvalues
  $\lambda_1, \ldots, \lambda_n \geq 0$. For $\delta > 0$, the
  matrix $\bA + \delta \bI_n$ has the same eigenvectors, but now with
  eigenvalues
  $\lambda_1 + \delta, \ldots, \lambda_n + \delta > 0$, so we may
  write $(\bA + \delta \bI_n)^{-1}
  = \sum_{i=1}^n \frac{1}{\lambda_i + \delta} v_i v_i^\T$.
  Therefore
  $1_n^\T (\bA + \delta \bI_n)^{-1} 1_n
  = \sum_{i=1}^n \frac{1}{\lambda_i + \delta} (1_n^\T v_i)^2$.
  We now proceed by cases. Suppose first that there exists $v \in \R^n$
  with $\bA v = 0$ and $1_n^\T v = 1$. Then we can find $i_0 \in [n]$
  with $\lambda_{i_0} = 0$ and $1_n^\T v_{i_0} \neq 0$. Thus
  $\frac{1}{\lambda_{i_0} + \delta} (1_n^\T v_{i_0})^2 \to \infty$
  as $\delta \searrow 0$, so $1_n^\T (\bA + \delta \bI_n)^{-1} 1_n \to
  \infty$ as $\delta \searrow 0$.
  Otherwise, for all $i \in [n]$, either
  $\lambda_i > 0$ or $1_n^\T v_i = 0$. In this case,
  $\sum_{i=1}^n \frac{1}{\lambda_i + \delta} (1_n^\T v_i)^2
  \to \sum_{i=1}^n \lambda_i^+ (1_n^\T v_i)^2
  = 1_n^\T \sum_{i=1}^n \lambda_i^+ v_i v_i^\T 1_n^\T
  = 1_n^\T \bA^+ 1_n$
  as $\delta \searrow 0$,
  where $\lambda_i^+ \vcentcolon= 1/\lambda_i$ if $\lambda_i > 0$ and
  $\lambda_i^+ \vcentcolon= 0$ if $\lambda_i = 0$. This verifies the
  desired equivalence of the two expressions for $\kappa(\bA)$.

  Next, we show that
  \begin{align}
    \label{eq:pos_def}
    \kappa(\bA)
    = \sup \bigl\{
      s \geq 0:
      \bA - s 1_n 1_n^\T \succeq 0
    \bigr\}.
  \end{align}
  Suppose for now that $\bA$ is positive definite, so invertible.
  Then for $s > 0$, we have
  $\bA - s 1_n 1_n^\T \succeq 0$ if and only if
  $v^\T 1_n 1_n^\T v/v^\T \bA v \leq 1/s$ for all $v \in \R^n
  \setminus \{0\}$.
  Equivalently, defining $w \vcentcolon= \bA^{1/2}v$,
  this holds precisely when
  \[
    \frac{1}{s} \geq \sup_{w \in \R^n \setminus \{0\}}
    \frac{w^\T \bA^{-1/2} 1_n 1_n^\T \bA^{-1/2} w}{\|w\|_2^2}
    = 1_n^\T \bA^{-1} 1_n.
  \]
  It follows that the largest such $s$ is given by
  $1/1_n^\T \bA^{-1} 1_n$.
  We now consider the case where~$\bA$ is not invertible.
  Define $g: [0, \infty)^2 \to \R$ by
  $g(s, \delta) \vcentcolon=
  \lambda_{\min}(\bA + \delta \bI_n - s 1_n 1_n^\T)$,
  which is continuous in $(s, \delta)$ by Weyl's inequality
  \citep[e.g.,][Theorem~A.10.17]{samworth2026modern}.
  Further, $s \mapsto g(s,\delta)$ is decreasing for each $\delta
  \geq 0$ and $\delta \mapsto g(s,\delta)$ is increasing for each $s \geq 0$.
  For $\delta \geq 0$, define
  $s^\star(\delta) \vcentcolon= \sup\{s \geq 0: g(s, \delta) \geq 0\}$,
  noting that $s^\star$ is an increasing function on $[0,\infty)$,
  that $s^\star(\delta) = 1/1_n^\T (\bA + \delta \bI_n)^{-1} 1_n$ for
  $\delta > 0$, and $g\bigl(s^\star(\delta),\delta\bigr) \geq 0$ by
  continuity. Now take $\varepsilon > 0$, so that
  $g\bigl(s^\star(0)+\varepsilon, 0\bigr) < 0$. By continuity, there
  exists $\delta > 0$ such that
  $g\bigl(s^\star(0)+\varepsilon, \delta\bigr) < 0$, so
  $s^\star(\delta) < s^\star(0)+\varepsilon$.
  But since $\varepsilon > 0$ was arbitrary, we deduce that
  \[
    s^\star(0) = \lim_{\delta \searrow 0} s^\star(\delta)
    = \lim_{\delta \searrow 0} \frac{1}{1_n^\T (\bA + \delta \bI_n)^{-1}
    1_n} = \kappa(\bA).
  \]
  This establishes~\eqref{eq:pos_def}.

  Now $k(x, y) - s$ is a kernel if and only if
  $\bK(x_1, \ldots, x_n) - s 1_n 1_n^\T \succeq 0$
  for all $n \in \N$ and $x_1, \ldots, x_n \in \cX$,
  or equivalently if
  $s \leq \inf_{n \in \N}
  \inf_{x_1, \ldots, x_n \in \cX}
  \kappa\bigl(\bK(x_1, \ldots, x_n)\bigr)$.
  The result follows by Lemma~\ref{lem:paulsen}.
\end{proof}

\subsection{Proofs for Appendix~\ref{sec:kernel_selection}}

\begin{proof}[Proof of Proposition~\ref{prop:kernel_selection}]
  We write $C^0_1, C^0_2, \ldots$ for quantities depending only on
  $\max_{k \in \cK} \|f_0\|_{\cH_k}$, $\Lambda$, $q_1$,
  $\max_{k \in \cK} \|1_\cX\|_{\cH_k}$, $\max_{k \in \cK} K_k$ and $\xi$.
  Let $\Gamma'$ be as in the proof of Theorem~\ref{thm:parameter_tuning}.
  For each $\gamma \in \Gamma$ and $k \in \cK$,
  by \eqref{eq:parameter_tuning_lipschitz},
  there exists $\gamma' \in \Gamma'$ with
  \begin{align*}
    \bigl\|\hat f_{n,\gamma',k} - \hat f_{n,\gamma,k}\bigr\|_\infty
    &\leq
    \frac{2 \sqrt K}{n}.
  \end{align*}
  By Lemma~\ref{lem:infty_bounds}, Theorem~\ref{thm:rate}
  and a union bound over $\gamma \in \Gamma'$ and $k \in \cK$,
  there exists $C^0_1 > 0$
  such that if $\gamma^+ H_{\infty,\gamma^+,\cK} \leq 1 / C^0_1$
  and $1 \leq s \leq \sqrt{n} / (C^0_1 H_{\infty,\gamma^-,\cK})$,
  then with probability at least $1 - e^{-s^2}$,
  \begin{align*}
    \max_{\gamma \in \Gamma',\, k \in \cK}
    \bigl\| \hat f_{n,\gamma,k} - f_0 \bigr\|_{\infty}
    &\leq
    \max_{\gamma \in \Gamma',\, k \in \cK}
    \Bigl\{
      \sqrt{H_{\infty,\gamma,k}} \bigl\| \hat f_{n,\gamma,k} - f_0
      \bigr\|_{\cH_{k,\gamma}}
    \Bigr\} \\
    &\leq C^0_1
    \max_{\gamma \in \Gamma',\, k \in \cK}
    \biggl\{
      \bigl(s + \sqrt{\log |\cK| + \log |\Gamma'|}\bigr)
      \frac{H_{\infty,\gamma,k}}{\sqrt n}
      + \sqrt{\gamma H_{\infty,\gamma,k}}
    \biggr\}\\
    &\leq C^0_1
    \biggl\{
      \bigl(s + \sqrt{\log |\cK| + \log |\Gamma'|}\bigr)
      \frac{H_{\infty,\gamma^-,\cK}}{\sqrt n}
      + \sqrt{\gamma^+ H_{\infty,\gamma^+,\cK}}
    \biggr\}.
  \end{align*}
  Thus, there exists $C^0_2 > 0$ such that if
  $\gamma^+ H_{\infty,\gamma^+,\cK} \leq 1 / C^0_2$,
  $H_{\infty,\gamma^-,\cK} \sqrt{\log |\cK| + \log |\Gamma'|}
  / \sqrt n \leq 1/C^0_2$ and
  $1 \leq s \leq \sqrt{n} / (C^0_2 H_{\infty,\gamma^-,\cK})$,
  then with probability at least $1 - e^{-s^2}$,
  \begin{align}
    \label{eq:kernel_tuning_infty}
    \max_{\gamma \in \Gamma',\, k \in \cK}
    \bigl\| \hat f_{n,\gamma,k} - f_0 \bigr\|_{\infty}
    \leq 1.
  \end{align}
  Moreover, for every $\gamma \in \Gamma'$ and $k \in \cK$, we have
  \begin{align*}
    \bigl\|P_X(\hat f_{n,\gamma,k})1_{\mathcal{X}}\bigr\|_{\infty}
    &= \bigl|P_X(\hat f_{n,\gamma,k})\bigr|
    = \bigl|P_X(\hat f_{n,\gamma,k} - f_0)\bigr|
    \leq \bigl\|\hat{f}_{n,\gamma,k} - f_0\bigr\|_\infty.
  \end{align*}
  Applying Lemma~\ref{lem:general_cv} to $\tilde \ell_n$ and
  $\bigl\{\hat f_{n,\gamma,k} - P_X(\hat f_{n,\gamma,k}) 1_\cX: \gamma \in
  \Gamma', k \in \cK\bigr\}$
  conditionally on the training data, on the event where
  \eqref{eq:kernel_tuning_infty} holds, there exists
  $C^0_3 > 0$ such that with probability at least $1 - e^{-s^2}$,
  \begin{align*}
    &\bigl\|\hat f_{n,\hat\gamma_\cK, \hat k}
    - P_X(\hat f_{n,\hat\gamma_\cK, \hat k}) 1_\cX - f_0\bigr\|_{L_2} \\
    &\qquad\leq
    C^0_3
    \biggl(
      \min_{\gamma \in \Gamma',\, k \in \cK}
      \bigl\|\hat f_{n,\gamma,k} - P_X(\hat f_{n,\gamma,k}) 1_\cX
      - f_0\bigr\|_{L_2}
      + \frac{s + \sqrt {\log |\cK| + \log |\Gamma'|}}{\sqrt n}
    \biggr).
  \end{align*}
  The result now follows by arguing as in the proof of
  Theorem~\ref{thm:parameter_tuning}.
\end{proof}

\subsection{Inequalities for \texorpdfstring{$U$}{U}-processes}
\label{sec:u_statistics}

Let $(\cX,\cA)$ be a measurable space and $X_1, \ldots, X_n$ be independent
and identically distributed random variables taking values in a
measurable space $\cX$
where $n \geq 2$.
Let $\cF$ be a non-empty class of bounded, symmetric, Borel measurable functions
$f : \cX \times \cX \to \R$ that is separable with respect to the
topology of pointwise convergence. In other words, there exists a
countable subset $\mathcal{F}' \subseteq \mathcal{F}$ such that for
every $f \in \mathcal{F}$, there exists a sequence $(f_n)$ in
$\mathcal{F}'$ with $f_n(x) \rightarrow f(x)$ as $n \rightarrow
\infty$ for every $x \in \mathcal{X}$. Assume that $\E \bigl\{
f(X_1, x) \bigr\} = 0$ for all $x \in \cX$ and $f \in \cF$.
Suppose that there exists a function $F : \cX \times \cX \to \R$
with $f(x, y) \leq F(x, y)$ for all $x, y \in \cX$ and $f \in \cF$,
and further $\E \bigl\{ F(X_1, X_2)^2 \bigr\} < \infty$.
Define $\mathcal{I} := \{(i,j):i,j \in [n],i<j\}$ and
for $f \in \cF$, define the degenerate second-order $U$-process
and its decoupled version by
\begin{align*}
  U_n(f) &:= \frac{2}{n(n-1)}
  \sum_{(i,j) \in \mathcal{I}}
  f(X_i, X_j),
  &V_n(f) &:=
  \frac{2}{n(n-1)}
  \sum_{(i,j) \in \mathcal{I}}
  f(X_i, X_j'),
\end{align*}
where $X_1', \ldots, X_n'$ are independent copies of $X_1, \ldots,
X_n$. Now let $\varepsilon_1, \ldots, \varepsilon_n,
\varepsilon_1', \ldots, \varepsilon_n'$ denote independent and
identically distributed Rademacher variables that are independent of
$X_1, \ldots, X_n, X_1', \ldots, X_n'$.
Define randomised versions of $V_n$ with and without decoupling by
\begin{align*}
  W_n(f)
  &:=
  \frac{2}{n(n-1)}
  \sum_{(i,j) \in \mathcal{I}}
  \varepsilon_i \varepsilon_j
  f(X_i, X_j),
  &Y_n(f)
  &:=
  \frac{2}{n(n-1)}
  \sum_{(i,j) \in \mathcal{I}}
  \varepsilon_i \varepsilon_j'
  f(X_i, X_j').
\end{align*}
For a stochastic process $\bigl(Z(f)\bigr)_{f \in \mathcal{F}}$, we
write $\|Z\|_\cF := \sup_{f \in \cF} |Z(f)|$.

\begin{lemma}[Decoupling]%
  \label{lem:decoupling}
  Let $\phi: [0, \infty) \to [0, \infty)$ be convex and increasing. Then
  \begin{align*}
    \E \biggl\{
      \phi \biggl(
        \frac{1}{4}
        \|V_n\|_\cF
      \biggr)
    \biggr\}
    &\leq
    \E \bigl\{
      \phi \bigl(
        \|U_n\|_\cF
    \bigr)\bigr\}
    \leq
    \E \bigl\{
      \phi \bigl(
        8\|V_n\|_\cF
    \bigr)\bigr\}.
  \end{align*}
\end{lemma}

\begin{proof}[Proof of Lemma~\ref{lem:decoupling}]
  See the proof of \citet[][Theorem~3.1.1]{de1999decoupling}.
\end{proof}

\begin{lemma}[Randomisation]%
  \label{lem:randomisation}
  Let $\phi: [0, \infty) \to [0, \infty)$ be convex and increasing. Then
  \begin{align*}
    \E \bigl\{
      \phi \bigl(
        \|V_n\|_\cF
      \bigr)
    \bigr\}
    &\leq
    \E \bigl\{
      \phi \bigl(
        4
        \|Y_n\|_\cF
      \bigr)
    \bigr\}.
  \end{align*}
\end{lemma}

\begin{proof}[Proof of Lemma~\ref{lem:randomisation}]
  We apply twice the randomisation lemma for independent processes from
  \citet[Lemma~2.3.6]{van1996weak},
  noting that the outer expectations may be replaced with
  expectations by the separability condition on $\cF$; this also
  means that we do not need to insist that the random variables are
  coordinate projections on a product space. Hence
  \begin{align*}
    \E \bigl\{
      \phi \bigl(
        \|V_n\|_\cF
      \bigr)
    \bigr\}
    &=
    \E \biggl[
      \E \biggl\{
        \phi \biggl(
          \sup_{f \in \cF}
          \biggl|
          \frac{2}{n(n-1)}
          \sum_{(i,j) \in \mathcal{I}}
          f(X_i, X_j')
          \biggr|
        \biggr)
        \biggm|
        X_1', \ldots, X_n'
      \biggr\}
    \biggr] \\
    &\leq
    \E \biggl\{
      \phi \biggl(
        2 \sup_{f \in \cF}
        \biggl|
        \frac{2}{n(n-1)}
        \sum_{(i,j) \in \mathcal{I}}
        \varepsilon_i
        f(X_i, X_j')
        \biggr|
      \biggr)
    \biggr\} \\
    &=
    \E \biggl[
      \E \biggl\{
        \phi \biggl(
          2 \sup_{f \in \cF}
          \biggl|
          \frac{2}{n(n-1)}
          \sum_{(i,j) \in \mathcal{I}}
          \varepsilon_i
          f(X_i, X_j')
          \biggr|
        \biggr)
        \biggm|\varepsilon_1,\ldots,\varepsilon_n,X_1, \ldots, X_n
      \biggr\}
    \biggr] \\
    &\leq
    \E \bigl\{
      \phi \bigl(
        4
        \|Y_n\|_\cF
      \bigr)
    \bigr\}.
    \qedhere
  \end{align*}
\end{proof}

\begin{corollary}[Decoupling and randomisation]
  \label{cor:decoupling}
  We have
  \[
    \E \bigl\{
      \phi \bigl(
        \|U_n\|_\cF
      \bigr)
    \bigr\}
    \leq
    \E \bigl\{
      \phi \bigl(
        128 \|W_n\|_\cF
      \bigr)
    \bigr\}.
  \]
\end{corollary}

\begin{proof}[Proof of Corollary~\ref{cor:decoupling}]
  Let
  \[
    \mathcal{G} := \bigl\{g:\bigl(\{0,1\} \times \mathcal{X}\bigr)
      \times \bigl(\{0,1\} \times \mathcal{X}\bigr) \rightarrow
      \mathbb{R}: g\bigl((\varepsilon,x),(\varepsilon',x')\bigr) =
    128\varepsilon \varepsilon'f(x,x') \text{ with } f \in \mathcal{F}\bigr\}.
  \]
  Then by Lemmas~\ref{lem:decoupling} and~\ref{lem:randomisation},
  \begin{align*}
    \E \bigl\{
      \phi \bigl(
        \|U_n\|_\cF
      \bigr)
    \bigr\} &\leq \E \bigl\{
      \phi \bigl(8
        \|V_n\|_\cF
      \bigr)
    \bigr\} \leq \E \bigl\{
      \phi \bigl(
        32\|Y_n\|_\cF
      \bigr)
    \bigr\} =
    \mathbb{E}\biggl\{\phi\biggl(\frac{1}{4}\|V_n\|_\mathcal{G}\biggr)\biggr\}
    \\
    &\leq \mathbb{E}\bigl\{\phi\bigl(\|U_n\|_\mathcal{G}\bigr)\bigr\}
    \leq
    \E \bigl\{
      \phi \bigl(
        128 \|W_n\|_\cF
      \bigr)
    \bigr\}.
    \qedhere
  \end{align*}
\end{proof}

\begin{lemma}[Rademacher chaos moment equivalence]%
  \label{lem:moment_equivalence}
  Let $\bX = (X_1, \ldots, X_n)$.
  For $k \in [1,\infty)$ and $f \in \cF$, define
  $\|W_n(f)\|_{k \mid \bX} \vcentcolon=
  \E \bigl\{ |W_n(f)|^k \mid \bX \bigr\}^{1/k}$. Then
  \begin{align*}
    \| W_n(f) \|_{k \mid \bX}
    &\leq
    \max(1, k - 1) \cdot \| W_n(f) \|_{2 \mid \bX}.
  \end{align*}
\end{lemma}

\begin{proof}[Proof of Lemma~\ref{lem:moment_equivalence}]
  For $k \in [1,2]$, the result follows from H{\"o}lder's inequality.
  For $k > 2$, we can apply \citet[][Theorem~3.2.2]{de1999decoupling}
  conditionally on $\bX$
  with $d = 2$, $p = 2$ and $q = k$ to obtain the result.
\end{proof}

Define $\psi:[0,\infty) \rightarrow [0,\infty)$ by $\psi(x) := e^x -
1$, and for a random variable $A$, the Orlicz norms
$\|A\|_\psi := \inf\bigl\{c > 0 : \E \bigl( \psi(|A| / c)\bigr) \leq 1\bigr\}$
and
$\|A\|_{\psi \mid \bX} := \essinf\bigl\{c(\bX) \text{ measurable}: \E
  \bigl[ \psi\bigl(\frac{|A|}{c(\bX)}\bigr) \mid \bX \bigr] \leq 1
\text{ almost surely}\bigr\}$, where
$\essinf$ denotes the essential infimum.

\begin{lemma}[Exponential bound for Rademacher chaoses]%
  \label{lem:exponential_chaos}
  For every $f \in \mathcal{F}$,
  \begin{align*}
    \big\| W_n(f) \big\|_{\psi \mid \bX}
    \leq
    3 \|W_n(f)\|_{2 \mid \bX}.
  \end{align*}
\end{lemma}

\begin{proof}[Proof of Lemma~\ref{lem:exponential_chaos}]
  By Stirling's formula, $k! > \sqrt{2 \pi k} (k/e)^k$ for $k \in \mathbb{N}$,
  so by the monotone convergence theorem
  and Lemma~\ref{lem:moment_equivalence}, with $a:=3$
  and noting that $(1-1/k)^k \leq 1/e$,
  \begin{align*}
    &\E \biggl\{
      \psi \biggl(
        \frac{|W_n(f)|}
        {a \|W_n(f)\|_{2 \mid \bX} }
      \biggr)
      \biggm| \bX
    \bigg\}
    =
    \sum_{k=1}^\infty
    \frac{1}{k!}
    \frac{\|W_n(f)\|_{k \mid \bX}^k}
    {a^k \|W_n(f)\|_{2 \mid \bX}^k} \\
    &\quad\leq
    \frac{1}{a}
    + \frac{1}{2 a^2}
    + \frac{4}{3 a^3}
    + \sum_{k=4}^\infty
    \frac{(k-1)^k}{a^k k!}
    \leq
    \frac{1}{a}
    + \frac{1}{2 a^2}
    + \frac{4}{3 a^3}
    + \sum_{k=4}^\infty
    \frac{1}{\sqrt{2 \pi k}}
    \left( \frac{e}{a} \right)^k
    (1-1/k)^k \\
    &\quad\leq
    \frac{1}{a}
    + \frac{1}{2 a^2}
    + \frac{4}{3 a^3}
    + \frac{1}{e \sqrt{8 \pi}}
    \sum_{k=4}^\infty
    \left( \frac{e}{a} \right)^k
    \leq
    \frac{1}{a}
    + \frac{1}{2 a^2}
    + \frac{4}{3 a^3}
    + \frac{1}{e \sqrt{8 \pi}}
    \frac{e^4}{a^3(a-e)}
    \leq 1.
    \qedhere
  \end{align*}
\end{proof}

\begin{definition}[Covering numbers]%
  \label{def:covering}
  Let $(\cT, d)$ be a non-empty pseudometric space.
  For $\varepsilon > 0$,
  we say that a non-empty finite set $\cC \subseteq \cT$
  is an \emph{$\varepsilon$-cover} of $(\cT, d)$ if
  $\sup_{t \in \cT} \min_{t' \in \cC} d(t, t') \leq \varepsilon$.
  We write $N(\varepsilon, \cT, d)$ for the minimal cardinality
  of such a cover, and set
  $N(\varepsilon, \cT, d) = \infty$ if no such cover exists.
  We say that $(\cT, d)$ is \emph{totally bounded} if
  $N(\varepsilon, \cT, d) < \infty$ for each $\varepsilon > 0$.

\end{definition}

\begin{lemma}[Chaining with Orlicz norms]%
  \label{lem:orlicz_chaining}
  Let $(\cT, d)$ be a non-empty, totally bounded pseudometric space with
  diameter $D \in (0,\infty)$
  and let $(X_t)_{t \in \cT}$ be a stochastic process.
  Suppose that $(X_t)_{t \in \cT}$ is separable in the sense that there exists
  a countable set $\mathcal T' \subseteq \mathcal T$ such that for all
  $t \in \cT$, there exists a sequence $(t_n)$ in $\cT'$ with
  $t_n \to t$ and $X_{t_n} \to X_t$ almost surely.
  Assume that $\|X_s - X_t\|_\psi \leq d(s, t)$ for all $s, t \in \cT$.
  Then for any $t_0 \in \cT$,
  \begin{align*}
    \Big\|
    \sup_{t \in \cT} |X_t|
    \Big\|_\psi
    \leq
    \| X_{t_0} \|_\psi
    + 8 e \int_0^{D/4}
    \log \bigl(1 + N(\varepsilon, \cT, d)\bigr)
    \diffi \varepsilon,
  \end{align*}
  where $N(\varepsilon, \cT, d)$ denotes the $\varepsilon$-covering
  number of $\cT$ with respect to $d$.
\end{lemma}

\begin{proof}[Proof of Lemma~\ref{lem:orlicz_chaining}]
  Let $N \geq 2$.
  Following the proof of \citet[][Lemma~2.2.2]{van1996weak}
  with $c=1$ and $y \geq 1$,
  for random variables $Z_1, \ldots, Z_N$,
  we have
  \begin{align*}
    \E \biggl\{
      \psi \biggl(
        \frac{\max_{j \in [N]} |Z_j|}
        {y \max_{j \in [N]} \|Z_j\|_\psi}
      \biggr)
    \biggr\}
    \leq \frac{N}{\psi(y)} + \psi(1).
  \end{align*}
  Now if $Z$ is a non-negative random variable, $a > 0$ and $b \geq 1$, then
  $\psi\bigl(Z/(ab)\bigr) \leq \psi(Z/a)/b$
  by convexity of $\psi$, so
  $\|Z\|_\psi \leq
  \inf_{a > 0} a\max\bigl\{\E \bigl(\psi(Z/a) \bigr),1\bigr\}$.
  Hence, with $y = \psi^{-1}(N) = \log(1 + N) \geq 1$,
  \begin{align}
    \label{eq:maximal_inequality}
    \Bigl\| \max_{j \in [N]} |Z_j| \Bigr\|_\psi &\leq y \max_{j \in
    [N]} \|Z_j\|_\psi \max\biggl[
      \E \biggl\{
        \psi \biggl(
          \frac{\max_{j \in [N]} |Z_j|}
          {y \max_{j \in [N]} \|Z_j\|_\psi}
      \biggr) \biggr\}
    , 1 \biggr] \nonumber \\
    &\leq e \log (1 + N)
    \max_{j \in [N]} \|Z_j\|_\psi.
  \end{align}
  Let $\delta > 0$ and
  suppose that $t_1, \ldots, t_N, s_1, \ldots, s_N \in \cT$
  satisfy $\max_{j \in [N]} d(s_j, t_j) \leq \delta$.
  By~\eqref{eq:maximal_inequality},
  \begin{align*}
    \Bigl\|
    \max_{j \in [N]}
    \big|X_{t_j} - X_{s_j}\big|
    \Bigr\|_\psi
    &\leq
    e \delta \log (1 + N).
  \end{align*}
  Let $\tilde \cT$ be
  a finite subset of $\cT$ that has
  strictly positive diameter and contains $t_0$.
  Set $\tilde\cT_0 = \{t_0\}$, which is a $D$-cover of
  $(\tilde\cT, d)$,
  and for $k \in \mathbb{N}$ let
  $\tilde\cT_k$ be a $(2^{-k}D)$-cover of $(\tilde\cT, d)$
  with cardinality
  $N_k \vcentcolon= N\big(2^{-k}D, \tilde\cT, d\big)$.
  Now, for $k \in \mathbb{N}_0$, define $\pi_k:\tilde\cT \rightarrow
  \tilde\cT_k$
  by\footnote{Formally, we order
    $\tilde \cT_k$
    and in the definition of $\pi_k$ we take the
  smallest element of the argmin set.}
  $\pi_k(t) \vcentcolon= \argmin_{s \in \tilde \cT_k} d(s,t)$,
  so that $d\bigl(t, \pi_k(t)\bigr) \leq 2^{-k}D$
  for all $t \in \tilde\cT$.
  Choose $K \in \mathbb N$ sufficiently large that
  $2^{-K}D < \min\{d(s,t): s, t \in \tilde \cT, d(s,t) > 0\}$.
  Observing that if $d(s, t) = 0$ then $\|X_s - X_t\|_\psi = 0$, so
  $X_s = X_t$ almost surely, and from the fact that $\pi_K$ is
  surjective, we have
  \begin{align*}
    \max_{t \in \tilde \cT} |X_t - X_{t_0}|
    &= \max_{t \in \tilde\cT_K} |X_t - X_{t_0}|
    = \max_{t \in \tilde\cT} |X_{\pi_K(t)} - X_{t_0}|
  \end{align*}
  almost surely.
  Noting that
  $d\bigl(\pi_{k+1}(t), \pi_k \circ \pi_{k+1}(t)\bigr)
  \leq 2^{-k} D$, we have
  \begin{align*}
    \Bigl\|
    \sup_{t \in \tilde \cT}
    \big|
    X_{\pi_{k+1}(t)} - X_{\pi_k \circ \pi_{k+1}(t)}
    \big|
    \Bigr\|_\psi
    &\leq
    e \cdot 2^{-k} D
    \log (1 + N_{k+1}).
  \end{align*}
  Since
  $N(\varepsilon, \tilde\cT, d) \leq N(\varepsilon/2, \cT, d)$,
  \begin{align*}
    \Bigl\|
    \max_{t \in \tilde\cT}
    \big|X_t - X_{t_0}\big|
    \Bigr\|_\psi
    &\leq
    \sum_{k=0}^{K-1}
    \Bigl\|
    \max_{t \in \tilde\cT}
    \big|
    X_{\pi_{k+1} \circ \cdots \circ \pi_K(t)}
    - X_{\pi_k \circ \pi_{k+1} \circ \cdots \circ \pi_K(t)}
    \big|
    \Bigr\|_\psi \\
    &\leq
    e \sum_{k=0}^{K-1}
    2^{-k}D
    \log (1 + N_{k+1})
    \leq
    4 e \sum_{k=0}^{\infty}
    \int_{2^{-(k+2)}D}^{2^{-(k+1)}D}
    \log \big(1 + N(\varepsilon, \tilde\cT, d) \big)
    \diffi{\varepsilon} \\
    &\leq
    4 e
    \int_{0}^{D / 2}
    \log \big(1 + N(\varepsilon, \tilde\cT, d)\big)
    \diffi{\varepsilon}
    \leq
    4 e
    \int_{0}^{D / 2}
    \log \big(1 + N(\varepsilon/2, \cT, d)\big)
    \diffi{\varepsilon} \\
    &=
    8 e \int_{0}^{D / 4}
    \log \big(1 + N(\varepsilon, \cT, d)\big)
    \diffi{\varepsilon}.
  \end{align*}
  Taking an increasing sequence of finite sets whose union is the
  countable set $\cT'$, and since the above bound does not depend on
  the set $\tilde\cT$,
  the monotone convergence theorem yields
  \begin{align*}
    \Bigl\|
    \sup_{t \in \cT'}
    \big|X_t\big|
    \Bigr\|_\psi
    &\leq
    \bigl\|X_{t_0}\bigr\|_\psi
    + \Bigl\|
    \sup_{t \in \cT'}
    \big|X_t - X_{t_0}\big|
    \Bigr\|_\psi
    \leq
    \bigl\|X_{t_0}\bigr\|_\psi
    + 8 e \int_{0}^{D / 4}
    \log \big(1 + N(\varepsilon, \cT, d)\big)
    \diffi{\varepsilon}.
  \end{align*}
  Finally the result follows by separability, since
  $\sup_{t \in \cT} |X_t| = \sup_{t \in \cT'} |X_t|$ almost surely.
\end{proof}

\begin{lemma}[Orlicz norm supremum bound for symmetrised $U$-statistics]%
  \label{lem:orlicz_symmetrised}
  Let $\mathcal{Q}$ denote the set of finitely supported
  probability measures on $\mathcal{X} \times \mathcal{X}$,
  and for a measurable $f : \cX \times \cX \to \R$ and $Q \in
  \mathcal{Q}$, define
  $\|f\|_{Q, 2}^2 \vcentcolon= \E_{(Y_1, Y_2) \sim Q} \bigl\{ f(Y_1,
  Y_2)^2 \bigr\}$
  and $\|f\|_{2}^2 \vcentcolon= \E \bigl\{ f(X_1, X_2)^2 \bigr\}$.
  Let $\|F\|_\infty \vcentcolon= \sup_{x_1, x_2 \in \cX} |F(x_1, x_2)|$ and
  \begin{align*}
    J &\vcentcolon=
    \sup_{Q \in \mathcal Q} \int_0^1
    \log \bigl(1 + N(\varepsilon \|F\|_{Q,2}, \cF, \|\cdot\|_{Q,2}) \bigr)
    \diffi \varepsilon,
  \end{align*}
  with the covering number taken
  with respect to the pseudo-metric
  induced by $\|\cdot\|_{Q,2}$.
  Then
  \begin{align*}
    \bigl\|
    \|n W_n\|_\cF
    \bigr\|_\psi
    &\leq
    140 \|F\|_\infty J.
  \end{align*}
\end{lemma}

\begin{proof}[Proof of Lemma~\ref{lem:orlicz_symmetrised}]
  The result is clear if $\|F\|_\infty = 0$, so assume that
  $\|F\|_\infty > 0$. Note that $W_n$ is a Rademacher chaos of degree two
  conditional on $\bX$, and for $f \in \cF$ define
  $\|f\|_{n,2}^2 \vcentcolon= \frac{2}{n(n-1)} \sum_{(i, j) \in \mathcal I}
  f(X_i, X_j)^2$. If $(i, j), (i', j') \in \mathcal I$, then
  $\E (\varepsilon_i \varepsilon_j \varepsilon_{i'} \varepsilon_{j'}) = 0$
  unless $(i, j) = (i', j')$. Hence, by Lemma~\ref{lem:exponential_chaos},
  \begin{align}
    \nonumber
    \big\| W_n(f) \big\|_{\psi \mid \bX}^2
    &\leq
    9 \|W_n(f)\|_{2 \mid \bX}^2
    =
    \frac{36}{n^2(n-1)^2} \,
    \E \biggl\{
      \biggl(
        \sum_{(i,j) \in \mathcal{I}}
        \varepsilon_i \varepsilon_j
        f(X_i, X_j)
      \biggr)^2
      \biggm| \bX
    \biggr\} \\
    \nonumber
    &=
    \frac{36}{n^2(n-1)^2} \,
    \sum_{(i,j) \in \mathcal{I}}
    \sum_{(i',j') \in \mathcal{I}}
    f(X_i, X_j)
    f(X_{i'}, X_{j'})
    \E ( \varepsilon_i \varepsilon_j
    \varepsilon_{i'} \varepsilon_{j'} ) \\
    \label{eq:Wn_Orlicz}
    &=
    \frac{36}{n^2(n-1)^2}
    \sum_{(i,j) \in \mathcal{I}}
    f(X_i, X_j)^2
    =
    \frac{18}{n(n-1)}
    \|f\|_{n,2}^2
    \leq
    \frac{36}{n^2}
    \|f\|_{n,2}^2.
  \end{align}
  For $f, f' \in \cF$, by linearity of $W_n$
  and~\eqref{eq:Wn_Orlicz}, we have
  $\big\| n W_n(f) - n W_n(f') \big\|_{\psi \mid \bX}^2
  \leq 36\|f - f'\|_{n,2}^2$.
  The diameter of $\cF$ under $\|\cdot\|_{n,2}$ is at most
  $2 \|F\|_{n,2}$, so by Lemma~\ref{lem:orlicz_chaining}
  with $d(f, f') = 6 \|f - f'\|_{n, 2}$,
  \begin{align*}
    \bigl\|
    \|n W_n\|_\cF
    \bigr\|_{\psi \mid \bX}
    &\leq
    \|n W_n(f_0) \|_{\psi \mid \bX}
    + 8 e \int_0^{3 \|F\|_{n, 2}}
    \log \bigl(1 + N(\varepsilon, \cF, 6 \|\cdot\|_{n,2})\bigr)
    \diffi \varepsilon \\
    &\leq
    6 \|f_0\|_{n, 2}
    + 48 e \|F\|_{n, 2} \int_0^{1/2}
    \log \bigl(1 + N(\varepsilon \|F\|_{n, 2}, \cF, \|\cdot\|_{n,2})\bigr)
    \diffi \varepsilon \\
    &\leq
    6 \|F\|_{n, 2} + 48 e \|F\|_{n, 2} J
    \leq
    \|F\|_{n, 2} J\Bigl(\frac{6}{J} + 48 e\Bigr)
    \leq
    140 \|F\|_{\infty} J,
  \end{align*}
  where the last inequality follows as $J \geq \log 2$
  and $6 / \log 2 + 48 e \leq 140$. Therefore,
  \begin{align*}
    \E \biggl\{
      \psi \biggl(
        \frac{\|n W_n\|_\cF}{ 140 \|F\|_\infty J}
      \biggr)
    \biggr\}
    &=
    \E \biggl[
      \E \biggl\{
        \psi \biggl(
          \frac{\|n W_n\|_\cF}{ 140 \|F\|_\infty J}
        \biggr)
        \Bigm|
        \bX
      \biggr\}
    \biggr]
    \leq 1,
  \end{align*}
  as required.
\end{proof}

\begin{lemma}[Orlicz norm supremum bound for
  $U$-statistics]\label{lem:orliczUStat}
  We have
  \begin{align*}
    \bigl\| \|n U_n\|_\cF \bigr\|_\psi
    \leq 17920 \|F\|_\infty J.
  \end{align*}
  Therefore, for $s \geq 1$,
  \begin{align*}
    \P \Big(
      \|n U_n\|_\cF \geq
      33104 s \|F\|_\infty J
    \Big)
    &\leq
    e^{-s}.
  \end{align*}
\end{lemma}

\begin{proof}[Proof of Lemma~\ref{lem:orliczUStat}]
  By Corollary~\ref{cor:decoupling},
  we have
  \begin{align*}
    \E \biggl\{
      \psi \biggl(
        \frac{\|U_n\|_\cF}{128 \bigl\| \|W_n\|_\cF \bigr\|_\psi}
      \biggr)
    \biggr\}
    \leq
    \E \biggl\{
      \psi \biggl(
        \frac{\|W_n\|_\cF}{\bigl\| \|W_n\|_\cF \bigr\|_\psi}
      \biggr)
    \biggr\}
    \leq 1,
  \end{align*}
  so that by
  Lemma~\ref{lem:orlicz_symmetrised}, %
  \begin{align*}
    \bigl\| \|n U_n\|_\cF \bigr\|_\psi \leq 128 \bigl\| \|nW_n\|_\cF
    \bigr\|_\psi
    \leq 128 \times 140 \|F\|_\infty J
    = 17920 \|F\|_\infty J.
  \end{align*}
  Now, for $s \geq 1$ and a random variable $Z$, we have by Markov's
  inequality that
  \[
    \P \bigl( |Z| \geq s \|Z\|_\psi \bigr) =
    \P\biggl\{\psi\biggl(\frac{|Z|}{\|Z\|_\psi}\biggr) \geq
    \psi(s)\biggr\} \leq \frac{1}{\psi(s)}
    = \frac{1}{e^s-1}
    \leq \frac{1}{(e-1)^s}
    = e^{-s \log(e-1)},
  \]
  so
  \begin{align*}
    \P \Big(
      \|n U_n\|_\cF \geq
      33104 s \|F\|_\infty J
    \Big) &\leq \P \bigg(
      \|n U_n\|_\cF \geq
      \frac{17920 s}{\log(e-1)} \|F\|_\infty J
    \bigg) \\
    &\leq \P \bigg(
      \|n U_n\|_\cF \geq
      \frac{s}{\log(e-1)} \bigl\|\|n U_n\|_\cF\bigr\|_\psi
    \bigg) \leq
    e^{-s}.
    \qedhere
  \end{align*}
\end{proof}

\bibliographystyle{apalike}
\bibliography{refs}

\end{document}